%% file: notes.tex
\documentclass[10pt,a4]{amsart}

\usepackage{amssymb}
\usepackage{amsmath}
\usepackage{verbatim}
\usepackage{enumerate}
\usepackage{color}
\usepackage{graphicx}
\usepackage{xcolor}
\usepackage{amscd}
\usepackage{amsthm}
\usepackage{mathabx,epsfig}

\newcommand{\mr}{\mathrm}
\newcommand{\mc}{\mathcal}
\newcommand{\dd}{\partial}
\newcommand{\HH}{\mathcal{H}}
\newcommand{\ra}{\rightarrow}
\newcommand{\xra}{\xrightarrow}
\newcommand{\hra}{\hookrightarrow}
\newcommand{\CC}{\mathbb{C}}
\newcommand{\DD}{\mathcal{D}}
\newcommand{\FF}{\mathcal{F}}
\newcommand{\LL}{\mathcal{L}}
\newcommand{\bt}{\bullet}
\newcommand{\g}{\mathfrak{g}}
\newcommand{\PP}{\mathcal{P}}
\newcommand{\tr}{\mathrm{tr}\,}
\newcommand{\ZZ}{\mathbb{Z}}
\newcommand{\RR}{\mathbb{R}}
\newcommand{\Conn}{\mathrm{Conn}}
\newcommand{\GGauge}{\mathrm{Gauge}}
\newcommand{\FlatConn}{\mathrm{FlatConn}}
\newcommand{\MM}{\mathcal{M}}
\newcommand{\s}{\mathsf}
\newcommand{\til}{\widetilde}
\newcommand{\raslash}{\ra\!\!\!\!\!\!\!\slash\;\;}
\newcommand{\xraslash}[1]{\xrightarrow{#1}\!\!\!\!\!\!\!\!\!\slash\;\;\;}
\newcommand{\Sym}{\mathrm{Sym}}
\newcommand{\simra}{\stackrel{\sim}{\rightarrow}}
\newcommand{\match}{\mathrm{Matchings}}
\newcommand{\lan}{\left\langle}
\newcommand{\ran}{\right\rangle}
\newcommand{\pf}{\mathrm{pf}}
\newcommand{\proj}{\twoheadrightarrow}
\newcommand{\OO}{\mathcal{O}}
\newcommand{\even}{\mathrm{even}}
\newcommand{\odd}{\mathrm{odd}}
\newcommand{\NN}{\mathcal{N}}
\newcommand{\Ber}{\mathrm{Ber}}
\newcommand{\BER}{\mathrm{BER}}
\newcommand{\ola}{\overleftarrow}
\newcommand{\ora}{\overrightarrow}
\newcommand{\HDens}{\mathrm{Dens}^{\frac12}}
\newcommand{\BRST}{\mathrm{BRST}}
\newcommand{\BV}{\mathrm{BV}}
\newcommand{\gh}{\mathrm{gh}}
\def\acts{\mathrel{\reflectbox{$\righttoleftarrow$}}}
\newcommand{\diag}{\mathsf{diag}}

\theoremstyle{definition}
\newtheorem{remark}{Remark}[section]
\newtheorem{example}[remark]{Example}
\newtheorem{lemma}[remark]{Lemma}
\newtheorem{theorem}[remark]{Theorem}
\newtheorem{definition}[remark]{Definition}
\newtheorem{corollary}[remark]{Corollary}
\newtheorem{assumption}[remark]{Assumption}

\begin{document}
\title[ 
BV formalism and applications
]{Lectures on Batalin-Vilkovisky formalism and its applications in topological quantum field theory}

\begin{abstract}
Lecture notes for the course ``Batalin-Vilkovisky formalism and applications in topological quantum field theory''  given at the University of Notre Dame in the Fall 2016 for a mathematical audience. In these lectures we give a slow introduction to the perturbative path integral for gauge theories in Batalin-Vilkovisky formalism and the associated mathematical concepts.
\end{abstract}

\author{Pavel Mnev}
\address{
University of Notre Dame, Notre Dame, Indiana 46556, USA
}
\address{
St. Petersburg Department of V. A. Steklov Institute of Mathematics of the Russian Academy of Sciences, Fontanka 27, St. Petersburg, 191023 Russia}
\email{
pmnev @nd.edu
}

\thanks{The author acknowledges partial support of RFBR Grant No. 17-01-00283a.}

\date{\today}
\maketitle

\setcounter{tocdepth}{3} 
\tableofcontents

\allowdisplaybreaks
\section*{Preface}
The Batalin-Vilkovisky (``BV'') formalism arose in the end of 1970's/beginning of 1980's as a tool of mathematical physics designed to define the path integral for gauge theories. Since then the construction turned out to be very useful for applications in algebraic topology -- invariants of 3-manifolds and knots, Chas-Sullivan string topology, operations on rational cohomology of CW complexes. Another spectacular application of the BV formalism is Kontsevich's deformation quantization of Poisson manifolds. The general direction these applications go is via applying the BV formalism to define the path integral for particular models of topological field theory and then finding an appropriate interpretation for the value of the path integral (and proving the desired properties).

These  lectures 
were given at the University of Notre Dame in the Fall 2016 for graduate mathematical audience; a previous iteration of this course was given in the Fall 2014 in the Max Planck Institute for Mathematics, Bonn, jointly with the University of Bonn. The aim of the course was to 
give an introduction, oriented towards mathematical audience and not requiring any prior physics background, to the perturbative path integral for gauge theories (in particular, topological field theories) 
in Batalin-Vilkovisky formalism, 
and some of its 
applications.
To elucidate the picture, we were mostly focusing on finite-dimensional models for gauge systems and path integrals, while giving comments on what has to be amended in the infinite-dimensional case relevant to local field theory. Our motivating examples included the Alexandrov-Kontsevich-Schwarz-Zaboronsky sigma models; the perturbative expansion for Chern-Simons invariants of $3$-manifolds, given in terms of integrals over configurations of points on the manifold; the $BF$ theory on cellular decompositions of manifolds. 

\subsection*{Acknowledgements}
These lectures were strongly influenced by numerous inspiring discussions with Alberto S. Cattaneo, Andrei Losev and Nicolai Reshetikhin.
I am also most grateful for questions and enthusiasm to the audience of the course when it was given in the University of Notre Dame in the Fall 2016 and to the audience of its predecessor in Max Planck Institute for Mathematics in Bonn and University of Bonn in the Fall 2014.

\marginpar{\LARGE{Lecture 1, 08/24/2016.}}
\section{Introduction/motivation}

Idea of locality (in the interpretation of Atiyah-Segal): a quantum field theory (QFT) assigns some values (``partition functions'') to manifolds. It can be evaluated on manifolds and satisfies a gluing/cutting property. So, a manifold can be chopped into simple (small) pieces, then the QFT can be evaluated on those pieces and then assembled to the value of the QFT on the entire manifold.\footnote{An alternative way to describe locality in quantum field theory is provided by the language of factorization algebras \cite{CostelloGwilliam}. There, instead of cobordisms, one evaluates the theory on open subsets of the spacetime and instead of cutting an $n$-manifold into submanifolds with boundary, one considers open covers (subject to certain condition -- so-called Weiss covers).}
\subsection{Atiyah's axioms of topological quantum field theory}
An $n$-dimensional topological quantum field theory (TQFT) is the following set of data.
\begin{itemize}
	\item To a \textbf{closed  $(n-1)$-dimensional manifold} $\Sigma$, the TQFT assiociates a vector space $\mathcal{H}_\Sigma$ over $\mathbb{C}$ -- the ``space of states''.
	\item To an \textbf{$n$-manifold $M$ with  boundary} split into in- and out-parts, $\partial M= \bar{\Sigma}_\mr{in}\sqcup \Sigma_\mr{out}$ (bar refers to reversing the orientation on the in-boundary), the TQFT associates a $\mathbb{C}$-linear map $Z_M:\; \HH_{\Sigma_\mr{in}}\ra \HH_{\Sigma_\mr{out}}$ -- the ``partition function''.\footnote{Another possible name for $Z_M$ is the ``evolution operator''.}
	$$\vcenter{\hbox{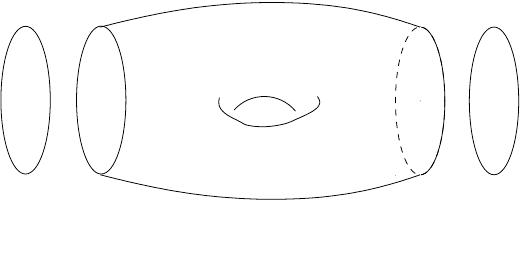}}$$
	We call such $M$ a \textit{cobordism} between $\Sigma_\mr{in}$ and $\Sigma_\mr{out}$, and we denote $$\Sigma_\mr{in}\xra{M}\Sigma_\mr{out}$$
	\item \textbf{Diffeomorphisms of closed $(n-1)$-manifolds} act on spaces of states: to $\phi: \Sigma\ra\Sigma'$ a diffeomorphism, the TQFT associates an isomorphism $\rho(\phi): \HH_\Sigma\ra \HH_{\Sigma'}$ (in the way compatible with composition of diffeomorphisms). For $\phi$ orientation-preserving, $\rho(\phi)$ is $\CC$-linear; for $\phi$ orientation-reversing, $\rho(\phi)$ is $\CC$-anti-linear.
\end{itemize}

This set of data should satisfy the following axioms:
\begin{itemize}
	\item \textbf{Multiplicativity}: disjoint unions are mapped to tensor products. Explicitly,
	$$\HH_{\Sigma\sqcup \Sigma'}=\HH_\Sigma\otimes \HH_{\Sigma'},\qquad Z_{M\sqcup M'}=Z_M\otimes Z_{M'}$$
	\item \textbf{Gluing}: given two cobordisms $\Sigma_1\xra{M'} \Sigma_2$ and $\Sigma_2\xra{M''} \Sigma_3$, with out-boundary of the first one coinciding with the in-boundary of the second one,
	$$\vcenter{\hbox{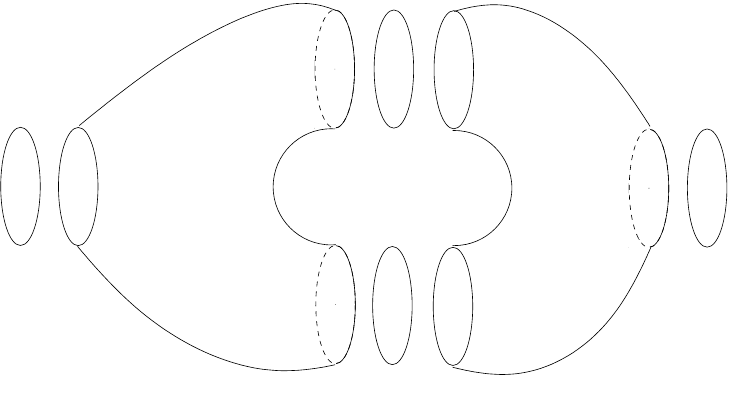}}$$ 
	we can \textit{glue} (or ``concatenate'') them over $\Sigma_2$ to a new cobordism $M:=M'\cup_{\Sigma_2} M''$, going as $\Sigma_1\xra{M} \Sigma_3$.
	$$\vcenter{\hbox{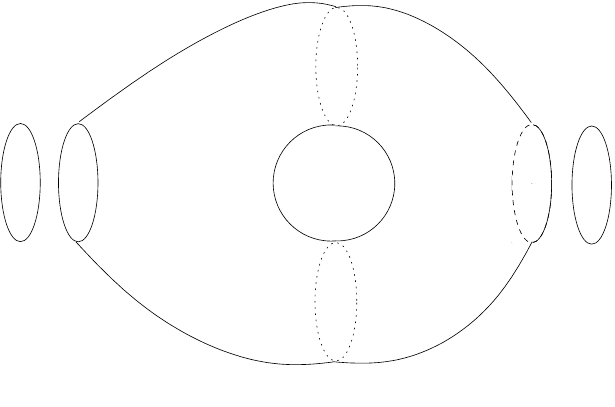}}$$ 
	Then the partition function for $M$ is the \textit{composition} of partition functions for $M'$ and $M''$ as linear maps:
	$$\boxed{Z_M=Z_{M''}\circ Z_{M'}}:\quad \HH_{\Sigma_1}\ra \HH_{\Sigma_3}$$
	\item \textbf{Normalization}: 
	\begin{itemize}
	\item For $\varnothing$ the empty $(n-1)$-manifold, $$\HH_\varnothing=\CC$$
	\item For $\Sigma$ a closed $(n-1)$-manifold, the partion function for the \textit{cylinder} $\Sigma\xra{\Sigma\times [0,1]} \Sigma$ is the identity on $\HH_\Sigma$.
	\end{itemize}
	\item For $\phi: M\ra M'$ a diffeomorphism between two cobordisms, denote $\phi|_\mr{in}$, $\phi|_\mr{out}$ the restrictions of $\phi$ to the in- and out-boundary. We have a commutative diagram
	$$\begin{CD}
	\HH_{\Sigma_\mr{in}} @>Z_M>> \HH_{\Sigma_\mr{out}} \\
	@V\rho(\phi|_\mr{in})VV @VV\rho(\phi|_\mr{out})V \\
	 \HH_{\Sigma'_\mr{in}} @>>Z_{M'}> \HH_{\Sigma'_\mr{out}} 
	\end{CD}$$
\end{itemize}

\begin{remark}
	Atiyah's TQFT is a functor of symmetric monoidal categories, $\mr{Cob}_n\ra \mr{Vect}_\CC$,
	where the structure is as follows:
	
	\begin{tabular}{l||l|l}
		& $\mr{Cob}_n$ & $\mr{Vect}_\CC$ \\
		\hline
		objects & closed $(n-1)$-manifolds & vector spaces$/\CC$ \\
		morphisms & cobordisms  $\Sigma_\mr{in}\xra{M} \Sigma_\mr{out}$& linear maps \\
		composition & gluing & composition of maps \\
		identity morphism & cylinder $\Sigma\xra{\Sigma\times [0,1]}\Sigma$ & identity map $\mr{id}: V\ra V$ \\
		monoidal product & disjoint union $\sqcup$ & tensor product $\otimes$ \\
		monoidal unit & $\varnothing$ & $\CC$
		\end{tabular}
\end{remark}

\begin{remark}
	A \emph{closed} $n$-manifold $M$ can be viewed as a cobordism from $\varnothing$ to $\varnothing$, thus $Z_M: \CC\ra \CC$ is a multiplication by some number $z\in\CC$. By abuse of notations, we denote $Z_M:=z\in\CC$. Thus, with this convention, the partition function for a closed $n$-manifold is a complex number, invariant under diffeomorphisms and compatible with gluing-cutting. E.g., for $n=2$, we can cut any closed surface into disks and pairs of pants
	$$\vcenter{\hbox{\includegraphics[scale=1.5]{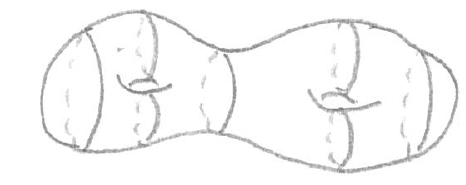}}}$$
	Thus, $Z$ for any surface can be calculated from the gluing axiom, provided that $Z$ is known for a disk and for a pair of pants. 
\end{remark}

\begin{remark}
	In Segal's approach to (not necessarily topological) quantum field theory, one allows manifolds to carry a local geometric structure (of the type depending on the particular QFT): Riemannian metric, conformal structure, complex structure, framing, local system,\ldots Atiyah's axioms above have to be modified slightly to accommodate for the geometric structure.
\end{remark}

\begin{example}[Quantum mechanics]\label{l1_exQM}
Consider the 1-dimensional Segal's QFT with geometric structure the Riemannian metric on 1-cobordisms. Objects are points with $+$ orientation, assigned a vector space $\HH$ and points with $-$ orientation, assigned the dual space $\HH^*$. Consider an interval of length $t>0$ (our partition functions depend on a metric on the interval considered modulo diffeomorphisms, thus only on the length), $I_t=[0,t]$. Denote $Z(t):=Z_{I_t} \in \mr{End}(\HH)$. By the gluing axiom (from considering the gluing $[0,t_1]\cup_{\{t_1\}} [t_1,t_1+t_2]=[0,t_1+t_2]$), we have the semi-group law $Z(t_1+t_2)=Z(t_2)\circ Z(t_1)$. It implies in turn that 
\begin{equation}\label{l1eq1}
Z(t)=Z(\frac{t}{N})^N
\end{equation} for $N$ an arbitrarily large integer. Assume that for $\tau$ small, we have $Z(\tau)=\mr{id}-\frac{i}{\hbar}\hat{H}\cdot \tau+O(\tau^2)$, for $\hat{H}\in \mr{End}(\HH)$ some operator. Then (\ref{l1eq1}) implies that
$$\boxed{Z(t)=\exp\left(-\frac{i}{\hbar}\hat{H}t\right)}$$
This system is the quantum mechanics, with $Z(t)$ the \textit{evolution operator} in time $t$ and $\hat{H}$ the Schr\"odinger operator (or \textit{quantum Hamiltonian}), describing the infinitesimal evolution of the system. 

E.g. the choice $\HH=L^2(X)$ for $X$ a Riemannian manifold and $\hat{H}=-\frac{\hbar^2}{2m}\Delta_X + U(x)\cdot$ would correspond to the quantum particle of mass $m$ moving on the manifold $X$ in the force field with potential $U$. In this case $Z(t): \psi(x)\mapsto \int_{X\ni y}dy\; Z(t;x,y)\psi(y) $ is the integral operator whose integral kernel $Z(t;x,y)$ is intepreted as the propagation amplitude of the particle from position $y$ to position $x$ in time $t$.
\end{example}

\subsection{The idea of path integral construction of quantum field theory}

\subsubsection{Classical field theory data}
We start by fixing the data of \textit{classical field theory} on an $n$-manifold:
\begin{itemize}
	\item A space of fields $F_M=\Gamma(M,\mathbb{F}_M)$ -- a space of sections of some sheaf $\mathbb{F}_M$ over $M$. Typical examples of $F_M$ are:
	\begin{itemize}
		\item $C^\infty(M)$
		\item Space of connections on a principal $G$-bundle $\mathcal{P}$ over $M$. (This exaple is typical for some of \textit{gauge theories} e.g. Chern-Simons theory, Yang-Mills theory,\ldots)
		\item Mapping space $\mr{Map}(M,N)$ with $N$ some fixed target manifold. This is typical for so-called \textit{sigma models}.
	\end{itemize}
	\item The \textit{action functional} $S_M:\; F_M\ra \mathbb{R}$ of form
	$$S_M(\phi)=\int_M L(\phi,\dd \phi,\dd^2 \phi,\ldots)$$
	where $L$ is \textit{the Lagrangian density} -- a density on $M$ depending on the value of the field $\phi\in F_M$ and its derivatives (up to fixed finite order) at the point of integration on $M$. Variational problem of extremization of $S$  (i.e. the critical point equation $\delta S=0$) leads to Euler-Lagrange PDE on $\phi$.
\end{itemize}

\begin{example}[Free massive scalar field]
	Let $(M,g)$ be a Riemannian manifold, we set $F_M=C^\infty(M)\,\ni \phi$ with the action
	$$S_M(\phi)=\int_M \left(\frac12 \langle d\phi , d\phi \rangle_{g^{-1}} + \frac{m^2}{2}\phi^2\right) d\mr{vol}$$
	Here $m\geq 0$ is a parameter of the theory -- the \textit{mass}; 
	$d\mr{vol}$ is the Riemannian volume element on $M$. The associated Euler-Lagrange equation on $\phi$ is: $(\Delta+m^2)\phi=0$.
\end{example}

\subsubsection{Idea of path integral quantization}
The idea of quantization is then to construct the partition function for $M$ a closed manifold as
\begin{equation}\label{l1PI}
Z_M(\hbar):=``\int_{F_M}\DD\phi\; e^{\frac{i}{\hbar}S_M(\phi)} \;\;"
\end{equation}
Here $\hbar$ is a parameter of the quantization (morally, $\hbar$ measures the ``distance to classical theory''); $\DD \phi$ is a symbol for a reference measure on the space $F_M$. Integral (\ref{l1PI}) is problematic to define directly as a measure-theoretic integral, however it can be defined as an asymptotic series in $\hbar\ra 0$, as we will discuss in a moment. So far, r.h.s. of (\ref{l1PI}) is a heuristic expression which is to be made mathematical sense of.

Consider $M$ with boundary $\Sigma$. Denote $B_\Sigma$ the set of boundary values of fields on $M$; we have a map of evaluation of the field at the boundary (or pullback by the inclusion $\Sigma\hra M$) $F_M\ra B_\Sigma$ sending $\phi\mapsto \phi|_\dd$. For the space of states on $\Sigma$, we set $\HH_\Sigma:=\mr{Fun}_\CC(B_\Sigma)$ -- complex-valued functions on $B_\Sigma$. For the partition function $Z_M$, we set
\begin{equation}
Z_M(\phi_\Sigma;\hbar):=\int_{\phi\in F_M\, \mr{s.t.}\, \phi|_\dd=\phi_\Sigma} \DD \phi\; e^{\frac{i}{\hbar}S_M(\phi)}
\end{equation}
This path integral gives us a function on $B_\Sigma\ni \phi_\Sigma$ and thus a vector in $Z_M(-;\hbar)\in\HH_\Sigma$.

\subsubsection{Heuristic argument for gluing}
Let a closed (for simplicity) $n$-manifold $M$ be cut by a codimension 1 submanifold $\Sigma$ into two $M'$ and $M''$, i.e. $M=M'\cup_\Sigma M''$. Then the integral (\ref{l1PI}) can be performed in steps:
\begin{enumerate}[(i)]
	\item Fix $\phi_\Sigma$ on $\Sigma$.
	\item Integrate over fields on $M'$ with boundary condition $\phi_\Sigma$ on $\Sigma$.
	\item Integrate over fields on $M''$ with boundary condition $\phi_\Sigma$ on $\Sigma$.
	\item Integrate over $\phi_\Sigma\in B_\Sigma$.
\end{enumerate}

This yields $$Z_M=\int_{B_\Sigma\ni \phi_\Sigma} \DD \phi_\Sigma\; Z_{M'}(\phi_\Sigma)\cdot Z_{M''}(\phi_\Sigma)$$
One can recognize in this formula the Atiyah-Segal gluing axiom: $M'$ and $M''$ yield two vectors in $\HH_\Sigma$ which are paired in $\HH_\Sigma$ to a number -- the partition function for the whole manifold.

\subsubsection{How to define path integrals?}
Let us first look at finite-dimensional \textit{oscillating} integrals: consider $X$ a compact manifold with $\mu$ a fixed volume form and $f\in C^\infty(X)$ a function. The asymptotics, as $\hbar\ra 0$, of the integral 
$$\int_X \mu \; e^{\frac{i}{\hbar}f(x)}$$
is given by the \textit{stationary phase formula}\footnote{See e.g. \cite{Etingof,Reshetikhin} }
\begin{equation*}
\int_X \mu \; e^{\frac{i}{\hbar}f(x)} \underset{\hbar\ra 0}{\sim}
\sum_{x_0\in \{\mr{crit.\, points\,of\,} f\}} e^{\frac{i}{\hbar}f(x_0)} |\det f''(x_0)|^{-\frac12} e^{\frac{\pi i}{4}\mr{sign}f''(x_0)}(2\pi\hbar)^{\frac{\dim X}{2}}
\end{equation*}
The rough idea here is that the rapid oscillations of the integrand cancel out except in the neighborhood  of critical points $x_0$ of $f$ (i.e. points with $df(x_0)=0$), which are the ``stationary phase points'' for the integrand -- points around which oscillations slow down.

This formula can be improved to accommodate corrections in powers of $\hbar$:
\begin{multline}\label{l1statphase}
\int_X \mu \; e^{\frac{i}{\hbar}f(x)} \underset{\hbar\ra 0}{\sim}
\sum_{x_0\in \{\mr{crit.\, points\,of\,} f\}} e^{\frac{i}{\hbar}f(x_0)} |\det f''(x_0)|^{-\frac12} e^{\frac{\pi i}{4}\mr{sign}f''(x_0)}(2\pi\hbar)^{\frac{\dim X}{2}}\cdot\\
\cdot \sum_{\Gamma} \hbar^{-\chi(\Gamma)}\Phi_\Gamma
\end{multline}
where $\Gamma$ ranges over graphs with vertices of valence $\geq 3$ (possibly disconnected, including $\Gamma=\varnothing$); $\chi(\Gamma)\leq 0$ is the Euler characteristic of the graph. Graphs $\Gamma$ are called the \textbf{Feynman diagrams}. Assume that $\Gamma$ has $E$ edges and $V$ vertices. We decorate all half-edges of $\Gamma$ with labels $i_1,\ldots,i_{2E}$ each of which can take values $1,2,\ldots,p:=\dim X$. The weight of the graph $\Gamma$,  $\Phi_\Gamma$, is defined as follows.
\begin{itemize}
	\item We assign to every edge $e$ consisting of half-edges $h_1,h_2$ the decoration $f''(x_0)^{-1}_{i_{h_1}i_{h_2}}$ -- the matrix element of the inverse Hessian given by the labels of the constituent half-edges.
	\item We assign to every vertex $v$ of valence $k$ with adjacent half-edges $h_1,\ldots,h_k$ the decoration $\dd_{i_{h_1}}\cdots \dd_{i_{h_k}}f(x_0)$ -- a $k$-th partial derivative of $f$ at the critical point.
	\item We take the product of all the decorations above and sum over all possible values of labels on the half-edges. $\Phi_\Gamma$ is this sum times the factor $\frac{{i}^{E+V}}{|\mr{Aut}(\Gamma)|}$ with $\mr{Aut}(\Gamma)$ the automorphism group of the graph.
\end{itemize}
I.e., we have
$$\Phi_\Gamma:=\frac{{i}^{E+V}}{|\mr{Aut}(\Gamma)|}\cdot \sum_{i_1,\ldots,i_{2E}\in \{1,\ldots,p\}} \prod_{\mr{edges\,}e=(h_1 h_2)}  f''(x_0)^{-1}_{i_{h_1}i_{h_2}} \cdot\prod_{\mr{vertices\,} v}\dd_{i_{h_1}}\cdots \dd_{i_{h_{\mr{val}(v)}}}f(x_0)$$

\begin{example}
	Consider the ``theta graph'' $$\vcenter{\hbox{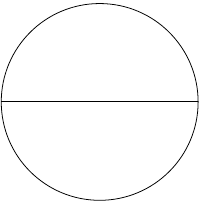}}$$
	(Note that its Euler characterestic is $-1$, hence it enters in (\ref{l1statphase}) in the order $\hbar^1$.) For its weight, we obtain
	$$\Phi\left(\;\vcenter{\hbox{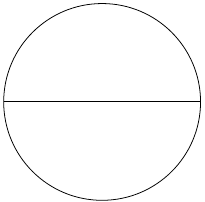}}\;\;\right)=\frac{i^{3+2}}{12}\cdot\sum_{i,j,k,l,m,n\in \{1,\ldots,p\}} f''(x_0)^{-1}_{il}f''(x_0)^{-1}_{jm}f''(x_0)^{-1}_{kn} f'''(x_0)_{ijk}f'''(x_0)_{lmn}$$
\end{example}

Stationary phase formula (\ref{l1statphase}) replaces, in the asymptotics $\hbar\ra 0$, a measure-theoretic integral on the l.h.s. with the purely algebraic expression on the r.h.s., involving only values of derivatives of $f$ at the critical points $x_0$.

The idea then is to define the path integral (\ref{l1PI}) by formally applying the stationary phase formula, as the r.h.s. of (\ref{l1statphase}), i.e. as a series in $\hbar$ with coefficients given by weights of Feynman diagrams.

We expect that if we started with a classical field theory with $S_M$ invariant under diffeomorphisms of $M$, the partition functions $Z_M$ coming out of the path integral quantization procedure yield manifold invariants and arrange into a TQFT.

\textbf{Problem:} Stationary phase formula requires critical points of $f$ to be \textit{isolated} (more precisely, we need the Hessian of $f$ at critical points to be non-degenerate). However, diffeomorphism invariant classical field theories are \textit{gauge theories}, i.e. there is a tangential distribution $\mc{E}$ on  $F_M$ which preserves the action $S_M$ (in some examples, $\mc{E}$ corresponds to an action of a group $\mc{G}$ -- the \textit{gauge group} -- on $F_M$). Thus, critical points of $S_M$ come in $\mc{E}$-orbits and therefore are not isolated. Put another way, the Hessian of $S_M$ is degenerate in the direction of $\mc{E}$. So, the stationary phase formula cannot be applied to the path integral (\ref{l1PI}) in the case of a gauge theory.

The cure for this problem comes from using the Batalin-Vilkovisky construction.

\subsubsection{Towards Batalin-Vilkovisky (BV) formalism}

Batalin-Vilkovisky construction \cite{BV1,BV2} replaces the classical field theory package $F, S$ with a new package consisting of:
\begin{itemize}
	\item A $\mathbb{Z}$-graded supermanifold $\FF$ (the ``space of BV fields'') endowed with odd-symplectic structure $\omega$ of internal degree $-1$.
	\item A function $S_{BV}$ on $\FF$ -- the ``master action'', satisfying the ``master equation''
	$$\{S_{BV},S_{BV}\}=0$$
	In particular, this implies that the corresponding Hamiltonian vector field $Q=\{S_{BV},\bt\}$ is \textit{cohomological}, i.e. satisfies $Q^2=0$. Thus, $Q$  endows $C^\infty(\FF)$ with the structure of a cochain complex. In other words, $(\FF,Q)$ is a \textit{differential graded (dg) manifold}. 
\end{itemize}

The idea is then to replace 
\begin{equation}\label{l1BV_replacement}
\int_{F} e^{\frac{i}{\hbar}S} \ra \int_{\LL\subset \FF} e^{\frac{i}{\hbar}S_{BV}}
\end{equation}
with $\LL\subset \FF$ a Lagrangian submanifold w.r.t. the odd-symplectic structure $\omega$.

The integral on the l.h.s. of (\ref{l1BV_replacement}) is ill-defined (by means of stationary phase formula) in the case of a gauge theory whereas the integral on the r.h.s. is well-defined, for a good choice of Lagrangian submanifold $\LL\subset \FF$ and moreover is invariant under deformations of $\LL$.

\begin{remark}\footnote{See \cite{Stasheff}.}
Space $\FF$ is constructed, roughly speaking, as Spec of a two-sided resolution of $C^\infty(F)$ construted out of
\begin{itemize}
	\item Chevalley-Eilenberg resolution for the subspace of gauge-invariant functions of fields $C^\infty(F)^\mc{G}$ and
	\item Koszul-Tate resolution for functions on the space of solutions of Euler-Lagrange equations $C^\infty(EL\subset F)$.
\end{itemize}
So, coordinates on $\FF$ of nonzero degree arise as either Chevalley-Eilenberg generators (in positive degree) or Koszul-Tate generators (in negative degree). In particular, this is the reason why $\FF$ has to be a supermanifold (since C-E and K-T generators anti-commute).
\end{remark}

\begin{remark}
In the case of a gauge field theory, one could try to remedy the problem of degenerate critical points in the path integral by passing to the integral over the quotient, $\int_F\ra \int_{F/\mc{G}}$. The latter may indeed have nondegenerate critical points. But the issue is then that we know how to make sense of Feynman diagrams for the path integral over the space of sections of a sheaf over $M$, but the quotient $F/\mc{G}$ would not be of this type. In this sense, one may think of the r.h.s. of (\ref{l1BV_replacement}) as a resolution of the integral over a quotient $F/\mc{G}$ by an integral over a locally free object -- the space of sections of a sheaf over $M$.
\end{remark}

\begin{remark}
	There are finite-dimensional cases when l.h.s. of (\ref{l1BV_replacement}) exists as a measure-theoretic integral (despite having non-isolated critical points). Then, under certain assumptions, 
	one has a comparison theorem that l.h.s. and r.h.s. of (\ref{l1BV_replacement}) coincide. We will return to this when talking about Faddeev-Popov construction and how it embeds into BV.
\end{remark}

\subsection{Tentative program of the course}\marginpar{\LARGE{Lecture 2, 08/29/2016.}}
\begin{itemize}
	\item Classical Chern-Simons theory.
	\item Feynman diagrams (in the context of finite-dimensional integrals) 
	\cite{Etingof,Reshetikhin}.
	\begin{itemize}
		\item Stationary phase formula.
		\item Wick's lemma for moments of a Gaussian integral. Perturbed Gaussian integral.
		\item Berezin integral over an odd vector space \cite{Losev}.
		Feynman diagrams for integrals over a super vector space.
	\end{itemize}
	\item Introduction to BV formalism:
	\begin{itemize}
		\item ($\mathbb{Z}$-graded) supergeometry: odd-symplectic geometry (after \cite{Schwarz}), dg manifolds (partly after \cite{AKSZ}), integration on supermanifolds.
		\item BV Laplacian, classical and quantum master equation (CME and QME).
		\item $\frac12$-densities on odd-symplectic manifolds, BV integrals, fiber BV integral as a pushforward of solutions of quantum master equation (\cite{DiscrBF,CM}).
		\item BV as a solution to the problem of gauge-fixing: Faddeev-Popov construction, BRST (as a homological algebra interpretation of Faddeev-Popov), BV (as a ``doubling'' of BRST). Reference: e.g. \cite{DiscrBF}.
	\end{itemize}
	\item AKSZ (Alexandrov-Kontsevich-Schwarz-Zaboronsky) construction \cite{AKSZ}.
\end{itemize}

\textbf{Applications:}
\begin{enumerate}[(I)]
	\item A topological quantum field theory (not in Atiyah sense, but in the sense of compatibility with cellular subdivisions/aggregations) on CW complexes $X$ -- cellular non-abelian $BF$ theory \cite{DiscrBF,CMRcell}. 
	$$\vcenter{\hbox{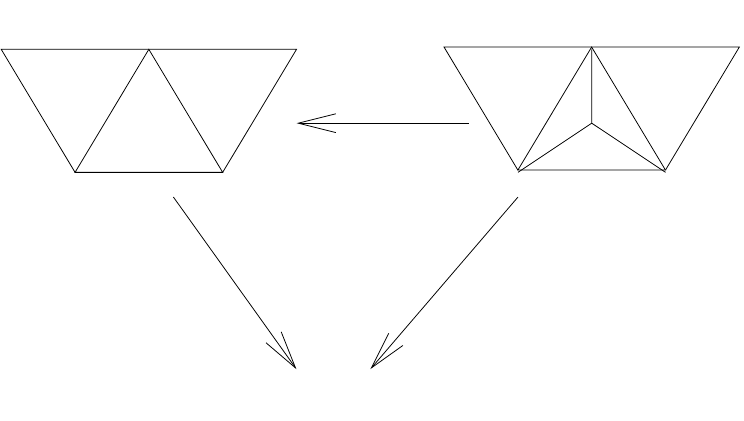}}$$ 
	Here a CW complex $X$ gets assigned a BV package -- a space of fields comprised of cellular cochains and chains twisted by a $G$-local system $E$, $\FF_X=C^\bt(X,E)\oplus C_\bt(X,E^*)$ (with certain homological degree shifts which we omitted here); $G$ is a fixed Lie group -- the structure group of the theory. $\FF_X$ carries a natural odd-symplectic structure (coming from pairing chains with cochains). The action is given as a sum, over cells $e\subset X$ of all dimensions, of certain universal local building blocks $\bar{S}_e$ depending only on combinatorial type of the cell and on values of fields restricted to the cell. 
	
	One calculates certain invariant $\psi(X)$ of $X$ by pushing forward the BV package to the (cellular) cohomology of $X$, via a \textit{finite-dimensional} fiber BV integral. If $X'$ is a cellular subdivision of $X$ (then we say that $X$ is an ``aggregation'' of $X'$), the pushforward of the BV package on $X'$ to $X$ yields back the package on $X$, and for the invariant one has $\psi(X')=\psi(X)$. 
More precisely, one gets a simple-homotopy invariant of CW complexes.
	
We will also discuss here:
	\begin{itemize}
		\item Solutions of the QME vs. infinity algebras  (relevant case for this model: unimodular $L_\infty$ algebras). Fiber BV integral as homotopy transfer of infinity algebras. Feynman diagrams from homological perturbation theory.
		\item 	Relation to rational homotopy type, to formal geometry (neighborhoods of singularities) of the moduli space $\mathcal{M}_{X,G}$ of local systems on $X$, 
		to behavior of the $R$-torsion near the singularities of $\mathcal{M}_{X,G}$.
	\end{itemize}
	\item Perturbative Chern-Simons theory (after Axelrod-Singer \cite{AS1,AS2}). Perturbative invariants of 3-manifolds $M$ given in terms of integrals over Fulton-MacPherson-Axelrod-Singer compactifications of configuration spaces of $n$  distinct points on $M$.
	\item \label{l2_III}\marginpar{\textbf{A posteriori comment:} we did not have time to discuss applications (\ref{l2_III}) and (\ref{l2_IV}) in the end.}
	Kontsevich's deformation quantization of Poisson manifolds $(M,\pi)$ \cite{Kontsevich_Feynman}, partly following \cite{CF}. Here the problem is to costruct a family (parameterized by $\hbar$) of associative non-commutative deformations of the pointwise product on $C^\infty(M)$, of the form
	\begin{equation}\label{l2_star}
	f*_\hbar g(x)=f\cdot g(x)-\frac{i\hbar}{2}\{f,g\}_\pi+\sum_{n\geq 2} (i\hbar)^n B_n(f,g)(x)	\end{equation}
	where $B_n$ are some bi-differential operators (of some order depending on $n$). The idea of the construction (following \cite{CF}) is to write the star-product as as path integral representing certain expectation value for a 2-dimensional topological field theory (the \textit{Poisson sigma model}) on a disk $D$, with two observables placed on the boundary, at points $0$ and $1$:
	\begin{equation}\label{l2_star_via_PI}
	f*_\hbar g(x_0)=\int_{X(\infty)=x_0, \eta_{\dd D}=0} \DD X\, \DD\eta\quad e^{\frac{i}{\hbar}S_{PSM}(X,\eta)} f(X(0))\cdot g(X(1))
	\end{equation}
	Here the fields $X,\eta$ are the base and fiber components of a bundle map
	$$\begin{CD}
	TD @>\eta >> T^*M\\
	@VVV @VVV\\
	D @>>X> M
	\end{CD}$$
	and the action is: $S_{PSM}=\int_D \langle \eta,dX \rangle+\frac12 \langle X^*\pi, \eta\wedge \eta \rangle$. This action possesses a rather complicated gauge symmetry (given by a non-integrable distribution on the space of fields) and one needs BV to make sense of the integral (\ref{l2_star_via_PI}). The final result is the explicit construction of operators $B_n$ in (\ref{l2_star}) in terms of integrals over compactified configuration spaces of points on the 2-disk $D$.
	\item \label{l2_IV} BV formalism for field theories on manifolds with boundary, with Atiyah-Segal's gluing/cutting -- ``BV-BFV formalism'' \cite{CMR,CMRpert} (a very short survey in \cite{CMR_ICMP}). Examples:
	\begin{itemize}
		\item Non-abelian $BF$ theory on cobordisms endowed with CW decomposition \cite{CMRcell}.
		\item AKSZ theories on manifolds with boundary.
	\end{itemize}
\end{enumerate}

\section{Classical Chern-Simons theory}

\subsection{Chern-Simons theory on a closed $3$-manifold}
Let, for simplicity, $G=SU(2)$ (we will comment on generalization to other Lie groups later) and let $M$ be a closed oriented 3-manifold. Let $\PP$ be the trivial $G$-bundle over $M$.

\subsubsection{Fields} We define the space of fields to be the space of principal connections on $\PP$. 
Since $\PP$ is trivial, we can use the trivialization to identify connections with $\g$-valued $1$-forms on $M$  (by pulling back the connection 1-form $\mc{A}\in \Omega^1(\PP,\g)$ on the total space of $\PP$ to $M$ by the trivializing section $\sigma:M\ra \PP$). Here $\g$ is the Lie algebra of $G$, i.e. in our case $\g=\mathfrak{su}(2)$. So, we have $F_M=\mr{Conn}_{M,G}\simeq \Omega^1(M,\g)$.

\subsubsection{Action} We define the action functional on $F_M$ as
$$S_{CS}(A):=\int_M \tr \frac12 A\wedge dA + \frac13 A\wedge A\wedge A$$
with $A\in \Omega^1(M,\g)$ a connection $1$-form in fundamental representation of $\mathfrak{su}(2)$. 
\begin{remark}
It can be instructive to rewrite the action as $\int_M \tr \frac12 A\wedge dA + \frac16 A\wedge [A, A]$ where $[,]$ is the (super-)Lie bracket on the differential graded Lie algebra of $\g$-valued forms, $\Omega^\bt(M,\g)$; here $[A,A]$ is simply $A\wedge A+A\wedge A$. But this rewriting exhibits denominators $1/2!$, $1/3!$ and suggests that there might be some ``homotopy Chern-Simons'' action associated to infinity algebras where higher terms would appear, which is indeed correct \cite{CM}.
\end{remark}

\subsubsection{Euler-Lagrange equation}
Let us calculate the variation of the action:
$$\delta S_{CS}=\int_{M}\tr \frac12 \delta A\wedge dA +\frac12 A\wedge d\delta A+ \delta A\wedge A\wedge A = \int_{M}\tr \delta A\wedge (\underbrace{dA+A\wedge A}_{\mr{curvature}\; F_A})$$
Here in the second equality we used integration by parts to remove $d$ from $\delta A$. Note that the coefficient of $\delta A$ in the final expression is the curvature $2$-form of the connection $A$, $F_A=dA+A\wedge A=A+\frac12 [A,A]\in \Omega^2(M,\g)$. Thus, the Euler-Lagrange equation $\delta S_{CS}=0$ (the critical point equation for $S_{CS}$) reads
$$\boxed{F_A=0}$$
-- flatness condition on the connection.

\subsubsection{Gauge symmetry}\label{sss: CS gauge invariance}
For any group-valued  map $g:M\ra G$ and a connection $A\in \Omega^1(M,\g)$, we define the \textit{gauge transformation} as mapping
\begin{equation}\label{l2_gauge_transf}
A \quad \mapsto \quad \boxed{A^g:=g^{-1}A g+g^{-1}dg}
\end{equation}
This defines a (right) action of the gauge group $\mr{Gauge}_{M,G}=\mr{Map}(M,G)$ on $F_M=\mr{Conn}_{M,G}$.

One can understand the transformation formula (\ref{l2_gauge_transf}) as the effect of a change of trivialization of the principal bundle $\PP$: assume that the connection $1$-form on total space $\mc{A}\in \Omega^1(\PP,\g)$ is fixed but we are given two different trivializations $\sigma,\sigma': M\ra \PP$ with $\sigma'=\sigma\cdot g$. Then, the corresponding $1$-forms on the base, $A_\sigma=\sigma^* \mc{A}$ and $A_{\sigma'}=(\sigma')^* \mc{A}$ are related by (\ref{l2_gauge_transf}).

Alternatively, one can interpret (\ref{l2_gauge_transf}) as the action of a bundle automorphism 
\begin{equation}\label{l2_S_gauge_transf}
\begin{CD} \PP @>{\cdot g}>{\simeq}> \PP \\ @VVV @VVV \\ M @= M \end{CD}
\end{equation}
on a connection.

Note that $A^g$ is flat iff $A$ is flat.

Chern-Simons action changes under the gauge transformation (\ref{l2_gauge_transf}) as
$$S_{CS}(A^g)-S_{CS}(A)=-\frac16 \int_M\tr (g^{-1}dg)^{\wedge 3}$$ 
where $(g^{-1}dg)^{\wedge 3}=(g^{-1}dg)\wedge (g^{-1}dg)\wedge (g^{-1}dg)$ is a 3-form on $M$ with coefficients in matrices (endomorphisms of the space where $\g$ is represented).

Recall that for $G\subset U(N)$ 
a simple 
compact group, one has the \textit{Cartan $3$-form}
$$\theta=-\frac{1}{24 \pi^2}\tr\,(g^{-1}dg)^{\wedge 3} \in \Omega^3(G) $$
-- a closed $G$-invariant form on $G$ with integral periods representing the generator of $H^3(G,\mathbb{Z})\simeq \ZZ$. In particular, for $G=SU(2)$, $\theta$ is the volume form on $SU(2)$ viewed as the 3-sphere, normalized to have total volume $1$.

Therefore, (\ref{l2_S_gauge_transf}) implies the following
\begin{lemma}[Gauge (in)dependence of Chern-Simons action]\label{l2_lem}
$$\frac{1}{4\pi^2}\left(S_{CS}(A^g)-S_{CS}(A)\right)=\int_M g^*\theta = \langle [M], g^*[\theta] \rangle\quad \in \ZZ$$
\end{lemma}
Note that, for $G=SU(2)$, the r.h.s. is simply the degree of the map $g: M\ra SU(2)\sim S^3$.

Thus, $S_{CS}(A)$ is invariant under infinitesimal gauge transformations; more precisely, it is invariant under $\mr{Gauge}_{M,G}^0\subset \mr{Gauge}_{M,G}$ -- the connected component of trivial transformation $g=1$ in $\mr{Gauge}_{M,G}$. However, under a general gauge transformation $S_{CS}(A)$ can change by an integer multiple of $4\pi^2$.

\marginpar{\LARGE{Lecture 3, 08/31/2016.}}
Introduce a function 
\begin{equation}\label{l3_psi_k}
\boxed{\psi_k(A):=e^{\frac{ik}{2\pi}S_{CS}(A)}}
\end{equation}
with $k\in \ZZ$ a parameter -- the ``level'' of Chern-Simons theory. By Lemma \ref{l2_lem}, $\psi_k$ is a $\GGauge_{M,G}$-invariant function on $F_M=\Conn_{M,G}$. In particular, we can regard $\psi_k$ as a function on the quotient $\Conn_{M,G}/\GGauge_{M,G}$.

\subsubsection{Chern-Simons invariant on the moduli space of flat connections}
Restriction of the function $\psi_k$ to flat connections yields a \textit{locally constant} function on the quotient 
$$\MM_{M,G}=\FlatConn_{M,G}/\GGauge_{M,G}=\frac{\{A\in \Omega^1(M,\g)\;\;\mr{s.t.}\;\; dA+\frac12[A,A]=0\}}{A\sim g^{-1}Ag+g^{-1}dg\quad\quad \forall g:M\ra G}$$
-- the \textit{moduli space of flat connections}. The locally constant property of $\psi_k$ on the moduli space follows immediately from the fact that flat connections solve the Euler-Lagrange equation $\delta S_{CS}=0$. 

Recall that $\MM_{M,G}$ can be identified\footnote{The identification goes via mapping a flat connection $A$ to a map associating to based loops $\gamma$ on $M$ the holonomy of $A$ around $\gamma$. 
Flatness of $A$ implies that this map on loops descends to homotopy classes of loops and implies the group homomorphism property of the map. Final quotient by $G$ corresponds to quotienting out the changes of trivialization of the fiber of $\PP$ over the base point.
} 
with $\mr{Hom}(\pi_1(M),G)/G$ -- the space of group homomorphisms $\pi_1(M)\ra G$, modulo action of $G$ on such homomorphisms by conjugation on the target $G$.\footnote{The identification $\MM_{M,G}\simeq \mr{Hom}(\pi_1(M),G)/G$ is true for $M$ of arbitrary dimension, if one allows flat connections in all -- possibly non-trivial -- $G$-bundles over $M$. Thus, $\MM_{M,G}$ is in fact the moduli space of flat bundles, rather than just flat connections in a trivial bundle.}

Moduli space $\MM_{M,G}$ is typically disconnected and $\psi_k$ can take different values on different connected components.

\begin{example} Take $G=SU(2)$ and take $M$ to be a \textit{lens space}:
$$M=L(p,q):=\frac{\{(z_1,z_2)\in \CC^2\;\; \mr{s.t.}\;\; |z_1|^2+|z_2|^2=1\}}{(z_1,z_2)\sim (\zeta\cdot z_1,\zeta^q\cdot z_2)}\sim S^3/\ZZ_p$$
where $\zeta=e^{\frac{2\pi i}{p}}$ the $p$-th root of unity; we assume that $(p,q)$ are coprime (otherwise $L(p,q)$ is not a smooth manifold).
\end{example}

The moduli space $\MM_{M,G}$ is the space of elements of order $p$ in $SU(2)$ considered modulo conjugation. Thus, $\MM_{M,G}$ consists of $\left[\frac{p+1}{2}\right]$ isolated points corresponding to classes of flat connections $[A]_0,\ldots, [A]_{\left[\frac{p-1}{2}\right]}$ where class $[A]_r$ has the holonomy around the loop $\gamma$, representing the generator of $\pi_1(M)=\ZZ_p$, of the form 
$$\mr{hol}_{\gamma}[A]_r=\left(\begin{array}{ll} e^{\frac{2\pi i r}{p}} & 0 \\ 0 & e^{-\frac{2\pi i r}{p}} \end{array} \right)\in SU(2)$$
We consider consider $r$ as defined $\bmod\; p$, and moreover $r$ and $-r$ correspond to conjugate elements in $SU(2)$. Therefore choices $r\in \{0,1,\ldots,\left[\frac{p-1}{2}\right] \}$ do indeed exhaust all distinct points of $\MM_{M,G}$.

The value of the function $\psi_k$ (\ref{l3_psi_k}) on the point $[A]_r\in \MM_{M,G}$ is:
$$\psi_k([A]_r)=e^{\frac{2\pi ik q^* r^2}{p}}$$
(This is the result of a non-trivial calculation.)
Here $q^*$ is the residue mod $p$ reciprocal to $q$, i.e. defined by $q^*q=1\bmod\; p$. In particular, the set of values of $\psi_k$ on $\MM_{M,G}$ distinguishes between non-homotopic lens spaces, e.g.  distinguishes between $L(5,1)$ and $L(5,2)$.

\subsubsection{Remark: more general $G$}
We can allow $G$ to be any connected, simply-connected, simple, compact Lie group (e.g. $G=SU(N)$) without having to change anything. 

We can also allow $G$ to be semi-simple, $G=G_1\times\cdots \times G_n$ with $G_n$ the simple factors -- the corresponding Chern-Simons theory is effectively a collection of $n$ mutually non-interacting Chern-Simons theories for groups $G_1,\ldots,G_n$. In this case we can introduce independent levels $k_1,\ldots,k_n\in\ZZ$ for different factors.

The assumption that $\pi_0(G)$ and $\pi_1(G)$ are trivial is crucial. By a result of W. Browder, 1961, $\pi_2(G)$ is trivial for any finite-dimensional Lie group (in fact, even for any finite-dimensional $H$-space). Thus, under our assumptions $G$ is $2$-connected and the classifying space $BG$ is $3$-connected. Therefore, for $M$ of dimension $\leq 3$, $[M,BG]=*$ -- all classifying maps are homotopically trivial. Thus \textit{a $G$-bundle $\PP$ over $M$ has to be trivial}. And then we can globally identify connections in $\PP$ with $\g$-valued $1$-forms and can make sense of Chern-Simons action. However, if either $\pi_0(G)$ or $\pi_1(G)$ is nontrivial, then there can exist non-trivial $G$-bundles (and one has to allow connections in all possible $G$-bundles as valid fields for the theory, if one wants ultimately to construct a field theory compatible with gluing/cutting). In this case special techniques are needed to construct $S_{CS}$ (e.g. by defining the action on patches where the bundle is trivialized and then gluing the patches while taking into account the corrections arising from the change of trivialization on overlaps). In particular, for $G=U(1)$,  $S_{CS}$ is constructed in terms of  Deligne cohomology.


\subsubsection{Relation to the second Chern class}
We assume again that $G=SU(2)$ (or, more generally, any simply-connected subgroup of $U(N)$).

\textbf{Fact:} any closed oriented $3$-manifold $M$ is \textit{null-cobordant}, i.e. there exists a $4$-manifold $N$ with boundary $\dd N=M$.

As before, let $\PP$ be the trivial $G$-bundle over $M$ and let $\Tilde\PP$ be the trivial $G$-bundle over $N$

\begin{lemma}\label{l3_lem}
Let $A\in \Omega^1(M,\g)$ be a connection in $\PP$ and $a\in \Omega^1(N,\g)$ its extension to a connection in $\Tilde\PP$ (i.e. the pullback by the inclusion of the boundary $\iota:M\hookrightarrow N$ is $a|_M:=\iota^*a=A$). Then we have
\begin{equation}
S_{CS}(A)=\frac12\int_N \tr  F_a\wedge F_a
\end{equation}
where $F_a=da+\frac12 [a,a]\in \Omega^2(N,\g)$ is the curvature of $a$.
\end{lemma}

\begin{proof}
Indeed, we have
\begin{equation}\label{l3_e1}
d\, \tr(\frac12 a\wedge da+\frac13 a\wedge a\wedge a)=\tr(\frac12 da\wedge da +da\wedge a\wedge a)
\end{equation}
and 
\begin{equation}\label{l3_e2}
\tr \frac12 F_a\wedge F_a= \tr \frac12(da+a\wedge a)\wedge (da+a\wedge a)=\tr(\frac12 da\wedge da+da\wedge a\wedge a+a\wedge a\wedge a\wedge a)
\end{equation}
Note that the last term on the r.h.s. vanishes under trace: $\tr a^{\wedge 4}=\tr a\wedge a^{\wedge 3}=-\tr a^{\wedge 3}\wedge a=-\tr a^{\wedge 4}$, hence $\tr a^{\wedge 4}=0$. Thus, (\ref{l3_e1})=(\ref{l3_e2}) and the statement follows by Stokes' theorem.
\end{proof}

Let $N_+$, $N_-$ be two copies of $N$ (with $N_-$ carrying the opposite orientation). Let $\bar N=N_+\cup_M N_-$ be the closed $4$-manifold obtained by gluing $N_+$ and $N_-$ along $M$. 

Fix $g: M\ra G$ and construct a (generally, non-trivial) $G$-bundle $\bar\PP_g$ over $\bar{N}$ which is trivial over $N_+$ and $N_-$ and has transition function $g$ on the tubular neighborhood of $M\subset \bar{N}$.

Let $A$ be some connection on $M$; let $a_+$ be its extension over $N_+$ and let $a_-$ be an extension of the \textit{gauge transformed} connection $A^g=g^{-1}A g+g^{-1}dg$ over $N_-$. The pair $(a_+,a_-)$ defines a connection $\bar{a}$ in $\bar\PP_g$.

By Lemma \ref{l3_lem}, we have
\begin{multline}\label{l3_e3}
\frac{1}{8\pi^2}\int_{\bar{N}}\tr F_{\bar{a}}\wedge F_{\bar{a}}=\frac{1}{8\pi^2}\int_{N_+\cup N_-}\tr F_{\bar{a}}\wedge F_{\bar{a}}=\frac{1}{8\pi^2} \left(\int_N \tr F_{a_+}\wedge F_{a_+}- \int_N \tr F_{a_-}\wedge F_{a_-} \right)\\
=\frac{1}{4\pi^2}(S_{CS}(A)-S_{CS}(A^g))
\end{multline}

\textbf{Input from Chern-Weil theory.} Recall that for $\s{P}$ a $G$-bundle over $\s{M}$ (with $\s{M}$ of arbitrary dimension  and with $G$ a subgroup of $U(N)$), for $\s{A}$ an arbitrary connection in $\s{P}$, the closed $4$-form
\begin{equation}\label{l3_Chern-Weil_4-form}
\frac{1}{8\pi^2}\tr F_{\s{A}}\wedge F_{\s{A}} \in \Omega^4(\s{M})_\mr{closed}
\end{equation}
represents the image of the second Chern class of $\s{P}$,\footnote{More precisely, this is the second Chern class of the associated vector bundle $\s{P}\times_G \CC^N$.} $c_2(\s{P})\in H^4(M,\ZZ)$ in de Rham cohomology $H^4(M,\mathbb{R})$. In particular, $4$-form (\ref{l3_Chern-Weil_4-form}) has integral periods independent of $\s{A}$.

We conclude that the gauge transformation property of Chern-Simons action can be expressed in terms of characteristic classes for $G$-bundles on $4$-manifolds as follows.
\begin{lemma}
$$
\frac{1}{4\pi^2}(S_{CS}(A^g)-S_{CS}(A))=\langle [\bar{N}], c_2(\bar{P}_g) \rangle\quad \in\ZZ
$$
\end{lemma}

\subsection{Chern-Simons theory on manifolds with boundary}

Let now $M$ be an oriented 3-manifold with boundary $\dd M=\Sigma$ - a closed surface, or several closed surfaces.

As in the case of $M$ closed, fields are connections on $M$ and the action is unchanged, $S_{CS}(A)=\int_M\tr \frac12 A\wedge dA+\frac16 A\wedge [A,A]$.

\subsubsection{Phase space}
We define the \textit{phase space} $\Phi_\Sigma$ associated to the boundary $\Sigma$ as the space of pullbacks of fields (connections) on $M$ to the boundary. Thus, $\Phi_\Sigma=\Conn_{\Sigma,G}$ -- connections on $\Sigma$,  and we have a natural projection from fields on $M$ to the boundary phase space 
\begin{equation}\label{l3_pi}
\begin{CD} F_M=\Conn_{M,G} \\ @V{\pi=\iota^*}VV \\ \Phi_\Sigma=\Conn_{\Sigma,G} \end{CD}
\end{equation}
-- the pullback by the inclusion of the boundary $\iota: \Sigma\hookrightarrow M$.

\subsubsection{$\delta S$, Euler-Lagrange equations}
Let us calculate $\delta S$. Now we will interpret $\delta$ as the exterior derivative on the space of fields, i.e. $\delta S \in \Omega^1(F_M)$ is a $1$-form on fields and one can contract it with a tangent vector $v\in T_A F_M\simeq \Omega^1(M,\g)$ to produce a number. This is a (marginally) different interpretation from $\delta$ as a variation in variational calculus; the computations are the same but sign conventions are affected as now we treat $\delta$ as an odd operator.

Note that now we have two de Rham differentials: $d$ -- the ``geometric'' de Rham operator on $M$ (or $\Sigma$) and the ``field'' de Rham operator $\delta$ on $F_M$ (resp. $\Phi_\Sigma$).

The computation is as follows:
\begin{multline}\label{l3_deltaS}
\delta S= \int_M\tr\left( -\frac12 \delta A\wedge d A -\frac12 A\wedge d\delta A-\frac12 \delta A\wedge [A,A]\right)\\
=\underbrace{-\int_M\tr \delta A\wedge F_A}_{\mr{``bulk\; term"}}+\underbrace{\int_\Sigma \tr\frac12 A|_\Sigma\wedge \delta A|_\Sigma}_{\mr{``boundary\; term"}}
\end{multline}
Here we used Stokes' theorem to remove $d$ from $\delta A$, and, unlike in the computation for $M$ closed, a boundary term appeared as a result.

Euler-Lagrange equation read off from the first term in the r.h.s. of (\ref{l3_deltaS}) -- the equation that $\langle \delta S,v\rangle=0$ for a field variation $v\in \Omega(M,\g)$ \textit{supported away from the boundary} -- is 
\begin{equation}\label{l3_EL}
F_A=0
\end{equation}
-- the flatness equation, as for $M$ closed.

\subsubsection{Noether $1$-form, symplectic structure on the phase space}
We interpret the boundary term in the r.h.s. of (\ref{l3_deltaS}) as $\pi^*\alpha_\Sigma$ -- the pullback by the projection (\ref{l3_pi}) of the \textit{Noether $1$-form}  on the phase space $\alpha_\Sigma\in \Omega^1(\Phi_\Sigma)$ defined as
$$\alpha_\Sigma = \int_\Sigma \tr \frac12 A_\Sigma\wedge \delta A_\Sigma$$
I.e., for $A_\Sigma\in \Conn_{\Sigma,G}$ a fixed connection on the boundary and for $v\in T_{A_\Sigma}\Phi_\Sigma\simeq \Omega^1(\Sigma,\g)$ a tangent vector (``a variation of boundary field''), we have 
$$\iota_v \alpha_\Sigma = -\int_\Sigma \tr \frac12 A_\Sigma\wedge v\quad \in \RR$$ (symbol $\iota_v$ stands for the contraction with a vector or vector field).

The exterior derivative of the $1$-form $\alpha_\Sigma$ yields a $2$-form 
\begin{equation}
\omega_\Sigma:=\delta \alpha_\Sigma=\int_\Sigma \tr \frac12 \delta A_\Sigma \wedge \delta A_\Sigma\quad \in \Omega^2(\Phi_\Sigma)
\end{equation}
In particular, for $u,v\in T_{A_\Sigma}\Phi_\Sigma\simeq \Omega^1(M,\g)$ a pair of tangent vectors, we have
$$\iota_u \iota_v\omega_\Sigma = \int_\Sigma \tr\; u\wedge v\quad \in\RR$$

The $2$-form $\omega_\Sigma$ is closed by construction. Also, it is weakly non-degenerate (in the sense that the induced sharp-map $\omega^\#: T\Phi_\Sigma \ra T^* \Phi_\Sigma$ is \textit{injective}). Thus, $\omega_\Sigma$ defines a symplectic structure on $\Phi_\Sigma$, viewed as an infinite-dimensional (Fr\'echet) manifold.

\subsubsection{``Cauchy subspace''}
We define the \textit{Cauchy\footnote{Or ``constraint'' or ``coisotropic'' (see below).} subspace} $C_\Sigma\subset \Phi_\Sigma$ as the subspace of fields on the boundary which can be extended to a neighborhood of the boundary, $\Sigma\times [0,\epsilon)\subset M$, as solutions to Euler-Lagrange equations.\footnote{Thus, a ``Cauchy subspace'' -- space of valid (in the sense of guaranteeing existence of a solution) initial data on $\Sigma\times\{0\}$ for the Cauchy problem for Euler-Lagrange equations on $\Sigma\times [0,\epsilon)$.}

For Chern-Simons theory, this means that $C_\Sigma$ is comprised of connections on $\Sigma$ which can be extended to flat connections on $\Sigma\times [0,\epsilon )$. Thus, $C_\Sigma=\FlatConn_{\Sigma,G}\subset \Conn_{\Sigma,G}$ is simply the space of all flat connections on $\Sigma$. 

\marginpar{\LARGE{Lecture 4, 09/05/2016.}}
Recall that, a vector subspace $U$ of a symplectic vector space $(V,\omega)$ is called
\begin{itemize}
\item \textit{isotropic} if $U\subset U^\perp$, with $U^\perp=\{w\in V\;\mr{s.t.}\;  \omega(w,u)=0\;\forall u\in U\}$ -- the symplectic orthogonal complement of $U$ (equivalently, $U\subset (V,\omega)$ is isotropic if $\omega$ vanishes on pairs of vectors from $U$);
\item \textit{coisotropic} if $U^\perp\subset U$;
\item \textit{Lagrangian} if $U=U^\perp$.
\end{itemize}
Similarly, a submanifold $N\subset (\Phi,\omega)$ of a symplectic manifold is \textit{isotropic/coisotropic/Lagrangian} if, for any point $x\in N$, the tangent space $T_x N$ is a isotropic/coisotropic/Lagrangian subspace in $(T_x \Phi,\omega_x)$.

Recall that, for $C\subset (\Phi,\omega)$ a coisotropic submanifold, the \textit{characteristic distribution} is defined as 
$(T C)^\perp \subset TC$ -- a subbundle of the tangent bundle of $C$ assigning to $x\in C$ a subspace $(T_x C)^\perp$ in $T_x C$. This distribution is integrable (by Frobenius theorem and $d\omega=0$) and thus induces a foliation of $C$ by the leaves of charactersitic foliation. We denote $\underline{C}$ the corresponding space of leaves (the ``coisotropic reduction'' of $C$). The reduction $\underline{C}$ inherits a symplectic structure $\underline\omega$ characterized by $p^*\underline\omega=\omega|_C$ where $p: C\ra\underline{C}$ is the quotient map.\footnote{Put another way, forgetting about the ambient symplectic manifold, $(C,\omega|_C)$ is itself a \textit{pre-symplectic manifold}, i.e. one equipped with a \textit{pre-symplectic structure} -- a closed $2$-form which can be degenerate but its kernel is required to be a subbundle of the tangent bundle $TC$ (in particular, is required to have constant rank). From this point of view, $\underline{C}$ is the space of leaves of the kernel of pre-symplectic structure $\ker \omega|_C\subset TC$.}

\begin{lemma}\begin{enumerate}[(i)]
\item \label{l4_lm1_i} The submanifold $C_\Sigma\subset \Phi_\Sigma$ is coisotropic.
\item \label{l4_lm1_ii} The characteristic distribution $(TC_\Sigma)^\perp$ on $C_\Sigma$ is given by infinitesimal gauge transformations.
\end{enumerate}
\end{lemma}
\begin{proof}
Fix $A_\Sigma\in C_\Sigma$ a flat connection on $\Sigma$. The tangent space $T_{A_\Sigma}C_\Sigma$ is the space of first order deformations of $A_\Sigma$ as a flat connection. For the curvature of a generic small deformation of $A_\Sigma$, we have $F_{A_\Sigma+\epsilon\cdot \alpha}=\epsilon\cdot \underbrace{d_{A_\Sigma}\alpha}_{=:\ d\alpha+[A_\Sigma,\alpha]} +O(\epsilon^2)$ for a deformation $\alpha \in \Omega^1(\Sigma,\g)$ and $\epsilon\ra 0$ a small deformation parameter. Hence, 
$$T_{A_\Sigma}C_\Sigma=\{\alpha\in \Omega^1(\Sigma,\g)\;\;\mr{s.t.}\;\; d_{A_\Sigma}\alpha=0\}= \Omega^1(\Sigma,\g)_{d_{A_\Sigma}\mr{-closed}}$$
Let us calculate the symplectic orthogonal:
\begin{equation}\label{l4_e1}
(T_{A_\Sigma}C_\Sigma)^\perp=\{\beta\in \Omega^1(\Sigma,\g)\;\; \mr{s.t.} \;\;\int_\Sigma \tr\alpha\wedge \beta = 0\quad  \forall \alpha \in \Omega^1(\Sigma,\g)_{d_{A_\Sigma}\mr{-closed}}\}
\end{equation}
Let us put a metric on $\Sigma$ and let $*$ be the corresponding Hodge star operator. We then continue (\ref{l4_e1}) making a change $\beta=*\gamma$:
\begin{equation}
(T_{A_\Sigma}C_\Sigma)^\perp= *\{\gamma \in \Omega^1(\Sigma,\g)\;\; \mr{s.t.} \;\;(\alpha,\gamma)= 0\quad  \forall \alpha \in \Omega^1(\Sigma,\g)_{d_{A_\Sigma}\mr{-closed}}\}
\end{equation}
where $(\alpha,\gamma)=\int_\Sigma \tr \alpha\wedge \gamma$ is the positive definite Hodge inner product on $\Omega^\bt(\Sigma,\g)$. By Hodge decomposition theorem, we have
$$\Omega^\bt(\Sigma,\g)=\underbrace{\Omega^\bt(\Sigma,\g)_{d_{A_\Sigma}\mr{-exact}}\oplus \Omega^\bt(\Sigma,\g)_\mr{harmonic}}_{\Omega^\bt(\Sigma,\g)_{d_{A_\Sigma}\mr{-closed}}} \oplus\Omega^\bt(\Sigma,\g)_{d^*_{A_\Sigma}\mr{-exact}}$$
Thus, the orthogonal complement of $\Omega^1(\Sigma,\g)_{d_{A_\Sigma}\mr{-closed}}$ w.r.t. Hodge inner product is $\Omega^1(\Sigma,\g)_{d^*_{A_\Sigma}\mr{-exact}}$. Therefore,
\begin{equation}\label{l4_e2}
(T_{A_\Sigma}C_\Sigma)^\perp=*\left(\Omega^\bt(\Sigma,\g)_{d^*_{A_\Sigma}\mr{-exact}}\right)=\Omega^1(\Sigma,\g)_{d_{A_\Sigma}\mr{-exact}}
\end{equation}
Since exact forms are a subspace of closed forms,  we have $(T_{A_\Sigma}C_\Sigma)^\perp\subset T_{A_\Sigma}C_\Sigma$ which proves item (\ref{l4_lm1_i}) -- coisotropicity of $C_\Sigma$.

Infinitesimal gauge transformations are the action of the Lie algebra $\mr{gauge}_{\Sigma}=\mr{Lie}(\GGauge_{\Sigma})\simeq \mr{Map}(\Sigma,\g)$ by vector fields on $\Conn_{\Sigma}$; this infinitesimal action arises from considering the action of a path of gauge transformations, $g_t\in \GGauge_{\Sigma}$ with $t\in [0,\epsilon)$, starting at $g_{t=0}=1$, on a connection $A_\Sigma$ and taking the derivative at $t=0$. Thus the gauge transformation formula $$g\in \GGauge_{\Sigma}\quad \mapsto\quad  (A_\Sigma\mapsto A_\Sigma^g=g^{-1}A_\Sigma g+g^{-1}dg)\in \mr{Diff}(\Conn_{\Sigma})$$ implies that infinitesimal gauge transformations are given by
\begin{equation}\label{l4_inf_gt}
\gamma\in \mr{gauge}_{\Sigma}\quad \mapsto \quad (A_\Sigma\mapsto \underbrace{d_{A_\Sigma} \gamma}_{\in T_{A_\Sigma}\Conn_\Sigma}) \in \mathfrak{X}(\Conn_\Sigma) 
\end{equation}

Note that, fixing $A_\Sigma$ and varying $\gamma$ in (\ref{l4_inf_gt}), we obtain the subspace 
$$\{d_{A_\Sigma}\gamma \;|\; \gamma\in \Omega^0(\Sigma,\g)\}\subset T_{A_\Sigma}\Conn_\Sigma$$ which coincides with the value  (\ref{l4_e2}) of the characteristic distribution on $C_\Sigma$ at $A_\Sigma\in C_\Sigma$. This proves item (\ref{l4_lm1_ii}).
\end{proof}

\subsubsection{$L_{M,\Sigma}$}
Let $EL_M=\FlatConn_M$ be the space of solutions of Euler-Lagrange equation on $M$ -- the space of flat connections, and let 
$L_{M,\Sigma}:=\pi(EL_M)\subset \Phi_\Sigma$ be the set of boundary values of flat connections on $M$. Since a solution of E-L equation on $M$ is in particular a solution of E-L equation on the neighborhood of $\Sigma$, we have $$L_{M,\Sigma}\subset C_\Sigma \subset \Phi_\Sigma$$

\begin{remark}[\textbf{Aside on the evolution relation in classical mechanics.}] Consider a classical mechanical system in Hamiltonian formalism as a 1-dimensional field theory on an interval. It assigns to a point with $+$ orientation a phase space $\Phi$ (a symplectic manifold $(\Phi,\omega)$) and to a point with $-$ orientation the same space with the opposite sign of symplectic structure, $\bar\Phi$ (i.e. $(\Phi,-\omega)$). To an interval $[t_0,t_1]$ it assigns $L=\pi(EL_{[t_0,t_1]})\subset \bar\Phi\times\Phi$; $L$ consists of pairs of (initial state, final state) related by time evolution of the system from time $t_0$ to time $t_1$. In the case of a non-degenerate classical system, any point in $\Phi_{t_0}$ defines a unique solution for the Cauchy problem for E-L equation and evaluating it at $t=t_1$ we obtain an evolution map $U_{[t_0,t_1]}:\; \Phi_{t_0} \ra \Phi_{t_1}$ which is a symplectomorphism (since the equations of motion are Hamiltonian), and then $L=\mr{graph}\; U_{[t_0,t_1]}$. Being a graph of a symplectomorphism, $L\subset \bar\Phi_{t_0}\times \Phi_{t_1}$ is Lagrangian. One can think of $L$ as a set-theoretic relation between $\Phi_{t_0}$ and $\Phi_{t_1}$ with additional Lagrangian property (such relations are called ``canonical relations''). Since $L$ encodes the time evolution of the system (or ``dynamics''), it deserves a name of the ``evolution relation''  or ``dynamic relation''.
\end{remark}

Now we are back to Chern-Simons.
\begin{lemma}\label{l4_lm2}
$L_{M,\Sigma}\subset \Phi_\Sigma$ is isotropic.
\end{lemma}

\begin{proof}
Let $A_\Sigma\in L_{M,\Sigma}$ be the boundary value of a flat connection $\Tilde A$ on $M$. The tangent space to $L_{M,\Sigma}$ is $$T_{A_\Sigma}L_{M,\Sigma}=\{\alpha\in \Omega^{1}(\Sigma,\g)\;\;\mr{s.t.}\;\;\alpha=\Tilde\alpha|_\Sigma\;\mr{for\;some}\;\Tilde\alpha\in \Omega^1(\Sigma,\g)_{d_{\Tilde A}\mr{-closed}}\}$$
Thus, for $\alpha,\beta\in T_{A_\Sigma}L_{M,\Sigma}$, we have
$$\omega_\Sigma(\alpha,\beta)=\int_\Sigma\tr \alpha\wedge\beta \underset{\mr{Stokes'}}{=}\int_M\tr (d_{\Tilde A}\Tilde\alpha\wedge \Tilde\beta-\Tilde\alpha\wedge d_{\Tilde A}\Tilde\beta) =0 $$
(Note that replacing $d\ra d_{\Tilde A}$ under trace is an innocent operation, as $\tr [\Tilde A,\bullet]=0$.) Thus, $\omega_\Sigma$ vanishes on $L_{M,\Sigma}$, which is the isotropic property.
\end{proof}

\subsubsection{Reduction of the boundary structure by gauge transformations}
Let $\underline{C}_\Sigma=C_\Sigma/\GGauge_\Sigma$ be the coisotropic reduction of $C_\Sigma$ (by definition, this is the space of leaves of characterisitc distribution on $C_\Sigma$) -- the space of classes of flat connections on $\Sigma$ module gauge transformations. Thus, $$\underline{C}_\Sigma=\MM_\Sigma\simeq \mr{Hom}(\pi_1(\Sigma),G)/G$$ is the moduli space of flat connections on $\Sigma$.

Note that the tangent space to the moduli space is
$$T_{[A_\Sigma]}\MM_\Sigma=\frac{T_{A_\Sigma}C_\Sigma}{(T_{A_\Sigma}C_\Sigma)^\perp}=\frac{\Omega^1(\Sigma,\g)_{d_{A_\Sigma}\mr{-closed}}}{\Omega^1(\Sigma,\g)_{d_{A_\Sigma}\mr{-exact}}}=H^1_{d_{A_\Sigma}}(\Sigma,\g)$$
-- the twisted (by a flat connection $A_\Sigma$) first de Rham cohomology.

Symplectic structure $\underline\omega_\Sigma$ on $\MM_\Sigma$ (the Atiyah-Bott symplectic structure) is:
$$\underline{\omega}_\Sigma([\alpha],[\beta])=\int_\Sigma \tr \alpha\wedge \beta$$
-- the standard Poincar\'e duality pairing (with coefficients in a local system determinaed by $A_\Sigma$), $\langle,\rangle_\Sigma:\;H^1_{d_{A_\Sigma}}\otimes H^1_{d_{A_\Sigma}}\ra\RR$.

Let $\underline{L}_{M,\Sigma}=L_{M,\Sigma}/\GGauge_\Sigma\subset \underline{C}_\Sigma$ be the reduction of the evolution relation by gauge symmetry, i.e. $\underline{L}_{M,\Sigma}$ is the space of gauge classes of connections on $\Sigma$ which can be extended as flat connection over all $M$.

\subsubsection{Lagrangian property of $L_{M,\Sigma}$}

\begin{lemma}\label{l4_lm3}
Submanifold $\underline{L}_{M,\Sigma}\subset \MM_\Sigma$ is Lagrangian.
\end{lemma}

\begin{proof}
Fix some $A_\Sigma\in L_{M,\Sigma}$ with $\Tilde A$ a flat extension into $M$. Then the tangent space $T_{[A_\Sigma]}\underline{L}_{M,\Sigma}=\frac{\{\alpha\in \Omega^1(\Sigma,\g)\;\; \mr{s.t.}\;\exists\; \Tilde\alpha\in \Omega^1(M,\g)_{d_{\Tilde A}\mr{-closed}}\;\;\mr{with}\; \alpha=\Tilde\alpha|_\Sigma\}}{\Omega^1(\Sigma,\g)_{d_{A_\Sigma}\mr{-exact}}}=\mr{im}(\Pi)$
the image of the map $\Pi$ in the long exact sequence of cohomology of the pair $(M,\Sigma)$:
\begin{equation}\label{l4_LES}
\cdots \ra  H^1_{d_{\Tilde A}} (M;\g)\xra{\Pi} H^1_{d_{A_\Sigma}}(\Sigma;\g)\xra{\varkappa} H^2_{d_{\Tilde A}}(M,\Sigma;\g) \ra\cdots
\end{equation}
Let us calculate the symplectic complement $\mr{im}(\Pi)^\perp$ in $H^1(\Sigma)$:
\begin{equation}\label{l4_perp}
\mr{im}(\Pi)^\perp=\{[\alpha]\in H^1(\Sigma)\;\;\mr{s.t.}\; \langle [\alpha],\Pi[\Tilde\beta] \rangle_\Sigma=0\;\forall [\Tilde\beta]\in H^1(M) \}
\end{equation}
Note that 
$$\langle [\alpha],\Pi[\Tilde\beta] \rangle_\Sigma = \int_\Sigma\tr \alpha\wedge \Tilde\beta|_\Sigma\underset{\mr{Stokes'}}{=} \int_M d\,\tr \Tilde\alpha \wedge \Tilde \beta =
 \int_M \tr  d_{\Tilde A}\Tilde \alpha \wedge \Tilde \beta-\Tilde\alpha\wedge \underbrace{d_{\Tilde A} \Tilde\beta}_0 
$$
Here $\Tilde\alpha$ is an arbitrary (not necessarily closed) extension of the closed $1$-from $\alpha$ into the bulk of $M$. Note that $d_{\Tilde A}\Tilde \alpha$ is a closed $2$-form on $M$ vanishing on $\Sigma$. The class $[d_{\Tilde A}\Tilde \alpha]$ in relative cohomology $H^2(M,\Sigma)$ is $\varkappa[\alpha]$, by construction of the connecting homomorphism $\varkappa$. Thus, we have
$\langle [\alpha],\Pi[\Tilde\beta] \rangle_\Sigma=\langle \varkappa [\alpha], [\Tilde \beta] \rangle_M$
where $\langle,\rangle_M: H^1(M)\otimes H^2(M,\Sigma)\ra \RR$ is the Lefschetz pairing between relative and absolute cohomology. We then continue the calculation (\ref{l4_perp}):
\begin{multline*}
\mr{im}(\Pi)^\perp=\{[\alpha]\in H^1(\Sigma)\;\;\mr{s.t.}\; \langle \varkappa[\alpha],[\Tilde\beta] \rangle_\Sigma=0\;\forall [\Tilde\beta]\in H^1(M) \}\\
=\{[\alpha]\in H^1(\Sigma)\;\;\mr{s.t.}\;  \varkappa[\alpha]=0\}=\ker\varkappa=\mr{im}(\Pi)
\end{multline*}
Here we used non-degeneracy of the Lefschetz pairing and, in the last step, used exactness of the sequence (\ref{l4_LES}). This finishes the proof that $\underline{L}_{M,\Sigma}\subset \MM_\Sigma$ is Lagrangian. 
\end{proof}

A corollary of this is the following.
\begin{theorem}
Submanifold $L_{M,\Sigma}\subset \Phi_\Sigma$ is Lagrangian.
\end{theorem}

\begin{proof}
Fix $A_\Sigma\in L_{M,\Sigma}$. Denote $\Theta:=T_{A_\Sigma}L_{M,\Sigma}$ and $V:=T_{A_\Sigma}\Phi_\Sigma$. We know by Lemma \ref{l4_lm2} that $\Theta$ is isotropic in $V$, i.e. $\Theta\subset \Theta^\perp$. Let also $U:=T_{A_\Sigma}C_\Sigma$ and $H:=U^\perp\subset U$. We have then a sequence of subspaces 
$$ H\subset \Theta\subset \Theta^\perp \subset U \subset V $$
Note that $\Lambda:=\Theta/H=T_{[A_\Sigma]}\underline L_{M,\Sigma}$
Note that 
\begin{multline*}
\Lambda^\perp = \{[v]\in U/H\;\; \mr{s.t.}\; \underline\omega([v],[\theta])=0\;\forall [\theta]\in \Theta/H \}\\
=\{v\in U\;\;\mr{s.t.}\; \omega(v,\theta)=0\;\forall \theta\in \Theta\}/H=\Theta^\perp/H
\end{multline*}
On the other hand, $\Lambda$ is the space we have proven to be a Lagrangian subspace $U/H=T_{[A_\Sigma]}\underline{C}_\Sigma$ in Lemma \ref{l4_lm3}. Thus 
$$\Theta/H=\Lambda=\Lambda^\perp=\Theta^\perp/H$$
which, in combination with $\Theta\subset \Theta^\perp$, proves $\Theta=\Theta^\perp$.
\end{proof}

\marginpar{\LARGE{Lecture 5, 09/07/2016.}}
\subsubsection{Behavior of $S_{CS}$ under gauge transformations, Wess-Zumino cocycle}
For a manifold $M$ with boundary $\Sigma$, Chern-Simons action changes w.r.t. gauge transformation of a connection in following way (result of a straightforward calculation):
\begin{equation}  \label{l5_e1}
S_{CS}(A^g)-S_{CS}(A)=\int_\Sigma\tr \frac12 g^{-1}A g\wedge g^{-1}dg \underbrace{-\int_M\tr \frac16 (g^{-1}dg)^{\wedge 3}}_{=:W_\Sigma(g)}  
\end{equation}
The last term here is called the \textit{Wess-Zumino term}.

\begin{lemma}
$W_\Sigma(g) \;\bmod 4\pi^2 \ZZ$ depends only on the restriction of $g$ to the boundary, $g|_\Sigma\in \GGauge_\Sigma$.
\end{lemma}

\begin{proof}
Let $M'$ be a second copy of $M$ and let $\til M=M\cup_\Sigma \overline{M'}$ be the closed $3$-manifold obtained by gluing $M$ and $M'$ along $\Sigma$. Let $g: M\ra G$ and $g': M'\ra G$ be two maps to the group which agree on $\Sigma$, $g|_\Sigma=g'|_\Sigma$. The pair $(g,g')$ determines a map $\til g: \til M\ra G$. We have
$$W_\Sigma(g)-W_\Sigma(g')=-\int_{\til M}\tr \frac16 (\til g^{-1} d\til g)^{\wedge 3} = 4\pi^2 \langle [M] \til g^*[\theta]\rangle \quad \in 4\pi^2\cdot \ZZ$$
where $[\theta]$ is the class of Cartan's $3$-form in $H^3(G)$.
\end{proof}

Denote 
$$c_\Sigma^k(A,g):=e^{\frac{ik}{2\pi}(\int_\Sigma\tr \frac12 g^{-1}A g\wedge g^{-1}dg + W_\Sigma(g))}$$
By the Lemma above, for $k\in \ZZ$, it this is a well-defined function of a pair $(A,g)\in \Conn_\Sigma\times \GGauge_\Sigma$.

In particular, (\ref{l5_e1}) can be rewritten as the gauge transformation rule for the (normalized) exponential of Chern-Simons action $\psi_k(A)=e^{\frac{ik}{2\pi}S_{CS}(A)}$   (which we introduced earlier in the closed case):
\begin{equation}\label{l5_e2}
\psi_k(A^g)=\psi_k(A)\cdot c_\Sigma^k(A|_\Sigma,g|_\Sigma)
\end{equation}

\begin{remark}
$c_\Sigma^k$ can be viewed as a $1$-cocycle in the cochain complex of the group $\GGauge_\Sigma$ acting on $\mr{Map}(\Conn_\Sigma,S^1)$.  Group cocycle property amounts to
$$ \left(g\circ c_\Sigma^k(A,h)\right)\cdot \left( c_\Sigma^k(A,gh)\right)^{-1}\cdot \left( c_\Sigma^k(A,g) \right) =1 $$
(here $\cdot$ refers to the product in abelian group $S^1$ and $g\circ \phi(A)=\phi(A^g)$ is the $\GGauge_\Sigma$ action on the module $\{\phi(A)\}=\mr{Map}(\Conn_\Sigma,S^1)$). This property in turn follows from (\ref{l5_e1}) by exponentiating the obvious relation
$$0=\left(S_{CS}(A^{gh})-S_{CS}(A^g)\right)-\left(S_{CS}(A^{gh})-S_{CS}(A)\right)+\left(S_{CS}(A^g)-S_{CS}(A)\right)$$
\end{remark}

\begin{remark}
The construction of $c_\Sigma^k$ from $\psi_k$ is similar to the \textit{transgression}
in the \textit{inflation-restriction} exact sequence in group cohomology:
$$ \cdots \ra H^{j}(\mc{G}/\mc{N}, \mc{A}^\mc{N} )\ra H^j(\mc{G},\mc{A}) \ra H^j(\mc{N},\mc{A})^{\mc{G}/\mc{N}}\xra{\mathbb{T}}  
H^{j+1}(\mc{G}/\mc{N}, \mc{A}^\mc{N} )\ra \cdots
$$
which holds for $\mc{G}$ a group, $\mc{N}\subset \mc{G}$ a normal subgroup and $\mc{A}$ a $\mc{G}$-module (this exact sequence is related to the \textit{Lyndon-Hochschild-Serre spectral sequence}). In our case, $\mc{G}=\GGauge_M$, $\mc{N}=\{g:M\ra G\;\mr{s.t.}\; g|_\Sigma=1\}$, with the quotient $\mc{G}/\mc{N}\cong \GGauge_\Sigma$; the module is $\mc{A}=\mr{Map}(\Conn_M,S^1)$. In particular, invariants $\mc{A}^\mc{N}$ are the functionals of connections on $M$ which are gauge-invariant w.r.t. gauge transformations \textit{relative to the boundary} (i.e. fixed to $1$ at the boundary). We can view $\psi_k$ as a class in $H^0(\mc{N},\mc{A})$ and $c^k_\Sigma=``\mathbb{T}(\psi_k)"$ as a class in $H^1(\mc{G}/\mc{N},\mc{A}^\mc{N})$.
\end{remark}

Let $\LL_M=S^1\times \Conn_M$ be the trivial circle bundle over $\Conn_M$. We define the action of $\GGauge_M$ on $\LL_M$ by
$$g: (\lambda,A) \mapsto (\lambda\cdot c_\Sigma^k(A|_\Sigma,g|_\Sigma), A^g) $$
with $\lambda\in S^1$.  By (\ref{l5_e2}), $\psi_k$ is a $\GGauge_M$-invariant section of $\LL_M$.

Similarly, on the boundary, we have a trivial bundle $\LL_\Sigma=S^1\times \Conn_\Sigma$ with action of $\GGauge_\Sigma$ defined as
$$g_\Sigma: (\lambda,A_\Sigma) \mapsto (\lambda\cdot c_\Sigma^k(A_\Sigma,g_\Sigma), A_\Sigma^{g_\Sigma}) $$
The $1$-form 
$$\alpha^k_\Sigma=\frac{ik}{2\pi}\underbrace{\int_\Sigma\tr\frac12 A_\Sigma\wedge \delta A_\Sigma}_{\alpha_\Sigma}  \quad \in\;\;  \Omega^1(\Conn_\Sigma,\mathfrak{u}(1))$$ 
defines a $\GGauge_\Sigma$-invariant connection in $\LL_\Sigma$  (here $\mathfrak{u}(1)=i\RR$ is the Lie algebra of $S^1=U(1)$). Its curvature is
$$\omega_\Sigma^k=\frac{ik}{2\pi}\underbrace{\int_\Sigma\tr\frac12 \delta A_\Sigma\wedge \delta A_\Sigma}_{\omega_\Sigma}\quad \in\;\; \Omega^2(\Conn_\Sigma,\mathfrak{u}(1))$$

Exponential of the action $\psi_k$ restricted to flact connections satisfies the following property (instead of beilng locally constant as a function on $\FlatConn_M$ as in the case of $M$ closed):
$$(\delta-\pi^*\alpha^k_\Sigma)\psi_k=0$$
with $\pi:\Conn_M\ra \Conn_\Sigma$ the pullback of connections to the boundary; $\pi^*\alpha^k_\Sigma$ is the pullback of an $S^1$-connection $\alpha^k_\Sigma$ on $\Conn_\Sigma$ to an $S^1$-connection on $\Conn_M$.

\subsubsection{Prequantum line bundle on the moduli space of flat connections on the surface}
Restricting the circle bundle $\LL_\Sigma$ to flat connections and taking the quotient over gauge transformations, we obtain a non-trivial circle bundle $\underline\LL_\Sigma^k$ over the moduli space $\MM_\Sigma$ with connection $\underline\alpha_\Sigma^k$ with curvature $\underline\omega^k_\Sigma=\frac{ik}{2\pi}\underline\omega_\Sigma$ -- a multiple of the standard Atiyah-Bott symplectic structure on the moduli space $\MM_\Sigma$. In fact, $\underline\LL_\Sigma^k=(\underline\LL_\Sigma^1)^{\otimes k}$ (here we implicitly 
identify a circle bundle and the associated complex line bundle $\underline\LL\times_{S^1} \CC$). 
$\underline\LL_\Sigma^1$ is known as the \textit{prequantum line bundle} on the moduli space of flat connections on the surface.

Another point of view on the line bundle $\underline \LL_\Sigma^k$ is as follows. Consider $C_\Sigma$ as a space with $\GGauge_\Sigma$-action with quotient $\MM_\Sigma$. Restriction of the symplectic form $\omega_\Sigma|_{C_\Sigma}$ is a basic form (horizontal and invariant) w.r.t. $\GGauge_\Sigma$ and thus is a pullback of a form $\underline\omega_\Sigma$ on the quotient. 
But $\omega_\Sigma$ before reduction is exact, with primitive $1$-form $\alpha_\Sigma$. The first question is: can we reduce $\alpha_\Sigma$ to a primitive $1$-form for the reduced symplectic structure? The answer is: NO, because $\alpha_\Sigma|_{C_\Sigma}$ is not basic (in particular, not horizontal).\footnote{Also, we could not have succeded in constructing a primitive $1$-form for $\underline{\omega}_\Sigma$ because, being a symplectic structure on a compact manifold (for $G$ compact, $\MM_{\Sigma}$ is also compact), it has to define a nontrivial class in $H^2(\MM_\Sigma)$.} 

The solution is to promote $\alpha_\Sigma$ to a connection $\nabla$ in the trivial circle bundle over $C_\Sigma$, then one can identify the circle fibers along $\GGauge_\Sigma$-orbits on $C_\Sigma$. Locally this identification is consistent because $\nabla$ is flat when restricted to the orbit (since $F_\nabla=\omega_\Sigma$ and orbits are isotropic submanifolds). For the identification to be globally consistent, the holonomy of $\nabla$ on the orbit has to be trivial. This turns out to be true precisely if we normalize the connection $1$-form as $\frac{ik}{2\pi}\alpha_\Sigma=\alpha_\Sigma^k$ with $k$ an integer! The resulting consistent identification of circle fibers along gauge orbits on $C_\Sigma$ yields the circle bundle $\LL_\Sigma^k$ over the moduli space $C_\Sigma/\GGauge_\Sigma=\MM_\Sigma$.

\begin{remark}
The Chern-Weil representative of the first Chern class of $\underline \LL_\Sigma^k$ is 
$\frac{1}{2\pi  i}\omega_\Sigma^k=\frac{k}{4\pi^2}\omega_\Sigma$ - the (normalized) curvature of the connection in $\underline\LL_\Sigma^k$. In particular, this implies that the 2-form $\frac{1}{4\pi^2}\omega_\Sigma$ on the moduli space $\MM_\Sigma$ has integral periods.
\end{remark}

Exponential of the action $\psi_k$ restricted to flat connections, after reduction modulo gauge symmetry, yields a section $\underline\psi_k\in \Gamma(\MM_M,(\pi_*)^*\underline\LL_\Sigma^k)$ which satisfies 
$$(\delta-(\pi_*)^*\underline\alpha_\Sigma^k)\underline\psi_k=0$$
i.e. is horizontal w.r.t. the connection $\underline\LL_\Sigma^k$ pulled back to $\MM_M$ by the map $\pi_*:\MM_M\ra \MM_\Sigma$ sending the gauge class of a connection on $M$ to the gauge class of its restriction to the boundary.

\begin{remark} Existence of a global section $\underline\psi_k$ of the line bundle $(\pi_*)^*\underline \LL_\Sigma^k$ over $\MM_M$ implies, in particular, that the latter is trivial. Put another way, the pullback of the (nontrivial) first Chern class $c_1(\underline\LL_\Sigma^k)\in H^2(\MM_\Sigma)$ by $\pi_*:\MM_M\ra \MM_\Sigma$ is zero.
\end{remark}

\subsubsection{Two exciting formulae 
}

Symplectic volume of the moduli space of flat connections on a surface of genus $h\geq 2$ is given by
\begin{equation}\label{l5_sympVol}
\mr{Vol}(\MM_\Sigma):=\int_{\MM_\Sigma}\frac{(\underline\omega_\Sigma)^{\wedge m}}{m!}=\#Z(G)\cdot (\mr{Vol}(G))^{2h-2}\sum_{R\in\{\mr{irrep\;of\;} G\} } \frac{1}{(\dim R)^{2h-2}}
\end{equation}
where $m:=\frac{1}{2}\dim \MM_\Sigma=\dim G\cdot (h-1)$ and $\#Z(G)$ is the number of elements in the center of $G$; $R$ runs over irreducible representations of $G$ (see \cite{Witten_2Dgauge}).

A related result is the celebrated Verlinde formula for the dimension of the space of holomorphic sections of the line bundle $\underline{\LL}^k_\Sigma$ over $\MM_\Sigma$ (with respect to some a priori chosen complex structure on the surface $\Sigma$ which in turn endows $\MM_\Sigma$ with a complex structure -- and, moreover, makes $\MM_\Sigma$ a K\"ahler manifold). For simplicity, we give the formula for $G=SU(2)$:
\begin{equation}\label{l5_Verlinde}
\dim H^0_{\bar\dd}(\MM_\Sigma,\underline\LL^k_\Sigma)=\left(\frac{k+2}{2}\right)^{h-1}\sum_{j=0}^{k-1}\frac{1}{\left(\sin\frac{\pi (j+1)}{k+2}\right)^{2h-2}}
\end{equation}
This formula gives the dimension of the space of states which \textit{quantum} Chern-Simons theory assigns to the surface $\Sigma$ (see \cite{Witten_CS}). The r.h.s. here is, in fact, a polynomial in $k$ of degree $m=3h-3$ (in case $h\geq 2$), with the coefficient of the leading term given by (\ref{l5_sympVol}). This follows from Riemann-Roch-Hirzebruch formula which gives the following for the dimension of the space of holomorphic sections:
\begin{multline*}
\dim H^0_{\bar\dd}(\MM_\Sigma,\underline\LL^k_\Sigma)=\int_{\MM_\Sigma}  \mr{Td}(\MM_\Sigma)\cdot e^{\frac{k}{2\pi}\underline\omega_\Sigma}\\
=\left(\frac{k}{2\pi}\right)^{m}\mr{Vol}(\MM_\Sigma)+\mbox{polynomial of degree $< m$ in $k$}
\end{multline*}

The sum in (\ref{l5_Verlinde}) runs, secretely, over ``integrable'' irreducible representations of the affine Lie algebra 
$\hat\g$ (with $\g=\mathfrak{su}(2)$ in the case at hand)
at level 
$k$ (resp. 
irreducible representations of the quantum group $SL_q(2)$ with $q=e^{\frac{\pi i}{k+2}}$ a root of unity).

\marginpar{\LARGE{Lecture 6, 09/12/2016.}}
\subsubsection{
Classical field theory as a functor to the symplectic category}
\begin{definition}\footnote{See \cite{Weinstein}.}
Let $(\Phi_1,\omega_1)$ and $(\Phi_2,\omega_2)$ be two symplectic manifolds. A \emph{canonical relation} $L$ between $\Phi_1$ and $\Phi_2$ is a Lagrangian submanifold $L\subset \overline\Phi_1\times \Phi_2$ where $\overline\Phi_1=(\Phi_1,-\omega_1)$ is the \emph{symplectic dual} of $\Phi_1$, i.e. $\Phi_1$ endowed with symplectic structure of opposite sign. The notation is: $L:\Phi_1\raslash\Phi_2$. 
Composition of canonical relations $L:\Phi_1\raslash\Phi_2$ and $L':\Phi_2\raslash\Phi_3$ is defined as the set-theoretic composition of relations:
\begin{multline}\label{l6_comp_of_can_rel}
L'\circ L:=\{(x,z)\in \Phi_1\times \Phi_3\;\;\mr{s.t.}\; \exists y\in \Phi_2 \;\mr{s.t.}\; (x,y)\in L\;\mr{and}\; (y,z)\in L'\}
\\
=P\left((L\times L')\cap (\Phi_1\times \mr{Diag}_{\Phi_2}\times \Phi_3) \right)
\end{multline}
where $\mr{Diag}_{\Phi_2}=\{(y,y)\in \Phi_2\times \Phi_2\}$ -- the diagonal Lagrangian in $\Phi_2\times \overline\Phi_2$ and $P:\overline\Phi_1\times\Phi_2\times\overline\Phi_2\times \Phi_3\ra \overline\Phi_1\times \Phi_3$ is a projection to the outmost factors.
\end{definition}

Composition of canonical relations is guaranteed to be a canonical relation in the context of finite-dimensional symplectic vector spaces. More generally (for symplectic manifolds, possibly infinite-dimensional), the composition is always isotropic but may fail to be Lagrangian, if the intersection in (\ref{l6_comp_of_can_rel}) fails to be transversal. Also, the composition may fail to be smooth.

Thus, we have a symplectic category of symplectic manifolds and canonical relations between them with partially-defined composition. Unit morphisms are the diagonal Lagrangians $\mr{id}_\Phi=\mr{Diag}_\Phi: \Phi\ra \Phi$. The monoidal structure is given by direct products and the monoidal unit is the point (regarded as a symplectic manifold).

For $C\subset (\Phi,\omega)$ a coisotropic submanifold, introduce a special canonical relation $r_C:\Phi\raslash\Phi$ defined as the set of pairs $(x,y)\in C\times C$ such that $x$ and $y$ are on the same leaf of the characterictic distribution on $C$. Note that this relation is an idempotent: $r_C\circ r_C=r_C$ . Also note that for $C=\Phi$, $r_C$ is the identity (diagonal) relation on $\Phi$.

One can formulate an $n$-dimensional classical field theory in the spirit of Atiyah-Segal axiomatics of QFT, as the following association.\footnote{See \cite{CMR2} for an overview of this approach and examples.}
\begin{itemize}
\item 
To an $(n-1)$-manifold $\Sigma$ (possibly with geometric structure), the classical field theory assigns a symplectic manifold $(\Phi_\Sigma,\omega_\Sigma)$ -- the \textit{phase space}.\footnote{The idea of construction of the phase space from variational calculus data (fields and action functional) of the field theory is to first construct $\Phi^\mr{pre}_\Sigma$ as normal $\infty$-jets of fields at $\Sigma$ on some manifold $M$ containing $\Sigma$ as a boundary component. Thus, tautologically, one has a projection $\pi^\mr{pre}:F_M\ra \Phi^\mr{pre}_\Sigma$ -- evaluation of the normal jet of a field at $\Sigma$. By integrating by parts in the variation of action $\delta S_M$, one gets the pre-N\"other $1$-form $\alpha^\mr{pre}_\Sigma\in \Omega^1(\Phi_\Sigma^\mr{pre})$. Setting $\omega^\mr{pre}_\Sigma$, one performs the symplectic reduction by the kernel of $\omega^\mr{pre}_\Sigma$. Phase space $\Phi_\Sigma$ is the result of this reduction. By construction, it comes with a symplectic structure and a projection $\pi: F_M\ra \Phi_\Sigma$.}
\item To an $n$-cobordism $\Sigma_\mr{in}\xra{M}\Sigma_\mr{out}$, the classical field theory assigns a canonical relation
$L_M:\Phi_{\Sigma_\mr{in}} \raslash
\Phi_{\Sigma_\mr{out}}$.\footnote{The idea is to consider the space $EL_M$ of solutions of Euler-Lagrange equations on $M$ (as defined by the bulk term of the variation of action $\delta S_M$), and to construct $L_M:=(\pi_\mr{in}\times \pi_\mr{out})(EL_M)\subset \overline\Phi_{\Sigma_\mr{in}}\times \Phi_{\Sigma_\mr{out}} $ -- the set of boundary values of solutions of Euler-Lagrange equations. \textbf{Warning:} though it is automatic that $L_M$ is isotropic, fact that is Lagrangian has to be proven for individual field theories and there exist (pathological) examples where Lagrangianity fails, e.g. 2-dimensional scalar field on Misner's cylinder \cite{CMwave}.}
\item Composition (gluing) of cobordisms $\Sigma_1\xra{M}\Sigma_2\xra{M'}\Sigma_3$ is mapped to the set-theoretic composition of relations $\Phi_{\Sigma_1}\xraslash{L_M}\Phi_{\Sigma_2}\xraslash{L_{M'}}\Phi_{\Sigma_3}$.
\item Disjoint unions are mapped to direct products.
\item Null $(n-1)$-manifold is mapped to the point as its phase space.
\item A short\footnote{In a topological theory, we can think of a unit cylinder $\Sigma\times [0,1]$ and in a theory e.g. with cobordisms endowed with metric, we should think of taking a limit $\epsilon\ra 0$.} cylinder $\Sigma\times [0,\epsilon]$ is mapped to the relation $r_{C_\Sigma}: \Phi_\Sigma\raslash\Phi_\Sigma$ for some distinguished coisotropic $C_\Sigma\subset \Phi_\Sigma$ -- the \textit{Cauchy subspace}.
\end{itemize}

Thus, from this point of view, a classical field theory, similarly to quantum field theory, is a functor of monoidal categories from the category of $n$-cobordisms (possibly with geometric structure) to the symplectic category. With two corrections:
\begin{itemize}
\item The target category has only partially defined composition. On the other hand, if we know that the space of solutions of Euler-Lagrange equations induces a Lagrangian submanifold in the phase space on the boundary for any spacetime manifold $M$ (which is the case in all but pathological examples), then we know that there is no problem with composition of relations in the image of cobordisms under the given field theory.
\item Units do not go to units (if we deal with a gauge theory; for a non-degenerate/unconstrained theory, we have $C_\Sigma=\Phi_\Sigma$ and then units do go to units). One can then pass to a \textit{reduced} field theory, by replacing phase spaces $\Phi_\Sigma$ with coisotropic reductions $\underline C_\Sigma=:\Phi^\mr{reduced}_\Sigma$ and replacing relations $L_M$ with respective reduced relations $L_M^\mr{reduced}:=\underline L_M:\; \underline C_{\Sigma_\mr{in}}\raslash \underline C_{\Sigma_\mr{out}}$ (pushforwards of $L_M$ along the coisotropic reduction). The reduced theory is a functor to the symplectic category and takes units to units, but there may be a problem with reductions not being smooth manifolds.
\end{itemize}

\begin{example}[Non-degenerate classical mechanics] This is a 1-dimensional classical field theory. A point with positive orientation $\mr{pt}^+$ is mapped to some symplectic manifold $\Phi$ and $\mr{pt}^-$ is mapped to the symplectic dual $\overline \Phi$. An interval $[t_0,t_1]$ (our cobordisms are equipped with Riemannian metric and so have length) is mapped to a relation $L_{[t_0,t_1]}:\Phi\raslash \Phi$. Using the gluing axiom, by the argument similar to Example \ref{l1_exQM} (where we considered quantum mechanics as an example of Segal's axioms), we have that 
$$L_{[t_0,t_1]}=\{(x,y)\subset \Phi\times \Phi\;\;\mr{s.t.}\;y=\mr{Flow}_{t_1-t_0}(X)\circ x\}$$
-- the graph of the flow, in  time $t_1-t_0$, of a vector field $X$ on $\Phi$ preserving symplectic structure. If $\Phi$ is simply connected, $X$ has to be a Hamiltonian vector field, $X=\{H,\bt\}_\omega$ for some Hamiltonian $H\in C^\infty(\RR)$.
On the other hand $L_{[t_0,t_1]}$ is constructed out of the action of the classical field theory (e.g. in the case of second-order Lagrangian, $S[x(\tau)]=\int_{t_0}^{t_1}d\tau \left(\frac{\dot{x}^2}{2m}-U(x(\tau))\right)$) as
$$L_{[t_0,t_1]}=\{(x,y)\in \Phi\times \Phi \;\; |\;\;\exists\; \mr{sol.\;of\; EL\; eq.\;} x(\tau)\;\mr{s.t.}\; x(t_0)=x,\; x(t_1)=y\}$$
In particular, evolution in infinitesimal time relates the Lagrangian density and the Hamiltonian (the relation being the Legendre transform).
\end{example}

\begin{example}[Classical Chern-Simons theory] Classical Chern-Simons theory as we discussed it here is the prototypical example of a functorial classical field theory, with $n=3$, $\Phi_\Sigma=\Conn_\Sigma$, $C_\Sigma=\FlatConn_\Sigma$ and $L_M=\mr{im}(\FlatConn_M\ra \overline\Conn_{\Sigma_\mr{in}}\times \Conn_{\Sigma_\mr{out}})$.
\end{example}

\section{Feynman diagrams}
Here we will discuss how Feynman diagrams arise in the context of finite-dimensional integrals. References: \cite{Etingof,Reshetikhin}.

\subsection{Gauss and Fresnel integrals}
Gauss integral:\footnote{Sometimes also called Poisson integral.}
\begin{equation}\label{l6_e0}
\int_{-\infty}^\infty dx\; e^{-x^2}=\sqrt\pi\quad \mbox{or more generally}\;\; \int_{-\infty}^\infty dx\; e^{-\alpha x^2}=\sqrt\frac\pi\alpha 
\end{equation}
with $\mr{Re}\,\alpha>0$ needed for absolute convergence. Multi-dimensional version:
\begin{equation}\label{l6_e1}
\int_{\RR^n}d^n x\; e^{-Q(x,x)}=\pi^{\frac{n}{2}}(\det Q)^{-\frac12}
\end{equation}
Here $Q(x,x)=\sum_{i,j=1}^n Q_{ij} x_i x_j$ is a \textit{positive-definite} (as necessary for convergence) quadratic form and $\det Q$ stands for the determinant of the matrix $(Q_{ij})$.\footnote{\label{l6_footnote} (\ref{l6_e1}) is proven e.g. by making an orthogonal change of coordinates on $\RR^n$ which diagonalizes $Q$; then the integration variables split and the problem is reduced to a product of 1-dimensional Gaussian integrals.}

Fresnel integral is the oscillating version of Gauss integral:
\begin{equation}\label{l6_e2}
\int_{-\infty}^\infty dx \; e^{ix^2}=\sqrt\pi\cdot e^{\frac{i\pi}{4}},\qquad \int_{-\infty}^\infty dx \; e^{-ix^2}=\sqrt\pi\cdot e^{-\frac{i\pi}{4}} 
\end{equation}
To calculate e.g. the first one, one way is to take the limit $\alpha\ra -i$ in (\ref{l6_e0}). Equivalently, one views it as an integral over the real line in the complex plane $\RR\subset \CC$ and rotates the integration contour counterclockwise $\RR\ra e^{\frac{\pi i}{4}}\cdot \RR \subset \CC $. On the new contour, the integrand becomes the standard Gaussian integrand (not oscillating but decaying). 
Note that we could not have rotated the contour clockwise because then the integral would have diverged.

Fresnel integrals are only conditionally convergent, as opposed to Gaussian integrals which are absolutely convergent.

Multi-dimensional Fresnel integral:
\begin{equation}\label{l6_e3}
\int_{\RR^n}d^n x\; e^{i\, Q(x,x)}=\pi^{\frac{n}{2}}\cdot e^{\frac{\pi i}{4}\mr{sign}\,Q}\cdot |\det Q|^{-\frac12}
\end{equation}
with $Q(x,x)=\sum_{i,j=1}^n Q_{ij}x_i x_j$ a \textit{non-degenerate} quadratic form (not required to be positive-definite); $\mr{sign}\,Q$ is the \textit{signature} of $Q$ -- the number of positive eigenvalues minus the number of negative eigenvalues. (Proven as in footnote \ref{l6_footnote}, by diagonalization of $Q$).

\begin{remark}[On convergence of Fresnel integrals]\label{l6_rem_convergence} Although one-dimensional integral can be made sense of as a limit of integrals with cut-off integration domain, $\lim_{\Lambda\ra\infty} \int_{-\Lambda}^\Lambda dx\; e^{ix^2}$ (the cut-off integral oscillates as a function of $\Lambda$ but the amplitude of oscillation goes to zero as $\Lambda\ra\infty$), in the higher-dimensional case there are problems. E.g. if $Q=x_1^2+\cdots+x_n^2$, then cutting-off the integration domain to a ball of radius $\Lambda$, we obtain
$\int_{||x||^2<\Lambda}d^n x\; e^{i ||x||^2}\propto \int_0^\infty d\Lambda \, \Lambda^{n-1}e^{i\Lambda^2}$ -- here the amplitude of oscillations in $\Lambda$ does not decrease for $n=2$ and actually increases for $n\geq 3$. The solution is to say that the limit $\Lambda\ra \infty$ exists not pointwise, but in the distributional sense, i.e. convolving with a smoothing function $\rho$ (that is, replacing $\lim_{\Lambda_\ra\infty }\int^\Lambda\cdots$ with $\lim_{\Lambda_0\ra\infty} \int d\Lambda\; \rho(\frac{\Lambda}{\Lambda_0}) \int^\Lambda\cdots $). This is equivalent to replacing an abrupt cut-off of the integration domain by ``smeared cut-off'' (e.g. multiplying the integrand by a bump function which realizes the smeared cut-off). A technically convenient way of arranging a smeared cut-off is simply to multiply the integrand by $e^{-\epsilon Q_0(x,x)}$ for some fixed positive-definite $Q_0$, and then take the limit $\epsilon\ra 0$. I.e. the meaning of the integral (\ref{l6_e3}) is:
$$\lim_{\epsilon\ra 0} \int_{\RR^n}d^n x \; e^{i Q(x,x)-\epsilon\, Q_0(x,x)}$$
The integral now is absolutely convergent $\forall \epsilon >0$; the result is independent of $Q_0$ and is equal to the r.h.s. of (\ref{l6_e3}). Note that, for $Q$ of diagonal (Morse) form $Q=\sum_{i=1}^p x_i^2-\sum_{i=p+1}^n x_{i}^2$, our regularization is equivalent to infinitesimally rotating the integration contour for $x_i$ counterclockwise for $i=1,\ldots,p$ and clockwise for $i=p+1,\ldots,n$.
\end{remark}

\subsection{Stationary phase formula}
\begin{theorem}\label{thm: statphase}
Let $X$ be an oriented $n$-manifold, $\mu\in \Omega^n_c(X)$ a top-degree form with compact support, $f\in C^\infty(X)$ a smooth function which has only non-degenerate
critical points $x_0^{(1)},\ldots, x_0^{(m)}$ on $\mr{Supp}\;\mu\subset X$. Then the integral $I(k):=\int_X \mu\; e^{ik f(x)}$ has the following asymptotics at $k\ra\infty$:
\begin{equation}\label{l6_statphase}
I(k)\sim\sum_{x_0\in \{\mr{crit.\;points\;of\;}f\}}e^{ik f(x_0)} \left(\frac{2\pi}{k}\right)^{\frac{n}{2}} |\det f''(x_0)|^{-\frac12}\cdot e^{\frac{\pi i}{4}\mr{sign}\,f''(x_0)}\cdot \mu_{x_0}+ O(k^{-\frac{n}{2}-1})
\end{equation}
Here:
\begin{itemize}
\item We assume that we have chosen, arbitrarily, a coordinate chart $(y_1,\ldots,y_n)$ near each critical point $x_0$.
\item Critical point $x_0$ of $f$ is said to be non-degenerate if the Hessian matrix $f''(x_0)=\left.\frac{\dd^2}{\dd y_i \dd y_j}\right|_{y=0} f$ is non-degenerate. (In particular, a non-degenerate critical point has to be isolated and therefore there can be only finitely many of them on the compact $\mr{Supp}\,\mu$.)
\item $\mu_{x_0}$ is the density of $\mu$ at $x_0$ in local coordinates $y_1,\ldots,y_n$. I.e., if $\mu$ is written in local coordinates as $\mu=\rho(y)dy_1\cdots dy_n$ for some local density $\rho(y)$, then $\mu_{x_0}:=\rho(y=0)$.
\end{itemize}
\end{theorem}

\begin{remark}\label{l6_rem_Jac}
Note that, although the Hessian $f''(x_0)$ and the density of $\mu$ at a critical point depend on the choice of local coordinates near $x_0$, this dependence cancels out in the r.h.s. of (\ref{l6_statphase}): if we change the coordinate chart $(y_1,\ldots,y_n)\mapsto (y_1',\ldots,y_n')$, then $\det f''(x_0)$ changes by the square of the Jacobian of the transformation at $y=0$  (we assume that charts are centered at $x_0$), and $\mu_{x_0}$ changes by the Jacobian. Thus the product $|\det f''(x_0)|^{-\frac12}\cdot \mu_{x_0}$ is, in fact, invariant.
\end{remark}

\begin{lemma}\label{l6_lm1}
Let $g\in C^\infty_c(\RR)$ a compactly-supported function on $\RR$ and let 
$$I(k):=\int_{-\infty}^\infty dx\; g(x)e^{ikx}$$ Then $I(k)$ decays faster than any power of $k$ as $k\ra \infty$,
$$I(k)\underset{k\ra \infty}{\sim} O(k^{-\infty})$$
In other words, for any $N$ there exists  some $C_N\in \RR$ such that $ |k^N I(k)|\leq C_N$.
\end{lemma}

\begin{proof}
We have
$$k^N\cdot I(k)=\int_{-\infty}^\infty dx \, g(x)\left(-i\frac{\dd}{\dd x}\right)^N e^{ikx} \underset{\mr{Stokes'}}{=}\int_{-\infty}^\infty dx\, e^{ikx} \left(i\frac{\dd}{\dd x}\right)^N g(x)$$
In the second step we have integrated by parts $N$ times, removing derivatives from the exponential and putting them on $g$.
The integral on the r.h.s. is certainly bounded by $\int_{-\infty}^\infty dx\; |\dd^N g(x)|=: C_N$. This proves the Lemma.
\end{proof}

\begin{lemma} \label{l6_lm2}
Let $g\in C^\infty_c(\RR^n)$ and let $f\in C^\infty(\RR^n)$ with no critical points on $\mr{Supp}\;g\subset \RR^n$. Then $$I(k):=\int_{\RR^n}d^n x\; g(x)e^{ik f(x)}\underset{k\ra\infty}\sim O(k^{-\infty})$$
\end{lemma}

\begin{proof}
Since $f$ has no critical points on $\mr{Supp}\;g$, it defines a submersion $f:\mr{Supp}\;g\ra \RR$. Thus, the pushforward (fiber integral) $f_*(d^n x\,g(x))\in \Omega^1_c(\RR)$ is a smooth 1-form on $\RR$. Thus, we can calculate $I(k)$ by first integrating over the level sets of $f$, $f(x)=y$ (the same as computing the pushforward $f_*$) and then integrating over the values $y$ of $f$:
$$I(k)=\int_\RR e^{iky} f_*(g \,d^n x)$$
This integral behaves as $O(k^{-\infty})$ by Lemma \ref{l6_lm1}.
\end{proof}

\begin{lemma} \label{l6_lm3}
Let $g\in C^\infty_c(\RR^n)$ such that $g$ and its derivatives of all orders vanish at $x=0$. Let $Q(x,x)$ be a non-degenerate quadratic form on $\RR^n$. Then:
$$I(k):=\int_{\RR^n}d^n x\; g(x)\, e^{ik Q(x,x)}\underset{k\ra\infty}\sim O(k^{-\infty})$$ 
\end{lemma}

\marginpar{\LARGE{Lecture 7, 09/14/2016.}}
\begin{proof}
First consider the case when $Q$ is \textit{positive-definite}. 
Then $Q:\RR^n-\{0\}\ra (0,\infty)$ is a submersion; we can calculate $I(k)$, similarly to the proof of Lemma \ref{l6_lm2}, by integrating first over the level sets of $Q$ and then over values $y$ of $Q$:
$$I(k)=\int_0^\infty e^{ik y} Q_*(d^n x\; g(x))$$
The pushforward $Q_*(d^n x\; g(x))\in \Omega^1_c[0,\infty)$ has vanishing $\infty$-jet at $y=0$ (because of the assumption on $\infty$-jet of $g$ at the origin $x=0$). Thus one can repeat the proof of Lemma \ref{l6_lm1} and no boundary terms at $y=0$ will appear when performing integration by parts multiple times. Thus we obtain $I(k)\underset{k\ra\infty}\sim O(k^{-\infty})$.

For $Q$ not positive-definite, we can assume without loss of generality (by making a linear change of coordinates) that $Q$ has Morse form $Q=\sum_{i=1}^p x_i^2-\sum_{i=p+1}^n x_i^2$. We can present $g(x)$ as a limit of finite sums of functions of form $g'(x_1,\ldots,x_p)\cdot g''(x_{p+1},\ldots,x_n)$ (since $ C^\infty_c(\RR^p)\otimes C^\infty_c(\RR^{n-p})$ is dense in $C^\infty_c(\RR^n)$). For such products we have $\int_{\RR^n} d^n x \;g'\cdot g''\, e^{ik Q(x,x)}=\left(\int_{\RR^p}dx_1\cdots dx_p g'(x_1,\ldots,x_p)e^{ik (x_1^2+\cdots+x_p^2)}\right)\cdot \left(\int_{\RR^{n-p}}dx_{p+1}\cdots dx_n g''(x_{p+1},\ldots,x_n)e^{-ik (x_{p+1}^2+\cdots+x_n^2)}\right)\sim O(k^{-\infty})$ by the result in the positive-definite case. One can check that the bound we get is uniform and one can pass to the limit.

\end{proof}

\begin{corollary}\label{l7_cor}
Let $g\in C^\infty_c(\RR^n)$ and let $Q(x,x)$ be a non-degenerate quadratic form on $\RR^n$. Let 
\begin{equation}\label{l7_e1}
I(k):=\int_{\RR^n}d^n x\;g(x)\; e^{ik Q(x,x)}
\end{equation}
Then:
\begin{enumerate}[(i)]
\item \label{l7_i} $I(k)$ modulo $O(k^{-\infty})$-terms depends only on the $\infty$-jet of $g$ at $x=0$.
\item \label{l7_ii} In particular $I(k)=g(0)\cdot\left(\frac{\pi}{k}\right)^{\frac{n}{2}}|\det Q|^{-\frac12}e^{\frac{\pi i}{4}\mr{sign}\, Q}+O(k^{-\frac{n}{2}-1})$
\end{enumerate}

\end{corollary}
\begin{proof}
(\ref{l7_i}) is an immediate consequence of Lemma \ref{l6_lm3}. 

For (\ref{l7_ii}), 
write $g(x)=g(0)+(x,dg(0))+R(x)$ -- a constant term, a linear term (which, being an odd function of $x$, vanishes when integrated with $e^{iQ(x,x)}$), and the ``error term'' which has zero of order two at $x=0$. 
Thus, we have 
$$I(k)=g(0)\cdot \int_{\RR^n}d^n x \; e^{ik Q(x,x)}+\underbrace{\int_{\RR^n}d^n x\; R(x) e^{ik Q(x,x)}}_{r(k)}$$ 
The first term on the r.h.s. is the standard Fresnel integral and we need to show that the error $r(k)$ behaves as $O(k^{-\frac{n}{2}-1})$. Write
$$r(k)=\int_{\RR^n}d^n x\;R(x) e^{ik Q(x,x)}=k^{-\frac{n}{2}-1}\int_{\RR^n}d^n y\; k R\left(\frac{y}{\sqrt{k}}\right) e^{iQ(y,y)}$$
Here we made a change $x=\frac{y}{\sqrt{k}}$. Integrand on the r.h.s. has a well-defined limit as $k\ra\infty$ (since $R$ has a zero of order $2$ at the origin) and converges to $e^{iQ(y,y)}$ times some quadratic form in $y$.\footnote{This is a bit sketchy: one has to explain why integration and limit can be interchanged; see a better argument below - Lemma \ref{l7_lm_afterthought}.} Thus $r(k)$ behaves as $k^{-\frac{n}{2}-1}$ times an integral which converges in the sense of Remark \ref{l6_rem_convergence}.
\end{proof}
The general idea is that in the integral (\ref{l7_e1}) one can replace $g$ with a piece of its Taylor series at the origin and the error will be estimated by the contribution of the first discarded term of the Taylor series (or the next one if the discarded term was of odd degree).

\marginpar{\textbf{An afterthought:} better/cleaner way (instead of Lemma \ref{l6_lm3} and Corollary \ref{l7_cor}).}
\begin{lemma}\label{l7_lm_afterthought} Let $g$ be a Schwartz class function on $\RR^n$, let $g_N$ be the Taylor series for $g$ truncated at $N$-th order for arbitrary $N$, so that $h(x):=g(x)-g_N(x)\underset{x\ra 0}\sim O(x^{N+1})$,  and let $Q(x,x)$ be a non-degenerate quadratic form on $\RR^n$. Then 
\begin{equation}\label{l7_e-1}
I(k):=\int_{\RR^n} d^n x\; h(x) e^{ik Q(x,x)} \underset{k\ra \infty}\sim O(k^{-\frac{n}{2}-\left[\frac{N+2}{2}\right]})
\end{equation}
\end{lemma}
\begin{proof}
Consider the differential operator $\mc{D}=-\frac{i}{2}\sum_{j,k=1}^n  (Q^{-1})_{jk} \frac{1}{x_j}  \frac{\dd}{\dd x_k}$  and its transpose $\mc{D}^T=\frac{i}{2}\sum_{j,k=1}^n (Q^{-1})_{jk} \frac{\dd}{\dd x_j} \frac{1}{x_k}$, acting on functions on $\RR^n$. Operator $\mc{D}$ is constructed so that we have the following property: $\mc{D}\; e^{ik Q(x,x)}=k\cdot e^{ik Q(x,x)}$. Thus, multiplying $I(k)$ by a power of $k$, we have
\begin{equation}\label{l7_e0}
k^m I(k)=\int_{\RR^n} d^n x\; h(x)\, \DD^m e^{ik Q(x,x)}= \int_{\RR^n} d^n x\; e^{ik Q(x,x)} (\DD^{T})^m h(x)
\end{equation}
Where we have integrated by parts $m$ times (we think of point $x=0$ as being punctured out of the integration domain). Note that $(\DD^T)^m h(x)\underset{x\ra 0}\sim O(x^{N+1-2m})$ and thus on the r.h.s. of (\ref{l7_e0}) we get an integrable singularity at the origin iff $N+1-2m>-n$ (e.g. $m=\left[\frac{N+n}{2}\right]$ satisfies this inequality); convergence at infinity holds in the sense of Remark \ref{l6_rem_convergence}. Thus we have proven that $I(k)\sim O(k^{-\left[\frac{N+n}{2}\right]})$. 

This is a slightly weaker estimate than claimed in (\ref{l7_e-1}); one can get the improved estimate considering a truncation of the Taylor series for $g$ three steps further, $g_{N+3}$. Then, by the result that we have proven, 
\begin{equation}\label{l7_e'}
\int_{\RR^n} d^n x\; (g-g_{N+2}) e^{ikQ(x,x)}\sim O(k^{-\left[\frac{N+n+3}{2}\right]})
\end{equation} 
(which is a better or equivalent estimate to the r.h.s. of (\ref{l7_e-1})).
On the other hand $g_N-g_{N+3}$ is a polynomial in $x$ containing monomials of degrees $N+1$, $N+2$ and $N+3$ only. Thus, 
\begin{equation}\label{l7_e''}
\int_{\RR^n}d^n x\; (g_N-g_{N+2})e^{ikQ(x,x)} =C_{N+1}k^{-\frac{n+N+1}{2}}+C_{N+2}k^{-\frac{n+N+2}{2}}+C_{N+3}k^{-\frac{n+N+3}{2}}
\end{equation} 
where the constant $C_{N+j}$ vanishes if $N+j$ is odd for $j=1,2,3$. Thus, (\ref{l7_e'}) and (\ref{l7_e''}) together imply (\ref{l7_e-1}).
\end{proof}
In particular: (\ref{l7_ii}) of Corollary \ref{l7_cor} is the $N=0$ case of (\ref{l7_e-1}).
Also note that Lemma \ref{l6_lm3} is a special case of the new Lemma (for $g$ with vanishing jet at the origin and $N$ arbitrarily large) - here we avoid splitting coordinates into positive and negative eigenspaces of $Q$ (and the painful discussion of approximating $g$ by products) by the trick with the differential operator $\DD$.
\marginpar{\textbf{End of the afterthought.}
}

\begin{proof}[Proof of Theorem \ref{thm: statphase}]
We can assume without loss of generality that $X$ is compact (since we only care about $\mr{Supp}\;g$ anyway which is compact by assumption). Choose a covering $\{U_\alpha\}$ of $X$ by open subsets such that 
\begin{itemize}
\item each $U_\alpha$ contains at most one critical point of $f$,
\item each critical point of $f$ is contained in exactly one $U_\alpha$.
\end{itemize}
Choose a partition of unity $\{\psi_\alpha\in C^\infty(X)\}$ subordinate to the covering $\{U_\alpha\}$, i.e.
\begin{itemize}
\item $\mr{Supp}\; \psi_\alpha\subset U_\alpha$,
\item $\psi_\alpha\geq 0$,
\item $\sum_\alpha \psi_\alpha=1$.
\end{itemize}
Then $I(k)=\sum_\alpha I_\alpha(k)$ with $I_\alpha(k)=\int_{U_\alpha}\mu\psi_\alpha(x)\, e^{ik f(x)}$.
We should consider two case:
\begin{enumerate}[(i)]
\item $U_\alpha$ does not contain critical points of $f$. Then $I_\alpha(k)\sim O(k^{-\infty})$ by Lemma \ref{l6_lm2}.
\item $U_\alpha$ contains a critical point $x_0$ of $f$. By Morse Lemma, we can introduce local coordinates $(y_1,\ldots,y_n)$ on $U_\alpha$ such that $f=f(x_0)+\underbrace{y_1^2+\cdots+y_p^2-y_{p+1}^2-\cdots-y_n^2}_{Q(y,y)}$. Then, by (\ref{l7_ii}) of Corollary \ref{l7_cor} (or by Lemma \ref{l7_lm_afterthought} for $N=0$), we have
\begin{multline*}
I_\alpha(k)=\int_{\RR^n}d^n y\;\rho(y)\;\psi_\alpha\;e^{ikf(x_0)+ik Q(y,y)}\sim \\
\sim \rho(0)e^{ik f(x_0)}\left(\frac{\pi}{k}\right)^{\frac{n}{2}}|\det Q|^{-\frac12} e^{\frac{\pi i}{4}\mr{sign}\,Q}+ O(k^{-\frac{n}{2}-1}) 
\end{multline*}
where $d^n y\,\rho(y)$ is $\mu$ expressed in coordinates $y$. Note that $Q_{ij}=\frac12 \frac{\dd^2}{\dd y_i \dd y_j} f$, thus 
$$I_\alpha(k)\sim \mu_{x_0} e^{ik f(x_0)}\left(\frac{2\pi}{k}\right)^{\frac{n}{2}}|\det f''(x_0)|^{-\frac12}e^{\frac{\pi i}{4}\mr{sign}\,f''(x_0)}+O(k^{-\frac{n}{2}-1})$$
\end{enumerate}
Summing over $\alpha$, we obtain the stationary phase formula for $I(k)$. Note that, by Remark \ref{l6_rem_Jac}, it does not matter that we have chosen the Morse chart around every critical point: the result is independent of this choice.

\end{proof}

\subsection{Gaussian expectation values. Wick's lemma}
Consider normalized expectation values with respect to Gaussian measure
\begin{equation}\label{l7_expectation}
\ll p \gg:= \frac{\int_{\RR^n} d^n x\; e^{-\frac{1}{2} Q(x,x)}\cdot p(x)}{\int_{\RR^n} d^n x\; e^{-\frac{1}{2} Q(x,x)}} 
\end{equation}
with $Q(x,x)=\sum_{i,j}Q_{ij}x_i x_j$ a positive-definite quadratic form on $\RR^n$, for $p(x)$ a polynomial on $\RR^n$.

\begin{definition} For $H$ a finite set with even number of elements we call partitions of $H$ into two-element subsets \emph{perfect matchings} on $H$.
\end{definition}
Note that a perfect matching is the same as an involution $\gamma$ on $H$ with no fixed points. Then the two-element subsets are the orbits of $\gamma$.
\begin{example} On the set $\{1,2,3,4\}$ there exist three different perfect matchings:
$$\{1,2\}\cup\{3,4\},\qquad \{1,3\}\cup \{2,4\},\qquad \{1,4\}\cup \{2,3\}$$
\end{example}
More generally, on the set of order $2m$ there are $(2m-1)!!=1\cdot 3\cdot 5\cdots (2m-1)$ perfect matchings.\footnote{Indeed, the first element of the set has to be matched with one of $2m-1$ other elements, first element among those left has to be matched with one of $(2m-3)$ remaining elements etc.}

The following lemma allows one to calculate the expectation $\ll p \gg$ for any monomial (and hence every polynomial) $p$.
\begin{lemma}[``Wick's lemma'']\footnote{The original Wick's lemma, though a similar statement, was formulated in the context of expressing words constructed out of creation and annihilation operators in terms of normal ordering.}\label{l7_lm_Wick}
\begin{enumerate}[(i)]
\item \label{l7_Wick (i)} $\ll 1\gg=1$.
\item \label{l7_Wick (ii)} $\ll x_{i_1}\cdots x_{i_{2m-1}} \gg=0$.
\item \label{l7_Wick (iii)} $\ll x_i x_j \gg=(Q^{-1})_{ij}$ -- the $(i,j)$-th matrix element of the inverse matrix to the matrix of the quadratic form $Q(x,x)$.
\item \label{l7_Wick (iv)} 
\begin{multline}\label{l7_Wick}
\ll x_{i_1}\cdots x_{i_{2m}}\gg=\\
=\sum_{\mr{perfect\; matchings}\;\{1,\ldots,2m\}=\{a_1,b_1\}\cup \cdots\cup \{a_m,b_m\}} \underbrace{\ll x_{i_{a_1}}x_{i_{b_1}}\gg}_{(Q^{-1})_{i_{a_1} i_{b_1}}}\cdot\; \cdots\ \cdot \underbrace{\ll x_{i_{a_m}} x_{i_{b_m}} \gg}_{(Q^{-1})_{i_{a_m} i_{b_m}}}
\end{multline}
\end{enumerate}
\end{lemma}

\begin{remark} We can identify perfect matchings on the set $H=\{1,\ldots,2m\}$ with elements of the quotient of the symmetric group $S_{2m}$ of permutations of $H$ by the group of permutations of two-element subsets constituting the partition and transpositions of the elements inside the two-element subsets. In other words, the set of perfect matchings can be presented as $S_{2m}/(S_m\ltimes \ZZ_2^m)$. Thus, in particular, expectation value (\ref{l7_Wick}) can be written as
\begin{equation}\label{l7_Wick2}
\ll x_{i_1}\cdots x_{i_{2m}}\gg=\sum_{\sigma\in S_{2m}/(S_m\ltimes \ZZ_2^m)}\ll x_{i_{\sigma_1}i_{\sigma_2}} \gg\cdot \;\cdots\cdot \ll x_{i_{\sigma_{2m-1}}i_{\sigma_{2m}}} \gg
\end{equation}
\end{remark}

\begin{example}
$$\ll x_i x_j x_k x_l \gg= \ll x_i x_j\gg\cdot \ll x_k x_l\gg+ \ll x_i x_k \gg\cdot \ll x_j x_l\gg+ \ll x_i x_l \gg\cdot \ll x_j x_k \gg$$
Pictorially, the three terms on the r.h.s. can be drawn as follows:
\begin{equation}   \label{l7_Wick_visualization}
\vcenter{\hbox{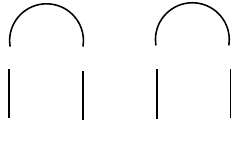}}\qquad \qquad\qquad
\vcenter{\hbox{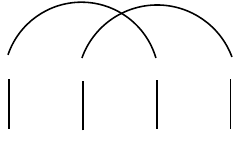}}\qquad \qquad \qquad
\vcenter{\hbox{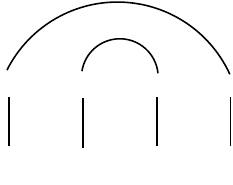}}
\end{equation}
\end{example}

\begin{example} From the count of perfect matchings and Wick's formula, we deduce, for 1-dimensional moment of Gaussian measure $dx\,e^{-\frac{x^2}{2}}$, that
$$ \ll x^{2m} \gg= (2m-1)!! $$
or equivalently
$$ \int_{-\infty}^\infty dx\; e^{-\frac{x^2}{2}} x^{2m}=\sqrt{2\pi}\cdot (2m-1)!! $$
\end{example}

\begin{proof}[Proof of Lemma \ref{l7_lm_Wick}]
Item (\ref{l7_Wick (i)}) is obvious, and (\ref{l7_Wick (ii)}) also (integrand in the numerator of (\ref{l7_expectation}) is odd with respect to $x\ra -x$, hence the integral is zero). For (\ref{l7_Wick (iii)}) and (\ref{l7_Wick (iv)}), consider an auxiliary integral
\begin{equation}\label{l7_e2}
W(J):=\int_{\RR^n} d^n x\; e^{-\frac12 Q(x,x)+\langle J,x \rangle}
\end{equation}
with $J\in \RR^n$ the \textit{source}. The integral is easily calculated by completing the expression in the exponential to the full square:
\begin{multline}\label{l7_e3}
W(J)=\int_{\RR^n} d^n x\; \underbrace{e^{-\frac12 Q(x,x)+\langle J,x \rangle - \frac12 \langle J, Q^{-1} J \rangle }}_{e^{-\frac12 Q( x-Q^{-1}J,x-Q^{-1}J ) }}\cdot e^{\frac12 \langle J, Q^{-1} J \rangle }=\\
=e^{\frac12 \langle J, Q^{-1} J \rangle }\cdot \int_{\RR^n} d^n y \; e^{-\frac12 Q(y,y)}=e^{\frac12 \langle J, Q^{-1} J \rangle }\cdot (2\pi)^{\frac{n}{2}}\left(\det Q\right)^{-\frac12}
\end{multline}
Here in the second step we made a shift $x\mapsto y=x-Q^{-1}J$. 

From definition (\ref{l7_e2}), we have
\begin{multline}\label{l7_e4}
\ll x_{i_1}\cdots x_{i_{2m}} \gg=\frac{1}{W(0)}\left|\frac{\dd}{\dd J_{i_1}}\cdots \frac{\dd}{\dd J_{i_{2m}}} W(J)\right|_{J=0} =\\
=\frac{1}{2^m m!}\;\frac{\dd}{\dd J_{i_1}}\cdots \frac{\dd}{\dd J_{i_{2m}}}\; \underbrace{Q^{-1}(J,J)\cdot\; \cdots \; \cdot Q^{-1}(J,J)}_{\left(\sum_{j_1,k_1}Q^{-1}_{j_1 k_1}J_{j_1} J_{k_1}\right)\cdots \left(\sum_{j_m,k_m}Q^{-1}_{j_m k_m}J_{j_m} J_{k_m}\right) }
\end{multline}
Here in the second step we replaced $W(J)$ by $m$-th term of the Taylor series for the exponential in the explicit formula (\ref{l7_e3}) for $W(J)$ (lower terms do not contribute because they are killed by the $2m$ derivatives in the source $J$ and higher terms do not contribute as tehy are killed by setting $J=0$ after taking the derivatives). Then (\ref{l7_Wick (iv)}) follows by evaluating the multiple derivative in the source in (\ref{l7_e4}) by Leibniz rule. Item (\ref{l7_Wick (iii)}) is the trivial $m=1$ case of this computation.

\end{proof}

\begin{remark}\label{l7_rem_abs}
In a slightly more invariant language, replace $\RR^n$ by an abstract finite-dimensional $\RR$-vector space $V$. Our input is a positive-definite quadratic form $Q\in \mr{Sym}^2 V^*$. We are interested in the map $\ll-\gg:\; \mr{Sym} V^*\ra \RR$ sending 
$$p\quad \mapsto \quad  \ll p \gg=\frac{\int_V \mu\; e^{-\frac12 Q}\; p}{\int_V \mu\; e^{-\frac12 Q}}$$ 
with $\mu\in \wedge^\mr{top} V^*$ a fixed constant volume form (irrelevant for the normalized expectation values). Then the Wick's lemma (\ref{l7_Wick2}) can be formulated as
\begin{equation}\label{l7_e5}
\ll \phi_1\odot \cdots\odot \phi_{2m} \gg= \sum_{\sigma\in S_{2m}/(S_m\ltimes \ZZ_2^m)}\langle \sigma\circ (Q^{-1})^{\otimes m},\phi_1\otimes\cdots\otimes \phi_{2m} \rangle
\end{equation}
Here $\phi_1,\ldots,\phi_{2m}\in V^*$ are linear functions on $V$, $\odot$ is the commutative product in $\mr{Sym} V^*$. We understand the inverse to $Q$ as an element in the symmetric square of $V$, $Q^{-1}\in \mr{Sym}^2 V$; $\sigma$ acts on $V^{\otimes 2m}$ by permuting the copies of $V$; the pairing in the r.h.s. is the pairing between $V^{\otimes 2m}$ and $(V^*)^{\otimes 2m}$
\end{remark}

\begin{remark} Another visualization (as opposed to (\ref{l7_Wick_visualization})) of the terms on the r.h.s. of Wick's lemma, corresponding to the presentation (\ref{l7_e5}) is like as follows:
$$\vcenter{\hbox{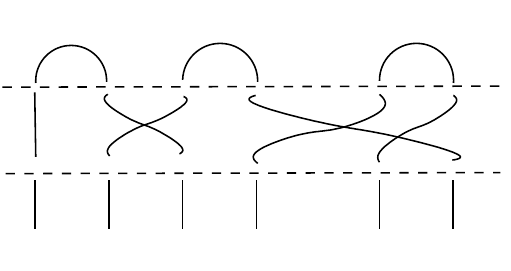}}$$
Here the lower strip presents $\phi_1\otimes\cdots\otimes \phi_{2m}\in (V^*)^{\otimes 2m}$, the upper strip presents pairing with $(Q^{-1})^{\otimes m}\in V^{\otimes 2m}$ and middle strip presents the action of $\sigma$ by permuting the $V$-factors (if we read the diagram from top to bottom), or equivalently the action of $\sigma^{-1}$ by permuting $V^*$-factors (if we read te diagram from bottom to top).
\end{remark}

\marginpar{\LARGE{Lecture 8, 09/19/2016.}}
\begin{remark}\label{l8_rem_inv_Gaussian} In the setup of Remark \ref{l7_rem_abs}, the value of the Gaussian integral itself, $\int_V \mu\, e^{-\frac12 Q}$, can be understood as follows (without referring explicitly to the matrix of $Q$ or, in other words, without identifying bilinears on $V$ with endomorphisms). To $Q\in \Sym^2 V^*$, there is an associated sharp map $Q^\#: V\ra V^*$. Raising it to the maximal exterior power, we obtain a map of determinant lines $\wedge^n Q^\#: \wedge^n V\ra \wedge^nV^*$ (with $n=\dim V$) or equivalently, dualizing the domain line and putting it to the right side, 
$\mr{Det}\, Q:= \wedge^nQ^\#\in (\wedge^n V^*)^{\otimes 2}$.\footnote{
Here we implicitly used the identification $(\wedge ^nV)^*\cong \wedge^n V^*$. It is induced by the  pairing $\wedge^n V\otimes  \wedge^n V^*\ra \RR$ which sends $(v_1\wedge\cdots\wedge v_n)\otimes (\theta_1\wedge\cdots\wedge \theta_n)\mapsto \det \Big(\langle  v_i,\theta_j \rangle\Big)_{i,j=1}^n$, where on the r.h.s.$\langle , \rangle$ is the canonical pairing between $V$ and $V^*$.
} Thus, $\mr{Det}\,Q$ in this context is not a number, but an element of the line $(\wedge^n V^*)^{\otimes 2}$. 
(Whenever a basis in $V$ is chosen, we have a trivialization $(\wedge^n V^*)^{\otimes 2}\simeq \RR$, and then $\mr{Det}\,Q$ gets assigned the number value, which coincides with the determinant of the matrix of the bilinear $Q$ in the chosen basis). Note that $\mu^{\otimes 2}$ is a nonzero element of the same line, thus we can form a quotient $\frac{\mr{Det}\,Q}{\mu^{\otimes 2}}\in \RR$. Value of the Gaussian integral is then
$$ \int_V \mu\, e^{-\frac12 Q}=(2\pi)^{\frac{n}{2}} \left(\frac{\mr{Det}\,Q}{\mu^{\otimes 2}}\right)^{-\frac12} $$
\end{remark}

\subsection{A reminder on graphs and graph automorphisms}
\begin{definition}
A graph is the following set of data: 
\begin{itemize}
\item A set $V$ of \textit{vertices}.
\item A set  $HE$ of \textit{half-edges}.
\item A map $i: HE\ra V$ -- \textit{incidence}.
\item A perfect matching $E$ on $HE$, i.e. a partition of $E$ into two-element subsets -- \textit{edges}. Put differently, we have a fixed-point-free involution $\gamma$ on $HE$ and its orbits are the edges.
\end{itemize} 
\end{definition}
We will only consider \textit{finite} graphs, i.e. with $V$ and $HE$ finite.
Here is a picture of a generic graph.
$$\vcenter{\hbox{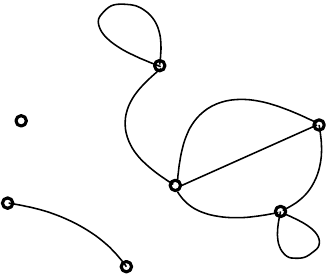}}$$ 

\begin{definition}
For $v\in V$ a vertex, one calls $i^{-1}(v)\subset HE$ the \textit{star} (or \textit{corolla}) of $v$ and the number of incident halh-edges to the vertex, $\# i^{-1}(v)$, is called the \textit{valency} of $v$.
\end{definition}

\begin{definition}
For two graphs $\Gamma=(V,HE,i,E)$, $\Gamma'=(V',HE',i',E')$, a \emph{graph isomorphism} $\Gamma\stackrel{\sim}{\ra}\Gamma'$ is a pair of bijections $\sigma_{V}:V\simra V'$, $\sigma_{HE}: HE\simra HE'$ commuting with the incidence maps (satisfying $i'\circ \sigma_{HE}=\sigma_V\circ i$) and preserving the partition into edges (i.e. $\gamma'\circ \sigma_{HE}=\sigma_{HE}\circ \gamma$ with $\gamma,\gamma'$ the respective involutions on half-edges).
\end{definition}

\begin{example}
Vertices: $V=\{a,b,c\}$, half-edges: $HE=\{1,1',2,2',3,3'\}$, incidence: 
$$i:\quad \begin{array}{lcl} 1 & \mapsto &  a \\ 1' & \mapsto & a \\ 2 & \mapsto & b \\ 2' & \mapsto & b \\ 3 & \mapsto & c \\ 3' & \mapsto & c  \end{array}$$
\end{example}
Edges: $E=\{1,2'\}\cup \{2,3'\}\cup \{3,1'\}$. Equivalently, the invloution is $\gamma: 1\leftrightarrow 2', 2 \leftrightarrow 3', 3\leftrightarrow 1'$. Here is the picture:
$$\vcenter{\hbox{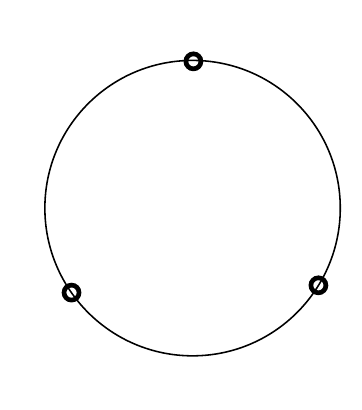}}$$ 
Example of an automorphism of this graph:
$$\sigma_V:\;(a,b,c)\mapsto (b,a,c),\quad \sigma_{HE}:\; (1,1',2,2',3,3')\mapsto (2',2,1',1,3',3) $$
(Check explicitly that this pair of permutations commutes with incidence maps and with involutions!) On the picture of the graph above, this automorphism corresponds to reflection w.r.t. the vertical axis.

We will be interested in the group of automorphisms $\mr{Aut}(\Gamma)$ of a graph $\Gamma$.

\begin{example}[Automorphism groups] 
\begin{enumerate}[(i)] 
\item A ``polygon graph'' with $n\geq 3$ vertices and $n$ edges:
$$\vcenter{\hbox{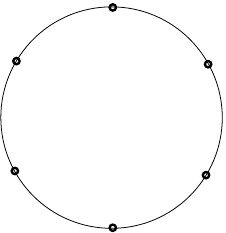}}$$ 
Automorphism group: $\mr{Aut}(\Gamma)=\ZZ_2\ltimes \ZZ_n$.
\item ``Theta graph'':
$$\vcenter{\hbox{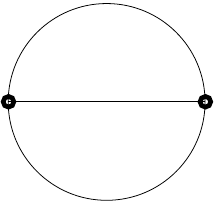}}$$ 
Automorphism group: $\mr{Aut}(\Gamma)=\ZZ_2\times S_3$.
\item ``Figure-eight graph'':
$$\vcenter{\hbox{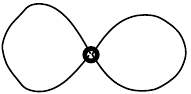}}$$ 
Automorphism group: $\mr{Aut}(\Gamma)=\ZZ_2\ltimes (\ZZ_2\times \ZZ_2)$.
\end{enumerate}
\end{example}

\begin{remark} A graph automorphism has to preserve valencies of vertices, in particular it permutes vertices of any given valency and maps the star of a source vertex to the star of a target vertex (via some permutation). Therefore, for a graph $\Gamma$ which has $V_d$ vertices of valency $d$ for $d=0,\ldots,D$,  the automorphism group can be seen as a subgroup of permutations of vertices for each valency $d$ and permutations of incident half-edges for each vertex:
$$\mr{Aut}(\Gamma)\subset \prod_{d=0}^D S_{V_d}\ltimes S_d^{\times V_d}$$
\end{remark}

\begin{remark} Graphs naturally form a groupoid, with morphisms given by graph isomorphisms. Consider the restriction $\mr{Graph}_{V_0,\ldots,V_D}$ of this groupoid to graphs with number of vertices of valency $d$ fixed to $V_d$ for $d=0,\ldots,D$ (and no vertices of higher valency than $D$). One can realize objects of  $\mr{Graph}_{V_0,\ldots,V_D}$ as all $(2m-1)!!$ (for $2m=\sum_{d=1}^D d\cdot V_d $) perfect matchings on the set of half-edges constituting the given vertex stars. The total group of isomorphisms is then $\prod_{d=0}^D S_{V_d}\ltimes S_d^{\times V_d}$. Thus the groupoid volume of $\mr{Graph}_{V_0,\ldots,V_D}$ is:
\begin{equation}\label{l8_vol}
\mr{Vol}\left(\mr{Graph}_{V_0,\ldots,V_D}\right) =\underbrace{\sum_{\Gamma} \frac{1}{|\mr{Aut}(\Gamma)|}}_{\mr{Vol\;\pi_0 \left(\mr{Graph}_{V_0,\ldots,V_D} \right)}} = \frac{(2m-1)!!}{\prod_{d=0}^D V_d! \cdot d!^{V_d}}
\end{equation}
where $\Gamma$ runs over \textit{isomorphism classes} of graphs; on the r.h.s. the numerator and denominator are the numbers of objects and morphisms of $\mr{Graph}_{V_0,\ldots,V_D}$, respectively.
\end{remark}

\begin{remark}
One can also define graphs as 1-dimensional CW complexes. From this point of view, the automorphism group of $\Gamma$ is $\pi_0$ of the group of cellular homeomorphisms of $\Gamma$ viewed as a CW complex.
\end{remark}

\subsection{Back to integrals: Gaussian expectation value of a product of homogeneous polynomials}
Fix $Q\in \Sym^2 V^*$ a positive-definite quadratic form on $V=\RR^n$. Let $\Psi_a \in \Sym^{d_a} V^*$ for $a=1,\ldots,r$ be a collection of homogeneous polynomials of degrees $d_1,\ldots,d_r$ on $V$. In coordinates, $\Psi_a=\sum_{i_1,\ldots,i_{d_a}=1}^n (\Psi_a)_{i_1\cdots i_{d_a}} x_{i_1}\cdots x_{i_{d_a}}$.  Consider the Gaussian expectation value $\ll \frac{1}{d_1!}\Psi_1\cdots \frac{1}{d_r!}\Psi_r\gg$. Denote $2m=\sum_{a=1}^r d_a$. Also denote $$\mr{Matchings}_{2m}:=S_{2m}/(S_m\ltimes \ZZ_2^m)$$ 
the set of perfect mathcings on $2m$ elements. We have the following:
\begin{multline*}
\ll \frac{1}{d_1!}\Psi_1\cdots \frac{1}{d_r!}\Psi_r\gg = \frac{1}{d_1!\cdots d_r!}\sum_{\sigma\in 
\match_{2m}
} \langle \sigma\circ (Q^{-1})^{\otimes m}, \Psi_1\otimes \cdots\otimes \Psi_r \rangle\\
=\sum_{[\sigma]\in (\prod_{a=1}^r S_{d_a})\backslash\; 
\match_{2m}
} \frac{1}{|\mr{Stab}_{[\sigma]}|} \langle \sigma\circ (Q^{-1})^{\otimes m}, \Psi_1\otimes \cdots\otimes \Psi_r \rangle
\end{multline*}
Here in the first step we have applied the Wick's lemma to calculate the Gaussian expectation value and in the second step we collected similar terms in the sum. In the second sum $[\sigma]$ runs over classes of perfect matchings under the action of $\prod_{a=1}^r S_{d_a}\subset S_{2m}$ (in other words, $[\sigma]$ is a class in the two-sided quotient of the symmetric group, $[\sigma]\in (\prod_{a=1}^r S_{d_a})\backslash\, S_{2m}/(S_m\ltimes \ZZ_2^m)$). This action is not free and has stabilizer subgroups $\mr{Stab}_{[\sigma]}\subset \prod_{a=1}^r S_{d_a}$. Note that the coefficient $\frac{1}{|\mr{Stab}_{[\sigma]}|}$ arises as 
$$\frac{1}{|\mr{Stab}_{[\sigma]}|}=\frac{\#\{\mr{orbit\;of\;}\sigma\; \mr{under}\;S_{d_1}\times\cdots\times S_{d_r}\mr{-action}\}}{|S_{d_1}\times\cdots\times S_{d_r}|}$$
where the denominator is $d_1!\cdots d_r!$.

\begin{example}
Let $\Psi=\sum_{i,j,k,l=1}^n \Psi_{ijkl}\; x_i x_j x_k x_l \in \Sym^4 V^*$ be a quartic polynomial. Then we have
\begin{multline*}
\ll \frac{1}{4!}\Psi\gg = \frac{1}{4!} \sum_{\sigma\in 
\match_4} 
\ll \sigma\circ (Q^{-1})^{\otimes 2}, \Psi  \gg= \\
=\frac{1}{4!}\left( \lan \vcenter{\hbox{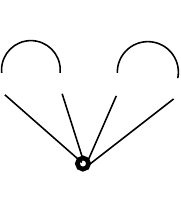}} \ran +
 \lan \vcenter{\hbox{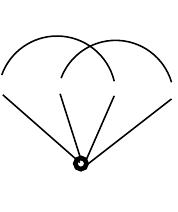}} \ran +
  \lan \vcenter{\hbox{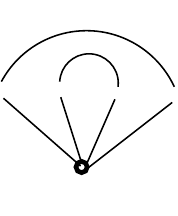}} \ran 
\right)
\\= \frac{3}{4!} \left\langle\vcenter{\hbox{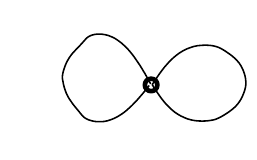}}\quad\;\;\right\rangle = \frac{1}{8}\sum_{i,j,k,l=1}^n \Psi_{ijkl}\; (Q^{-1})_{ij} (Q^{-1})_{kl}
\end{multline*}
Here all three matchings give the same contribution to the expectation value (correspondingly, $S_4\backslash \match_4 \ni [\sigma]$ consists of a single class).
\end{example}

\begin{example}
Let $\Psi_1=\sum_{i,j,k=1}^n(\Psi_1)_{ijk}\; x_i x_j x_k ,\;\Psi_2=\sum_{i',j',k'=1}^n (\Psi_2)_{i'j'k'} \; x_{i'} x_{j'} x_{k'}\;\in \Sym^3 V^*$ be two cubic polynomials. Then we have
\begin{multline*}
\ll \frac{1}{3!}\Psi_1 \cdot \frac{1}{3!}\Psi_2 \gg= \frac{1}{3!\, 3!} \sum_{\sigma\in\match_6} \ll \sigma\circ (Q^{-1})^{\otimes 3}, \Psi_1\otimes \Psi_2 \gg=\\
=\frac{1}{3!\, 3!} \left(
\lan \vcenter{\hbox{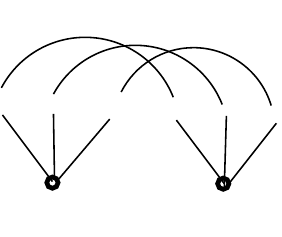}} \ran +\mbox{5 similar terms}  +
\lan \vcenter{\hbox{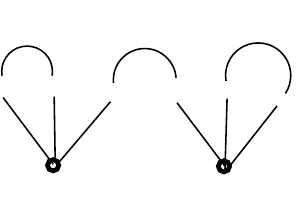}} \ran +\mbox{8 similar terms} 
\right)
\\=\frac{6}{3!3!} \lan \vcenter{\hbox{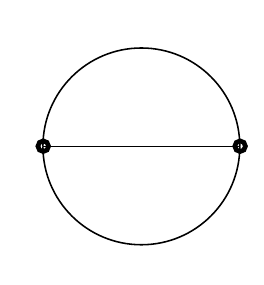}}\quad \ran +
\frac{9}{3!3!} \lan \vcenter{\hbox{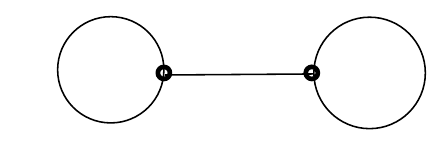}}\quad\;\;  \ran
\\
=\frac{1}{6} \sum_{i,j,k,i',j',k'=1}^n (\Psi_1)_{ijk} (\Psi_2)_{i'j'k'} (Q^{-1})_{ii'} (Q^{-1})_{jj'} (Q^{-1})_{kk'}+\\
+\frac14 \sum_{i,j,k,i',j',k'=1}^n (\Psi_1)_{ijk} (\Psi_2)_{i'j'k'} (Q^{-1})_{ii'} (Q^{-1})_{jk} (Q^{-1})_{j'k'}
\end{multline*}
\end{example}
Here $(S_3\times S_3)\backslash \match_6\ni [\sigma]$ consists of two different classes: 
\begin{itemize}
\item one with 6 representatives in $\match_6$ (i.e. with stabilizer subgroup of order $\frac{3!\, 3!}{6}=6$), corresponding to the ``theta graph'';
\item the second with 9 representatives in $\match_6$ (i.e. with stabilizer subgroup of order $\frac{3!\, 3!}{9}=4$), corresponding to the ``dumbbell graph''.
\end{itemize}

\subsection{Perturbed Gaussian integral}\label{sec: perturbed Gaussian integral}
Fix again $Q(x,x)$ a positive-definite quadratic form on $V=\RR^n$. We are interested in the integrals of form
\begin{equation}\label{l8_pert_Gaussian_1}
\int_V d^n x \; e^{-\frac12 Q(x,x)+p(x)}
\end{equation}
with $p$ a small polynomial perturbation of the quadratic form in the exponential. More precisely,  consider the perturbation $p(x)$ of the  form
\begin{equation}\label{l8_polynomial_perturbation}
p(x)=\sum_{d=0}^D\frac{g_d}{d!} P_d(x)
\end{equation}
with $D$ some fixed degree, $P_d=\sum_{i_1,\ldots,i_d=1}^n (P_d)_{i_1\cdots i_d}x_{i_1}\cdots x_{i_d}\in \Sym^d V^*$ a homogeneous polynomial of degree $d$, and $g_0,\ldots,g_D$ -- infinitesmial formal parameters (``coupling constants''). Note that then the exponential of the perturbation $e^{p(x)}$ is a formal power series in the couplings $g_0,\ldots,g_D$ where the coefficient of each monomial $g_0^{v_0}\cdots g_D^{v_D}$  is a finite-degree polynomial in $x$, i.e.
$$e^{p(x)}\quad\in \Sym V^* \otimes \RR[[g_0,\ldots, g_D]]=\Sym V^*[[g_0,\ldots,g_D]]$$

\begin{definition}
We define the \emph{perturbative evaluation} of the integral (\ref{l8_pert_Gaussian_1}) as follows:
\begin{equation}\label{l8_pert_Gaussian_formal}
\int_V^{\mr{pert}} d^n x \; e^{-\frac12 Q(x,x)+p(x)}:= \underbrace{\left(\int_V d^n x\; e^{-\frac12 Q(x,x)}\right)}_{(2\pi)^{\frac{n}{2}}(\det Q)^{-\frac12}}\cdot \ll e^{p(x)} \gg
\end{equation}
where the symbol $\ll e^{p(x)} \gg$ is to be understood as the evaluation on $e^{p(x)}\in \Sym V^*[[g_0,\ldots,g_D]]$ of the Gaussian expectation value $\ll \cdots \gg: \Sym V^* \ra \RR$, extended by linearity to a map $\ll \cdots \gg: \Sym V^*[[g_0,\ldots,g_D]] \ra \RR$.
\end{definition}

\begin{remark}
Perturbative integral (\ref{l8_pert_Gaussian_formal}) is well-defined for any perturbation $p(x)$ of form (\ref{l8_polynomial_perturbation}), while  (\ref{l8_pert_Gaussian_1}) as a measure-theoretic integral may fail to exist for non-zero coupling constants. E.g. the integral 
$$\int_\RR dx\; e^{-\frac{x^2}{2}+\frac{\alpha}{3!} x^3}$$ 
diverges for any non-zero coefficient $\alpha=g_3$ (except for the case of $\alpha\in i\cdot\RR$ purely imaginary), while 
$$\int_\RR dx\; e^{-\frac{x^2}{2}+\frac{\lambda}{4!}  x^4}$$
converges for $\lambda=g_4$ negative (or, more generally, 
for $\mr{Re}\,\lambda\leq 0$
) and diverges for $\lambda$ positive (resp. $\mr{Re}\, \lambda >0$).
\end{remark}

\marginpar{\LARGE{Lecture 9, 09/26/2016.}}

\begin{definition} Let $\Gamma$ be a graph (``Feynman diagram''). Fix a collection of symmetric tensors (the ``Feynman rules''):
\begin{itemize}
\item The ``propagator'' $$\eta=\sum_{i,j=1}^n \eta_{ij}\, e_i\odot e_j \quad \in\Sym^2 V$$
with $\{e_i\}$ the standard basis in $\RR^n$ (or, more abstractly, \emph{a} basis in $V$).
\item ``Vertex functions\footnote{Or, more appropriately, ``vertex tensors''.} for vertices of valency $d$'', 
$$p_d=\sum_{i_1,\ldots ,i_d=1}^n (p_d)_{i_1\cdots i_d}\, x_{i_1}\cdots x_{i_d}\quad \in \Sym^d V^*$$
for $d=0,\ldots,D$; $\{x_i\}$ is the basis in $V^*$ dual to $\{e_i\}$.
\end{itemize}
We define the \emph{Feynman weight} (or the ``value of the Feynman diagram'') of $\Gamma$ as 
$$\frac{1}{|\mr{Aut}(\Gamma)|}\Phi_{\eta; p_0,\ldots,p_D 
}(\Gamma)$$
where $\Phi_{\eta; p_0,\ldots,p_D}(\Gamma)$ is defined as the following \emph{state sum}.
\begin{itemize}
\item We define a \emph{state} $s$ on $\Gamma$ as a decoration of all half-edges of $\Gamma$ by numbers in $\{1,\ldots,n\}$.
\item To a state $s: HE\ra \{1,\ldots,n\}$ we assign a \emph{weight}
$$w_s:= \prod_{\mr{edges}\; e=(h,h')}\eta_{s(h)s(h')}\times \prod_{\mr{vertices}\; v} (p_d)_{s(h_1)\cdots s(h_d)}$$
In the first product, $h,h'$ are the two constituent half-edges of the edge $e$. In the second product, $d$ is the valency of the vertex $v$ and $h_1,\ldots,h_d$ are the half-edges adjacent to $v$.
\item We define $\Phi$ as the sum over states on $\Gamma$:
$$\Phi_{\eta; p_0,\ldots,p_D}(\Gamma):=\sum_{\mr{states}\; s:\, HE\ra \{1,\ldots,n\}} w_s $$
\end{itemize}
\end{definition}

\begin{example} Consider $\Gamma$ the theta-graph; we label the half-edges by $\{A,B,C,D,E,F\}$:
$$\vcenter{\hbox{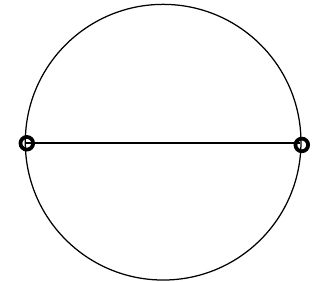}}$$
A state $s$ on $\Gamma$ maps half-edges to numbers $s:(A,B,C,D,E,F)\mapsto (i,j,k,i',j',k')$ each of which can take values from $1$ to $n$:
$$\vcenter{\hbox{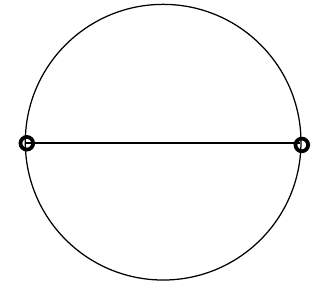}}$$
The weight of the state is:
$$w_s= \eta_{ii'} \eta_{jj'}\eta_{kk'}\times (p_3)_{ijk} (p_3)_{i'j'k'}$$
And thus the Feynman value of the theta graph is
$$\frac{1}{12}\sum_{i,j,k,i',j',k'=1}^n \eta_{ii'} \eta_{jj'}\eta_{kk'}\times (p_3)_{ijk} (p_3)_{i'j'k'} $$
\end{example}

\begin{theorem}[Feynman]\label{l9_Feyn_thm}
For $Q$ a positive-definite quadratic form on $V=\RR^n$ and $p(x)=\sum_{d=0}^D \frac{g_d}{d!}P_d(x)$ a polynomial perturbation with homogeneous terms $P_d\in \Sym^d V^*$, the perturbative evaluation of the integral (\ref{l8_pert_Gaussian_1}) is given by the sum over all finite graphs (up to graph isomorphism) of their Feynman weights:
\begin{equation}\label{l9_Feyn}
\int^\mr{pert}_V d^n x \; e^{-\frac12 Q (x,x)+p(x)}=(2\pi)^{\frac{n}{2}}(\det Q)^{-\frac12} \sum_{\mr{graphs}\; \Gamma}\frac{1}{|\mr{Aut}(\Gamma)|} \Phi_{Q^{-1};g_0 P_0,\ldots, g_D P_D}(\Gamma)
\end{equation}
\end{theorem}

\begin{proof}
By definition (\ref{l8_pert_Gaussian_formal}), we need to compute the Gaussian expectation value $\ll e^{p(x)} \gg$. Writing $e^p=\prod_{d=1}^D e^{\frac{g_d}{d!}P_d(x)}$ and expanding each exponential in Taylor series, we obtain
\begin{multline}
\ll e^p \gg=\ll \prod_{d=1}^D e^{\frac{g_d}{d!}P_d(x)} \gg =\sum_{v_0,\ldots,v_D\geq 0} \prod_{d=1}^D
\frac{g_d^{v_d}}{v_d!\, d!^{v_d}} \ll P_0(x)^{v_0}\cdots P_D(x)^{v_D} \gg \\ 
\underset{\mr{Wick's\; lemma}}{=}
\sum_{v_0,\ldots,v_D\geq 0} \frac{g_0^{v_0}\cdots g_D^{v_D}}{|\mc{V}_{v_0\ldots v_D}|} \sum_{\sigma\in \match_{2m}} \lan \sigma\circ (Q^{-1})^{\otimes m}, \otimes_{d=0}^D P_d^{\otimes v_d} \ran
\cdot \lan \sigma\circ (Q^{-1})^{\otimes m}, \otimes_{d=0}^D P_d^{\otimes v_d} \ran\\
\end{multline}
Here we denoted $\mc{V}_{v_0\cdots v_D}=\prod_{d=0}^D S_{v_d}\ltimes (S_d)^{\times v_d}$ -- group of ``vertex symmetries'' which we understand as a subgroup of $S_{2m}$ with $2m=\sum_{d=0}^D d\, v_d$. The picture is that for each $d=0,1,\ldots, D$, we have $v_d$ of $d$-valent stars decorated with $P_d$ (the vertex tensors); thus, in total, we have $2m=\sum_{d=0}^D d\, v_d$ half-edges. Then we attach $m$ edges decorated by $Q^{-1}$ according to all possible perfect matchings $\sigma$ of half-edges. The sum over matchings contains many similar terms, collecting which we get:
\begin{multline*}
\ll e^p \gg=\\
=\sum_{v_0,\ldots,v_D\geq 0} g_0^{v_0}\cdots g_D^{v_D} \sum_{[\sigma] \in \mc{V}_{v_0\cdots v_D}\backslash \match_{2m}} \frac{|\mr{orbit\; of\;}\sigma\;\mr{in}\;\match_{2m}\;\mr{under}\; \mc{V}_{v_0\cdots v_D} |}{|\mc{V}_{v_0\cdots v_D}|} \cdot \\
\cdot \lan \sigma\circ (Q^{-1})^{\otimes m}, \otimes_{d=0}^D P_d^{\otimes v_d} \ran
\end{multline*}
Equivalence classes of matchings 
$$[\sigma] \in \mc{V}_{v_0\cdots v_D}\backslash \match_{2m}= \left(\prod_{d=0}^D S_{v_d}\ltimes (S_d)^{\times v_d}\right) \backslash\, S_{2m}\,/(S_m\ltimes \ZZ_2^{\times m})$$ 
are in bijection with isomorphism classes of graphs with $v_0$ of $0$-valent vertices, \ldots, $v_D$ of $D$-valent vertices; the weight of the class $[\sigma]$ is easily seen to be the Feynman weight of the corresponding graph:
$$ \ll e^p \gg=\sum_{v_0,\ldots,v_D\geq 0} g_0^{v_0}\cdots g_D^{v_D} \sum_{\mr{graphs}\; \Gamma\;\mr{with}\; v_d\;d-\mr{valent\; vertices},\,d=0,\ldots,D} \frac{1}{|\mr{Aut}(\Gamma)|}\Phi_{Q^{-1},\{P_d\}_{d=0}^D}(\Gamma) $$
We can absorb $g_d$-factors into the normalization of vertex tensors, getting the r.h.s. of (\ref{l9_Feyn}).
\end{proof}

\begin{example} The contribution of the following graph 
$$\vcenter{\hbox{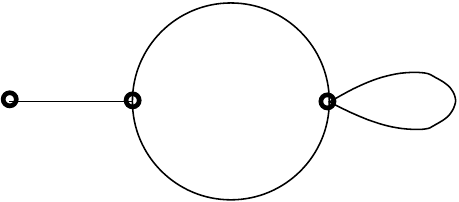}}$$
to the r.h.s. of (\ref{l9_Feyn})is:
\begin{multline*}
\frac{g_1\, g_3\, g_4}{|\mr{Aut}|}\;\;\Phi\left( \vcenter{\hbox{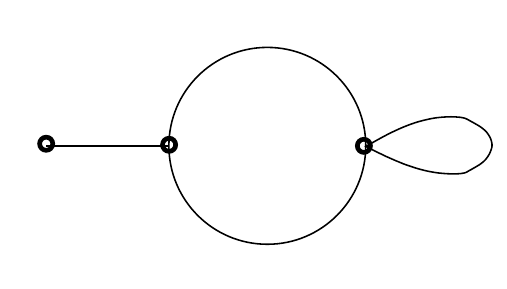}}\quad\;\;  \right) =\\
= \frac{g_1\, g_3\, g_4}{4}\sum_{i,j,k,l,m,n,o,p=1}^{n}(Q^{-1})_{kl} (Q^{-1})_{im} (Q^{-1})_{jn} (Q^{-1})_{op}\times (P_3)_{ijk} (P_1)_l (P_4)_{mnop} 
\end{multline*}
\end{example}

\begin{remark}
We can see the sum over graphs in the r.h.s. of (\ref{l9_Feyn}) as the volume of the groupoid of graphs with standard groupoid measure $\frac{1}{|\mr{Aut}(\Gamma)|}$ on objects (graphs) deformed by Feynman rules to $\frac{1}{|\mr{Aut}(\Gamma)|}\Phi_{Q^{-1},\{g_d\, P_d\}}(\Gamma)$.
\end{remark}

\begin{example}\label{l9_exa_quartic_pert}
Consider 
\begin{equation}\label{l9_e1}
I(\lambda)=\int_\RR dx\; e^{\frac{x^2}{2}+\frac{\lambda}{4!} x^4}
\end{equation}
By (\ref{l9_Feyn}), the perturbative evaluation yields the sum over $4$-valent graphs:
\begin{multline} \label{l9_e1.5}
\int_\RR^\mr{pert} dx\; e^{\frac{x^2}{2}+\frac{\lambda}{4!} x^4} =\sqrt{2\pi}  
\sum_{4\mr{-valent\; graphs\;}\Gamma} \frac{\lambda^{\#\mr{vertices}}}{|\mr{Aut}(\Gamma)|}\\
=
\sqrt{2\pi} \left(1+\frac18 \lambda+ \left(\frac{1}{2\cdot 8^2}\lambda^2+\frac{1}{2\cdot 4!}\lambda^2+\frac{1}{16}\lambda^2\right)+\cdots \right)
\end{multline}
The first contributing graphs here are: the empty graph, $\vcenter{\hbox{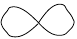}}$, $\vcenter{\hbox{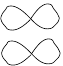}}$, $\vcenter{\hbox{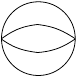}}$, $\vcenter{\hbox{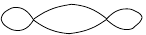}}$. Note that, using (\ref{l8_vol}), we can evaluate the total coefficient of $\lambda^n$:
\begin{equation}\label{l9_e2}
 \int_\RR^\mr{pert} dx\; e^{\frac{x^2}{2}+\frac{\lambda}{4!} x^4} =\sqrt{2\pi} \sum_{n=0}^\infty \lambda^n \frac{(4n-1)!!}{n!\, 4!^n} 
\end{equation}
Coefficients of this power series in $\lambda$ grow super-exponentially (roughly, as $n!$), therefore the convergence radius in $\lambda$ is zero! On the other hand, for $\lambda = -\nu <0$ the integral (\ref{l9_e1}) converges, as a usual measure-theoretic integral, to the function
\begin{equation}\label{l9_e3}
\sqrt{\frac{3}{\nu}}\cdot e^{\frac{3}{4\nu}} K_{\frac14}\left(\frac{3}{4\nu}\right)
\end{equation}
where $K_\alpha(x)=\int_0^\infty dt \; e^{-x \cosh t}\cosh(\alpha t)$ is the modified Bessel's function.
The relation between formal power series (\ref{l9_e2}) and the measure-theoretic evaluation(\ref{l9_e3}) is that the former is the \emph{asymptotic series} for the latter at $\lambda=-\nu \ra -0$ (i.e. $\lambda$ approaching zero along the negative half-axis).
\end{example}

\begin{definition}
Let $\phi(z)\in C^\infty(0,\infty)$ a function on the open positive half-line and let $f_n(z)\in C^\infty(0,\infty)$ be a collection of functions for $n=0,1,\ldots$. One says that $\sum_n f_n(z)$ is a \emph{Poincar\'e asymptotic series} for the function $\phi(z)$ at $z=0$ (notation: $\phi(z)\underset{z\ra 0}\sim  \sum_n f_n(z)$) if:
\begin{enumerate}[(i)]
\item $\phi(z)-\sum_{n=0}^N f_n(z)\underset{z\ra 0}\sim O(f_{N+1}(z))$ for any $N\geq 0$ and
\item $f_{n+1}(z)\underset{z\ra 0}\sim  o(f_n(z))$ for any $n\geq 0$, i.e. $\lim_{z\ra +0} \frac{f_{n+1}(z)}{f_n(z)}=0$.
\end{enumerate}

\end{definition}

\marginpar{\LARGE{Lecture 10, 09/28/2016.}}

\subsubsection{Aside: Borel summation}

Introduce an operation which assigns to a power series $f(z)=\sum_{n\geq 0} a_n z^n$ a new power series $\mc{B}f(t):=\sum_{n\geq 0} \frac{a_n}{n!} t^n$. 

We can recover $f(z)$ from $\mc{B}f(t)$ by certain integral transform $\mathbb{T}$ (the Laplace transform, up to a change of variable):
$$ \mathbb{T}(\mc{B} f)(z):=\int_0^\infty dt\; e^{-t} \mc{B}f(tz) \qquad = \sum_{n\geq 0} \frac{a_n}{n!}\underbrace{\int_0^\infty dt\; e^{-t} (tz)^n}_{n! z^n}=f(z)  $$

Note that the map $f(z)\mapsto \mc{B}f(t)$ improves convergence properties: if $f(z)$ has finite convergence radius in $z$, then $\mc{B}f(t)$ is an entire function in $t$.

Borel's summation method amounts to taking a possibly divergent series as $f(z)$ (e.g. with zero convergence radius); then $\mc{B}f(t)$ can still be convergent (possibly, with a finite convergence radius but possessing an analytic continuation). Then one can define $f_\mr{Borel}(z)$ -- the \emph{Borel summation} of $f(z)$, as a function which can be evaluated for nonzero $z$, rather than just a formal power series, as $\mathbb{T}(\mc{B}f)$.

\begin{example}
Consider the power series $f(z)=\sum_{n\geq 0} (-1)^n n! z^n$ -- it clearly has zero convergence radius in $z$. We have $\mc{B}f(t)=\sum_{n\geq 0} (-1)^n t^n$ -- this power series converges to $\frac{1}{1+t}$  with convergence radius $1$ and extends to an analytic function in $t\in\CC\backslash\{-1\}$. Thus, the Borel summation of $f(z)$ is:
$$f_\mr{Borel}(z):=\mathbb{T}\left(\frac{1}{1+t}\right)=\int_0^\infty dt\; e^{-t}\frac{1}{1+tz}=z^{-1}\,e^{z^{-1}} E_1(z^{-1})$$
where $E_1(x)=\int_x^\infty ds\; \frac{e^{-s}}{s}$ is the \emph{exponential integral}.
\end{example}

\textbf{General fact:} Original power series $f(z)$ is the asymptotic series for the Borel summation $f_\mr{Borel}(z)$ at $z\ra 0$.

In application to perturbative integral, the idea is that one may be able to recover the value of the integral at finite value of coupling constants from the perturbation series by means of Borel summation (which is particularly interesting for path integrals where a direct measure theoretic definition at finite coupling constants/Planck constant is not accessible and one only has the perturbative expansion).

If $F(z)$ is a function and $f(z)=\sum_{n\geq 0}a_n z^n$ is the asymptotic series for $F$ at $z\ra 0$ then under some assumptions it is guaranteed that the Borel summation of $f(z)$ gives back $F(z)$ (i.e. the question is when is the function uniquely determined by its asymptotic expansion).

\begin{theorem}[Watson]
Assume that, for some positive constants $R,\varkappa,\epsilon,c$, we have the following:
\begin{itemize}
\item $F(z)$ is holomorphic in the region 
$$\DD:=\{z\in \CC\;|\;\;|z|<R,|\arg(z)|< \varkappa\frac{\pi}{2}+\epsilon\} $$
\item In this region $F(z)$ is ``well approximated'' by its asymptotic series $f(z)$:
$$\left|F(z)-\sum_{n=0}^{N-1} a_n z^n\right|< c^N (\varkappa n)!\; z^N$$
\end{itemize}
Then, in the region $\DD$, we $F(z)$ coincides with Borel summation of its asymptotic series $f(z)=\sum_{n\geq 0}a_n z^n$.
\end{theorem}

\begin{example}
Function $F(z)=e^{-\frac1z}$ has zero asymptotic series $f(z)=0$ and thus cannot be recovered by Borel summation of $f(z)$. On the other hand, $F(z)$ fails the assumptions of Watson's theorem for any value of $\varkappa$. (Check this!)
\end{example}

\subsubsection{Connected graphs}

It turns out, one can reformulate the r.h.s. of Feynman's formula (\ref{l9_Feyn}) in terms of summation over \textit{connected} graphs only.
\begin{theorem} For a positive-definite quadratic form $Q$ and a polynomial perturbation $p(x)=\sum_{d=0}^D\frac{g_d}{d!}P_d(x)$ as in Theorem \ref{l9_Feyn_thm}, we have
\begin{equation}\label{l10_Feyn_connected}
\int_V^\mr{pert} d^n x\; e^{-\frac12 Q(x,x)+p(x)}=(2\pi)^{\frac{n}{2}}(\det Q)^{-\frac12}\cdot \exp\left( \sum_{\mr{connected\; graphs\;}\gamma} \frac{1}{|\mr{Aut}(\gamma)|}\Phi_{Q^{-1},\{g_d P_d\}}(\gamma)\right)
\end{equation}
\end{theorem}

\begin{proof}
Note that any graph $\Gamma$ can be uniquely split into connected components:
\begin{equation}\label{l10_e0}
\Gamma=\gamma_1^{\sqcup r_1}\sqcup\cdots\sqcup \gamma_k^{\sqcup r_k} 
\end{equation}
where $\gamma_1,\ldots,\gamma_k$ are pairwise non-isomorphic connected graphs and $r_1,\ldots, r_k$ are multiplicities with which they appear in the graph $\Gamma$. Automorphisms of $\Gamma$ are generated by automorphisms of individual connected components and permutations of connected components of same isomorphism type:
\begin{equation}\label{l10_e1}
\mr{Aut}(\Gamma)=\prod_{i=1}^k S_{r_i}\ltimes \mr{Aut}(\gamma_i)^{\times r_i}
\end{equation}

Choose some total ordering on the set of isomorphism classes of connected graphs.
Let us calculate $\exp \sum_{\gamma\;\mr{connected}}\frac{1}{|\mr{Aut}(\gamma)|}\Phi(\gamma)$ by expanding the exponential in the Taylor series:
\begin{multline}\label{l10_e2}
\exp \sum_{\gamma\;\mr{connected}}\frac{1}{|\mr{Aut}(\gamma)|}\Phi(\gamma)=\prod_{\gamma\mr{\; connected}}\sum_{r=0}^\infty\frac{1}{|\mr{Aut}(\gamma)|^r\, r!}\Phi(\gamma)^r =\\ 
=\sum_{k=0}^\infty\; \sum_{\gamma_1<\ldots<\gamma_k}\;\sum_{r_1,\ldots,r_k=1}^\infty 
\frac{1}{\prod_{i=1}^k r_i! |\mr{Aut}(\gamma_i)|^{r_i}} \Phi(\gamma_1)^{r_1}\cdots \Phi(\gamma_k)^{r_k} \\
= \sum_{k=0}^\infty\; \sum_{\gamma_1<\ldots<\gamma_k}\;\sum_{r_1,\ldots,r_k=1}^\infty 
\frac{1}{|\mr{Aut}(\Gamma)|}\Phi(\Gamma)
\end{multline}
where in the last step we set $\Gamma:=\gamma_1^{\sqcup r_1}\sqcup\cdots\sqcup \gamma_k^{\sqcup r_k} $ and we used (\ref{l10_e1}) and multiplicativity of Feynman state sum on graphs: $\Phi(\Gamma_1\sqcup \Gamma_2)=\Phi(\Gamma_1)\cdot \Phi(\Gamma_2)$. The sum in the final expression in (\ref{l10_e2}) corresponds simply to summing over all $\Gamma$ (by uniqueness of decomposition (\ref{l10_e0})). Thus, we have proven that
\begin{equation}
\exp \sum_{\gamma\;\mr{connected}}\frac{1}{|\mr{Aut}(\gamma)|}\Phi(\gamma)=
\sum_\Gamma \frac{1}{|\mr{Aut}(\Gamma)|}\Phi(\Gamma)
\end{equation}
which, together with Feynman's formula (\ref{l9_Feyn}) implies (\ref{l10_Feyn_connected}).

\end{proof}

\begin{example} Returning to the Example \ref{l9_exa_quartic_pert}, we can now rewrite (\ref{l9_e1.5}) as a sum over \emph{connected} graphs with $4$-valent vertices:
\begin{multline*}
\int_\RR^\mr{pert} dx\; e^{-\frac12 x^2+\frac{\lambda}{4!} x^4}=\sqrt{2\pi}\cdot \exp\left( \sum_{\gamma\;\mr{connected},\;4\mr{-valent}} \frac{\lambda^{\#\mr{vertices}}}{|\mr{Aut}(\gamma)|} \right)\\
= \sqrt{2\pi} \cdot \exp\left( \frac{\lambda}{8}+\frac{\lambda^2}{2\cdot 4!}+\frac{\lambda^2}{16}+\cdots\right)
\end{multline*}
where the first contributing graphs are $\vcenter{\hbox{\input{l9_graph1.pdftex_t}}}$, $\vcenter{\hbox{\input{l9_graph3.pdftex_t}}}$, $\vcenter{\hbox{\input{l9_graph4.pdftex_t}}}$. Note that the empty graph and $\vcenter{\hbox{\input{l9_graph2.pdftex_t}}}$ are disconnected and do not contribute here.\footnote{Empty graph is regarded as disconnected: it has zero connected components whereas a connected graph should have one connected component.}

\end{example}

\subsubsection{Introducing the ``Planck constant'' and bookkeeping by Euler characteristic of Feynman graphs}

Consider the integral 
\begin{equation}\label{l10_e3}
\int_V d^n x \; e^{\frac{1}{\hbar} (-\frac12 Q(x,x)+p(x))}
\end{equation}
with $\hbar$ an infinitesimal parameter, $Q$ a positive-definite quadratic form and $p(x)=\sum_{d= 3}^D \frac{1}{d!}P_d(x)$ with $P_d\in \Sym^d V^*$. Note that here, unlike in (\ref{l8_polynomial_perturbation}), we did not scale terms of the perturbation $p(x)$ with coupling constants, however here we only allow at least cubic terms in $p(x)$. We define the perturbative evaluation of (\ref{l10_e3}) by rescaling the integration variable $x=\sqrt\hbar\, y$ which converts it to the perturbative integral of the type defined in (\ref{l8_pert_Gaussian_formal}):
\begin{equation}\label{l10_e4}
\int_V^\mr{pert} d^n x \; e^{\frac{1}{\hbar} (-\frac12 Q(x,x)+p(x))}:= \hbar^{\frac{n}{2}} \int_V^\mr{pert} d^n y\; e^{-\frac12 Q(y,y)+\sum_{d=3}^D \frac{\hbar^{\frac{d}{2}-1}}{d!}P_d(y)}\quad \in \hbar^{\frac{n}{2}}\RR[[\hbar^{\frac12}]]
\end{equation}
Note that, in the integral on the r.h.s., the terms of the perturbation got scaled with ``coupling constants'' $\hbar^{\frac{d}{2}-1}$ -- positive powers of $\hbar$ (as we only allowed terms with $d\geq 3$ in $p(x)$). Moreover, there are finitely many Feynman graphs contributing to each order in $\hbar$.

\begin{lemma}[``Loop expansion'']
We have
\begin{multline}\label{l10_e5}
\int_V^\mr{pert} d^n x \; e^{\frac{1}{\hbar} (-\frac12 Q(x,x)+p(x))}= (2\pi\hbar)^{\frac{n}{2}}(\det Q)^{-\frac12}\sum_{\mr{graphs}\; \Gamma} \frac{\hbar^{-\chi(\Gamma)}}{|\mr{Aut}(\Gamma)|}\Phi_{Q^{-1},\{P_d\}_{d=3}^D}(\Gamma) \\
=(2\pi\hbar)^{\frac{n}{2}}(\det Q)^{-\frac12} \exp\left(\frac{1}{\hbar}\sum_{\gamma\;\mr{connected}} \frac{\hbar^{l(\gamma)}}{|\mr{Aut}(\gamma)|}\Phi_{Q^{-1},\{P_d\}_{d=3}^D}(\gamma)\right)
\end{multline}
where $\chi(\Gamma)$ is the Euler characteristic of the graph and $l(\gamma)=B_1(\gamma)$ is the ``number of loops'' (the first Betti number of a connected graph). Feynman graphs in these expansions are assumed to have valency $\geq 3$ for all vertices (in particular, this implies $l(\gamma)\geq 2$). 
\end{lemma}

\begin{proof}
Applying Feynman's formula (\ref{l9_Feyn_thm}) to the r.h.s. of the definition (\ref{l10_e4}), we get the following Feynman wights of graphs:
$$\frac{1}{|\mr{Aut}(\Gamma)|}\Phi_{Q^{-1},\{\hbar^{\frac{d}{2}-1}P_d\}_{d=3}^D}(\Gamma)=\hbar^{\sum_{\mr{vertices}\; v}(\frac{\mr{val(v)}}{2}-1)}\frac{1}{|\mr{Aut}(\Gamma)|}\Phi_{Q^{-1},\{P_d\}_{d=3}^D}(\Gamma) $$
with $ \mr{val}(v)$ the valency of a vertex $v$ of $\Gamma$. Note that $\sum_{\mr{vertices}\; v}\mr{val}(v)=\# HE$ -- the number of half-edges, therefore  
$$\sum_{\mr{vertices}\; v}(\frac{\mr{val(v)}}{2}-1)=\# E - \# V =- \chi(\Gamma)$$
Thus the Feynman weight of a graph is $\hbar^{-\chi(\Gamma)}\frac{1}{|\mr{Aut}(\Gamma)|}\Phi(\Gamma)$ which proves the first equality in (\ref{l10_e5}). For the second equality, we simply notice that, for $\gamma$ connected, $\hbar^{-\chi(\gamma)}=\frac{1}{\hbar}\cdot \hbar^{l(\gamma)}$.

\end{proof}

\begin{remark}
An intuitive way to recover the result (\ref{l10_e5}) is to interpret the normalization of the integrand of l.h.s. of (\ref{l10_e5}) by $\hbar$ as a change of normalization of the quadratic form $Q\mapsto \hbar^{-1}Q$ (and thus $Q^{-1}\mapsto \hbar Q^{-1}$), $p(x)\mapsto \hbar^{-1}p(x)$ with respect to (\ref{l9_Feyn}). Thus, each edge of a graph picks a factor $\hbar$ and each edge picks a factor $\hbar^{-1}$ which results in the value of the entire graph being scaled with the factor $\hbar^{-\chi(\Gamma)}$.

\end{remark}

\begin{remark}
If we allow terms of degree $<3$ in $p(x)$, in the integral (\ref{l10_e4}) (denote it by $I(\hbar)$) there will be infinitely many terms contributing in each order in $\hbar$, also, $I(\hbar)\in \hbar^{\frac{n}{2}}\RR[[\hbar^{-1},\hbar]]$ -- a two-sided formal Laurent series; more precisely, $I(\hbar)\in \hbar^{\frac{n}{2}}\exp\left(\hbar^{-1}\RR[[\hbar]]\right)$.\footnote{Note however that not every power series of form $\sum_{n\geq -1} a_n \hbar^n $  
can be exponentiated to a formal Laurent series -- certain convergence condition needs to hold for $a_n$ for the coefficients of $\exp(\sum_{n\geq -1} a_n \hbar^n )$ 
to be finite.}
\end{remark}

\subsubsection{Expectation values with respect to perturbed Gaussian measure}

We can consider graphs with vertices marked by elements of a set of colors $\mc{C}$. Then we only allow those graph automorphisms which preserve the vertex colors.

Here is the modification of Feynman's Theorem \ref{l9_Feyn_thm} for expectation values w.r.t. perturbed Gaussian measure:

\begin{theorem}\label{l10_thm_perturbed_expectation_values}
Let $Q$ be a positive-definite quadratic form, let $p(x)=\sum_{d=0}^D \frac{g_d}{d!}P_d(x)$ be a polynomial perturbation and let 
$\Psi_j=\sum_{d\geq 0}\frac{1}{d!}\Psi_{j,d}\in \Sym V^*$ for $j=1,\ldots,r$ be a collection of $r$ polynomials (``observables'') with $\Psi_{j,d}$ their respective homogeneous pieces of degree $d$. Then we have:
\begin{enumerate}[(i)]
\item
\begin{multline}\label{l10_e6}
\int_{V}^\mr{pert} d^n x\; e^{-\frac12 Q(x,x)+p(x)}\Psi_1(x)\cdots \Psi_r(x)=\\
=(2\pi)^{\frac{n}{2}} (\det Q)^{-\frac12}\sum_\Gamma \frac{1}{|\mr{Aut}(\Gamma)|} \Phi_{Q^{-1}; \{g_d P_d\},\{\Psi_{j,d}\}_{j=1}^r}(\Gamma)
\end{multline}
where in the r.h.s. we sum over graphs with vertices colored with elements of $\mc{C}=\{0;1,2,\ldots,r\}$ with the condition that vertices of each color $\neq 0$ occur in the graph exactly once (and there are arbitrarily many vertices of color $0$ -- the ``neutral color''). Vertices of color $0$ and valency $d$ are assigned the vertex tensor $g_d P_d$, while a vertex of color $j\in \{1,\ldots, r\}$ and valency $d$ is assigned the vertex tensor $\Psi_{j,d}$.
\end{enumerate}
\item The normalized expectation value of the product of observables w.r.t. the perturbed the Gaussian measure is:
\begin{multline}\label{l10_e7}
\ll\Psi_1\cdots \Psi_r \gg_\mr{pert}:=\frac{\int_{V}^\mr{pert} d^n x\; e^{-\frac12 Q(x,x)+p(x)}\Psi_1(x)\cdots \Psi_r(x)}{\int_{V}^\mr{pert} d^n x\; e^{-\frac12 Q(x,x)+p(x)}}=\\
=\sum_\Gamma \frac{1}{|\mr{Aut}(\Gamma)|} \Phi_{Q^{-1}; \{g_d P_d\},\{\Psi_{j,d}\}_{j=1}^r}(\Gamma)
\end{multline}
where the sum over graphs is as in (\ref{l10_e6}) with additional requirement that each connected component $\Gamma$ should contain at least one vertex of nonzero color. (Thus, $\Gamma$ can have at most $r$ connected components.)
\end{theorem}

The proof is a straightforward modification of the proof of Theorem \ref{l9_Feyn_thm}.

\marginpar{\LARGE{Lecture 11, 09/30/2016.}}

\begin{remark} If we normalize the perturbed Gaussian measure in Theorem \ref{l10_thm_perturbed_expectation_values} by a Planck constant, as $e^{\frac{1}{\hbar}(-\frac12 Q(x,x)+p(x))}$, then Feynman graphs will get weighed with $\hbar^{r-\chi(\Gamma)}$. We can interpret the power of $\hbar$ here as minus the Euler characteristic of the graph with vertices marked by nonzero colors removed (but the adjacent edges retained as half-open intervals).

\end{remark}

\subsubsection{Fresnel (oscillatory) version of perturbative integral}

Instead of considering perturbed Gaussian integrals, one can consider perturbed Fresnel integrals in the exact same manner. E.g. Fresnel version of (\ref{l10_e6}), with normalization by Planck constant, is as follows:
\begin{multline}
\int_V^\mr{pert} d^n x\; e^{\frac{i}{\hbar}(\frac12 Q(x,x)+p(x))}\Psi_1(x)\cdots \Psi_r(x)=\\
=(2\pi\hbar)^{\frac{n}{2}}|\det Q|^{-\frac12} e^{\frac{\pi i}{4}\mr{sign}Q}\sum_{\mr{graphs}\;\Gamma} \frac{\hbar^{r-\chi(\Gamma)}}{|\mr{Aut}(\Gamma)|} \Phi_{iQ^{-1}; \{iP_d\}; \{\Psi_{j,d}\}}
\end{multline}
Here $Q$ is a non-degenerate (not necessarily positive-definite) quadratic form and $p(x)=\sum_{d=3}^D \frac{1}{d!}P_d(x)$ a polynomial perturbation. Note that the effect of passing to Fresnel version (i.e. introducing the factor $i$ in the exponential in the integrand) amounts to introducing a factor $i$ in the Feynman rules for edges and vertices of neutral color (and the appearance of phase $e^{\frac{\pi i}{4}\mr{sign}Q}$ which comes from bare Fresnel integral and has nothing to do with perturbation).

\subsubsection{Perturbation expansion via exponential of a second order differential operator}
For a non-degenerate quadratic form $Q(x,x)$, introduce a second order differential operator $Q^{-1}(\frac{\dd}{\dd x},\frac{\dd}{\dd x}):= \sum_{i,j=1}^n (Q^{-1})_{ij} \frac{\dd}{\dd x_i} \frac{\dd}{\dd x_j}$.

One can rewrite perturbation expansion (\ref{l9_Feyn}) as follows:
\begin{equation}\label{l11_e1}
\frac{1}{(2\pi)^{\frac{n}{2}} (\det Q)^{-\frac12}}\int_{V}^\mr{pert} d^n x\; e^{-\frac12 Q(x,x)+p(x)}
=\left. e^{\frac12 Q^{-1}(\frac{\dd}{\dd x},\frac{\dd}{\dd x})}\circ e^{p(x)}\right|_{x=0}
\end{equation}
Here on the l.h.s. both exponentials are to be understood via expanding them in the Taylor series.

This follows from the fact that Wick's lemma can be rewritten as
$$\ll x_{i_1}\cdots x_{i_{2m}}\gg = \frac{1}{2^m m!} \left(Q^{-1}(\frac{\dd}{\dd x},\frac{\dd}{\dd x})\right)^m\circ ( x_{i_1}\cdots x_{i_{2m}}) = \left. e^{\frac12 Q^{-1}(\frac{\dd}{\dd x},\frac{\dd}{\dd x})}\circ ( x_{i_1}\cdots x_{i_{2m}}) \right|_{x=0}   $$
And, consequently, for any $f\in \mr{Sym}V^*$, the Gaussian expectation value can be written as
$$ \ll f(x) \gg=  \left. e^{\frac12 Q^{-1}(\frac{\dd}{\dd x},\frac{\dd}{\dd x})}\circ f(x) \right|_{x=0}$$
Setting $f(x)=e^{p(x)}$, we get (\ref{l11_e1}).

\begin{remark}
Pictorially, the mechanism of producing Feynman graphs from the r.h.s. of (\ref{l11_e1}) is as follows: $e^p$ produces, upon Taylor expansion, collections of stars of vertices (decorated with $g_d P_d$ for a $d$-valent vertex). Applying the operator $e^{\frac12 Q^{-1}(\frac{\dd}{\dd x},\frac{\dd}{\dd x})}$ connects some of the half-edges of those stars by arcs, into edges marked by $Q^{-1}$. Then, setting $x=0$, we kill all pictures where some half-edges were left unpaired, thus retaining only the perfect matchings on all available half-edges.
\end{remark}

\subsection{Stationary phase formula with corrections}

The following version of the stationary phase formula (Theorem \ref{thm: statphase}) explains that formal perturbative integrals we studied in Section \ref{sec: perturbed Gaussian integral} do indeed provide asymptotic expansions for measure-theoretic oscillating integrals in the limit of fast oscillation.

\begin{theorem}\label{thm: stat phase with corrections}
Let $X$ be an $n$-manifold, let $\mu\in \Omega^n_c(X)$ be a compactly supported top-degree form, and let $f\in C^\infty(X)$ be a function with only non-degenerate critical points on $\mr{Supp}\,\mu$. Let $I(\hbar):= \int_{X}\mu\; e^{\frac{i}{\hbar} f}$ -- a smooth complex-valued function on $\hbar\in (0,\infty)$. Then the behavior of $I(\hbar)$ at $\hbar\ra 0$ is given by the following asymptotic series:
\begin{multline}\label{l11_stat_phase_corrections}
I(\hbar)\underset{\hbar\ra 0}\sim \sum_{\mr{crit.\;points\;}x_0\;\mr{of}\; f\mr{\;on\;Supp}\,\mu}
e^{\frac{i}{\hbar}f(x_0)} (2\pi\hbar)^{\frac{n}{2}}|\det f''(x_0)|^{-\frac12}\cdot e^{\frac{\pi i}{4}\mr{sign}\, f''(x_0)}\mu_{x_0}\cdot\\
\cdot \exp \hbar^{-1}\left(\sum_{\gamma\mr{\;conn.\; graphs\; with\; vertices\;of\; val}\geq 3}
\frac{\hbar^{l(\gamma)}}{|\mr{Aut}(\gamma)|}\Phi_{if''(x_0)^{-1};\{i\dd^d f|_{x_0}\}_{d\geq 3}}(\gamma) \right)
\end{multline}
Here we assumed that around every critical point $x_0$ of $f$ on $\mr{Supp}\,\mu$ we have chosen some coordinate chart $(y_1,\ldots,y_n)$ with the property that locally near $x_0$ we have $\mu=d^n y\,\mu_{x_0}$ with $\mu_{x_0}$ a constant. Total $d$-th partial derivative appearing in the Feynman rules on the r.h.s. is understood as a symmetric tensor $\dd^d f|_{x_0}\in \Sym^d V^*$ with components $\left. \frac{\dd}{\dd y_{i_1}}\ldots \frac{\dd}{\dd y_d} f\right|_{y=0}$.

\end{theorem}

For the proof, see e.g. \cite{Hoermander, GuilleminSternberg, Zelditch}.

\begin{remark} One can drop the assumption that the density of $\mu$ in the local coordinates $(y_1,\ldots,y_n)$ around a critical point $x_0$ is constant. Let $\mu=\rho(y)\cdot d^n y$ with possibly non-constant $\rho(y)$. Then (\ref{l11_stat_phase_corrections}) becomes
\begin{multline}\label{l11_e2}
I(\hbar)\underset{\hbar\ra 0}\sim \sum_{\mr{crit.\;points\;}x_0\;\mr{of}\; f\mr{\;on\;Supp}\,\mu}
e^{\frac{i}{\hbar}f(x_0)} (2\pi\hbar)^{\frac{n}{2}}|\det f''(x_0)|^{-\frac12}\cdot e^{\frac{\pi i}{4}\mr{sign}\, f''(x_0)}\cdot\\
\cdot \sum_{\Gamma}
\frac{\hbar^{1-\chi(\Gamma)}}{|\mr{Aut}(\Gamma)|}\Phi_{if''(x_0)^{-1};\underbrace{\{i\dd^d f|_{x_0}\}_{d\geq 3}}_{\mr{color}\; 0}; \underbrace{\{\dd^d \rho|_{y=0}\}_{d\geq 0}}_{\mr{color\;}1}}(\Gamma) 
\end{multline}
where the sum on the r.h.s. is over (possibly disconnected) graphs $\Gamma$ with vertices of valency $\geq 3$ colored by neutral color $0$ and a single marked vertex, of arbitrary valency, colored by $1$. 
\end{remark}

\subsubsection{Laplace method}

Laplace method applies to integrals of form $I(\hbar)=\int dx \; e^{-\frac{1}{\hbar}f(x)}$. The idea is that the integrand is concentrated around the minimum $x_0$ of $f$, in the neighborhood of $x_0$ of size $\sim\sqrt\hbar$; in this neighborhood the integrand is well approximated by a Gaussian (given by expanding $f$ at $x_0$ in Taylor series and retaining only the constant and quadratic terms; higher Taylor terms may be accounted for as a perturbation, to obtain higher corrections in powers of $\hbar$).

Simplest version of this asymptotic result is as follows.
\begin{theorem}[Laplace]
Let $f\in C^\infty [a,b]$ be a function on an interval attaining a unique absolute minimum on $[a,b]$ at an interior point $x_0\in (a,b)$, with $f''(x_0)>0$. Let $g\in C^\infty [a,b]$ be another function on the interval with $g(x_0)\neq 0$. Then the integral 
$$I(\hbar):=\int_a^b dx\; g(x) e^{-\frac{1}{\hbar}f(x)}$$
as a smooth function of $\hbar>0$ has the following asymptotics as $\hbar\ra 0$:
\begin{equation}\label{l11_Laplace}
I(\hbar)\underset{\hbar\ra 0}\sim e^{-\frac{1}{\hbar}f(x_0)} \sqrt\frac{2\pi\hbar}{f''(x_0)}\cdot g(x_0)
\end{equation}
\end{theorem}

A more general multi-dimensional version, with $\hbar$-corrections is as follows.
\begin{theorem}[Feynman-Laplace]
Let $X$ be a compact $n$-manifold, possibly with boundary, and let $f\in C^\infty(X)$ be a function attaining a unique minimum on $X$ at an interior point $x_0\in \mr{int}(X)$ and assume that the Hessian $f''(x_0)$ is non-degenerate (thus, automatically, positive-definite); also, let $\mu\in \Omega^n(X)$ be a top-degree form. Assume that we have chosen some local coordinates $(y_1,\ldots,y_n)$ near $x_0$ and in these coordinates $\mu=\rho(y)\,d^n y$. Then the integral 
$$I(\hbar):=\int_X \mu \; e^{-\frac{1}{\hbar}f(x)}\qquad \in C^\infty (0,\infty)$$
has the following asymptotic expansion at $\hbar\ra 0$:
\begin{multline}\label{l11_Laplace_corr}
I(\hbar)\underset{\hbar\ra 0}\sim 
e^{-\frac{1}{\hbar}f(x_0)} (2\pi\hbar)^{\frac{n}{2}}(\det f''(x_0))^{-\frac12}
\cdot\\ 
\cdot\sum_{\Gamma}
\frac{\hbar^{1-\chi(\Gamma)}}{|\mr{Aut}(\Gamma)|}\Phi_{f''(x_0)^{-1};\{-\dd^d f|_{x_0}\}_{d\geq 3}; \{\dd^d \rho|_{y=0}\}_{d\geq 0}}(\Gamma) 
\end{multline}
where, as in (\ref{l11_e2}), the sum is over graphs with arbitrarily many vertices of color $0$ and valency $\geq 3$ and a single vertex of color $1$ and arbitrary valency.
\end{theorem}

\marginpar{\LARGE{Lecture 12, 10/03/2016.}}

\begin{example}[Stirling's formula with corrections]
Consider $z\ra\infty$ asymptotics of the Euler's Gamma function 
$$\Gamma(z)=\int_0^\infty dt\; t^{z-1} e^{-t}  $$
It is convenient to make a change of the integration variable $t=z\, e^x$, yielding
$$\Gamma(z)=z^z \int_{-\infty}^{\infty} dx\, e^{-z f(x)}$$
with $f(x)=e^x-x$; $f$ has unique absolute minimum at $x=0$ with Taylor expansion $f(x)=1+\frac12 x^2 + \frac{1}{3!}x^3+\cdots$. The asymptotics of this integral at $z\ra\infty$ can evaluated using Laplace's theorem (\ref{l11_Laplace}), with $\hbar:=\frac{1}{z}$:
$$\Gamma(z)\underset{z\ra\infty}\sim z^z e^{-z}\sqrt\frac{2\pi}{z}$$
Using (\ref{l11_Laplace_corr}), we can find corrections to this asymptotics in powers of $\frac{1}{z}$:
$$\Gamma(z)\underset{z\ra\infty}\sim z^z e^{-z}\sqrt\frac{2\pi}{z} \exp\sum_{n=1}^\infty \frac{c_n}{z^n} $$
with $c_n=\sum_{\Gamma}\frac{(-1)^{\#\mr{vertices}}}{|\mr{Aut}(\Gamma)|}$ where the sum goes over connected graphs with $n-1$ loops (all valencies $\geq 3$ allowed). E.g. the first coefficient $c_1$ gets contributions from the  three connected 2-loop graphs: 
$\vcenter{\hbox{\input{l9_graph1.pdftex_t}}}$, $\vcenter{\hbox{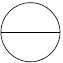}}$, $\vcenter{\hbox{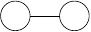}}$:
$c_1=-\frac{1}{8}+\frac{1}{12}+\frac{1}{8}=\frac{1}{12}$.\footnote{In fact, as can be obtained independently, e.g., from Euler-Maclaurin formula, $c_n=\frac{B_{n+1}}{n(n+1)}$, with $B_{n+1}$ the $(n+1)$-st Bernoulli number.}
In particular, this implies that the factorial of a large number $n!=n\Gamma(n)$ behaves as
$$n!\underset{n\ra\infty}\sim \sqrt{2\pi n}\; n^n e^{-n} \left(1+\frac{1}{12 n}+O(\frac{1}{n^2})\right)$$
\end{example}

\subsection{Berezin integral}

\subsubsection{Odd vector spaces}
Fix $n\geq 1$. Consider the ``odd $\RR^n$'', denoted as $\Pi \RR^n$ or $\RR^{0|n}$, -- space with anti-commuting\footnote{``Odd'' or ``Grassman'' or ``fermionic'' variables.} coordinates $\theta_1,\ldots,\theta_n$. I.e. $\Pi\RR^n$ is defined by its algebra of functions
$$\mr{Fun}(\Pi\RR^n):=\RR \lan \theta_1,\ldots, \theta_n \ran/\, \theta_i \theta_j=-\theta_j \theta_i$$
More abstractly, for $V$ a vector space over $\RR$, its odd version $\Pi V$ has the algebra of functions 
$$\mr{Fun}(\Pi V)=\wedge^\bt V^*$$
-- the exterior algebra of the dual (viewed as a super-commutative associative algebra), whereas for an even vector space $\mr{Fun}(V)=\widehat{\Sym} V^*$ -- the (completed) symmetric algebra of the dual.

\subsubsection{Integration on the odd line}
Consider the case $n=1$ -- the odd line $\Pi \RR$ with coordinate $\theta$ subject to relation $\theta^2=0$. Functions on $\Pi\RR$ have form $a+b\theta$ with $a,b\in \RR$ arbitrary coefficients. We define the integration map $\int_{\Pi \RR}D\theta \;(\cdots) :\; \mr{Fun}(\Pi \RR)\ra \RR$ by
\begin{equation}\label{l12_Berezin_1D}
\int_{\Pi\RR}D \theta\; (a+b\theta):=b
\end{equation}
I.e. the integration simply picks the coefficient of $\theta$ in the function being integrated. Integration as defined above is uniquely characterized by the following properties:
\begin{itemize}
\item integration maps is $\RR$-linear,
\item ``Stokes' theorem'': $\int_{\Pi \RR}D\theta \;\frac{\dd}{\dd \theta} g(\theta)=0$ for $g(\theta)$ an arbitrary function on $\Pi\RR$.\footnote{Derivatives are defined on $\Pi\RR^n$ in the following way: $\frac{\dd}{\dd\theta_i}$ is an odd derivation of $\mr{Fun}(\Pi\RR^n)$ (i.e. a linear map $\mr{Fun}(\Pi\RR^n)\ra \mr{Fun}(\Pi\RR^n)$ satisfying the Leibniz rule with appropriate sign $\frac{\dd}{\dd\theta_i}(f\cdot g)=(\frac{\dd}{\dd\theta_i}f)\cdot g+(-1)^{|f|}f\cdot (\frac{\dd}{\dd\theta_i}g)$) and defined on generators by $\frac{\dd}{\dd\theta_i}\theta_j=\delta_{ij}$.} This implies that the integral of a constant function has to vanish.
\item Normalization convention: $\int_{\Pi \RR}D\theta\;\theta=1$. 
\end{itemize}

\subsubsection{Integration on the odd vector space}
A function on $\Pi\RR^n$ can be written as 
$$f(\theta_1,\ldots,\theta_n)=\sum_{k=0}^n \;\sum_{1\leq i_1<\cdots < i_k \leq n} f_{i_1\cdots i_k} \theta_{i_1}\cdots \theta_{i_k}$$
with $f_{i_1\cdots i_k}\in \RR$ the coefficients. Berezin integral on $\Pi\RR^n$ is defined as follows:
\begin{equation}\label{l12_Berezin}
\int_{\Pi\RR^n} D\theta_n\,\cdots D\theta_1\;f :=f_{1\cdots n}\qquad =\mbox{coefficient of $\theta_1\cdots\theta_n$ in $f$}
\end{equation}
This definition can be obtained from the definition (\ref{l12_Berezin_1D}) for the 1-dimensional case by formally imposing the Fubini theorem, e.g. for $n=2$ and $f(\theta_1,\theta_2)=f_\varnothing+f_1\theta_1+f_2\theta_2 +f_{12}\theta_1 \theta_2$ we have
$$
\int D\theta_2\, D\theta_1\;f = \int D\theta_2 \underbrace{\left( \int D\theta_1\,f   \right)}_{f_1+f_{12}\theta_2} = f_{12}
$$
Case of general $n$ is treated similarly, by inductively integrating over odd variables $\theta_i$, in the order of increasing $i$.

\begin{remark}
Berezin integral can also be seen as an iterated derivative:
$$\int_{\Pi\RR^n} D\theta_n\cdots D\theta_1 \; f= \left.\frac{\dd}{\dd \theta_n}\cdots \frac{\dd}{\dd \theta_1} f\right|_{\theta=0}$$
\end{remark}

More abstractly, $f\in \mr{Fun}(\Pi V)$ a function on an odd vector space $\Pi V$ (for $V$ of dimension $n$) and $\mu\in \wedge^n V$ a ``Berezinian'' (a replacement of the notion of integration measure or volume form in the context of integration over odd vector spaces), Berezin integral is defined as
$$\int_{\Pi V}\mu\cdot f := \lan \mu,f \ran  $$
-- the pairing between the top component of $f$ in $\wedge^n V^*$ and $\mu\in \wedge^n V$. The pairing between $\wedge^n V$ and $\wedge^n V^*$ is defined by
$$\lan \psi_n\wedge\cdots\wedge \psi_1, \theta_1\wedge\cdots\wedge\theta_n  \ran: = \det \lan \psi_i,\theta_j\ran$$
for $\psi_i\in V$ vectors, $\theta_j\in V^*$ covectors and $\lan\psi_i,\theta_j\ran$ the canonical pairing between $V$ and $V^*$. 

Note that constant volume forms on an even space $V$ are (nonzero) elements of $\wedge^n V^*$ whereas Berezinians are elements of $\wedge^n V$. Note that there is no dual in the second case!

Given a basis $e_1,\ldots,e_n$ in $V$ and the associated dual basis regarded as coordinate functions on the odd space $\theta_1,\ldots,\theta_n\in V^*\subset \mr{Fun}(\Pi V)$, we have a ``coordinate Berezinian'' 
$$\mu=D\theta_n\cdots D\theta_1:= e_n\wedge\cdots\wedge e_1 \quad \in \wedge^n V$$
Note that, if we have a change of coordinates on $\Pi V$, $\theta_i=\sum_j A_{ij}\theta'_j$, the respective coordinate Berezinians are related by 
\begin{equation}
D^n\theta=(\det A)^{-1} D^n\theta'
\end{equation}
where $D^n\theta$ is a shorthand for $D\theta_n\cdots D\theta_1$ and similarly for $D^n\theta'$. Then we have a change of coordinates formula for the Berezin integral:
$$\int_{\Pi V} D^n\theta f(\theta)=\int_{\Pi V} (\det A)^{-1}D^n\theta'\; f(\theta_i=\sum_j A_{ij}\theta'_j)$$
Observe the difference from the case of a change of variables $x_i=\sum_{j}A_{ij} x'_j$ in an integral over an even space:
$$\int_V d^nx\; f(x)=\int_V |\det A|\;d^n x' \; f(x_i=\sum_j A_{ij} x'_j) $$
In even case we have the absolute value of the Jacobian of the transformation,\footnote{In the even case we either think of an integral over an oriented space against a top form, or of an integral over a non-oriented space against a measure (density). A measure transforms with the absolute value of the Jacobian, while a top form transform just with the Jacobian itself -- but then one has to take the change of orientation into account separately.}
whereas in the odd case we have the inverse of the Jacobian, without taking the absolute value.

\subsection{Gaussian integral over an odd vector space}
Let $Q(\theta,\theta)=\sum_{i,j=1^n}Q_{ij}\theta_i \theta_j$ be a quadratic form on $\Pi\RR^n$ with $Q_{ij}$ an anti-symmetric matrix, so that $\frac12 Q(\theta,\theta)=\sum_{i<j} Q_{ij} \theta_i \theta_j$. We assume that $n=2s$ is even. Then we have the following version of Gaussian integral over $\Pi\RR^n$:
\begin{equation}\label{l12_Berezin_Gaussian}
\int_{\Pi \RR^n}D^n\theta\; e^{\frac12 Q(\theta,\theta)}=\frac{1}{2^s s!}\sum_{\sigma\in S_{n}} (-1)^\sigma \prod_{i=1}^s Q_{\sigma_{2i-1}\sigma_{2i}} = \mr{pf}(Q)
\end{equation}
-- the Pfaffian of the anti-symmetric matrix $Q_{ij}$; here $(-1)^\sigma$ is the sign of permutation $\sigma$. We obtain the Pfaffian simply by expanding $e^{\frac12 Q}$ in Taylor series, picking the top monomial in $\theta$-s and evaluating its coefficients (as per definition of Berezin integral (\ref{l12_Berezin})).\footnote{Recall that an alternative definition of Pfaffian is as the coefficient on the r.h.s. of $\frac{1}{s!}(\sum_{i<j}Q_{ij}\theta_i\theta_j)^s=\pf(Q)\cdot\theta_1\cdots\theta_n$, which is precisely what we need to evaluate the Berezin integral (\ref{l12_Berezin_Gaussian}).} Note that, for $n$ odd, the integral on the l.h.s. of (\ref{l12_Berezin_Gaussian}) vanishes identically (the exponential contains only monomials of even degree in $\theta$, hence there is no monomial of top degree).

Recall the basic properties of Pfaffians:
\begin{itemize}
\item $\pf(Q)^2=\det Q$,
\item for $A$ any $n\times n$ matrix,  $\pf(A^T Q A)=\det A\cdot \pf(Q)$,
\item $\pf(Q_1\oplus Q_2)=\pf(Q_1)\cdot \pf(Q_2)$,
\item $\pf(\lambda Q)=\lambda^s \pf(Q)$.
\end{itemize}

\begin{example}
$$\pf\left(  \begin{array}{cccc}
\begin{array}{cc} 0 & a_1 \\ -a_1 & 0 \end{array} & & & \mbox{\Huge 0} \\
& \begin{array}{cc} 0 & a_2 \\ -a_2 & 0 \end{array} & & \\
& & \ddots & \\ 
\mbox{\Huge 0} & & & \begin{array}{cc} 0 & a_s \\ -a_s & 0 \end{array}
\end{array}
\right)\quad =\quad a_1\cdot\ldots \cdot a_s$$
\end{example}

\begin{example}
$$
\pf\left( \begin{array}{cccc} 
0 & a_{12} & a_{13} & a_{14} \\
-a_{12} & 0 & a_{23} & a_{24} \\
-a_{13} & -a_{23} & 0 & a_{34} \\
-a_{14} & -a_{24} & -a_{34} & 0
\end{array}  \right)\quad = \quad
a_{12} a_{34}-a_{13} a_{24} + a_{14} a_{23}
$$
\end{example}

\begin{remark}
Consider a special instance of Berezin Gaussian integral where odd variables come in pairs $\theta_i,\bar\theta_i $ (the bar does not stand for complex conjugation: $\bar\theta_i$ is an independent variable from $\theta_i$):
\begin{equation}\label{l12_Gaussian_theta_thetabar}
\int_{\Pi\RR^n\oplus \Pi\RR^n}(D\theta_n D\bar\theta_n)\cdots (D\theta_1 D\bar\theta_1) \; e^{B(\bar\theta,\theta)}=\det B
\end{equation}
here $B(\bar\theta,\theta)=\sum_{i,j=1}^n B_{ij} \bar \theta_i \theta_j$ where $B_{ij}$ is a matrix which does not have to be symmetric or anti-symmetric. The fact that the integral above is equal to $\det B$ is a simple calculation of the Berezin integral:
$$ \mbox{l.h.s.}= \frac{1}{n!}\sum_{\sigma,\sigma'\in S_n} (-1)^\sigma (-1)^{\sigma'} B_{\sigma_1\sigma'_1}\cdots B_{\sigma_n\sigma'_n}=\sum_{\sigma''\in S_n}(-1)^{\sigma''} B_{1 \sigma''_1}\cdots  B_{n \sigma''_n} = \det B$$
where $\sigma''=\sigma'\cdot \sigma^{-1}$. More invariantly, for endomorphism $B\in \mr{End}(V)\simeq V^*\otimes V\subset \mr{Fun}(\Pi V \oplus \Pi V^*)$, we have
$$\int_{\Pi V\oplus \Pi V^*} \mu^\mr{can}_{\Pi V\oplus \Pi V^*}\; e^B=\det B $$
where $\mu^\mr{can}_{\Pi V\oplus \Pi V^*}$ is the canonical Berezinian on $\Pi V\oplus \Pi V^*$, which, for any choice of coordinates $\theta_1,\ldots,\theta_n$ on $V$ and \emph{dual} coordinates $\bar\theta_1,\ldots,\bar\theta_n$ on $V^*$, takes the form 
$(D\theta_n D\bar\theta_n)\cdots (D\theta_1 D\bar\theta_1)$.
\end{remark}

\subsection{Perturbative integral over a vector superspace}

\subsubsection{``Odd Wick's lemma''}

We have the following version of Wick's lemma for integration over an odd vector space.

\begin{lemma}
Let $V$ be a vector space over $\RR$ of even dimension $n=2s$. Let $Q\in \wedge^2 V^*$ be a non-degenerate anti-symmetic bilinear, viewed as a quadratic form on $\Pi V$, and let $\xi_1,\ldots,\xi_{2m}$ be a collection of elements of $V^*$ (viewed as linear functions on $\Pi V$). Then the expectation value 
$$\ll \xi_1\cdots\xi_{2m} \gg:= \frac{\int_{\Pi V} \mu \; e^{-\frac12 Q} \xi_1\cdots \xi_{2m}}{\int_{\Pi V} \mu \; e^{-\frac12 Q} }$$ 
(here $\mu$ is an arbitrary non-zero Berezinian on $\Pi V$; the expectation value is clearly independent of $\mu$) is equal to the sum over perfect matchings with signs:
\begin{equation}
\ll \xi_1\cdots\xi_{2m} \gg =\sum_{\sigma\in S_{2m}/S_m\ltimes \ZZ_2^m} (-1)^\sigma \lan \sigma \circ (Q^{-1})^{\otimes m}, \xi_1\otimes\cdots\otimes \xi_{2m} \ran
\end{equation}
\end{lemma}

It is proven by the same technique as the usual Wick's lemma for an even Gaussian integral: one introduces a source $J$ (which is now odd) and obtains the expectation values as derivatives in $J$ of the Gaussian integral modified by the source term.

\marginpar{\LARGE{Lecture 13, 10/05/2016.}}

\begin{example}
Gaussian expectation value of a quartic monomial on $\Pi \RR^n$ is:
$$\ll \theta_i \theta_j \theta_k \theta_l \gg = \ll \theta_i \theta_j \gg\cdot \ll \theta_k \theta_l \gg -
\ll \theta_i \theta_k \gg \cdot \ll \theta_j \theta_l\gg +
\ll \theta_i \theta_l \gg \cdot \ll \theta_j \theta_k \gg$$
where e.g. the sign of the second term in the r.h.s. is $(-1)^{\left(\begin{array}{cccc} i&j&k&l\\ i&k&j&l \end{array}\right)}=-1$. Quadratic expectation values in turn are the matrix element of the inverse of $Q$:  
$$\ll \theta_i \theta_j\gg= (Q^{-1})_{ij}$$
\end{example}

\subsubsection{Perturbative integral over an odd vector space}

Perturbed Gaussian integral over an odd space can be treated similarly to the even case. Let $Q$ be a non-degenerate quadratic form on $\Pi V=\Pi \RR^n$ and let $p(\theta)=\sum_{d=0}^D \frac{g_d}{d!} P_d(\theta)$ be a polynomial perturbation where we allow only even degrees $d$ for the homogeneous components $P_d\in \wedge^d V^*$. Consider the integral
$$I:=\int_{\Pi V} D^n\theta \; e^{-\frac12 Q(\theta,\theta)+p(\theta)}$$
Evaluating it be expanding $e^{p(\theta)}$ in Taylor series and applying Wick's lemma termwise, we obtain:
\begin{multline}\label{l13_e1}
\int_{\Pi V} D^n\theta\;  e^{-\frac12 Q(\theta,\theta)+p(\theta)}=\\
=\pf(-Q)\cdot\sum_{v_0,\ldots,v_D=0}^\infty\sum_{[\sigma]\in (\prod_d S_{v_d}\ltimes S_d^{\times v_d})\backslash\; S_{2m}\; /S_m\ltimes \ZZ_2^m} \frac{1}{\mr{Stab}_{[\sigma]}}(-1)^\sigma \lan \sigma\circ (Q^{-1})^{\otimes m}, \prod_{d=0}^D(g_d P_d)^{\otimes v_d} \ran 
\\
= \pf(-Q)\cdot\sum_\Gamma \frac{1}{|\mr{Aut}(\Gamma)|}\Phi_{Q^{-1};\{g_dP_d\}}(\Gamma )
\end{multline}
Here $2m=\sum_d d\cdot v_d$; the sum runs over all graphs $\Gamma$ with vertices of even valency ranging between $0$ and $D$. Feynman state sum $\Phi(\Gamma)$ for a graph now contains the sign of a permutation $\sigma\in S_{2m}$ representing $\Gamma$.

\begin{remark}
\begin{itemize}
\item Note that (\ref{l13_e1}) is an exact evaluation of a Berezin integral (i.e. the perturbative evaluation and exact evaluation automatically coincide for integrals over finite-dimensional odd vector spaces). 
\item Since sufficiently high powers of $p(\theta)$ vanish identically, the r.h.s. of (\ref{l13_e1}) is a finite-degree polynomial in $g_0,\ldots,g_D$.
\item Graphs with $>n$ half-edges are guaranteed to cancel out on the r.h.s. (note that individual graphs with $\# HE>n$ can be still nonzero, but cancel out once all graphs are summed over).
\item R.h.s. of (\ref{l13_e1}) can be rewritten as
\begin{equation}\label{l13_e2}
\pf(-Q)\cdot \exp\sum_\gamma \frac{1}{|\mr{Aut}(\gamma)|}\Phi_{Q^{-1};\{g_dP_d\}}(\gamma ) 
\end{equation}
where the sum is over connected graphs $\gamma$. Here the sum in the exponential is, generally, not a polynomial in $g_d$ and contributions of connected graphs do not cancel out for graphs of high complexity.
\end{itemize}
\end{remark}

\begin{example}
Here is an example of a weight of a Feynman graph in the r.h.s. of (\ref{l13_e1}):
\begin{multline*}
\frac{1}{48}\Phi\left(\vcenter{\hbox{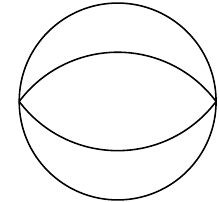}}\right)= \\
=\frac{(g_4)^2}{48}\sum_{s_1,\ldots,s_8=1}^n (Q^{-1})_{s_1 s_3} (Q^{-1})_{s_6 s_2} (Q^{-1})_{s_4 s_8} (Q^{-1})_{s_7 s_5} \cdot (P_4)_{s_1 s_6 s_4 s_7} (P_4)_{s_3 s_2 s_8 s_5} \cdot\\
\cdot (-1)^{\tiny \left(\begin{array}{cccccccc} 1 & 3 & 6 & 2 & 4 & 8 & 7 & 5 \\ 1 & 6 & 4 & 7 & 3 & 2 & 8 & 5 \end{array}\right)} 
\end{multline*}
Here we assigned arbitrary labels (from $1$ to $8$) to the half-edges; the sign factor is the sign of the permutation taking the order of labels for edges to the order of labels for the vertices.
\end{example}

\begin{example}
Let $B\in GL(V)$ and $P\in \mr{End}(V)$. Consider the following perturbation of the integral (\ref{l12_Gaussian_theta_thetabar}): 
\begin{equation}\label{l13_e4}
I(\alpha)=\int_{\Pi V\oplus \Pi V^*} \overleftarrow{\prod_{j=1}^n} D\theta_j D\bar\theta_j\; e^{-\sum_{i,j} B_{ij}\bar \theta_i \theta_j+\alpha \sum_{i,j} P_{ij} \bar \theta_i \theta_j} 
\end{equation}
with $\alpha$ a coupling constant. Using (\ref{l13_e1},\ref{l13_e2}) we find that
\begin{equation}\label{l13_e3}
I(\alpha)=\det(-B)\cdot\exp\left(-\sum_{k=1}^\infty\frac{\alpha^k}{k} \mr{tr}\,(B^{-1}P)^k\right)
\end{equation}
Terms in the exponential correspond to \emph{oriented} polygon graphs with $k$ vertices and $k$ edges 
$$\vcenter{\hbox{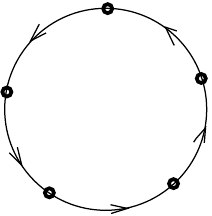}}$$ 
Oriented graphs appear if we label half-edges corresponding to variables $\theta_i$ as \emph{outgoing} and half-edges corresponding to $\bar\theta_i$ as \emph{incoming}. In the sum in the exponential in (\ref{l13_e3}) we can recognize the Taylor expansion of $\log(1-x)$, thus we obtain
$$I(\alpha)=\det(-B)\cdot \exp \tr\log(1-\alpha B^{-1} P)= \det(-B)\cdot \det (1-\alpha B^{-1} P)=\det (-B+\alpha P)$$
which is what we would have obtained if we evaluated (\ref{l13_e4}) directly as a Gaussian integral with quadratic form $B-\alpha P$ rather than treating $\alpha P$ as a perturbation. Note that the series in the exponential in (\ref{l13_e3}) has a finite convergence radius $|\alpha|< \frac{1}{|| B^{-1}P ||}$ where $||A||=\max_\lambda |\lambda|$ with $\lambda$ going over eigenvalues of $A$.
\end{example}

\subsubsection{Perturbative integral over a superspace}

Consider a vector superspace $$\mc{V}=V^e\oplus \Pi V^o$$ 
for $V^e, V^o$ two vector spaces of dimensions $n,m$ (superscripts $e,o$ stand for ``even'', ``odd''), with the algebra of functions $\mr{Fun}(\mc{V}):=C^\infty(V^e)\otimes \wedge^\bt (V^o)^*$. Let $x_1,\ldots,x_n$ be coordinates on $V^e$ and $\theta_1,\ldots,\theta_m$ be coordinates on $\Pi V^o$. Let $Q_e$ be a quadratic form on $V^e$ and $Q_o$ a quadratic form on $\Pi V^o$, and let $p(x,\theta)=\sum_{j,k}\frac{g_{jk}}{j!k!}P_{jk}(x,\theta)$ be a perturbation, with $P_{jk}\in \Sym^j (V^e)^*\otimes \wedge^k (V^o)^*$ the homogeneous parts; degree $k$ here is only allowed to take even values.
Consider the perturbative integral
\begin{equation}
I:=\int_{V^e\oplus \Pi V^o}^\mr{pert} d^n x\; D^m \theta\; e^{-\frac12 Q_e(x,x)-\frac12 Q_0(\theta,\theta)+p(x,\theta)}
\end{equation}
It is understood by formally imposing Fubini theorem: we first integrate over the odd variables and then -- perturbatively -- over even variables. The result is  the following generalization of Feynman's theorem (Theorem \ref{l9_Feyn_thm}) for integration over a superspace:
$$I= (2\pi)^\frac{n}{2}(\det Q_e)^{-\frac12}\pf(-Q_o)\sum_\Gamma \frac{1}{|\mr{Aut}(\Gamma)|}\Phi(\Gamma)$$
Feynman rules for evaluating $\Phi(\Gamma)$ are as follows:
\begin{itemize}
\item Graphs $\Gamma$ are allowed to have half-edges marked as $\vcenter{\hbox{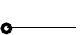}}$ (even) and $\vcenter{\hbox{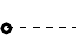}}$ (odd).
\item Edges are pairs of even half-edges $\vcenter{\hbox{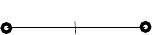}}$ (assigned $Q_e^{-1}$) or pairs of odd half-edges $\vcenter{\hbox{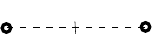}}$ (assigned $Q_o^{-1}$).
\item Vertices have bi-valency $(j,k)$ -- $j$ adjacent even half-edges and $k$ (an even number) adjacent odd half-edges $\vcenter{\hbox{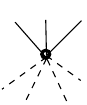}}$ (assigned $g_{jk}P_{jk}$).
\end{itemize}

Put another way, a graph $\Gamma$, with $E_e$, $E_o$ the numbers of even/odd half-edges and with $v_{jk}$ the number of vertices of bi-valency $(j,k)$, is identified with the class of a pair of permutations $(\sigma_e,\sigma_o)$ in the double coset
$$\prod_{j,k} S_{v_{jk}}\ltimes (S_j\times S_k)^{v_{jk}}\backslash\quad S_{2E_e}\times S_{2E_o}\quad/(S_{E_e}\ltimes \ZZ_2^{E_e})\times (S_{E_o}\ltimes \ZZ_2^{E_o})$$
Pictorially: 
$$\vcenter{\hbox{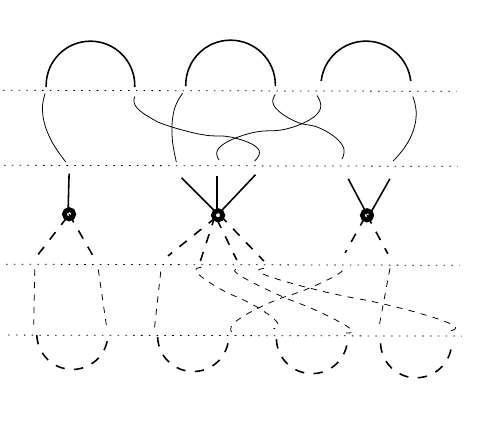}}$$
Note that, when defining automorphisms of a graph, we now only allow permutations of half-edges which preserve the parity.
The Feynman state sum of a graph is
\begin{equation}\label{l13_e5}
\Phi(\Gamma)=(-1)^{\sigma_o}\lan \left(\sigma_e\circ (Q^{-1}_e)^{\otimes E_e}\right)\otimes \left(\sigma_o\circ (Q^{-1}_o)^{\otimes E_o}\right), \bigotimes_{j,k} (g_{jk}P_{jk})^{\otimes v_{jk}} \ran
\end{equation}

\begin{example}[``Faux quantum electrodynamics'' integral]
Fix $V\simeq \RR^n,\;U\simeq \RR^m$ two vector spaces. Let $\mc{V}=V\oplus \Pi (U\oplus  U^*)$ with coordinates $x_i, \theta_a, \bar\theta_a$ -- ``photon'', ``electron'' and ``positron'' variables. We also need the following input data:
\begin{itemize}
\item quadratic form $Q_e(x,x)\in \Sym^2 V^*$,
\item quadratic form $Q_o(\bar\theta,\theta) = \lan\bar \theta, \mathfrak{D} \theta \ran$ with $\mathfrak{D}\in GL(U)$ -- ``faux Dirac operator'',
\item a tensor $P(x,\theta,\bar\theta)\in V^*\otimes U^*\otimes U$ -- ``photon-electron interaction''.
\end{itemize}
We then consider the following perturbative integral
\begin{equation}\label{l13_QED_integral}
\int_\mc{V}^\mr{pert} d^nx\; D^m\theta\; D^m\bar{\theta}\; e^{-\frac12 Q_e(x,x)-\lan\bar\theta,\mathfrak{D}\theta \ran+g P(x,\theta,\bar\theta)} = \left(\det\frac{Q_e}{2\pi}\right)^{-\frac12}\cdot \det (-\mathfrak{D})\cdot\sum_\Gamma \frac{1}{|\mr{Aut}(\Gamma)|}\Phi(\Gamma)
\end{equation}
Here $g$ is a coupling constant (``charge of the electron''). Graphs $\Gamma$ in the r.h.s. of (\ref{l13_QED_integral}) have three types of half-edges:
\begin{enumerate}[(i)]
\item $\vcenter{\hbox{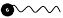}}$ for ``photon'' variables $x_i$,
\item $\vcenter{\hbox{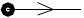}}$ for ``electron'' variables $\theta_a$,
\item $\vcenter{\hbox{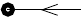}}$ for ``positron'' variables $\bar\theta_a$
\end{enumerate} 
Admissible edges are: $\vcenter{\hbox{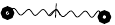}}$ (non-oriented, assigned the propagator $Q_e^{-1}$) and  
$\vcenter{\hbox{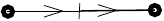}}$ (oriented, assigned the propagator $\mathfrak{D}^{-1}$). The only admissible vertex is $\vcenter{\hbox{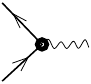}}$ (assigned $g\cdot P$).
Typical graph contributing to the r.h.s. of (\ref{l13_QED_integral}) looks like this:
$$\vcenter{\hbox{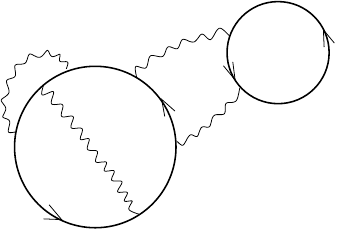}}$$
An admissible $\Gamma$ is always a collection of oriented solid (elctron/positron) cycles arbitrarily interconnected by photon edges. Here is an example of evaluation of a simple admissible graph:
$$ \frac12\; \Phi\left( \vcenter{\hbox{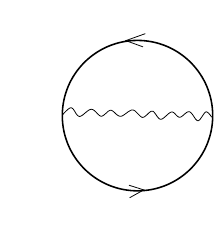}}\quad\quad \right) = 
-\frac{g^2}{2}\langle Q_e^{-1},\underbrace{\tr_U( \mathfrak{D}^{-1}P\mathfrak{D}^{-1}P) }_{\in (V^*)^{\otimes 2}} \rangle $$
Here we understand $P$ as an element of $V^*\otimes \mr{End}(U)$ and take compositions of endomorphisms of $U$. The minus sign here is $(-1)^{\sigma_o}$, cf. (\ref{l13_e5}).
\end{example}

\subsection{Digression: the logic of perturbative path integral}

In the case of finite-dimensional integrals of oscillatory type $I(\hbar)=\int_X \mu\; e^{\frac{i}{\hbar}f}$, asymptotics of the measure-theoretic integral (which exists for finite $\hbar$) at $\hbar\ra 0$ is given by the expansion in Feynman diagrams (Theorem \ref{thm: stat phase with corrections}). 

On the other hand a path (functional) integral
\begin{equation}\label{l13_PI}
I(\hbar)=``\int_{\Gamma(M,\sf{Fields})\ni \phi} \mc{D}\phi \; e^{\frac{i}{\hbar}S(\phi)}\;\;"
\end{equation}
with $M$ the spacetime manifold and $\sf{Fields}$ the sheaf of fields on $M$, and with action $S=\frac12 \int_M \lan \phi,\mathfrak{D}\phi \ran +\int_M \LL_\mr{int}(\phi)$ (here $\mathfrak{D}$ is some differential operator), is a heuristic expression which is \emph{defined} as an asymptotic series in $\hbar$ by its expansion in Feynman diagrams,
\begin{equation} \label{l13_PI_pert}
I(\hbar):=(\det \mathfrak{D})^{-\frac12}\cdot \sum_\Gamma \frac{\hbar^{-\chi(\Gamma)}}{|\mr{Aut}(\Gamma)|}\Phi(\Gamma)
\end{equation}
Here $\Phi(\Gamma)$ is given as an integral over $M^{\times V}$ ($V$ is the number of vertices in $\Gamma$) of certain differential form on $M^{\times V}$ (which we view as the space of configurations of $V$ points on $M$) which depends on $\Gamma$ and is constructed in terms of the \emph{propagator} -- the integral kernel of the inverse operator $\mathfrak{D}^{-1}$ assigned to edges and \emph{vertex functions}, read off from $\LL_\mr{int}$, assigned to vertices.  Expansion (\ref{l13_PI_pert}) is obtained by treating (\ref{l13_PI}) following the logic of finite-dimensional perturbed Gaussian integral: one expands $e^{\frac{i}{\hbar}\int_M \LL_\mr{int}(\phi)}$ in Taylor series, thereby producing integrals over configuration spaces of $V$ points on $M$ (with $V$ the term in the Taylor series for the exponential); then one averages individual terms with (Fresnel version of) Gaussian measure $\DD\phi\; e^{\frac{i}{2\hbar}\int_M \lan \phi,\mathfrak{D}\phi \ran }$ using (formally) Wick's lemma.

\subsubsection{Example: scalar theory with $\phi^3$ interaction}
Let $(M,g)$ be a compact Riemannian manifold. Consider the path integral 
\begin{equation}\label{l13_PI_phi3}
I(\hbar)=\int_{C^\infty(M)}\DD\phi \; e^{\frac{i}{\hbar}\int_M \left(\frac12\lan  d\phi,d\phi\ran_{g^{-1}}+\frac{m^2}{2}\phi^2 + \frac{g}{3!}\phi^3\right)d\mr{vol}}
\end{equation}
where $m>0$ is a parameter of the theory -- the ``mass'' (of the field quanta); $g$ is a coupling constant and we treat the $\phi^3$ as perturbing the Gaussian integral. Perturbative evaluation of (\ref{l13_PI_phi3}) yields
\begin{equation}\label{l13_PI_phi3_pert}
I_\mr{pert}(\hbar)={\det}^{-\frac12}(\Delta+m^2)\cdot \sum_\Gamma \frac{\hbar^{-\chi(\Gamma)}}{|\mr{Aut}(\Gamma)|}\Phi(\Gamma)
\end{equation}
where the sum goes over 3-valent graphs $\Gamma$, with 
\begin{equation}\label{l13_phi3_Phi}
\Phi(\Gamma)=g^V\int_{M^{\times V}}d^nx_1\cdots d^nx_V\; \prod_{e=(v_1,v_2)}G(x_{v_1},x_{v_2})
\end{equation}
where $V$ is the number of vertices in $\Gamma$, $dx_i$ stands for the Riemannian volume element on $i$-th copy of $M$, the product goes over edges $e$ of $\Gamma$ and $v_1,v_2$ are the vertices adjacent to the edge; $G(x,y)$ is the Green's function for the differential operator $\Delta+m^2$.

\begin{remark} One can represent the Green's function $G(x,y)$ by Feynman-Kac formula, as an integral over paths on $M$ going from $y$ to $x$. Then $\Phi(\Gamma)$ becomes represented as an integral over the mapping space $\mr{Map}(\Gamma, M)$. Note that this mapping space is fibered over $M^{\times V}$ (by evaluating the map at the vertices of $\Gamma$) and r.h.s. of (\ref{l13_phi3_Phi}) can be viewed as the result of the fiber integral over fibers of $\mr{Map}(\Gamma,M)\ra M^{\times V}$ (i.e. over paths on $M$ representing the edges of $\Gamma$, between vertices fixed at points $x_1,\ldots,x_V$ on $M$).
\end{remark}

\begin{example} The contribution of theta graph to the r.h.s. of (\ref{l13_PI_phi3_pert}) is:
\begin{equation}\label{l13_phi3_theta}
\frac{\hbar}{12}\;\Phi\left(\vcenter{\hbox{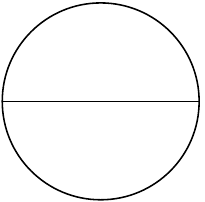}}\right)\quad = \quad \frac{\hbar\, g^2}{12} \int_{M\times M} d^nx\, d^n y\; G(x,y)^3
\end{equation}
And the contribution of the dumbbell graph is:
\begin{equation}\label{l13_phi3_dumbbell}
\frac{\hbar}{8}\;\Phi\left(\vcenter{\hbox{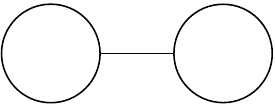}}\right)\quad = \quad \frac{\hbar\, g^2}{8} \int_{M\times M} d^nx\, d^ny\; G(x,x)\, G(x,y)\, G(y,y)
\end{equation}
\end{example}

Similarly, one can calculate expectation values, e.g. of products $\prod_{i=1}^m \phi(x_i)$ of the values of the field $\phi$ in several fixed points on $M$, with respect to the perturbed Gaussian measure (the integrand of (\ref{l13_PI_phi3})). The result is again given as a sum over graphs, with several unique marked vertices.

\begin{example} The following Feynman graph gives a contribution to the normalized expectation value (w.r.t. to the perturbed measure) $\ll \phi(x_1) \phi(x_2) \gg_\mr{pert}$:
\begin{equation}\label{l13_phi3_corr}
\frac{\hbar^2}{2}\;\Phi\left(\vcenter{\hbox{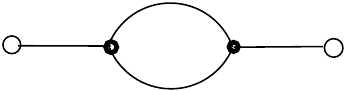}}\right)\quad = \quad
\frac{\hbar^2\, g^2}{2}\int_{M\times M} d^nx\,d^ny\; G(x_1,x) \, G(x,y)^2\, G(y,x_2)
\end{equation}
Here the two marked vertices are fixed at points $x_1,x_2$ whereas the unmarked vertices move around and we integrate over their possible positions on $M$.
\end{example}

\subsubsection{Divergencies!} \textbf{Problem:}
Green's function $G(x,y)$ for the operator $\Delta+m^2$ an $n$-dimensional Riemannian manifold $M$ behaves,as the points $x$ and $y$ approach each other, as 
$$G(x,y)\underset{x\ra y}\sim 
\frac{\mr{const}}{|x-y|^{n-2}} 
$$
(Case $n=2$ is special: then $G(x,y)\sim C\cdot\log|x-y|$.) This implies that the integrals over $M^{\times V}$ on the r.h.s. of (\ref{l13_phi3_Phi}) are, typically, 
(depending on $n=\dim M$ and on the combinatorics of $\Gamma$, see examples below) 
divergent: the integrand typically has non-integrable singularities near diagonals of $M^{\times V}$.

\marginpar{\LARGE{Lecture 14, 10/10/2016.}}



\textbf{Examples.}\\
\begin{enumerate}[(i)]
\item for $n=2$  
and $\Gamma$ any graph without ``short loops'' (edges connecting a vertex to itself), there is no divergency.
\item The integrand in (\ref{l13_phi3_theta}) behaves as $\frac{1}{|x-y|^{3n-6}}$ near the diagonal $x=y$; this singularity is non-integrable iff $3n-6\geq n$ or, equivalently, if $n\geq 3$. So, for $M$ of dimension $\geq 3$, theta graph for scalar $\phi^3$ theory is divergent.
\item By a similar argument, graph (\ref{l13_phi3_corr}) diverges iff $2\cdot (n-2)\geq n$  or equivalently $n\geq 4$.
\item For the graph (\ref{l13_phi3_dumbbell}), singularity of $G(x,y)$ on the diagonal $x=y$ is always integrable but evaluations of the propagator at coinciding points $G(x,x)$ and $G(y,y)$, corresponding to short loops of the graph, are ill-defined for $n\geq 2$.
\item Consider the graph
\begin{multline*}
\Phi\left(\vcenter{\hbox{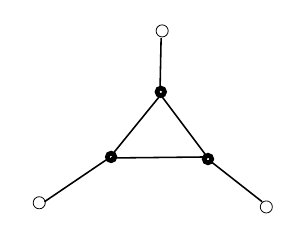}}\;\;\;\;\right)\quad =\\ =
g^3\int_{M\times M\times M} d^n x\, d^n y \, d^n z\underbrace{G(x_1,x) G(x_2,y) G(x_3,z) G(x,y) G(y,z) G(z,x)}_\psi
\end{multline*}
contribution to the 3-point correlation function $\ll \phi(x_1)\phi(x_2)\phi(x_3) \gg_\mr{pert}$. The integrand $\psi$ has integrable singularities at all diagonals where \emph{pairs} of points collide. However,near the diagonal $x=y=z$, when $x,y,z$ are within distance of order $r\ra 0$ of each other, we have  $\psi\sim\frac{1}{r^{3(n-2)}}$, and we think of the integral as $\int_{M} d^n x \int_{M\times M}d^n y\, d^n z$. The internal integral over $y,z$ for fixed $x$ diverges iff $3(n-2)\geq 2n$ or equivalently $n\geq 6$.
\end{enumerate}

Generally, one can say whether the graph diverges or not by analyzing the behavior of the integrand at all diagonals. The answer is as follows. Define the \textbf{weight} $w(\Gamma')$ of a graph $\Gamma'$ with $E_{\Gamma'}$ edges and $V_{\Gamma'}$ vertices as
$$w(\Gamma'):=E_{\Gamma'}\cdot (n-2)-(V_{\Gamma'}-1)\cdot n$$
\begin{lemma}\label{l14_lm_divergence}
$\Phi(\Gamma)$ diverges iff the graph $\Gamma$ contains a subgraph $\Gamma'\subset \Gamma$ with non-negative weight $w(\Gamma')\geq 0$.
\end{lemma}
This lemma applies to scalar theory with arbitrary polynomial interaction $p(\phi)$, not necessarily $\phi^3$ (monomials present in $p(\phi)$ restrict admissible valencies of vertices of contributiong graphs $\Gamma$).

\begin{remark} \label{l14: rem phi3 divergencies}
Consider $\phi^3$ theory on a manifold of dimension $n$. 
\begin{itemize}
\item For $n=3$,
a graph $\Gamma$ diverges iff $\Gamma$ either contains a short loop or contains a theta graph (\ref{l13_phi3_theta}) as a subgraph (a corollary of Lemma \ref{l14_lm_divergence}).
\item More generally, for $n<6$, there is a finite list of subgraphs with non-negative weight.
\item For $\Gamma'\subset \Gamma$, let us call ``leaves'' of $\Gamma'$ the edges connecting vertices of $\Gamma'$ to vertices of $\Gamma$ not belonging to $\Gamma'$.  For $n=6$, the weight of $\Gamma'$ is non-negative, iff the number of leaves of $\Gamma'$ is $\leq 3$. (There are infinitely many such subgraphs.)
\item For $n>6$, there are infinitely many divergent $\Gamma'$ and there is nor restriction on the number of leaves for them.
\end{itemize}
\end{remark}

\subsubsection{Regularization and renormalization}\footnote{Here is a standard easy-going textbook reference (intended for physics students): \cite{Peskin}.}
The logic of dealing with divergencies of Feynman graphs for the path integral is to first introduce a\\
\textbf{Step I: Regularization.} We want to replace the path integral $I(\hbar)$ by a regularized version $I_\epsilon (\hbar)$ with a small parameter $\epsilon$ the \emph{regulator}. Here are some of the ideas of regularization.
\begin{enumerate}[a)]
\item Replace $M\ra M_\epsilon$ -- a lattice or trianglation or cellular decomposition with spacing/typical cell size $\epsilon$. Space of fields $F$ gets replaces by a finite-dimensional space $F_\epsilon$ (modelled on functions on the set vertices or, e.g., cellular cochains of $M_\epsilon$). Action $S$ gets replaced by a finite-difference approximation $S_\epsilon$. Then $I_\epsilon(\hbar)=\int_{F_\epsilon}e^{\frac{i}{\hbar}S_\epsilon}$ is a well-defined finite-dimensional integral. It can be developed in Feynman graphs, $I_\epsilon(\hbar)\propto \sum_\Gamma \Phi_\epsilon(\Gamma)$ with $\Phi_\epsilon(\Gamma)$ the regularized (finite) weights of Feynman graphs.
\item Regularize the Feynman weights of graphs directly (without deriving this regularization from a regularization of the path integral itself), $\Phi(\Gamma)\ra \Phi_\epsilon (\Gamma)$. E.g. regularize the propagator $G(x,y)$ as follows (some of the possible options):
\begin{enumerate}[1.]
\item Proper time cut-off: $G_\epsilon(x,y)=\int_\epsilon^\infty dt\, K(x,y | t)$ with $K(x,y|t)$ the \emph{heat kernel} -- the integral kernel of the operator $e^{-t(\Delta+m^2)}$.
\item Spectral cut-off: $G_\Lambda(x,y)=\sum_{\lambda<\Lambda} \frac{1}{\lambda}\Psi_\lambda(x) \overline{\Psi_\lambda(y)}$ where $\lambda$ runs over eigenvalues of the operator $\Delta+m^2$ (up to $\Lambda$) and $\Psi_\lambda$ are the corresponding eigenfunctions. Here the cut-off $\Lambda=1/\epsilon$ is large rather than small.
\item Momentum cut-off (case of $M=\RR^n$): $G_\Lambda(x,y)=\int_{|k|<\Lambda} d^n k\, \frac{e^{i (k,x-y)}}{k^2+m^2}$ where the integral is over a ball of large radius $\Lambda=\epsilon^{-1}$ in the momentum space $(\RR^n)^*\ni k$.
\item Regularization 
$G_\epsilon(x,y)=\int_0^\infty dt\,t^\epsilon\, K(x,y|t)$, with $\epsilon$ the regulator. The integral over $t$ is convergent for $\mr{Re}(\epsilon)>\frac{n}{2}-1$, and possesses a meromorphic continuation to the entire $\CC\ni \epsilon$; we are interested in the limit $\epsilon\ra 0$ of the continuation.
\end{enumerate}
\end{enumerate}

\begin{remark} The functional determinant in (\ref{l13_PI_phi3_pert}) also has to be regularized, e.g. via zeta-regularization, as $\det_{\zeta\mr{-reg}} (\Delta+m^2):=e^{-\zeta'(0)}$ with $\zeta(s)=\sum_\lambda \lambda^{-s}$ the zeta function of the operator $\Delta+m^2$  ($\lambda$ runs over the eigenvalues and it is implied that we take the analytic continuation of the zeta function to $s=0$).
\end{remark}

Whichever way we go about regularization, we get regularized weights of Feynman graphs $\Phi_\epsilon (\Gamma)$. However, the limit of removing the regulator $\lim_{\epsilon\ra 0}\Phi_\epsilon (\Gamma)$ typically does not exist. To deal with this, we introduce\\
\textbf{Step II: Renormalization.}\\
We replace the action with the renormalized action 
\begin{equation}\label{l14: S renorm}
S(\phi)\ra \til{S}_\epsilon(\phi)=S(\phi)+\sum_i c_i(\epsilon) A_i(\phi)
\end{equation}
where corrections $A_i(\phi)=\int_M d^n x\; \mc{A}_i(\phi)$ are local expressions in the field $\phi$ -- \textit{counterterms}, with coefficients $c_i(\epsilon)$ diverging as $\epsilon^{-k}$ (for some positive $k$) or $\log\epsilon$ as $\epsilon\ra 0$. Replacement (\ref{l14: S renorm}) should be such that when we compute Feynman diagrams for the renormalized action $\til\Phi_\epsilon(\Gamma)$, the limit $\epsilon\ra 0$ exists.\footnote{To be more precise: counterterms in the renormalized action produce new vertices (with $\epsilon$-dependent coefficients) for the Feynman rules. Contributions of graphs containing these new vertices compensate for the divergence, in the limit $\epsilon\ra 0$, of the graphs of original theory.}

Thus, local action $S(\phi)$ is replaced by $\til{S}_\epsilon(\phi)$ with counterterms divergent as the regulator $\epsilon\ra 0$, but the path integral is now perturbatively well-defined: 
$$\lim_{\epsilon\ra 0} \til{I}_\epsilon(\hbar)=:\til{I}(\hbar)$$ 
where l.h.s. is defined by regularized Feynman diagrams for the renormalized action.

In practice, counterterms in (\ref{l14: S renorm}) correspond to the possible divergent subgraphs (cf. Lemma \ref{l14_lm_divergence}) and are introduced in order to compensate for these divergencies.
E.g. in scalar theory with polynomial perturbation $p(\phi)$, one can 
assign to a divergent subgraph $\Gamma'$ of weight $w(\Gamma')\geq 0$ with $d$ leaves the counterterm $
\mc{A}_{\Gamma'}(\phi)=\phi(x)^d$ with coefficient $c_{\Gamma'}(\epsilon)=c_{\Gamma'}\cdot \epsilon^{-w(\Gamma')}$ if the weight $w(\Gamma')>0$  and $c_{\Gamma'}(\epsilon)=c_{\Gamma'}\cdot\log\epsilon$ if $w(\Gamma')=0$ with $c_{\Gamma'}$ a constant.
\begin{remark}
 In particular, by Remark \ref{l14: rem phi3 divergencies}, for $\phi^3$ scalar theory in dimension $<6$, we need finitely many counterterms of form $\phi^d$ for some values of $d\geq 0$ (number of leaves of $\Gamma'$) in (\ref{l14: S renorm}).  In dimension $6$ there are infinitely many divergent subgraphs, but we only need counterterms $\phi^d$ with $0\leq d\leq 3$. In dimension $>6$, we need counterterms of form $\phi^d$ for all $d$. Thus, one says that in dimensions up to $6$, scalar $\phi^3$ theory is \emph{renormalizable} (finitely many counterterms) and in dimensions $>6$ it is \emph{non-renormalizable}.
\end{remark}

\subsubsection{Wilson's picture of renormalization (``Wilson's RG flow'')}\label{sss_Wilson_RG_flow}

In Wilson's picture \cite{Wilson},\footnote{See also \cite{Costello} for an interpretation of Wilson's RG flow via effective BV actions.} one considers the \emph{tower} of spaces of fields $F_\Lambda$ with different values of cut-off $\Lambda$ (originally, the momentum cut-off, though other realizations are possible, see below), equipped with associated actions $S_\Lambda$ ``at cut-off $\Lambda$'' (``Wilson's effective actions''):
\begin{equation}\label{l14 Wilson tower}
\underbrace{F=F_\infty,S}_{\mr{local\; theory}}\cdots  \twoheadrightarrow \;\; \underbrace{F_{\Lambda},S_\Lambda}_{\mr{theory\;at\;finite\;}\Lambda} \;\;\proj \;\; F_{\Lambda'},S_{\Lambda'}\;\; \proj \cdots \proj \;\; \underbrace{F_0, S_0}_{\mr{effective\; theory\; on\; zero-modes}}
\end{equation}
For $\Lambda > \Lambda'$, we have a projection 
\begin{equation}\label{l14: P}
P^{\Lambda\ra \Lambda'}: F_\Lambda\proj F_{\Lambda'}
\end{equation} 
and the actions are related by a pushfroward (fiber integral) $S_{\Lambda'}=P^{\Lambda\ra \Lambda'}_* S_\Lambda$ defined by
\begin{equation}\label{l14_Wilson_eff_action}
e^{\frac{i}{\hbar}S_{\Lambda'}(\phi')}:=\int \DD\til\phi\; e^{\frac{i}{\hbar}S_{\Lambda}(\phi'+\til\phi)}
\end{equation}
where we are integrating over $\til\phi$ in the fiber $\til{F}_{\Lambda,\Lambda'}$ of the projection (\ref{l14: P}) -- ``fields between $\Lambda$ and $\Lambda'$''.

Examples of realizations:
\begin{enumerate}
\item Wilson's original realization. For $M=\RR^n$, take $F_\Lambda$ to be the space of functions of form $\phi(x)=\int_{B_\Lambda\subset (\RR^n)^*} d^n k\; e^{i (k,x)} \psi(k)$ where $B_\Lambda=\{k\in (\RR^n)^*\;\mr{s.t.}\; ||k||\leq\Lambda\}$. I.e. $F_\Lambda$ consists of functions whose Fourier transform is supported inside the ball of radius $\Lambda$ in the momentum space $(\RR^n)^*\ni k$.
Then, for $\Lambda\ra \Lambda'$, pushforward $P^{\Lambda\ra \Lambda'}_*$ corresponds to integrating out fields in a spherical layer $\Lambda'<||k||\leq \Lambda$ in the momentum space.
$$\vcenter{\hbox{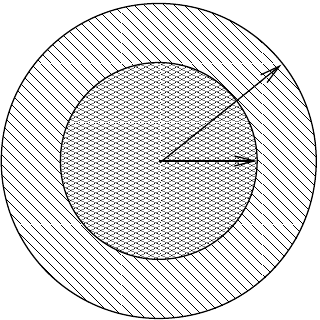}}$$ 
Picture of Wilson's ``renormalization group (RG) flow'' amounts to ``flowing'' from theory at large $\Lambda_\mr{big}$ (the cut-off) to theory at small $\Lambda$ by successively integrating out thin spherical layers in the momentum space.
\item For $M$ compact, we can take $F_\Lambda=\mr{Span}_{\lambda\leq \Lambda} \{\Psi_\lambda\}$ -- the span of eigenfunctions of the operator $\Delta+m^2$ with eigenvalues $\lambda\leq \Lambda$.
\item Let $\cdots \succ T_{i+1}\succ T_i \succ\cdots$ be a sequence of CW decompositions of $M$ such that $T_{i+1}$ is a subdivision $T_i$ (the we say that $T_i$ is an \emph{aggregation} of $T_{i+1}$) and mesh (typical size of cells) of $T_i$ decays fast enough as $i\ra\infty$.
$$\vcenter{\hbox{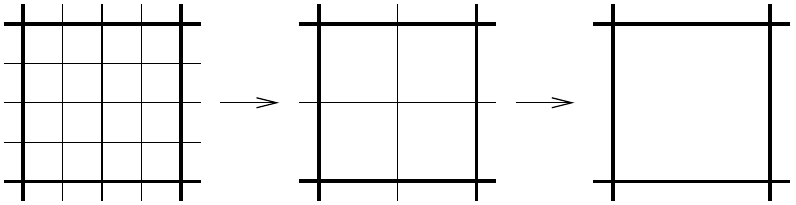}}$$
We can set $F_i=C^0(T_i)$ -- zero-cochains (functions on vertices of $T_i$), and $S_i\in \mr{Fun}(F_i)$ a suitable finite-difference replacement of the action satisfying the compatibility condition w.r.t. aggregations $S_i=P_*(S_{i+1})$.\footnote{See \cite{SimpBF,DiscrBF,CMRcell} for an example; there we need cochains of all degrees in $F_i$.}
\end{enumerate}

\begin{remark}
Pushforwards out of the top tier $F,S$ of the tower (\ref{l14 Wilson tower}) are ill-defined, and it has to be replaced with the asymptotic ``tail'' of the tower $F_{\Lambda_\mr{big}}, S_{\Lambda_\mr{big}}$ with $S_{\Lambda_\mr{big}}(\phi)\underset{\Lambda_\mr{big}\ra\infty }\sim \til{S}_{\Lambda_\mr{big}}(\phi)=S(\phi)+\sum_i c_i(\Lambda_\mr{big})  A_i(\phi)$ the renormalized action (\ref{l14: S renorm}). Then, if e.g. $F_0$ is a point, $S_0$ is given by the sum of connected Feynman diagrams for the renormalized action.
\end{remark}

\marginpar{\LARGE{Lecture 15, 10/12/2016.}}

\section{Batalin-Vilkovisky formalism}

\subsection{Faddeev-Popov construction}
Faddeev-Popov construction appeared in \cite{FaddeevPopov}  as a way to resolve the problem of degeneracy of critical points of the Yang-Mills action, in order to construct the perturbative path integral (Feynman diagrams) for the Yang-Mills theory. The construction in fact applies to a large class of gauge theories. Here we study a finite-dimensional model for this situation.

Let $G$ be a compact Lie group of dimension $m$ acting freely on a finite-dimensional $n$-manifold $X$ with 
\begin{equation}\label{l15_gamma}
\gamma: G\times X\ra X
\end{equation} the action map. Let $\g=\mr{Lie}(G)$ be the Lie algebra of $G$ and assume that we have chosen a basis $\{T_a\}$ in $\g$. Denote by $v_a\in\mathfrak{X}(X)$ the fundamental vector fields on $X$ by which the generators $T_a$ act on $X$. 

Let $S\in C^\infty(X)^G$ be a $G$-invariant function on $X$, 
and let $\mu \in \Omega^n_c(X)^G$ be a $G$-invariant top form with compact support.

We are interested in the integral
\begin{equation}\label{l15_I}
I=\int_X \mu\; e^{\frac{i}{\hbar}S}
\end{equation}
We can rewrite it as the integral over the quotient $X/G$: 
\begin{equation}\label{l15_I via X/G}
I=\mr{Vol}(G)\int_{X/G} \til \mu\; e^{\frac{i}{\hbar}\til S}
\end{equation}
Where $\til S\in C^\infty(X/G)$ is such that 
\begin{equation}\label{l15_S}
S=p^* \til S
\end{equation} where $p:X\ra X/G$ is the quotient map; $\til \mu\in \Omega^{n-m}(X/G)$ is a top form on the quotient constructed in such a way that 
$$\iota_{v_m}\cdots\iota_{v_1}\mu=p^*\til \mu$$ Note that the $(n-m)$-form on the l.h.s. here is \textit{basic} (invariant and horizontal w.r.t $G$-action) and hence is a pullback from the quotient. Note that we can write 
\begin{equation}\label{l15_mu}
\mu = p^*\til\mu \wedge \chi
\end{equation} 
where $\chi\in \Omega^m(X)$ is a (any) form on $X$ with the property that its restrictions to $G$-orbits in $X$ yield the volume form on the orbits induced from Haar measure on $G$ (via the identification of an orbit with $G$ by picking a base point on the orbit). The normalization of $\chi$ should be such that $\iota_{v_m\wedge\cdots \wedge v_1}\chi=1$. Note that (\ref{l15_S},\ref{l15_mu}) together imply (\ref{l15_I via X/G}).

Let $\phi:X\ra \g$ be a $\g$-valued function on $X$ such that:
\begin{itemize}
\item zero is a regular value of $\phi$,
\item $\sigma=\phi^{-1}(0)\subset X$ intersects every $G$-orbit transversally, exactly $N$ times, for some fixed $N\geq 1$.\footnote{Ideally, we would like to have a single intersection, i.e. $N=1$, but typically, for $G$ compact, there are topological obstructions for having a global section of $p:X\ra X/G$ defined as a zero locus of a globally defined function. E.g. for $G=U(1)$, orbits are circles, thus $\phi$ has to have some even number of zeroes on an orbit.} 
\end{itemize}
We think of $\sigma$ as a (local) section of $G$-orbits. We refer to $\sigma$ as the \textbf{gauge-fixing} (and to $\phi$ as the \textit{gauge-fixing function}).

Since $\sigma\subset X$ is an $N$-fold covering of the quotient $X/G$, (\ref{l15_I via X/G}) implies
\begin{equation}\label{l15_e1}
I=\frac{\mr{Vol}(G)}{N} \int_{\sigma}\left.\iota_{v_m\wedge\cdots\wedge v_1}\mu\;\; e^{\frac{i}{\hbar}S} \right|_\sigma =
\frac{\mr{Vol}(G)}{N} \int_X \delta^{(m)}(\phi)\;\;\iota_{v_m\wedge\cdots\wedge v_1}\mu\;\; e^{\frac{i}{\hbar}S}
\end{equation}
Here $\delta^{(m)}(\phi)=\delta(\phi)\cdot \bigwedge_a d\phi^a$ is the distributional $m$-form supported on $\sigma$; $\delta(\phi)=\prod_a \delta(\phi^a(x))$ is the delta-distribution (not a form) supported on $\sigma\subset X$. We can think of $\delta(\phi)$ and $\delta^{(m)}(\phi)$ as the pullbacks by $\phi$ of the standard Dirac delta function and delta form, respectively, centered at the origin in $\g$.

Note that, generally, for $C\subset X$ a $k$-cycle, we have a distributional form $\delta_C^{(n-k)}:\Omega^k(X)\ra\RR$ mapping 
$$\omega\mapsto \int_C \omega|_C=: ``\;\;{\int_X \delta_C^{(n-k)}\wedge \omega}\;\;"$$
Formula (\ref{l15_e1}) is a special case of this, for $C=\sigma$.

\begin{remark}
The delta form $\delta^{(m)}(\phi)$ depends only on zero-locus of $\phi$ and, in particular, does not change under rescaling $\phi\mapsto \lambda\cdot\phi$ with $\lambda\neq 0$ a constant. On the other hand, the delta function $\delta(\phi)$ changes with rescaling of $\phi$, by $\lambda^{-m}$.
\end{remark}

Let $J$ be a function on $X$ such that
\begin{equation}\label{l15_e2}
\bigwedge_a d\phi^a\wedge \iota_{v_m\wedge\cdots\wedge v_1}\mu = J\cdot \mu
\end{equation}
\begin{lemma} The coefficient $J$ in (\ref{l15_e2}) is:
\begin{equation}\label{l15_J}
J(x)={\det}_\g FP(x)
\end{equation}
where 
\begin{equation}\label{l15_FP}
FP(x)=d_x\phi\circ d_{1,x}\gamma:\;\g\ra\g
\end{equation}
is an endomorphism of $\g$ depending on a point $x\in X$; here $d_{1,x}\gamma: \g\ra T_x X$ is the infinitesimal action of $\g$ on $X$ viewed as a derivative of the group action (\ref{l15_gamma}); $d_x\phi: T_xX\ra \g$ is the derivative of $\phi$. In components, we have 
\begin{equation}\label{l15_FP_matrix}
FP(x)^a_b=\lan d\phi^a(x) ,v_b(x)\ran = v^b(\phi_a)|_x
\end{equation}
\end{lemma}

One calls $J(x)$ given by (\ref{l15_J}) the \textit{Faddeev-Popov determinant}.

\begin{proof} First note that nondegeneracy of $FP(x)$ is equivalent to $\phi^{-1}(\phi(x))\subset X$ intersecting the $G$-orbit through $X$ transversally. If the intersection is nontransversal, then l.h.s. of  (\ref{l15_e2}) is obviously vanishing ant the statement is trivial. So, we assume that the intersection is transversal, i.e. theat $FP(x)$ is non-degenerate.

Let $V=\mr{im}d_{1,x}\gamma=\mr{Span}(v_a(x))\subset T_xX$ be the tangent space to $G$-orbit through $x$ and let $\mr{Ann}(V)\subset T^*_x X$ be its annihilator in the cotangent space. Let $\alpha_1,\ldots,\alpha_{n-m}$ be a basis in $\mr{Ann}(V)$. We have a basis $(d\phi^1(x),\ldots,d\phi^m(x),\alpha_1,\ldots,\alpha_{n-m})$ in $T^*_x X$ (fact that this is a basis is equivalent to non-degeneracy of $FP(x)$ which we assumed). Without loss of generality (by normalizing $\alpha$s appropriately), we may assume $\mu =\bigwedge_{a=1}^m d\phi^a(x)\wedge \alpha_1\wedge\cdots\wedge \alpha_{n-m}$. Contracting with $v_m\wedge\cdots \wedge v_1$ and using orthogonality of $v$s and $\alpha$s, we have
$$\iota_{v_m\wedge\cdots \wedge v_1}\mu=\left(\sum_{s\in S_m}(-1)^s\prod_{a=1}^m  \lan d\phi^a,v_{s(a)} \ran \right)\alpha_1\wedge\cdots\wedge \alpha_{n-m}={\det}_\g FP(x)\cdot \alpha_1\wedge\cdots\wedge \alpha_{n-m} $$
Wedging with $\bigwedge_{a=1}^m d\phi^a(x)$, we get the statement of the Lemma.
\end{proof}

Thus, we have the following.
\begin{theorem}[Faddeev-Popov]
\begin{equation}\label{l15_FP1}
\int_X \mu\;e^{\frac{i}{\hbar}S} =\frac{\mr{Vol}(G)}{N} \int_X \mu\; \delta(\phi(x))\cdot {\det}_\g FP(x)\cdot e^{\frac{i}{\hbar}S}
\end{equation}
\end{theorem}

Next, we would like to deal with integrals of stationary phase type, i.e. with integrands of form $e^{\frac{i}{\hbar}(\cdots)}$. We can achieve that, at the cost of introducing auxiliary integration variables, by using integral presentations for the delta function (as a Fourier transform on the unit) and for the determinant (as a Gaussian integral over odd variables):
\begin{eqnarray}
\delta(\phi(x))&=&  \frac{1}{(2\pi\hbar)^m}\int_{\g^*} d^m\lambda \; e^{\frac{i}{\hbar}\lan \lambda , \phi(x) \ran} \label{l15_delta_as_integral}\\
{\det}_\g FP(x) &=& \left(\frac{\hbar}{i}\right)^m \int_{\Pi(\g\oplus \g^*)} \prod_{a=1}^m (Dc_a D \bar c_a)\; e^{\frac{i}{\hbar}\lan \bar{c}, FP(x) c \ran } \label{l15_detFP_as_integral}
\end{eqnarray}
Here the auxiliary odd variables $c_a$, $\bar c_a$ are called \textit{Faddeev-Popov ghosts}; $\lambda$ is the even \emph{Lagrange multiplier} variable. For brevity, we will denote the odd Berezin measure in 
(\ref{l15_detFP_as_integral}) by $D^m c\; D^m \bar{c}$. Plugging integral presentations (\ref{l15_delta_as_integral},\ref{l15_detFP_as_integral}) into (\ref{l15_FP1}), we obtain the following.
\begin{theorem}[Faddeev-Popov]
\begin{equation}\label{l15_FP2}
\int_X \mu\;e^{\frac{i}{\hbar}S} =\frac{\mr{Vol}(G)}{N\cdot (2\pi i)^m} \int_{X\times \g^*\times \Pi(\g\oplus \g^*)} \mu\; d^m \lambda \; D^m c\; D^m \bar{c}\;\;e^{\frac{i}{\hbar}S_{FP}(x,\lambda,c,\bar{c})}
\end{equation}
where
\begin{equation}\label{l15_S_FP}
S_{FP}(x,\lambda,c,\bar{c})=S(x)+\lan \lambda , \phi(x) \ran + \lan \bar{c} , FP(x) c \ran
\end{equation}
is the \emph{Faddeev-Popov action} associated to the gauge-fixing $\phi$.
\end{theorem}

The point of replacing the integral (\ref{l15_I}) with the r.h.s. of (\ref{l15_FP2}) is that the former cannot be calculated, in the asymptotics $\hbar\ra 0$, by stationary phase formula, since the critical points of $S$ are not isolated but rather come in $G$-orbits (hence the Hessian of $S$ at a critical point is always degenerate and one cannot construct Feynman rules in this case). On the other hand the integral in the r.h.s. of (\ref{l15_FP2}) has isolated critical points with non-degenerate Hessians of the extended action $S_{FP}$ and the stationary phase formula is applicable.

\subsubsection{Hessian of $S_{FP}$ in an adapted chart}
Let $x_0$ be a critical point of $S$ lying on a critical $G$-orbit $[x_0]\subset X$ and satisfying $\phi(x_0)=0$. Let $(y_1,\ldots,y_{n-m};z^1,\ldots,z^m)$ be an adapted local coordinate chart on $X$ near $x_0$, such that:
\begin{enumerate}[(i)]
\item $x_0$ is given by $y=z=0$.
\item\label{l15_item_ii} $[x_0]$ is given by $y=0$; moreover, $G$-orbits are locally given by $y=\mr{const}$.
\item Locally $\phi$ is given by $\phi^a=z^a$.
\end{enumerate}
$$\vcenter{\hbox{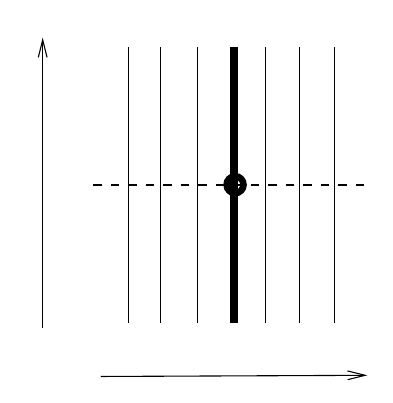}}$$
For instance, $G$-invariance of $S$ implies that $S=S(y)$ and $\frac{\dd}{\dd z^a}S=0$.

Hessian of $S$ has the form
$$\dd^2 S|_{x_0}=\left(\begin{array}{c|c} \left.\frac{\dd^2 S}{\dd y_i \dd y_j}\right|_{x_0} & 0 \\\hline 0 & 0 \end{array} \right) $$
where first $(n-m)$ rows/columns correspond to $y_i$ variables and the last $m$ rows/columns correspond to $z^a$ variables. We are assuming that the block $\left.\frac{\dd^2 S}{\dd y_i \dd y_j}\right|_{x_0} $ is non-degenerate, i.e. that all degeneracy of the Hessian of $S$ comes from $G$-invariance. In other words, we assume that $\mr{rank}(\dd^2 S|_{x_0})=n-m$.

The Hessian $\dd^2 S|_{x_0}$ is, obviously, degenerate. However, let us consider
\begin{equation}\label{l15_Hess1}
\underbrace{\left.\dd^2(S+\lan\lambda , \phi(x)\ran )\right|_{x_0,\lambda=0}}_{\in \mr{Sym}^2(T_{x_0}X\oplus \g^*)^*}=
\left(\begin{array}{c|c|c}
\left.\frac{\dd^2 S}{\dd y_i \dd y_j}\right|_{x_0} & 0 & 0 \\ \hline
0 & 0 & \delta^a_b \\ \hline 0 & \delta^b_a & 0
\end{array}
\right)
\end{equation}
Here rows correspond to $y_i, z^a, \lambda_a$ and columns correspond to $y_j,z^b,\lambda_b$. Note that this Hessian is \emph{non-degenerate}! The $z-\lambda$ blocks that appeared because of the new $\lan\lambda , \phi(x)\ran$ term make the matrix non-degenerate.

Next, note that assumption (\ref{l15_item_ii}) above implies that fundamental vector fields $v_a$ locally have the form $v_a=\sum_{b}f^b_{\;\;a}(y,z)\frac{\dd}{\dd z^b}$ with $(f^b_{\;\;a})(y,z)$ a non-degenerate $m\times m$ matrix. Thus, by (\ref{l15_FP_matrix}), we have $FP(x)^a_b=f^a_{\;\;b}(y,z)$. Therefore, the part of the Hessian corresponding to the ghost part of Faddeev-Popov action is:
\begin{equation}\label{l15_Hess2}
\left.\dd^2 \lan \bar{c}, FP(x) c \ran\right|_{x_0,c=\bar{c}=0} = 
\left(\begin{array}{c|c}
0 & - f^b_{\;\; a} \\ \hline
f^a_{\;\;b} & 0
\end{array}
\right)
\end{equation}
where rows correspond to $c^a,\bar{c}_a$ and columns correspond to $c^b$, $\bar{c}_b$.

Assembling (\ref{l15_Hess1}) and (\ref{l15_Hess2}), we get the full Hessian of Faddeev-Popov action
\begin{equation}\label{l15_full_Hessian}
\left.\dd^2 S_{FP} \right|_{x_0,\lambda=c=\bar{c}=0} = \left(\begin{array}{c|c|c||c|c}
\left.\frac{\dd^2 S}{\dd y_i \dd y_j}\right|_{x_0} &  &  &  & \\ \hline
& & \delta^a_b & & \\ \hline
& \delta^b_a & & & \\ \hline\hline
& & & & - f^{b}_{\;\; a} \\ \hline
& & & f^a_{\;\; b} & 
\end{array}
\right)
\end{equation}
with row variables $y_i,z^a,\lambda_a,c^a,\bar{c}_a$ and column variables $y_j,z^b,\lambda_b,c^b,\bar{c}_b$. All the non-filled blocks are zero. From this explicit form it is obvious that the full Hessian of the Faddeev-Popov action is non-degenerate.

\marginpar{\LARGE{Lecture 16, 10/24/2016.}}
\subsubsection{Stationary phase evaluation of Faddeev-Popov integral}
Critical point (Euler-Lagrange) equations for Faddeev-Popov action $S_{FP}(x,\lambda,c,\bar{c})$ (\ref{l15_S_FP}) read:
\begin{eqnarray}
c=\bar{c} &=& 0 \label{l16_EL_1} \\
\frac{\dd}{\dd x_i} S(x)+\lan \lambda,\frac{\dd}{\dd x_i} \phi \ran &=&0  \label{l16_EL_2} \\
\phi(x)&=&0  \label{l16_EL_3} 
\end{eqnarray}
Here (\ref{l16_EL_1}) is equivalent to the Euler-Lagrange equations $\frac{\dd}{\dd c^a}S_{FP}=0$, $\frac{\dd}{\dd \bar{c}_a}S_{FP}=0$, whereas (\ref{l16_EL_2}) corresponds to $\frac{\dd}{\dd x_i}S_{FP}=0$ (where we dropped the term bilinear in $c$ and $\bar{c}$ which is excluded by (\ref{l16_EL_1})); last equation (\ref{l16_EL_3}) is $\frac{\dd}{\dd \lambda_a}S_{FP}=0$.

Note that equations  (\ref{l16_EL_2},\ref{l16_EL_3}) together correspond to the fact that $x$ is a conditional extremum of $S$ restricted to submanifold $\sigma=\phi^{-1}(0)\subset X$ with $\lambda$ the Lagrange multiplier. On the other hand, $G$-invariance of $S$ together with transversality of the local section $\sigma$ and $G$-orbits, implies that a conditional extremum of $S$ on $\sigma$ is in fact a non-conditional extremum (i.e. $dS$ vanishes on the whole tangent space $T_x X$, not just on $T_x\sigma\subset T_x X$). Therefore, (\ref{l16_EL_2}) implies $\lambda=0$. Thus, a critical point of $S_{FP}$ has a form $(x_0,\lambda=c=\bar{c}=0)$ with $x_0$ an intersection point of the critical $G$-orbit of $S(x)$ with the gauge-fixing submanifold $\sigma=\phi^{-1}(0)$.

The Hessian of $S_{FP}$ at a critical point (written without using the adapted chart as in (\ref{l15_full_Hessian})), is
\begin{multline}
\left.\dd^2 S_{FP} \right|_{x_0,\lambda=c=\bar{c}=0} = \\
\left(\begin{array}{c|c||c|c}
\left.\frac{\dd^2 S}{\dd x^2} \right|_{x_0} & (d\phi|_{x_0})^T:\, \g^*\ra T^*_{x_0} X & & \\ \hline
d\phi|_{x_0}:\, T_{x_0}X\ra \g & 0 & & \\
\hline \hline & & & -FP(x_0)^T:\, \g^*\ra\g^* \\ \hline
& & FP(x_0):\, \g \ra\g &
\end{array}
\right)
\end{multline}
with blocks corresponding to variables $x,\lambda,c,\bar{c}$. Its inverse has the following structure:
\begin{equation}\label{l16_Hessian_inverse}
(\left.\dd^2 S_{FP} \right|_{x_0})^{-1} =
\left(\begin{array}{c|c||c|c}
\mc{D} & \beta & & \\
\hline
\beta^T & 0 & & \\
\hline \hline
& & & FP(x_0)^{-1} \\
\hline
& & -FP(x_0)^{-1 T} & 
\end{array}\right)
\end{equation}
Here $\beta: \g\ra T_{x_0}X$ is the section of the projection $d\phi|_{x_0}: T_{x_0}X \ra \g$ constructed as 
$$\beta=d_{1,x_0}\gamma\circ FP(x_0)^{-1}:\,\g\ra T_{x_0}X$$
where $d_{1,x_0}\gamma:\, \g \ra T_{x_0} X$ is, as in (\ref{l15_FP}), the infinitesimal action of the Lie algebra $\g$ on $X$ specialized at the point $x_0$. Thus, $\beta$ and $d\phi|_{x_0}$ together give us a splitting
\begin{equation}\label{l16_splitting}
T_{x_0}X\simeq T_{x_0}\phi^{-1}(0)\oplus \g
\end{equation}
The block $\mc{D}\in \Sym^2 T_{x_0}X$ in  (\ref{l16_Hessian_inverse}) is the image of $\til{\mc{D}}\in \Sym^2 T_{x_0}\phi^{-1}(0)$ under the splitting (\ref{l16_splitting}), where $\til{\mc{D}}$ is the inverse of $\dd^2_{x_0}(\left.S\right|_{\phi^{-1}(0)})$ -- the (invertible) Hessian of $S$ restricted to gauge-fixing submanifold $\phi^{-1}(x_0)$.

We say that $\mc{D}$ is the ``propagator'' or ``Green's function'' for $\dd^2S|_{x_0}$ in the gauge $\phi(x)=0$.

Applying the stationary phase formula to the Faddeev-Popov integral (\ref{l15_FP2}), we obtain the following.

\begin{theorem}[Stationary phase formula for Faddeev-Popov integral]\label{thm: FP stat phase}
\begin{multline}\label{l16_FP_stat_phase}
\int_X e^{\frac{i}{\hbar}S(x)}\mu = \\ 
\frac{\mr{Vol}(G)}{(2\pi i)^m} \sum_{\mr{crit.\;}G\mr{-orbits\;}[x_0]\mr{\;of\;}S
}
(2\pi \hbar)^{\frac{n+m}{2}}\left(\frac{i}{\hbar}\right)^{m}e^{\frac{i}{\hbar}S(x_0)} \left|\det \left.\dd^2_{x_0}S\right|_{\phi^{-1}(0)}\right|^{-\frac12} \cdot {\det}_\g FP(x_0)\cdot e^{\frac{\pi i}{4}\,\mr{sign}\left.\dd^2_{x_0}S\right|_{\phi^{-1}(0)}}\;\times\\
\times \sum_\Gamma \frac{\hbar^{1-\chi(\Gamma)}}{|\mr{Aut}(\Gamma)|}\cdot \Phi(\Gamma)
\end{multline}
Here in the r.h.s. we pick, for every critical $G$-orbit $[x_0]$ of $S$, a single representative $x_0$ -- one intersection point of $[x_0]$ with $\phi^{-1}(0)$. The Feynman rules for calculating $\Phi(\Gamma)$ are as follows.\\
\begin{tabular}{c|c}
Half-edge & field \\
\hline 
$\vcenter{\hbox{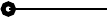}}$ & $y_i$\\
$\vcenter{\hbox{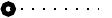}}$ & $\lambda_a$ \\
$\vcenter{\hbox{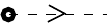}}$ & $c^a$ \\
$\vcenter{\hbox{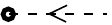}}$ & $\bar{c}_a$
\end{tabular}\quad
\begin{tabular}{c|c}
Edge & propagator \\
\hline
$\vcenter{\hbox{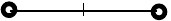}}$ & $i\mc{D}\;\in\Sym^2 T_{x_0}X$\\
$\vcenter{\hbox{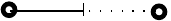}}$ & $i\beta:\;\g \ra T_{x_0}X$ \\
$\vcenter{\hbox{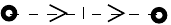}}$ & $iFP(x_0)^{-1}:\;\g\ra\g$
\end{tabular} \\
\begin{tabular}{c|c|c}
Vertex & $y$-valency & vertex tensor \\
\hline
$\vcenter{\hbox{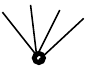}}$ & $k\geq 3$ & $\left.i\dd^k S\right|_{x_0}\;\in\Sym^k T_{x_0}^*X$ \\
$\vcenter{\hbox{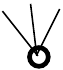}}$ & $l\geq 0$ & $\left.i\dd^l \rho\right|_{x_0}\;\in\Sym^l T_{x_0}^*X$ \\
$\vcenter{\hbox{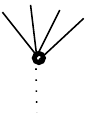}}$ & $j\geq 2$ & $\left.i\dd^j \phi\right|_{x_0}\;\in\Sym^j T_{x_0}^*X\otimes \g$ \\
$\vcenter{\hbox{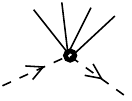}}$ & $q\geq 1$ & $\left.i\dd^q FP\right|_{x_0}\;\in\Sym^q T_{x_0}^*X\otimes \mr{End}(\g)$
\end{tabular}\\
Here we assume that local coordinates $y_i$ on $X$ are introduced near the critical point $x_0$.
``$y$-valency'' refers to the number of adjacent solid ($y$-)half-edges. The second vertex is the marked vertex that should appear in $\Gamma$ exactly once; $\rho$ is the density of the volume form $\mu$ in the local coordinates $y_i$, i.e. $\mu=\rho(y)d^n y$. 
\end{theorem}

\begin{remark}
In the special case when the gauge-fixing $\phi$ is linear in local coordinates $y_i$,\footnote{This is the finite-dimensional model for, e.g., the Lorentz gauge $d^*A=0$ in Yang-Mills theory, see Section \ref{sss: FP Yang-Mills} below} the third vertex above vanishes, and thus $\lambda$-half-edges do not appear in admissible graphs in the r.h.s. of (\ref{l16_FP_stat_phase}) at all. Here is a typical graph $\Gamma$ in such situation:
$$\vcenter{\hbox{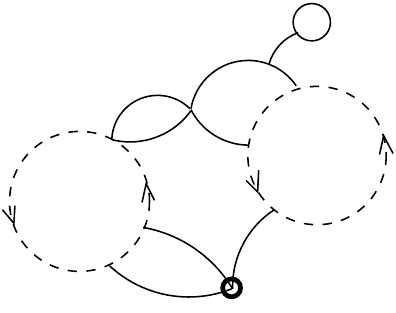}}$$
\end{remark}

\begin{remark} Assume that, in addition to $\phi$ being linear in $y_i$, fundamental vector fields have constant coefficients in local coordinates $y_i$ near $x_0$.\footnote{This is the finite-dimensional model for the Lorentz gauge in QED (abelian Yang-Mills theory) and explains why Faddeev-Popov ghosts do not appear in the Feynman diagrams for QED (but do appear in non-abelian Yang-Mills theory).} Then  both third and fourth vertex in the Feynman rules above vanish. In this case  one has only solid $y$-edges in admissible graphs $\Gamma$.
\end{remark}

\begin{remark}
In order to define invariantly (cf. Remark \ref{l8_rem_inv_Gaussian}) the determinant of the restricted Hessian $\det \dd^2_{x_0}S|_{\phi^{-1}(0)}$ appearing in the r.h.s. of (\ref{l16_FP_stat_phase}), we need a volume element on $T_{x_0}\phi^{-1}(0)$, i.e. an element in $\mr{Det}\,T^*_{x_0}\phi^{-1}(0)$.\footnote{Recall that, for $V$ a vector space, the determinant line $\mr{Det}\,V$ is the top exterior power of $V$,  $\mr{Det}\,V=\wedge^{\dim V}V$.} To construct it, we use the short exact sequence $T_{x_0}\phi^{-1}(0)\hra T_{x_0}X \xra{d\phi|_{x_0}}\g$ which induces a canonical isomorphism of determinant lines
$$\mr{Det}\, T^*_{x_0}X\cong \mr{Det}\, \g^*\otimes \mr{Det}\, T^*_{x_0}\phi^{-1}(0)$$
Using it, we can take the (canonically defined) ``ratio'' of $\mu|_{x_0}\in \mr{Det}\, T^*_{x_0}X$ (the volume form on $X$ evaluated at $x_0$) and $\mu_\g\in \mr{Det}\, \g^*$ -- the Lebesgue measure on $\g$, to obtain $\nu=\frac{\mu|_{x_0}}{\mu_\g}\;\in \mr{Det}\, T^*_{x_0}\phi^{-1}(0)$.
\end{remark}

\begin{remark}
In Theorem \ref{thm: FP stat phase}, instead of choosing the gauge-fixing $\phi:X\ra \g$ \emph{globally} on $X$, we can choose individual (\emph{local}) gauge-fixing $\phi_j:U_j\ra \g$ in a tubular neighborhood $U_j$ of $j$-th critical orbit $[x_0^{(j)}]$ of $S$, with $j$ going over 
all critical orbits.
\end{remark}

\subsubsection{Motivating example: Yang-Mills theory}\label{sss: FP Yang-Mills}

For $M$ a Riemannian (or pseudo-Riemannian) manifold, classical Yang-Mills theory on $M$ with structure group $G$ (a compact group with Lie algebra $\g$) has the space of fields 
$$F=\Conn_{M,G}\simeq \Omega^1(M)\otimes \g$$
-- the space of connections in a trivial $G$-bundle on $M$.\footnote{We restrict our discussion to the case of a trivial $G$-bundle for simplicity. This assumption can be relaxed.} The space of fields is acted on by the group of gauge transformations (principal bundle automorphisms), $\GGauge_{M,G}=C^\infty(M,G)$ and the action is given by $A\mapsto A^g=g^{-1}Ag+g^{-1}dg$. Infinitesimally, the Lie algebra of gauge transformations $\mr{gauge}_{M,G}\simeq \Omega^0(M,\g)$ acts by 
\begin{equation}\label{l16_e1}
A\mapsto d_A\alpha = d\alpha +[A,\alpha] \in T_A F
\end{equation}
for $\alpha\in \mr{gauge}_{M,G} $ the generator of the infinitesimal transformation.

Yang-Mills action is given by 
\begin{equation}\label{l16_S_YM}
S_{YM}(A)=\frac12 \int_M \tr F_A\wedge * F_A
\end{equation}
with $F_A=dA+\frac12 [A, A]\in \Omega^2(M,\g)$ the curvature of the connection; $*$ is the Hodge star associated to the metric on $M$; $\tr$ is the trace in the adjoint representation of $\g$.

Volume form $\mu$ on $F$ (thought of the ``Lebesgue measure on the space of connections'') and the Haar measure on $\GGauge_{M,G}$ are parts of the functional integral measure for Yang-Mills theory and are, certainly, problematic. One works around them by considering \emph{perturbative} Faddeev-Popov integral, as given by the Feynman graph expansion in the r.h.s. of (\ref{l16_FP_stat_phase}).

For the gauge-fixing $\phi: \Conn_{M,G}\ra \mr{gauge}$, one of the possible choices is the \emph{Lorentz} gauge, corresponding to 
\begin{equation}\label{l16_e2}
\phi(A)=d^*A
\end{equation}
In this case, Faddeev-Popov endomorphism of $\mr{gauge}$ is:
\begin{equation}\label{l16_e3}
FP(A)=d^* d_A :\; \Omega^0(M,\g)\ra \Omega^0(M,\g)
\end{equation}
-- as follows from (\ref{l16_e1}) and (\ref{l16_e2}).

We are interested in evaluating the perturbative contribution of the gauge orbit of zero connection. The fact that the intersection of $\phi^{-1}(0)$ and the gauge orbit through $A=0$ is transversal at $A=0$ follows from the Hodge decomposition theorem (which implies $\Omega^1(M,\g)=\Omega^1(M,\g)_\mr{exact}\oplus \Omega^1(M,\g)_\mr{coclosed}=\mr{im}(d_{1,A=0}\gamma)\oplus T_{A=0}\phi^{-1}(0)$).

The formal Faddeev-Popov integral for Yang-Mills theory in Lorentz gauge is:
\begin{equation}\label{l16_FP-YM}
Z=\int_{\Conn \oplus \mr{gauge}^*\oplus\Pi(\mr{gauge}\oplus \mr{gauge}^*)}\DD A \,\DD\lambda \, \DD c \,\DD\bar{c}\;\; e^{\frac{i}{\hbar}S_{FP}(A,\lambda,c,\bar{c})}
\end{equation}
with 
\begin{equation}
S_{FP}(A,\lambda,c,\bar{c})=S_{YM}(A)+\int_M \lan \lambda ,d^*A \ran + \int_M \lan \bar{c} ,d^* d_A c \ran
\end{equation}
Here $\lambda\in \Omega^\mr{top}(M,\g^*)$ where the r.h.s. is our model for the dual of the Lie algebra of gauge transformations. Likewise,  $\bar{c}\in \Pi\;\Omega^\mr{top}(M,\g^*)$ and $c\in \Pi\; \Omega^0(M,\g)$.

\marginpar{\LARGE{Lecture 17, 10/26/2016.}}
Feynman rules for perturbative calculation of the Faddev-Popov integral for Yang-Mills theory (\ref{l16_FP-YM}) in the case $M=\RR^{3,1}$ -- the flat Lorentzian space with metric $\eta_{\mu\nu}=\left( \begin{array}{cccc} -1 & & &\\ & 1 & &\\ & & 1 & \\ & & & 1 \end{array} \right)$ -- are as follows.

\begin{tabular}{c|c}
Hald-edge & field \\
\hline 
$\vcenter{\hbox{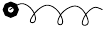}}$ & $A^a_\mu(x)$ \\
$\vcenter{\hbox{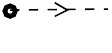}}$ & $c^a(x)$ \\
$\vcenter{\hbox{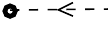}}$ & $\bar{c}_a(x)$
\end{tabular}\\
Here $A^a_\mu(x)$ are the local components of the connection evaluated at a point $x$, $A=\sum_{a=1}^{\dim\g}\sum_{\mu=1}^4 T_a A^a_\mu (x) dx^\mu$, with $\{T_a\}$ the chosen basis in $\g$ (which we assume to be orthonormal w.r.t. to the Killing form in $\g$). Likewise, $c^a(x)$ are the components of $c=\sum_{a=1}^{\dim\g} T_a c^a(x)$ and $\bar{c}_a(x)$ are the components of $\bar{c}=\sum_{a=1}^{\dim\g} T_a \bar{c}_a(x) d^4 x$.

\begin{tabular}{c|c}
Edge & propagator \\
\hline
$\vcenter{\hbox{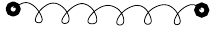}}$ & 
$\int \frac{d^4 k}{(2\pi)^4} e^{-i (k,x-y)}\frac{i\delta_{ab}\eta_{\mu\nu}}{k^2+i\epsilon}$ \\
$\vcenter{\hbox{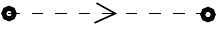}}$ & 
$\int \frac{d^4 k}{(2\pi)^4} e^{-i (k,x-y)}\frac{i\delta_{ab}}{k^2+i\epsilon}$
\end{tabular}\\
Here a limit $\epsilon\ra +0$ is implied. This provides a regularization for the propagators which, in pseudo-Riemannian case, are singular on the light-cone $(x-y,x-y)=0$, as opposed to the Riemannian case, where the singularity is just at $x=y$.

\begin{tabular}{c|c}
Vertex & vertex tensor \\ \hline
$\vcenter{\hbox{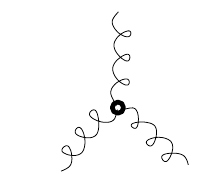}}$\qquad \qquad & $f^{abc}\eta^{\mu\nu}\left(i\frac{\dd}{\dd x_\rho}\left( \vcenter{\hbox{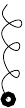}}  \right)-
i\frac{\dd}{\dd x_\rho}\left(  \vcenter{\hbox{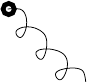}}  \right)\right)+\mr{cycl.\;perm.\;of\;}\{(a,\mu),(b,\nu),(c,\rho)\}$\\
$\vcenter{\hbox{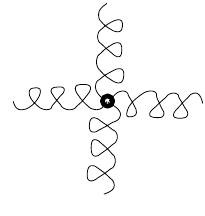}}$ & $-i \sum_e f^{abe}f^{cde}(\eta^{\mu\rho}\eta^{\nu\sigma}-\eta^{\mu\sigma}\eta^{\nu\rho})+ \mr{cycl.\;perm.\;of \;}\{(a,\mu),(b,\nu),(c,\rho),(d,\sigma)\}$\\
$\vcenter{\hbox{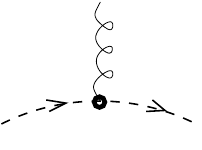}}$ & $i f^{abc}\frac{\dd}{\dd x_\mu}\left( \vcenter{\hbox{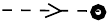}} \right)$
\end{tabular}

These vertices correspond to the cubic and quartic terms $\frac12 \int \tr [A,A]dA$, $\frac18 \int \tr [A,A]\wedge [A,A]$, $\int \lan \bar{c},d^*[A,c] \ran$ in the Taylor expansion in the fields of the Faddev-Popov extension of the Yang-Mills action (\ref{l16_FP-YM}).

One can also enhance the Yang-Mills theory by adding a \emph{matter term} to the action,
$$S_{YM}\mapsto S_{YM}+ \int_M dx \lan\bar\psi, (i\dd\!\!\!/_A+m) \psi\ran $$
Here the new matter field $\psi$ is an \emph{odd} complex Dirac fermion field on $M$ -- a section of 
$E\otimes R$ 
with $E\ra M$ the spinor bundle and $R$ a representation of the structure group $G$. Field $\psi$ has local components $\psi_\alpha^i(x)$ with $i$ the index of spanning the basis of the representation space $R$ and $\alpha$ the spinor index; $\dd\!\!\!/_A=\sum_{\mu,\alpha,\beta,i,j}(\gamma^\mu)_{\alpha\beta}(\delta_{ij}\dd_\mu+(T_a)_{ij}A^a_\mu(x))$ is the Dirac operator, with $\gamma^\mu$ the Dirac gamma-matrices and $(T_a)_{ij}$ the representation matrices of the basis elements $T_a$ of $\g$; $\lan,\ran$ is the inner product of Dirac spinors; $m$ is the \emph{mass} of the fermion.

Adjoining the matter field results in the extension of Feynman rules by new half-edges 
$$\vcenter{\hbox{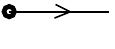}}\quad \mapsto\;\; \psi_\alpha^i(x),\qquad  \vcenter{\hbox{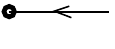}}\quad\mapsto \; \; \bar\psi^i_\alpha(x) $$ 
The new edge is:
$$ \vcenter{\hbox{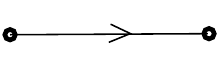}}\quad\mapsto\;\; \int \frac{d^4 k}{(2\pi)^4} e^{-i (k,x-y)}\left(\frac{-i}{k\!\!\!/+m}\right)_{\alpha\beta}\delta_{ij} $$
where the dash in $k\!\!\!/:=\sum_\mu k_\mu (\gamma^\mu)_{\alpha\beta}$ stands for contraction with Dirac gamma-matrices. The new vertex is:
$$ \vcenter{\hbox{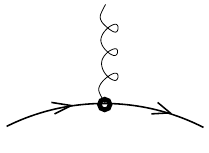}}\quad\mapsto\;\; i (\gamma^\mu)_{\alpha\beta}(T_a)_{ij} $$

\begin{remark} Yang-Mills theory for the group $G=SU(3)$ is the theory of the strong interaction (quantum chromodynamics). The Yang-Mills field $A$ corresponds to the \emph{gluon} -- the carrier of the strong interaction and the matter fields $\psi$ correspond to \emph{quarks}. Abelian case $G=U(1)$ corresponds to quantum electrodynamics, with $A$ the photon field and $\psi,\bar\psi$ the electron/positron field. Standard model of particle physics is the Yang-Mills theory with $G=U(1)\times SU(2)\times SU(3)$ (with the factors corresponding to the electromagnetic, weak and strong interactions).
\end{remark}

\begin{remark} Frequently, instead of scaling the Yang-Mills-Faddeev-Popov action in the path integral with $\frac{1}{\hbar}$, as in (\ref{l16_FP-YM}), one sets $\hbar=1$ but scales the Yang-Mills action as $S_{YM}\mapsto \frac{1}{2g^2}\int \tr F_A\wedge * F_A$ (instead of (\ref{l16_S_YM})) with $g$ the \emph{coupling constant} of the strong interaction.\footnote{Or equivalently, by rescaling $A\mapsto g\cdot A$, one has $S_{YM}=\frac12 \int\tr dA\wedge *dA+\frac{g}{2}\int \tr [A,A]\wedge dA+\frac{g^2}{8}\int\tr [A,A]\wedge [A,A]$. In the matter term, if present, the quark-gluon interaction term $\bar\psi A\psi$ also gets rescaled by a factor $g$.} This normalization can be converted back to ours by setting $\hbar=g^2$ and rescaling the auxiliary fields $\lambda,c,\bar{c}$ (and the matter fields $\psi$, $\bar\psi$, if present), by appropriate powers of $g$. Put another way, with the normalization by the coupling constant $g$, Feynman graphs are weighed with $g^{-2\chi(\Gamma)}$ instead of $\hbar^{-\chi(\Gamma)}$.
\end{remark}

\subsection{Elements of supergeometry
}\footnote{A reference for  the basic definitions on supermanifolds and $\ZZ$-graded (super)manifolds: Appendix B in \cite{CMR}.}
\subsubsection{Supermanifolds}
\begin{definition}\label{l17_def_smfd} An $(n|m)$-supermanifold $\MM$ is a sheaf $\OO_\MM$, over a smooth $n$-manifold $M$ (the \emph{body} of $\MM$), of supercommutative algebras locally isomorphic to algebras of form $C^\infty(U)\otimes \wedge^\bt V^*$ with $U\subset M$ open and $V$ a fixed $m$-dimensional vector space.  I.e., there is an atlas on $M$ comprised by open subsets $U_\alpha\subset M$ with chart maps $\phi_\alpha:U_\alpha\ra W=\RR^n$, with isomorphisms of supercommutative algebras $\Phi_\alpha: \OO_\MM(U_\alpha)\ra C^\infty(\phi_\alpha(U))\otimes \wedge^\bt V^*=:\mc{A}_\alpha$.
\end{definition}
Locally a function on $\MM$ is an element of $\mc{A}_\alpha$, i.e., has local form 
$$f|_{U_\alpha}=\sum_k \sum_{1\leq i_1<\cdots< i_k\leq m} f_{i_1\cdots i_k}(x 
)\cdot \theta_{i_1}\cdots \theta_{i_k}$$ 
with $x_1,\ldots,x_n$ the local \emph{even} coordinates on $M$ (pullbacks of the standard coordinates on $\RR^n$ by $\phi_\alpha$) and $\theta_1,\ldots,
\theta_m\in V^*$ the  \emph{odd} (anti-commuting) coordinates on $V$.

\begin{remark}\label{l17_rem_aug}
The augmentation map $\wedge^\bt V^*\ra \RR$ induces a globally well-defined augmentation map 
\begin{equation}\label{l17_aug}
\OO_\MM\ra C^\infty(M)
\end{equation}
\end{remark}

\begin{example}
Let $\mc{V}=V_\even\oplus \Pi V_\odd$ be a super-vector space. We can define an associated supermanifold, also denoted $\mc{V}$, by
$\OO_\mc{V}(U):=C^\infty(U)\otimes \wedge^* V_\odd^*$ for any open $U\subset V_\even$.
\end{example}

\begin{example}[Split supermanifolds]\label{l17_ex_split}
Let $E\ra M$ be a rank $m$ vector bundle over an $n$-manifold $M$. Then we can construct a ``split''  $(n|m)$-supermanifold $\Pi E$ with body $M$ and the structure sheaf $\OO_{\Pi E}=\Gamma(M,\wedge^\bt E^*)$ -- the space of smooth sections, over $M$, of the bundle of supersommutative algebras $\wedge^\bt E^*$.
\end{example}
E.g., for $M$ an $n$-manifolds, we have two distinguished $(n|n)$-supermanifolds, $\Pi TM$ and $\Pi T^*M$, obtained by applying the construction above to the tangent and cotangent bundle of $M$, respectively.

\begin{definition}
A morphism of supermanifolds $\phi:\MM\ra \NN$ consist of the data of:
\begin{itemize}
\item A smooth  map between the bodies $f: M\ra N$,
\item An extension of $f$ to a morphism of sheaves of supercommutative algebras $\phi^*:\OO_\NN\ra \OO_\MM$. In particular, for an open $U\subset \NN$, we have a morphism $\phi^*_U: \OO_\NN(U)\ra \OO_\MM(f^{-1}(U))$ commuting with the augmentation maps (\ref{l17_aug}):
$$\begin{CD}
\OO_\NN(U) @>\phi^*>> \OO_\MM(f^{-1}(U)) \\
@VVV @VVV \\
C^\infty(U) @>f^*>> C^\infty(f^{-1}(U))
\end{CD}$$
\end{itemize}
\end{definition}

\begin{theorem}[Batchelor]
Every smooth supermanifold with body $M$ is (non-canonically) isomorphic to a split-supermanifold $\Pi E$ for some vector bundle $E\ra M$.
\end{theorem}

\begin{example}\label{l17_ex1} Let us construct a morphism $\phi: \RR^{1|2}\ra \RR^{1|2}$ where the source $\RR^{1|2}$ has even coordinate $x$ and odd coordinates $\theta_1,\theta_2$ and the target $\RR^{1|2}$ has the even coordinate $y$ and odd coordinates $\psi_1,\psi_2$. We define $\phi$ by specifying the pullbacks of the target coordinates:
$$\phi^*: \begin{array}{ccc} y & \mapsto & x+\theta_1\theta_2 \\
\psi_1 & \mapsto & \theta_1 \\
\psi_2 & \mapsto & \theta_2
\end{array}$$
\end{example}

\begin{example} By Remark \ref{l17_rem_aug}, for $\MM$ any supermanifold the inclusion of the body $M\hra \MM$ is a canonically defined morphism of supermanifolds.
\end{example}

\begin{example} A morphism of vector bundles
$$\begin{CD}
E @>\phi_E>> E'\\
@VVV @VVV \\
M @>\phi_M>> M'
\end{CD}$$
induces a map of the corresponding split supermanifolds $\Pi E\ra \Pi E'$. \textbf{Warning:} the converse is not true -- there are morphisms of $\Pi E\ra \Pi E'$ not coming from morphisms of vector bundles! (E.g., the morphism constructed in Example \ref{l17_ex1} does not come from a morphism of vector bundles.)
\end{example}

\begin{definition}\label{l17_def_vect_field}
A vector field $v\in \mathfrak{X}(\MM)$ of parity $|v|\in \{0,1\}$ (with the convention $0$=even, $1$=odd) is a derivation of $\OO_\MM$ of parity $|v|$, i.e., an $\RR$-linear map $v: \OO_\MM\ra \OO_\MM$   satisfying 
\begin{eqnarray}
v(f\cdot g) &=& v(f)\cdot g +(-1)^{|v|\cdot |f|}f\cdot v(g) \\
|v(f)| &= & |v|+|f|\mod 2 \label{l17_e1}
\end{eqnarray}
Vector fields on $\MM$ form a Lie superalgebra with Lie bracket 
\begin{equation}\label{l17_Lie_bracket}
[v,w]:=v\circ w-(-1)^{|v|\cdot |w|}w\circ v
\end{equation}
\end{definition}

\subsubsection{$\ZZ$-graded (super)manifolds}
\begin{definition}[$\ZZ$-graded supermanifold] 
Let $\MM$ be a supermanifold.
Assume that, in terms of Definition \ref{l17_def_smfd},  both $V=\bigoplus_k V_k$ (the odd fiber) and $W=\bigoplus_k W_k$ (the target of even coordinate charts) are $\ZZ$-graded vector spaces (we assume that only finitely many of $V_k$, $W_k$ are nonzero). This grading induces a grading on the polynomial subalgebra $\Sym\; W^* \otimes \wedge V^*$ in $\mc{A}_\alpha$ where linear functions $x^i$ on $V_k$ are prescribed degree $|x^i|=-k$ and linear functions $\theta^\alpha$ on $W_k$ are prescribed degree $|\theta^\alpha|=-k$. If transition maps between the charts $\Phi_\alpha\circ \Phi_\beta^{-1}$ are compatible with this grading, we say that we have a (global) $\ZZ$-grading on $\MM$ or, equivalently, that $\MM$ is a \emph{$\ZZ$-graded supermanifold}.
\end{definition}
Using the grading of local coordinates, we can introduce, locally, a vector field 
\begin{equation}\label{l17_E_local}
\mathbb{E}:=\sum_{i} |x^i|\cdot x^i\frac{\dd}{\dd x^i}+\sum_\alpha |\theta^\alpha|\cdot \theta^\alpha\frac{\dd}{\dd \theta^\alpha}
\end{equation}
The fact that the grading in local charts is compatible with transitions between charts is equivalent to the local expression (\ref{l17_E_local}) gluing to a well-defined vector field $\mathbb{E}$ on $\MM$. It has the name \emph{Euler vector field} and has the property that for $f$ a function on $\MM$ of well-defined degree $|f|$, we have 
$$\mathbb{E}f=|f|\cdot f$$

Unless stated otherwise, we will be making the following simplifying assumption.
\begin{assumption}[Compatibility of $\ZZ$-grading and super-structure]
We assume that $W_k$ can be nonzero only for $k$ even and $V_k$ can be nonzero only for $k$ odd. Then one says that the $\ZZ$-grading and the super-structure on $\MM$ are compatible, or that the $\ZZ_2$-grading (responsible for the Koszul sign in the multiplication of functions) is $\bmod\; 2$ reduction of the $\ZZ$-grading. 
\end{assumption}

Similarly to the Definition \ref{l17_def_vect_field}, we can define a vector field of degree $k$ on a $\ZZ$-graded manifold $\MM$. The degree condition (\ref{l17_e1}) gets replaced by $|v(f)|=k+|f|$.   

\textbf{Notation:} we denote $C^\infty(\MM)_k$ or $(\OO_\MM)_k$ the space of functions of degree $k$ on a $\ZZ$-graded supermanifold.\footnote{We use notations $C^\infty(\MM)$ and $\OO_\MM$ for the algebra of functions on $\MM$ interchangeably.} Likewise, we denote $\mathfrak{X}(\MM)_k$ the space of vector fields of degree $k$.

In particular, $\mathbb{E}\in \mathfrak{X}(\MM)_0$ is a vector field of degree $0$ and for $v\in \mathfrak{X}(\MM)_k$ a vector field of degree $k$, we can probe its degree by looking at its Lie bracket with $\mathbb{E}$:
$$[\mathbb{E},v]=k\cdot v$$

\marginpar{\LARGE{Lecture 18, 10/31/2016.}}

\begin{example}
Let $E_\bt\ra M$ be a graded vector bundle with fibers graded by \emph{odd} integers. Then, similarly to the construction of Example \ref{l17_ex_split}, we can construct a $\ZZ$-graded manifold $\mc{E}$ with body $M$ and with
$$
\OO_\mc{E}
:=\Gamma(M,\Sym^\bt_\mr{gr}E^*)=\Gamma(M,\wedge^\bt E^*)$$
Here $\Sym^\bt_\mr{gr}$ stands for the graded-symmetric algebra of a graded vector bundle (i.e. symmetric algebra of the even part tensored with the exterior algebra of the odd part; the former vanishes in the present example).
\end{example}

\begin{example}\label{l18_ex_T[1]}
 $\MM=T[1]M$ -- the tangent bundle of $M$ with tangent fiber coordinates assigned grading $1$. Locally, we have coordinates $x^i$ in an open $U\subset M$.\footnote{
We adopt the following (standard) convention for shifts of homological degree: if $V^\bt$ is a $\ZZ$-graded vector space, then the degree-shifted vector space  $V[k]$ is defined by $(V[k])^i:=V^{k+i}$. In particular, e.g., for $V$ concentrated in degree zero, $V^{\neq 0}=0$, $V[k]$ is concentrated in degree $-k$.
} The corresponding chart on $T[1]M$ has local base coordinates $x^i$ of degree $0$ and fiber coordinates $\theta^i=``\; dx^i\;"$ of degree $1$. An element of $
\OO_\MM
$ locally has the form $\sum_{k=0}^n \sum_{1\leq i_1<\cdots <i_k\leq n} f_{i_1\cdots i_k}(x)\theta^{i_1}\cdots \theta^{i_k}$. Globally, we have an identification of functions on $T[1]M$ with forms on $M$, $
\OO_\MM
\cong \Omega^\bt(M)$ with the form degree providing the $\ZZ$-grading.
\end{example}

\begin{example}
$\MM=T^*[-1]M$ -- the cotangent bundle of $M$ with cotangent fiber coordinates assigned degree $-1$. Locally, we have base coordinates $x^i$, $\deg x^i=0$ and fiber coordinates $\psi_i=``\; \frac{\dd}{\dd x^i}\;"$, $\deg \psi_i=-1$. An element of $\OO_\MM$ locally has the form
$\sum_{k=0}^n \sum_{1\leq i_1<\cdots< i_k < n} f^{i_1\cdots i_k}\psi_{i_1}\cdots \psi_{i_k}$. Globally, we have an identification of function on $T[-1]M$ with polyvectors with reversed grading: $(\OO_\MM)_{-k}\cong \mc{V}^k(M)=\Gamma(M,\wedge^k TM)$. I.e., a function on $T[-1]M$ of degree $-k$ is the same as a $k$-vector field on $M$.
\end{example}

\subsubsection{Differential graded manifolds (a.k.a. $Q$-manifolds)}

\begin{definition}
For $\MM$ a $\ZZ$-graded supermanifold, one calls a vector field $Q$ on $\MM$ a \emph{cohomological vector field} if 
\begin{itemize}
\item $Q$ has degree $1$,
\item $Q^2=0$ (as a derivation of $\OO_\MM$). Or, equivalently, the Lie bracket of $Q$ with itself vanishes, $[Q,Q]=0$.\footnote{Note that, by (\ref{l17_Lie_bracket}), for an odd vector field, we have $[Q,Q]=2Q^2$. In particular, vanishing of $[Q,Q]$ is not a tautological property, unlike for a bracket of an even vector field with itself.}
\end{itemize}
Then we say that the pair $(\MM,Q)$ is a \emph{differential graded (dg) manifold} or, equivalently, a \emph{$Q$-manifold}.
\end{definition}

\begin{remark}
Note that $Q$ defines a differential on the algebra of functions, $Q: C^\infty(\MM_k)\ra C^\infty(\MM)_{k+1}$, thus endowing $C^\infty(\MM)$ with the structure of a commutative differential graded algebra.
\end{remark}

\begin{remark}[Carchedi-Roytenberg?] Vector fields $\mathbb{E},Q$ satisfy the commutation relations
$$[\mathbb{E},\mathbb{E}]\underset{\mr{tautologically}}{=}0=[Q,Q],\qquad [\mathbb{E},Q]=Q $$
Thus, the pair of vector fields $\mathbb{E},Q$ define an action on $\MM$ of a Lie superalgebra of automorphisms of the odd line $\RR^{0|1}$. This algebra is generated by infinitesimal dilatation $\sf{e}=-\theta\frac{\dd}{\dd \theta}$ and an infinitesimal translation $\sf{q}=\frac{\dd}{\dd \theta}$ (with $\theta$ the odd coordinate on $\RR^{0|1}$), satisfying same super Lie algebra relations as above.
\end{remark}

\begin{example} For $$\MM=T[1]M$$ the degree-shifted tangent bundle of $M$, we have a cohomological vector field $Q$ on $\MM$ corresponding to the de Rham operator $d_M$ on $M$, so that we have
$$\begin{CD}
C^\infty(\MM)_k @>Q>> C^\infty(\MM)_k\\
@|  @| \\
\Omega^k(M) @>d_M>> \Omega^{k+1}(M)
\end{CD}$$
Locally, in terms of local coordinates $(x^i,\theta^i=dx^i)$ (cf. Example \ref{l18_ex_T[1]}), we have 
$$Q=\sum_i \theta^i \frac{\dd}{\dd x^i}$$
This local formula glues, over coordinate charts on $\MM$, to a globally well-defined vector field $Q=d_M\in \mathfrak{X}(\MM)_1$.
\end{example}

\begin{example}
Let $\g$ be a Lie algebra. Consider a graded manifold $$\MM=\g[1]$$ 
with body a point and $C^\infty(\MM)=\wedge^\bt \g^*$. Note that functions on $\MM$ can be identified with Chevalley-Eilenberg cochains on $\g$, $C^\infty(\MM)\cong C^\bt_{CE}(\g)$. We define the cohomological vector field $Q$ on $\MM$ to be the Chevalley-Eilenberg differential $d_{CE}: \wedge^k\g^*\ra \wedge^{k+1}\g^*$, obtained from the dual of the Lie bracket $[,]^*:\g^*\ra \wedge^2 \g^*$ by extension to $\wedge^\bt\g^*$ as a derivation, by Leibniz identity. The property $d_{CE}^2=0$ then corresponds to the Jacobi identity in $\g$. Let $\{T_a\}$ be a basis in $\g$ and $\{\psi^a\}$ be the corresponding degree $1$ coordinates on $\MM$ (the dual basis to $\{T_a\}$); let also $f^c_{ab}$ be the structure constants of $\g$, i.e. $[T_a,T_b]=\sum_c f^c_{ab} T_c$. Then we have
$$Q=d_{CE}=\frac12 \sum_{a,b,c} f^c_{ab} \psi^a \psi^b \frac{\dd}{\dd \psi^c}\qquad \in \mathfrak{X}(\MM)_1$$
\end{example}

\begin{definition}\label{l18_def_L_infty}
An $L_\infty$ algebra is a graded vector space $\g^\bt$ endowed with multi-linear, graded skew-symmetric operations $l_k:\wedge_\mr{gr}^k\g\ra \g$ for each $k\geq 1$, such that:
\begin{itemize}
\item $l_k$ has degree $2-k$,
\item the following quadratic relations hold for each $n\geq 1$:
\begin{equation}\label{l18_L_infty_relations}
\sum_{n=r+s,\;r\geq 0,s\geq 1}\;\; \sum_{\sigma\in \mr{Sh}(r,s)}\pm \; l_{r+1}(x_{\sigma_1},\ldots,x_{\sigma_r},l_s(x_{\sigma_{r+1}},\ldots,x_{\sigma_n}))=0
\end{equation}
for $x_1,\ldots,x_n\in\g^\bt$ any $n$-tuple of vectors. Here $\mr{Sh}(r,s)$ stands for $(r,s)$-shuffles, i.e., permutations of numbers $1,\ldots,n=r+s$, such that $\sigma_1<\cdots < \sigma_r$ and $\sigma_{r+1}<\cdots<\sigma_n$.
\end{itemize}
\end{definition}

In particular, for small values of $n$, relations (\ref{l18_L_infty_relations}) have the following form:
\begin{itemize}
\item $n=1$: $l_1(l_1(x))=0$, i.e. $l_1=:d$ is a differential on $\g^\bt$.
\item $n=2$: $l_1(l_2(x,y))=l_2(l_1(x))+(-1)^{|x|}l_2(x,l_1(y))$ -- Leibniz identity, i.e. $d$ is a derivation of the binary operation $l_2=:[,]$.
\item $n=3$: Jacobi identity \emph{up to homotopy} for $l_2=[,]$, i.e. the Jacobiator equals a commutator (in appropriate sense) 
of a trinary operation $l_3$ with $l_1=d$:
$$[x,[y,z]]-[[x,y],z]-(-1)^{|x|\cdot |y|}[y,[x,z]]=\pm d l_3(x,y,z)\pm l_3(dx,y,z)\pm l_3(x,dy,z)\pm l_3(x,y,dz)$$
\end{itemize}

An alternative definition of an $L_\infty$ algebra is as follows.
\begin{definition}
An $L_\infty$ algebra is a graded vector space $\g^\bt$ together with a \emph{coderivarion}\footnote{Recall that a linear map $\DD:C\ra C$ is a coderivation of a coalgebra $C$ if the co-Leibniz identity holds: $\Delta\circ\DD=(\DD\otimes\mr{id})\circ \Delta+(\mr{id}\otimes \DD)\circ\Delta$, with $\Delta:C\ra C\otimes C$ the coproduct. In particular, if $\delta:A\ra A$ is a derivation of an algebra $A$, then the dual map $\delta^*:A^*\ra A^*$ is a coderivation of the dual coalgebra $C=A^*$.} $\DD$ of the cofree cocommutative coalgebra generated by $\g[1]$, 
$\DD: \Sym^\bt (\g[1]) \ra \Sym^\bt (\g[1])$, satisfying the following:
\begin{itemize}
\item $\DD^2=0$,
\item $p_0\circ \DD=0$ where $p_0: \Sym^\bt (\g[1])\ra \Sym^0 (\g[1])=\RR$ is the counit,
\item $\DD$ has degree $+1$.
\end{itemize}
\end{definition}

\begin{remark}\label{l18_rem_L_infty_from_coder} Coderivation $\DD$ is determined by its projection to (co)generators in $\g[1]$, i.e., by a sequence of maps
\begin{equation}\label{l18_D_proj}
p\circ \DD^{(k)}: \Sym^k (\g[1]) \ra \g[1]
\end{equation}
where $p: \Sym^\bt (\g[1]) \ra \Sym^1 (\g[1])=\g[1]$ is the projection to (co)generators. In (\ref{l18_D_proj}) we restricted the input of $\DD$ to $k$-th symmetric power of $\g[1]$, with $k\geq 1$. One has a tautological \emph{d\'ecalage isomorphism} $\alpha: \Sym^k(\g[1])\ra (\wedge^k \g)[k]$ which sends $\alpha: s(x_1)\odot \cdots \odot s(x_k)\mapsto \pm s^k(x_1\wedge \cdots\wedge x_n)$ for $x_1,\ldots,x_k\in\g$, with $s$ the suspension symbol. The relation of the $L_\infty$ operations $l_k$ from Definition \ref{l18_def_L_infty} with the components of the coderivation (\ref{l18_D_proj}) is via
$$l_k= p\circ \DD^{(k)}\circ \alpha^{-1}: \quad \wedge^k \g \ra \g$$
The quadratic relations on operations correspond to the equation $\DD^2=0$.
\end{remark}

\begin{example}
Let $(\g^\bt,\{l_k\})$ be an $L_\infty$ algebra. Then $(\g^\bt[1],Q=\DD^*)$ is a dg manifold. I.e., we identify the dual of $\Sym^\bt(\g[1])$ with a polynomial subalgebra in $C^\infty(\g[1])$. The dual of the coderivation $\DD$ is a derivation of polynomial functions on $\g[1]$ and thus yields a vector field on $\g[1]$. If $\{T_a\}$ is a basis in $\g$, $\{T^a\}$ the dual basis in $\g^*$, and $\theta^a$ the corresponding coordinates on $\g[1]$, we have
$$Q=\sum_{k=1}^\infty\frac{1}{k!}\sum_{a_1,\ldots,a_k,b}\pm \lan T^b, l_k(T_{a_1},\ldots,T_{a_k})\ran\theta^{a_1}\cdots \theta^{a_k}\frac{\dd}{\dd \theta^b} $$
Introducing a ``generating function for coordinates on $\g[1]$'' (or ``superfield'') $\underline\theta=\sum_a  \theta^a  T_a\in  \Sym^1 (\g[1])^*\otimes \g$, we can write
$$Q=\sum_{k=1}^\infty\frac{1}{k!} \lan l_k(\underline\theta,\ldots,\underline\theta),\frac{\dd}{\dd \underline\theta} \ran $$
where $\frac{\dd}{\dd \underline\theta}:=\sum_a T^a \frac{\dd}{\dd \theta^a} $,  operations $l_k$ act only on elements of $\g$ (the $T^a$s) and $\lan,\ran$ pairs $\g$ with $\g^*$.

The property $Q^2=0$ is equivalent to the quadratic relations (\ref{l18_L_infty_relations}) on operations $\{l_k\}$.
\end{example}

\begin{remark}[From \cite{AKSZ}]
If $(\MM,Q)$ is a dg manifold and $x_0\in M$ a point of the body such that $Q$ vanishes at $x_0$, then the shifted tangent space $\g:=T_{x_0}[-1]\MM$ inherits the structure of $L_\infty$ algebra: Taylor expansion of $Q$ at $x_0$ produces a sequence of elements 
$$Q^{(k)}\in \Sym^k T^*_{x_0}\MM\otimes T_{x_0}\MM=\Sym^k (\g[1])^*\otimes \g[1]$$
which, by the d\'ecalage isomorphism (cf. Remark \ref{l18_rem_L_infty_from_coder}), yield the $L_\infty$ operations $l_k: \wedge^k\g\ra \g$.\footnote{The $L_\infty$ structure induced this way on the shifted tangent space depends on the choice of a local chart near $x_0$. Choosing a different chart induces an isomorphism of $L_\infty$ algebras.}
\end{remark}

\begin{definition}
A \emph{Lie algebroid} is a vector bundle $E\ra M$ with skew-symmetric Lie bracket on sections $[,]:\Gamma(E)\times \Gamma(E)\ra \Gamma(E)$ satisfying Jacobi identity, endowed additionally with the \emph{anchor map} -- a bundle map $\rho: E\ra TM$ (covering the identity map on $M$), such that for $\alpha,\beta\in \Gamma(E)$ and $f\in C^\infty(M)$ the following version of Leibniz identity holds:
\begin{equation}\label{l18_anchor_eq}
[\alpha,f \cdot\beta]=f\cdot[\alpha,\beta]+\rho(\alpha)(f)\cdot \beta
\end{equation}
\end{definition}

\begin{example}[Vaintrob, \cite{Vaintrob}]\label{l18_ex_Lie_algbd} Let $(E\ra M;[,];\rho)$ be a Lie algebroid. Consider the graded manifold $E[1]$ with body $M$ and functions $C^\infty(E[1])=\Gamma(M,\wedge^\bt E^*)$. One can endow $E[1]$ with a cohomological vector field
$Q: \Gamma(M,\wedge^k E^*)\ra \Gamma(M,\wedge^{k+1}E^*)$ defined as follows: for $\psi\in \Gamma(M,\wedge^k E^*)$ and $\alpha_0,\ldots,\alpha_k\in \Gamma(M,E)$, we set
\begin{multline}\label{l18_Li_algbd_Q}
Q\psi(\alpha_0,\ldots,\alpha_k):=\sum_{\sigma\in\mr{Sh}(2,k)}(-1)^\sigma \rho(\alpha_{\sigma_0})\left(\psi(\alpha_{\sigma_1},\ldots,\alpha_{\sigma_k})\right)+\\
+\sum_{\sigma\in \mr{Sh}(2,k-1)}(-1)^{\sigma}\psi\left([\alpha_{\sigma_0},\alpha_{\sigma_1}],\alpha_{\sigma_2},\ldots, \alpha_{\sigma_k}\right)
\end{multline} 
Locally, let $\{x^i\}$ be local coordinates in a neighborhood $U$ on $M$ and $\{e_a\}$ be a basis of sections of $E$ over $U$. In particular, $[e_a,e_b]=\sum_c f^c_{ab}(x)e_c$ with $f^c_{ab}(x)$ the structure constants of the Lie bracket of sections of $E$. The anchor maps $e_a$ to a vector field $\sum_i \rho^i_a(x)\frac{\dd}{\dd x^i}$. On $E[1]$ we have local coordinates $x^i$, $\deg x^i=0$ and $\theta^a$, $\deg \theta^a=1$. The cohomological vector field (\ref{l18_Li_algbd_Q}) locally takes the form
\begin{equation}
Q=\frac12 \sum_{a,b,c} f^c_{ab}(x) \theta^a \theta^b \frac{\dd}{\dd \theta^c} + \sum_{a,i} \theta^a \rho^i_a(x)\frac{\dd}{\dd x^i}
\end{equation}
Equation $Q^2=0$ is equivalent to the structure relations of a Lie algebroid: 
\begin{itemize}
\item the Jacobi identity for sections of $E$, 
\item the condition that the anchor $\rho: \Gamma(M,E)\ra \mathfrak{X}(M)$ is a Lie algebra morphism (which follows from (\ref{l18_anchor_eq})).
\end{itemize}

\end{example}

\begin{example}\label{l18_ex_action_Lie_algbd}
A special case of Example \ref{l18_ex_Lie_algbd} is as follows. Let $G$ be a group acting on a manifold $M$ with $\gamma: G\times M\ra M$ the action. Let $d_{1,x}\gamma:\g\ra T_x M$ be the corresponding infinitesimal action, with $x\in M$. We can construct the \emph{action Lie algebroid}, with $E=\g\times M$ (as a trivial bundle over $M$), with the bracket of sections given by pointwise bracket in $\g$ and with the anchor map $\rho=d_{1,-}\gamma: E\ra TM$ given by the Lie algebra action. The corresponding graded manifold is $E[1]=M\times \g[1]$ with the algebra of functions
$$C^\infty(E[1])=\wedge^\bt \g^* \otimes C^\infty(M)=C^\bt_{CE}(\g,C^\infty(M))$$
-- Chevalley-Eilenberg cochains of $\g$ with coefficients in the module $C^\infty(M)$ with module structure given by $T_a\otimes f\mapsto v_a(f)$ with $v_a$ the fundamental vector fields of $\g$-action and with $f\in C^\infty(M)$ an arbitrary function. The cohomological vector field is the  Chevalley-Eilenberg differential twisted by the module $C^\infty(M)$. Locally on $M$:
$$Q=\sum_{a,b,c} f^c_{ab}\theta^a \theta^b\frac{\dd}{\dd \theta^c}+\sum_{a,i}\theta^a v_a^i(x)\frac{\dd}{\dd x^i}$$
\end{example}

\marginpar{\LARGE{Lecture 19, 11/2/2016.}}
\subsubsection{Integration on supermanifolds}
Let $p:E\ra M$ be a vecor bundle of rank $m$ over an $n$-manifold $M$. Let $\MM=\Pi E$ be the corresponding split $(n|m)$-supermanifold.

We define the Berezin line bundle of the supermanifold $\MM$ as the real line bundle $\mr{Ber}(\MM)=\wedge^n T^*M\otimes \wedge^m E$ over $M=\mr{body}(\MM)$. We call sections of $\mr{Ber}$ the \emph{Berezinians}.

Given a Berezinian $\mu\in \Gamma(M,\mr{Ber}(\MM))$, we have an $\RR$-linear \emph{integration map}
$$\int_\MM \mu\cdot \bt : \;\; C^\infty_c(\MM)\ra \RR$$
defined as follows:
\begin{equation}\label{l19_int}
\int_\MM \mu\;f=\int_M \lan \mu, (f)_m \ran
\end{equation}
where $\lan,\ran$ is the fiberwise pairing between line bundles $\wedge^m E$ and $\wedge^m E^*$; $(f)_m$ is the component of $f\in C^\infty(\Pi E)=\Gamma(M,\wedge^\bt E^*)$ in the top 
exterior power of $E^*$. Note that the integrand on the r.h.s. $\lan \mu, (f)_m \ran$ is a  section of $\wedge^n T^*M$ over $M$, i.e., a top degree form, and thus can be integrated. One can understand the definition (\ref{l19_int}) as doing a standard Berezin integral in odd fibers of $\Pi E$ and then integrating the result over the body in the ordinary (measure-theoretic) sense.

In fact, sections of $\Ber(\MM)$ over $M$ correspond to Berezinians that are \emph{constant in the fiber direction} of $\Pi E\ra M$. More generally, we can consider the super-vector bundle $\til\Ber(\MM)=\Ber(\MM)\otimes \wedge^\bt E^*$ over $M$, such that $\Gamma(M,\til\Ber(\MM))=\Gamma(M,\Ber(\MM))\otimes_{C^\infty(M)}C^\infty(\MM)$. We denote the space of sections $\BER(\MM):=\Gamma(M,\til\Ber(\MM))$. Its elements are the (general) 
Berezinians. By construction, $\BER(\MM)$ is a module over $C^\infty(\MM)$. Note that we can alternatively understand $\BER(\MM)$ as the space of sections of the pullback line bundle $p^*\Ber(\MM)$ over the whole of  $\MM$ rather than just the body $M$ (where $p:\Pi E\ra M$ is the bundle projection).
In the language of general Berezinians, integration (\ref{l19_int}) is simply a map
$$\int_\MM:\quad  \BER(\MM) \ra \RR$$

\begin{remark}
The notion of a Berezinian constant in the fiber direction depends on the splitting of the supermanifold $\MM$, i.e. on a particular identification of it with $\Pi E$ for $E\ra M$ a vector bundle. On the other hand, the general notion of a Berezian (element of $\BER(\MM)$) does not depend on the splitting.
\end{remark}

\begin{remark} Parity-shifted tangent bundle $\MM=\Pi TM$ carries a distinguished Berezinian $\mu_{\Pi TM}$, characterized as follows. For $f\in C^\infty(\Pi TM)\cong \Omega^\bt(M)$ denote $\til{f}$ the corresponding differential form on $M$. Then $\mu_{\Pi TM}$ satisfies
$$\int_{\Pi TM}\mu_{\Pi TM}\cdot f=\int_M \til{f}$$
where on the r.h.s. we have an ordinary integral over $M$ of a differential form. In the local coordinates (cf. Example \ref{l18_ex_T[1]}), we have $\mu_{\Pi TM}=\prod_i (dx^i D\theta^i)\;\in\BER(\Pi TM) $.
\end{remark}

When one considers integration over $\ZZ$-graded manifolds, only the underlying $\ZZ_2$-grading (superstructure) plays role for the integration theory.

\subsubsection{Change of variables formula for integration over supermanifolds}
\begin{definition}
Let $S$ be a supermanifold of parameters and $J\in  \mr{End}(\RR^{n|m})\otimes C^\infty(S)$ an $S$-dependent endomorphism of $\RR^{n|m}$ of block form
$$J=\left(\begin{array}{c|c}
A & B \\ \hline C & D
\end{array}\right)$$
with the blocks 
\begin{multline*}
A\in \left[\mr{End}(\RR^n)\otimes C^\infty(S)\right]_\even,\quad D\in \left[\mr{End}(\RR^m)\otimes C^\infty(S)\right]_\even,\\ 
B\in \left[\mr{Hom}(\RR^m,\RR^n)\otimes C^\infty(S)\right]_\odd,\quad 
C\in \left[\mr{Hom}(\RR^n,\RR^m)\otimes C^\infty(S)\right]_\odd
\end{multline*}
Assume that $D$ is invertible. Then the \emph{superdeterminant} of $J$ is defined as
\begin{equation}
\mr{Sdet}\, \left(\begin{array}{c|c}
A & B \\ \hline C & D
\end{array}\right)
= \det(A-B D^{-1}C)\cdot (\det D)^{-1}\qquad \in C^\infty(S)
\end{equation}
\end{definition}

\begin{remark} Superdeterminant is characterized by the following two properties:
\begin{itemize}
\item Multiplicativity: for $J,K\in \mr{End}(\RR^{n|m})\otimes C^\infty(S)$, we have 
$$\mr{Sdet}(J K)=\mr{Sdet}(J)\cdot \mr{Sdet}(K) $$
where $JK$ is the composition of $J$ and $K$ as endomorphisms of $\RR^{n|m}$.
\item For $j=\left(\begin{array}{c|c}
a & b \\ \hline c & d
\end{array}\right)$ an $S$-dependent endomorphism of $\RR^{n|m}$, we have
$$\mr{Sdet}\;(\mr{id}+\epsilon\cdot j)=1+\epsilon\cdot \mr{Str}\;j + O(\epsilon^2)$$
Here $\mr{Str}\;j=\tr a-\tr d$ is the \emph{supertrace} of $j$.
\end{itemize}
Note that these two properties imply that 
$$\mr{Sdet}\; e^j=e^{\mr{Str}\; j}$$
\end{remark}

\begin{theorem}[Change of variables formula]
Let $\RR^{n|m}_{I}$, $\RR^{n|m}_{II}$ be two copies of the $(n|m)$-dimesnional vector superspace, endowed with coordinates $x^i,\theta^a$ on the first copy and coordinates $y^i,\psi^a$ on the second copy. Let $\phi: \RR^{n|m}_{I}\ra \RR^{n|m}_{II}$ be a smooth map of supermanifolds and $f(y,\psi)\in C^\infty_c(\RR^{n|m}_{II})$ a compactly supported function. Then the integral of $f$ over $\RR^{n|m}_{II}$ against the standard coordinate Berezinian can be expressed as an integral of the pullback of $f$ by $\phi$ as follows:
\begin{multline}\label{l19_change_of_coords}
\int_{\RR^{n|m}_{II}} d^n y \DD^m \psi \; f(y,\psi)=\int_{\RR^{n|m}_I} d^n x \DD^m \theta\;\; \mr{sign}\det \left(\frac{\dd y^i(x,0)}{\dd x^j}\right)\cdot \mr{Sdet}\; \frac{\dd(y,\psi)}{\dd (x,\theta)}\cdot f(y(x,\theta),\psi(x,\theta))
\end{multline}
Here on the r.h.s. 
$$\frac{\dd(y,\psi)}{\dd (x,\theta)}=\left(
\begin{array}{c|c} 
\frac{\dd y^i}{\dd x^j} & \frac{\dd y^i}{\dd \theta^b} \\ \hline
\frac{\dd \psi^a}{\dd x^j} & \frac{\dd \psi^a}{\dd \theta^b}
\end{array}
\right)\qquad \in \mr{End}(\RR^{n|m})\otimes C^\infty(\RR^{n|m}_{I})$$
is the super-matrix of first derivatives of $\phi$. The sign factor in (\ref{l19_change_of_coords}) is the sign of the determinant of the even-even block of the matrix of derivatives.\footnote{It corresponds to the fact that in the change of variables formula for an ordinary integral, the \emph{absolute value} of the Jacobian appears.}
\end{theorem}

\subsubsection{Divergence of a vector field}
\begin{definition}\label{l19_def_div}
For $v\in \mathfrak{X}$ a vector field on a supermanifold $\MM$ and $\mu\in \BER(\MM)$ a Berezinian, we define the \emph{divergence} $\mr{div}_\mu (v)\in C^\infty(\MM)$ of $v$ with respect to $\mu$ via the property
\begin{equation}\label{l19_div}
\int_\MM \mu \; v(f) = - \int_\MM \mu\;\mr{div}_\mu(v)\cdot f
\end{equation}
for any compactly supported test function $f\in C^\infty_c(\MM)$.
\end{definition}

\begin{example} For $\MM=M$ an ordinary manifold and $\mu$ a volume form, by Stokes' theorem we have
$$0\underset{\mr{Stokes'}}{=}\int_M \LL_v(\mu\, f)=\int_M \mu\; v(f)+\underbrace{(\LL_v \mu)}_{\mu\cdot\mr{div}_\mu(v)}\cdot f$$
where $\LL_v$ is the Lie derivative along $v$. Thus, definition (\ref{l19_div}) is compatible, in the context of ordinary geometry, with the definition of divergence as $\mr{div}_\mu(v)=\frac{\LL_v \mu}{\mu}$. I.e., roughly speaking, the divergence measures how the flow by $v$ changes volumes of subsets of $M$, as measured using $\mu$.
\end{example}

The following is a straightforward consequence of the Definition \ref{l19_def_div}.
\begin{lemma}
Let $\mu,\mu_0$ be two Berezinians on $\MM$ with $\mu=\rho\cdot \mu_0$ where $\rho\in C^\infty(\MM)$ is a nonvanishing function. Then, for $v\in \mathfrak{X}(\MM)$ a vector field, divergences with respect to $\mu$ and $\mu_0$ are related as follows:
\begin{equation}\label{l19_div_change_of_Ber}
\mr{div}_\mu(v)=\mr{div}_{\mu_0}(v)+\underbrace{\frac{1}{\rho}\cdot v(\rho)}_{=v(\log\rho)}
\end{equation}
\end{lemma}

On a general supermanifold $\MM$, using local coordinates $x^i,\theta^a$, assume first that $\mu=\mu_\mr{coord}=d^n x\DD^m\theta$ -- the standard coordinate Berezinian. The vector field can be expressed locally as 
$$v=\sum_i v^i(x,\theta)\frac{\dd}{\dd x^i}+\sum_a v^a(x,\theta)\frac{\dd}{\dd \theta^a}$$
Then the divergence of $v$ is given by the local formula:
\begin{equation}\label{l19_div_locally}
\mr{div}_{\mu_\mr{coord}}v=\sum_i \frac{\dd}{\dd x^i}v^i-(-1)^{|v|}\sum_a \frac{\dd}{\dd \theta^a} v^a
\end{equation}
Note that, using derivatives acting on the left,\footnote{
For $f$ a function of commuting variables $x^i$ and anti-commuting variables $\theta^a$, let $y$ be one of $x$s or $\theta$s. One  denotes the ordinary derivative as $\frac{\ora\dd}{\dd y}f=\frac{\dd}{\dd y} f$ and sets $f\frac{\ola\dd}{\dd y}:=(-1)^{|y|\cdot (|f|+1)} \frac{\dd}{\dd y} f$. In particular $y\frac{\ola \dd}{\dd y}=1$. The idea is that, if $f$ is monomial, in order to calculate $f\frac{\ola\dd}{\dd y}$, if $y$ occurs in $f$, one commutes $y$ to the right in the monomial, using Koszul sign rule, and then $y$ gets killed by the derivative from the right (acting on the left).
} we can simplify the signs:
$$\mr{div}_{\mu_\mr{coord}}(v)=\sum_i v^i\frac{\ola\dd}{\dd x^i}-\sum_a v^a \frac{\ola\dd}{\dd \theta^a}$$

In a more general case one, when $\mu$ is not the coordinate Berezinian, one obtains the local formula by combining (\ref{l19_div_locally}) with (\ref{l19_div_change_of_Ber}).

\subsection{BRST formalism}
BRST formalism arose in \cite{BRS,Tyutin} independently as a cohomological formalism for treating gauge symmetry.

\subsubsection{Classical BRST formalism}
We will call a \emph{classical BRST theory} the following supergeometric data:  
\begin{itemize}
\item A $\ZZ$-graded supermanifold $\FF$ (the ``space of fields''),
\item A cohomological vector field -- a vector field $Q\in \mathfrak{X}(\FF)_1$ satisfying $\boxed{Q^2=0}$ -- the ``BRST operator'' (encoding the data of gauge symmetry),
\item A function $S\in C^\infty(\FF)_0$ --  the ``action'' satisfying $\boxed{Q(S)=0}$ (gauge-invariance property).
\end{itemize}

\begin{example}\label{l19_ex_FP_via_BRST}
Starting from Faddeev-Popov data -- action of a group $G$ on a manifold $X$ and an invariant function $S\in C^\infty(X)^G$, we construct the BRST package as follows:
$\FF=X\times \g[1]$ with 
$$Q=\frac12 \sum_{a,b,c}f^c_{ab} c^a c^b \frac{\dd}{\dd c^c}+\sum_{a,i}c^a v^i_a(x)\frac{\dd}{\dd x^i}$$
with $x^i$ local coordinates on $X$, $c^a$ the  degree $1$ coordinates on $\g[1]$; $v_a$ are the fundamental vector fields of $G$-action on $X$.

In other words, the functions of fields $C^\infty(\FF)=\wedge^\bt \g^*\otimes C^\infty(X)=C^\bt_{CE}(\g,C^\infty(X))$ are the Chevalley-Eilenberg cochains of $\g$ twisted by the module $C^\infty(X)$ with $Q=d_{CE}$ the corresponding Chevalley-Eilenberg differential. Equivalently, $(\FF,Q)$ is the dg manifold associated to the action Lie algebroid for the action of $G$ on $X$ (via the construction of Examples \ref{l18_ex_Lie_algbd}, \ref{l18_ex_action_Lie_algbd}).

Note that $Q^2=0$ is equivalent to the pair of properties: Jacobi identity for the bracket in $\g$ and the condition that the infinitesimal action $\g\ra \mathfrak{X}(X)$ is a Lie algebra homomorphism. The equation $Q(S)=0$ is equivalent to $\g$-invariance of $S$ (cast as $v_a(S)=0$ with $v_a$ the fundamental vector fields).
\end{example}

\subsubsection{Quantum BRST formalism}\label{sss: qBRST}
We define the quantum (finite-dimensional) BRST theory as the data of classical BRST theory $(\FF,Q,S)$ with an additional structure adjoined: a Berezinian $\mu$ on $\FF$ (the finite-dimensional toy model for the functional integral measure), such that the following property holds: 
\begin{equation}\label{l19_div_free}
\boxed{\mr{div}_\mu Q=0}
\end{equation}
-- compatibility of the integration measure on fields with gauge symmetry.

\begin{lemma}\label{l19_lm_int_of_Q_coboundary}
For any $f\in C^\infty_c(\FF)$
\begin{equation}
\int_\FF \mu\; Q(f)=0
\end{equation}
\end{lemma}
(Follows immediately from divergence-free condition (\ref{l19_div_free}) and the Definition \ref{l19_def_div}.)

\marginpar{\LARGE{Lecture 20, 11/7/2016.}}

A \emph{BRST integral} is an integral of $Q$-cocycle, $I=\int_\FF \mu\; f$ with $Q(f)=0$. By Lemma \ref{l19_lm_int_of_Q_coboundary}, the integral is invariant under shifts of the integrand by a $Q$-coboundary. I.e., the integrand can be considered modulo shifts $f\sim f+Q(g)$ for any $g$. In other words, the BRST integral is a map
$$\int_\FF \mu: \; H_Q(C^\infty(\FF))\ra \RR\;\mbox{(or $\CC$)}$$
assigning numbers to cohomology classes of $Q$. The relevant case for field theory is when the $Q$-cocycles are complex-valued, in which case the BRST integral takes values in $\CC$.

In particular, we are interested in the oscillatory BRST integral 
\begin{equation}\label{l20_shift_by_Q(Psi)}
Z=\int_\FF \mu \; e^{\frac{i}{\hbar}S}=\int_\FF \mu \; e^{\frac{i}{\hbar}(S+Q(\Psi))}
\end{equation}
Here $\Psi\in C^\infty(\FF)_{-1}$ is an arbitrary function generating the shift of the integrand by a $Q$-exact term;\footnote{
The fact that the integrand in the l.h.s. and r.h.s. of (\ref{l20_shift_by_Q(Psi)}) differs by a $Q$-exact term, i.e., that  $e^{\frac{i}{\hbar}(S+Q(\Psi))}-e^{\frac{i}{\hbar}S}=Q(\cdots)$, follows from a simple computation: $e^{X+Q(Y)}-e^X=e^X\sum_{n=1}^\infty \frac{1}{n!}Q(Y)^n=Q\left(e^X Y \Phi(Q(Y))\right)$. Here $X$ (of degree $0$) is assumed to be $Q$-closed and we denoted $\Phi(x):=\frac{e^x-1}{x}=\sum_{n=0}^\infty \frac{1}{(n+1)!}x^n$. Setting $X=\frac{i}{\hbar}S$, $Y=\frac{i}{\hbar}\Psi$, we obtain the statement above.
} 
in this context $\Psi$ is known as the \emph{gauge-fixing fermion}. 

\begin{remark} \emph{Observables} in BRST formalism are $Q$-cocycles $\OO\in C^\infty(\FF)$. Given a collection of observables $\OO_1,\ldots,\OO_N$, one can consider their expectation value (correlation function): 
$$
\lan \OO_1\cdots \OO_N \ran:= \frac{1}{Z}\int_\FF \mu\; \OO_1\cdots \OO_N\; e^{\frac{i}{\hbar}S}= \frac{1}{Z}\int_\FF \mu\; \OO_1\cdots \OO_N\; e^{\frac{i}{\hbar}(S+Q(\Psi))}
$$
The fact that $\OO_i$ are $Q$-cocycles imply that the entire integrand is a $Q$-cocycle, and thus one can again shift $S$ by a $Q$-coboundary.
\end{remark}

\textbf{The idea of gauge-fixing in BRST formalism:} L.h.s. of (\ref{l20_shift_by_Q(Psi)}) typically \emph{perturbatively ill-defined}, i.e., cannot be evaluated (in the aymptotic regime $\hbar\ra 0$) by the stationary phase formula, due to the degeneracy of critical points of $S$ arising from gauge symmetry. On the other hand, the r.h.s. of (\ref{l20_shift_by_Q(Psi)}) is \emph{perturbatively well-defined} (i.e. critical points are non-degenerate and the stationary phase formula is applicable), for a good choice of $\Psi$. So, the r.h.s. of (\ref{l20_shift_by_Q(Psi)}) is the gauge-fixed BRST integral which can be evaluated in terms of Feynman diagrams. By Lemma \ref{l19_lm_int_of_Q_coboundary}, the result is independent on the choice of gauge-fixing fermion $\Psi$ (though the particular Feynman rules for calculating the r.h.s. of (\ref{l20_shift_by_Q(Psi)}) do depend on $\Psi$; the result is independent of $\Psi$ once all contributing graphs are summed over).

\subsubsection{Faddeev-Popov via BRST}\label{sss: FP via BRST}
We start from Faddeev-Popov data: an $n$-manifold $X$ acted on by a compact group $G$, an invariant function $S\in C^\infty(X)^G$, an invariant volume form $\mu_X\in \Omega^n(X)^G$ and a gauge-fixing function $\phi:X\ra \g$ defining a local section of $G$-orbits $\phi^{-1}(0)\subset X$.

\textbf{First attempt.} Set, as in Example \ref{l19_ex_FP_via_BRST}, $\FF=X\times \g[1]$, with the cohomological vector field locally written as
\begin{equation}
Q=\frac12 \sum_{a,b,c}f^c_{ab} c^a c^b \frac{\dd}{\dd c^c}+\sum_{a,i}c^a v^i_a(x)\frac{\dd}{\dd x^i}
\end{equation}
and the Berezinian $\mu=\mu_X\cdot \DD^m c$. Here $\DD^m c=\ola\prod_a \DD c^a$ is the coordinate Berezinian on $\Pi\g$ 
(invariantly, it is the element of $\wedge^m\g$ compatible with the chosen normalization of Haar measure on $G$).

Note that the divergence of $Q$ is
$$\mr{div}_\mu Q=\sum_a c^a \left(\sum_b f^b_{ab} + \mr{div}_{\mu_X}v_a\right)$$
The two terms on the r.h.s. vanish individually because:
\begin{itemize}
\item $\mu_X$ is $G$-invariant and thus fundamental vector fields $v_a$ are divergence-free,
\item the contraction of the structure constants in $\g$, $\sum_b f^b_{ab}=\tr_\g [T_a,-]$ vanishes due to \emph{unimodularity} of $\g$,\footnote{Recall that a Lie algebra $\g$ is called \emph{unimodular} if the matrices of adjoint representation are traceless, $\tr_\g[x,-]=0$ for any $x\in\g$.} which in turn follows from the existence of Haar measure on $G$.
\end{itemize}

\textbf{Problem:}
\begin{enumerate}
\item $\FF$ has coordinates $x^i$ of degree $0$ and $c^a$ of degree $1$, in particular $\FF$ is non-negatively graded. Thus, there is no non-zero element $\Psi\in C^\infty(\FF)_{-1}$ which we would need for gauge-fixing (\ref{l20_shift_by_Q(Psi)}).
\item The integral $\int_\FF \mu\; e^{\frac{i}{\hbar}S(x)}=0$ vanishes because of the integral over $\g[1]$ (note that the integrand is independent of the odd variables $c^a$, hence the Berezin part of the integral vanishes trivially).
\end{enumerate}

\textbf{Solution/second attempt.}
Let us call the quantum BRST package constructed in the first attempt the \emph{minimal} BRST package, $(\FF_\mr{min},Q_\mr{min},\mu_\mr{min})$. We construct the new, non-minimal BRST package by setting: 
\begin{itemize}
\item Non-minimal fields: $\FF:=\FF_\mr{min}\times \FF_\mr{aux}$ where the added auxiliary fields are $\FF_\mr{aux}:=\g^*[-1]\oplus \g^*$ with degree $-1$ coordinates $\bar{c}_a$ and degree $0$ coordinates $\lambda_a$. Thus, we have the following local coordinates on $\FF=X\times\g[1]\times \g^*[-1]\times \g^*$:\\
\begin{tabular}{c|c}
coordinates & degree \\ \hline
$x^i$ on $X$ & 0 \\
$c^a$ on $\g[1]$ & 1 \\
$\bar{c}_a$ on $\g^*[-1]$ & -1 \\
$\lambda_a$ on $\g^*$ & 0
\end{tabular}
\item Non-minimal cohomological vector field: $Q:=Q_\mr{min}+Q_\mr{aux}$ with $Q_\mr{aux}=\sum_a \lambda_a\frac{\dd}{\dd\bar{c}_a}$. The added term can be regarded as a de Rham vector field on $\FF_\mr{aux}=T[1]\g^*[-1]$.\footnote{Note that 
the complex $C^\infty(\FF_\mr{aux}),Q_\mr{aux}$ has the cohomology of a point. Thus, complexes $C^\infty(\FF),Q$ and $C^\infty(\FF_\mr{min}),Q_\mr{min}$ are quasi-isomorphic.}
\item The non-minimal Berezinian $\mu=\mu_X\cdot \DD^m c\cdot \DD^m \bar{c}\cdot d^m\lambda$.
\end{itemize}

The integral 
\begin{equation}\label{l20_non-gauge-fixed}
\int_\FF \mu \; e^{\frac{i}{\hbar}S}
\end{equation}
contains a $0\cdot\infty$ indeterminacy: $0$ comes from the Berezin integral over the odd variables $c,\bar{c}$ of the integrand independent on them; $\infty$ comes from the integral over the even variable $\lambda$ of the integrand independent on $\lambda$.

However, let us replace the ill-defined integral $\int_\FF \mu \; e^{\frac{i}{\hbar}S}$ by a gauge-fixed integral 
\begin{equation}\label{l20_gauge-fixed}
I=\int_\FF \mu \; e^{\frac{i}{\hbar}(S+Q(\Psi_\phi))}
\end{equation} 
with the gauge-fixing fermion
\begin{equation}\label{l20_Psi_phi}
\Psi_\phi:=\lan \bar{c},\phi(x) \ran\quad \in C^\infty(\FF)_{-1}
\end{equation}
Note that this implies
$$Q(\Psi_\phi)=\lan \lambda, \phi(x) \ran+ \lan \bar{c}, FP(x) c \ran$$
Thus, the gauge-fixed integral (\ref{l20_gauge-fixed}) is precisely the Faddeev-Popov integral (\ref{l15_FP2}). In particular, the integral (\ref{l20_gauge-fixed})
\begin{enumerate}[(a)]
\item \label{l20_a} exists (converges) and is equal to $\frac{(2\pi i)^m}{\mr{Vol}(G)}\int_X \mu_X\;e^{\frac{i}{\hbar}S}$,
\item\label{l20_b} is invariant, by Lemma \ref{l19_lm_int_of_Q_coboundary}, under changes of the gauge-fixing fermion $\Psi_\phi$, and in particular invariant under changes of $\phi:X\ra \g$.
\end{enumerate}
\begin{remark} Note that the comparison of the ill-defined integral (\ref{l20_non-gauge-fixed}) with the gauge-fixed integral (\ref{l20_gauge-fixed}) 
is outside of the assumptions of the Lemma \ref{l19_lm_int_of_Q_coboundary}: the difference of the integrands is $Q$-exact but not compactly supported (in particular, in the direction of the Lagrange multiplier variables $\lambda_a$). This is why in this case the gauge-fixing (\ref{l20_non-gauge-fixed})$\ra$(\ref{l20_gauge-fixed}) is simultaneously a regularization of the ill-defined integral (\ref{l20_non-gauge-fixed}), rather than being an equality of two well-defined integrals as in (\ref{l20_shift_by_Q(Psi)}). Change of the gauge-fixing in (\ref{l20_b}) also leads, generally, to a non-compactly supported $Q$-exact shift of the integrand. However, as long as the integrals converge, (\ref{l20_shift_by_Q(Psi)}) still applies (in particular, we can deform the gauge-fixing $\phi:X\ra \g$ in a 1-parameter family $\phi_t$, $t\in [0,1]$, as long as $\phi^{-1}_t(0)\subset X$ is a local transversal section of $G$-orbits for all $t$).
\end{remark}

\begin{remark} One can employ more general gauge-fixing fermions than (\ref{l20_Psi_phi}). E.g.
\begin{itemize}
\item $\Psi=\lan \bar{c},\phi(x) \ran+ \frac{\varkappa}{2} (\bar{c},\lambda) $ with $(,)$ a non-degenerate pairing on $\g^*$ (e.g. the dual Killing form) and $\varkappa\in \RR$ a parameter of gauge-fixing. Then $Q(\Psi)=Q(\Psi_\phi)+\frac{\varkappa}{2}(\lambda,\lambda)$. Then one can take the Gaussian integral over $\lambda$ in $\int_\FF \mu\; e^{\frac{i}{\hbar}(S+Q(\Psi))}$. The result is a perturbatively well-defined integral over $X\times \g[1]\times \g^*[-1]$.\footnote{
E.g. in the case of Yang-Mills theory in Lorentz gauge, we have (writing only the quadratic in the fields part of the gauge-fixed action $S+Q(\Psi)$; we do not write the ghost term): $\int_M \tr \frac12 dA\wedge * dA+ \lambda\wedge d^*A+\frac{\varkappa}{2}\lambda\wedge *\lambda 
$. After integrating out the field $\lambda$, we obtain $\int_M \tr \frac12 A (d^*d-\frac{1}{\varkappa}d d^*)A$. In particular, taking $\varkappa=-1$, we obtain the standard Hodge-de Rham Laplacian as the $A-A$ part of the Hessian.
}
\item We can allow $\Psi$ to contain monomials of higher degree in $c$ and $\bar{c}$. This leads to new vertices in the Feynman rules for the gauge-fixed integral. E.g., if $\Psi$ contains a term $\propto \bar{c}\bar{c} c$ and thus $Q(\Psi)$ contains a term $\propto \bar{c}\bar{c}cc$ leading to the new vertex $$\vcenter{\hbox{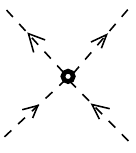}}$$
\end{itemize} 
\end{remark}

\begin{remark}
Due to freeness of the $G$-action on $X$, the BRST cohomology $H^\bt_Q\cong H^\bt_{Q_\mr{min}}$ is concentrated in degree zero and $H^0_Q\cong C^\infty(X)^G\cong C^\infty(X/G)$. In this sense, one may say that $(\FF,Q)$ is a \emph{resolution} of the quotient $X/G$.
\end{remark}

\begin{remark}\label{l20_rem_Lie_algbd_symmetry_via_BRST} In the construction of Faddeev-Popov setup cast withing the BRST framework, one can replace the symmetry given by a group action on $X$ by a symmetry given by an (injective) Lie algebroid $E\ra TX$. In this case the infinitesimal symmetry is given by an integrable distribution $\mr{im}(E)\subset TX$ on $X$ and gauge orbits are replaced by the leaves of the foliation on $X$ induced by this distribution via Frobenius theorem. In this case $\FF_\mr{min}=E[1]$ with the corresponding cohomological vector field (see Example \ref{l18_ex_Lie_algbd}). The full space of BRST fields is $\FF=E[1]\oplus E^*[-1]\oplus E^*$ (as a Whitney sum of graded vector bundles over $X$) with the homologically trivial extension of $Q$ to the auxiliary fields.  If the foliation is globally well-behaved (induces a fibration of $X$ over a smooth quotient $X/E$), then, similarly to (\ref{l20_a}) above, one has a comparison theorem \cite{MR09} asserting that the BRST integral equals 
$$(2\pi i)^{\mr{rk}(E)})\int_X \frac{\mu_X}{\mr{Vol}(\lambda_x)}e^{\frac{i}{\hbar}S}$$
where $\mr{Vol}(\lambda_x)$ is the volume of the leaf of the foliation passing through the point of integration  $x\in X$.
\end{remark}

\subsubsection{Remark: reducible symmetries and higher ghosts}\label{sss: higher ghosts}
BRST formalism can be applied to the case when the $G=G^1$ acts on $X$ with a stabilizer $G^2$ -- in this case, in addition to the ghosts of degree (ghost number) $1$ associated to the Lie algebra $\g=\g^1$, one introduces \emph{higher ghosts} for the Lie algebra $\g^2$. It may happen that it is convenient (in order to be compatible with locality structure on the underlying spacetime manifold) to have $G^2$ over-parameterizing the stabilizer of the $G=G^1$-action (i.e. different elements of $G^2$ may correspond to the same element in the stabilizer), then one introduces a second stabilizer $G^3$ and, respectively, new higher ghosts of degree $3$. This process can be iterated further. An example of this situation is the ``$p$-form electrodynamics'' -- a field theory on a Riemannian manifold $M$ with classical fields $X=\Omega^p(M)\ni \alpha$ and action $S=\frac12 \int_M d\alpha\wedge *d\alpha$. We have gauge-symmetry $\alpha\mapsto \alpha+d\beta$ with $\beta\in \Omega^{p-1}(M)=:G=G^1$. Clearly, $G^1$ acts on $X$ non-freely. In particular $\beta\sim \beta+d\gamma$, with $\gamma\in \Omega^{p-2}=:G^2$, correspond to the same gauge transformation\footnote{
Note that $G^2$ fails to parameterize the entire stabilizer of the $G$-action, if de Rham cohomology $H^{d-1}(M)\neq 0$. One solution is to twist the differential forms on $M$ by an acyclic local system. Another way is to allow this discrepancy. It will result in BRST cohomology not being concentrated just in degree zero (however, the degree nonzero cohomology will have \emph{finite rank}).
} etc. We have a tower of (infinitesimal) symmetry
$$ \underbrace{\Omega^0(M)}_{\g^{p}}\ra  \cdots\ra \underbrace{\Omega^{p-2}(M)}_{\g^2}\ra \underbrace{\Omega^{p-1}(M)}_{\g^1}\ra \underbrace{\Omega^p(M)}_{T_\alpha X}$$
-- a truncation of de Rham complex. The minimal BRST resolution in this case is $\FF=\bigoplus_{k=0}^p \Omega^{p-k}(M)[k] \;\ni (\alpha,c^{(1)},\ldots, c^{(p)})$ where field $c^{(k)}\in \Omega^{p-k}(M)$ is the $k$-th ghost (and has ghost number $k$).

\marginpar{\LARGE{Lecture 21, 11/9/2016.}}

\subsection{Odd-symplectic manifolds}(Main reference: \cite{Schwarz}.)

\subsubsection{Differential forms on super (graded) manifolds}
Let $E\ra M$ be a vector bundle and $\Pi E$ the corresponding split supermanifold. Then one defines the space of $p$-forms on $\Pi E$ as
$$\Omega^p(\Pi E)=\Gamma(M,\bigoplus_{j=0}^p \wedge^\bt E^*\otimes \wedge^j T^*M \otimes \Sym^{p-j}E^*)$$
Here the bundle of $p$-forms on the r.h.s. is split according to the base/fiber bi-degree $(j,p-j)$.

More generally, for $\MM$ a supermanifold, one can define $\Omega^\bt(\MM)$ as functions on the parity-shifted tangent bundle $\Pi T\MM$. If $(x^i,\theta^a)$ are local even and odd coordinates on $\MM$, then $\Pi T\MM$ has local coordinates $x^i,\theta^a$ (on the base of the tangent bundle) and $dx^i, d\theta^a$ (on the fiber). Here $x^i, d\theta^a$ are even and $\theta^a, dx^i$ are odd. Also, one prescribes form degree  (or de Rham degree) $0$ to the base coordinates and form degree $1$ to the fiber coordinates. Transition maps between charts on $\Pi T\MM$ are written naturally in terms of transition maps between the underlying charts on $\MM$.

For $\MM$ a $\ZZ$-graded supermanifold, differential forms $\Omega^\bt(\MM)$ have the following three natural gradings:
\begin{enumerate}
\item Form (de Rham) degree $\deg_\mr{dR}$.
\item Internal degree (also called ``grade'') $\mr{gr}$, coming from $\ZZ$-grading of coordinates on $\MM$. In particular, grades of $x^i$ and $dx^i$ are the same, and similarly for $\theta^a$ and $d\theta^a$.
\item Total degree $\deg_\mr{dR}+\mr{gr}$.
\end{enumerate}
By convention, it is the parity of the total degree that governs the signs. 

In $\ZZ$-graded context, we will use notation $\Omega^p(\MM)_k$ for $p$-forms of grade $k$.

\begin{example} $p$-forms on the odd line $\RR^{0|1}$ are functions $f(\theta,x)$ of an odd variable $\theta$ (the coordinate on $\RR^{0|1}$) and even variable $x=d\theta$, which are of degree $p$ in $x$. I.e., $\Omega^p(\RR^{0|1})=\{x^p(a+b\cdot \theta)\;|\;a,b\in\RR\}$. In particular, unlike forms on an ordinary $n$-manifold, whose degree is bounded above by $n$, there are differential forms of arbitrarily large degree $p\geq 0$ on $\RR^{0|1}$!
\end{example}

\subsubsection{Odd-symplectic supermanifolds} 
Let $\MM$ be a supermanifold. 
\begin{definition}
An odd-symplectic structure on $\MM$ is a $2$-form $\omega$ on $\MM$ which is:
\begin{itemize}
\item closed, $d\omega=0$;
\item odd, i.e., in local coordinates $x^i,\theta^a$ on $\MM$ (with $x^i$ even and $\theta^a$ odd) has the form $\sum_{i,a} \omega_{ia}(x,\theta)dx^i\wedge d\theta^a$, with $(\omega_{ia}(x,\theta))$ a matrix of local functions on $\MM$;
\item is \emph{non-degenerate}, i.e., the matrix of coefficients $(\omega_{ia}(x,\theta))$ is invertible.
\end{itemize}
A supermanifold $\MM$ endowed with an odd-symplectic structure $\omega$ is called an \emph{odd-symplectic supermanifold}.
\end{definition}

Note that it follows from non-degeneracy of $\omega$ that the dimension of $\MM$ is $(n|n)$ for some $n$. 

We survey the main properties of odd-symplectic supermanifolds and Lagrangian submanifolds in them without giving proofs. For proofs and details, see \cite{Schwarz}. 

\begin{theorem}[Schwarz, \cite{Schwarz}]\label{l21_thm_Darboux}
Let $(\MM,\omega)$ be an odd-symplectic manifold with body $M$.
\begin{enumerate}[(i)]
\item\label{l21_Darboux_i} In the neighborhood of any point of $M$, one can find local 
coordinates $(x^i,\xi_i)$ on $\MM$, such that $\omega=\sum_i dx^i\wedge d\xi_i$.
\item\label{l21_Darboux_ii} There exists a (global) symplectomorphism\footnote{I.e. an invertible map of supermanifolds, such that the pullback along it intertwines the symplectic forms.} $\phi:(\MM,\omega)\ra (\Pi T^* M,\omega_\mr{stand})$ where $\omega_\mr{stand}$ is the standard (odd-)symplectic structure on the (odd) cotangent bundle, locally written as $\omega_\mr{stand}=\sum_i dx^i\wedge d\xi_i$.
\end{enumerate}
\end{theorem}

Here (\ref{l21_Darboux_i}) is the analog of Darboux theorem in odd-symplectic case. As in the ordinary symplectic geometry, one calls local coordinates $(x^i,\xi_i)$ such that $\omega=\sum_i dx^i \wedge d\xi_i$ the \emph{Darboux coordinates}.  The global statement (\ref{l21_Darboux_ii}) is very much unlike the situation of ordinary symplectic geometry: it says that, up to symplectomorphism, all odd-symplectic manifolds are (odd) cotangent bundles.

\begin{definition} A submanifold $\LL$ of an odd-symplectic manifold $(\MM,\omega)$ is called \emph{Lagrangian} if it \emph{maximally isotropic} in $\MM$, i.e., if
\begin{itemize}
\item $\LL$ is isotropic: $\omega|_\LL=0$,
\item $\LL$ is not a proper submanifold of another isotropic submanifold of $\MM$.
\end{itemize}
\end{definition}

A Lagrangian $\LL$ in an $(n|n)$-dimensional odd-symplectic manifold $\MM$ has dimension $(k|n-k)$ for some $0\leq k \leq n$.

\begin{example}[``Conormal Lagrangian'']\label{l21_ex_conormal}
Given a $k$-dimensional submanifold $C$ in an (ordinary) $n$-manifold $M$, we can construct a $(k|n-k)$-dimensional Lagrangian $\LL_C\subset \Pi T^* M$ (with $\Pi T^* M$ equipped with the standard symplectic structure of the cotangent bundle). The Lagrangian $\LL_C$ is constructed as the \emph{odd conormal bundle} (conormal bundle\footnote{
Recall that, for $C\subset M$, the conormal bundle $N^*C\subset (T^*M)|_C$ has the fiber $N^*_x C:= \mr{Ann}(T_x C)=\{\alpha\in T^*_x M\;\mr{s.t.}\; \lan \alpha,v \ran=0 \; \forall v\in T_x C\}$ over a point $x\in C$. Here $\mr{Ann}$ stands for \emph{annihilator} (of the subspace $T_x C\subset T_x M$).
} with reversed parity of conormal fibers) of $C$: 
$$\LL_C=\Pi N^*C \quad \subset \Pi T^* M$$
\end{example}

The following theorem is a direct analog, in odd-symplectic context, of Weinstein's tubular neighborhood theorem in the context of ordinary symplectic manifolds.

\begin{theorem}[Tubular neighborhood theorem in odd-symplectic context, \cite{Schwarz}]
Given a Lagrangian $\LL$ in an odd-symplectic manifold $\MM$, there exists 
\begin{itemize}
\item a tubular neighborhood $U\subset \MM$ (with projection $p:U\ra \LL$) of the Lagrangian $\LL\subset \MM$,
\item a tubular neighborhood $U_0\subset \Pi T^*\LL$ (with projection $p_0: U_0\ra \LL$) of the zero-section $\LL_0\simeq \LL$ of the parity-shifted cotangent bundle $\Pi T^*\LL$ (endowed with the standard odd-symplectic structure of the cotangent bundle),
\item a symplectomorphism $\phi: U \stackrel\sim\ra U_0$,
\end{itemize}
such that $\phi$ sends the Lagrangian $\LL\subset \MM$ to the zero-section of $\Pi T^*\LL$ and intertwines the projections $p$, $p_0$.
\end{theorem}

The tubular neighborhood theorem above states, essentially, that in the neighborhood of a Lagrangian, the ambient odd-symplectic manifold always looks like (is locally symplectomorphic to) the odd cotangent bundle of the Lagrangian.

\begin{example}[``Graph Lagrangian'']\label{l21_ex_graph} Let $\NN$ be a $(k|n-k)$-supermanifold and $\Psi\in C^\infty(\NN)_\mr{odd}$ an odd function. One has a Lagrangian 
\begin{equation}
\NN_\Psi:=\mr{graph}(d\Psi)\subset \Pi T^* \NN
\end{equation}
Note that $\Pi T^*\NN$ has dimension $(n|n)$. If $X^\alpha$ are the local coordinates on $\NN$ (some of $X^\alpha$s are even and some are odd), we have coordinates $(X^\alpha,\Xi_\alpha)$ on $\Pi T^*\NN$ with parity of the cotangent fiber coordinate $\Xi_\alpha$ opposite to the parity of $X^\alpha$. Then the submanifold $\NN_\Psi$ is given by
$$\Xi_\alpha(X) = \frac{\dd}{\dd X^\alpha}\Psi(X)$$
For $\Psi=0$, $\NN_\Psi$ is the zero-section of $\Pi T^*\NN$. For $\Psi$ nonzero, we get a deformation of the zero-section in the cotangent bundle, given as a graph of the exact 1-form $d\Psi$ on the base.
\end{example}

\begin{theorem}[Classification of Lagrangians, \cite{Schwarz}]\label{l21_thm_Lagr}
\begin{enumerate}[(a)]
\item\label{l21_Lagr_a} Given a Lagrangian $\LL$ in an odd-symplectic manifold $\MM$ with body $M$, there exists a submanifold $C\subset M$ and a symplectomorphism $\phi:\MM\stackrel\sim\ra \Pi T^*M$ such that $\phi$ maps $\LL\subset \MM$ to $\LL_C=\Pi N^* C\subset  \Pi T^* M$ (cf. Example \ref{l21_ex_conormal}). 
\item\label{l21_Lagr_b} A Lagrangian $\LL$ in $\Pi T^* M$ can be obtained from a Lagrangian of the standard form $\LL_C=\Pi N^*C$ for some $C\subset M$, as a graph of $d\Psi$ for some $\Psi\in C^\infty(\LL_C)_\mr{odd}$ (cf. Example \ref{l21_ex_graph}). Here we use the tubular neighborhood theorem to identify $\Pi T^*M$ in the neighborhood of $\LL_C$ with $\Pi T^* \LL_C$.
\end{enumerate}
\end{theorem}

\subsubsection{Odd-symplectic manifolds with a compatible Berezinian. BV Laplacian.}

\begin{definition}\label{l21_def_SP} For $(\MM,\omega)$ an $(n|n)$-dimensional odd-symplectic manifold, a Berezinian $\mu$ on $\MM$ is called \emph{compatible} with $\omega$, if there exists an atlas of Darboux charts $(x^i,\xi_i)$ on $\MM$ such that locally $\mu=d^n x \DD^n\xi$ is the coordinate Berezinian in all charts of the atlas. (We will call the Darboux charts with this property the \emph{special} Darboux charts.) Note that, in particular, this implies that the transition functions between charts are \emph{unimodular}: $\mr{Sdet}\;\frac{\dd(x_\beta,\xi_\beta)}{\dd(x_\alpha,\xi_\alpha)}=1$. In the terminology of Schwarz \cite{Schwarz}, an odd-symplectic manifold $(\MM,\omega)$ endowed with a compatible Berezinian $\mu$ is called an \emph{$SP$-manifold} (where ``$P$-structure'' refers to the odd-symplectic form and ``$S$-structure'' refers to the Berezinian).
\end{definition}

For $(\MM,\omega,\mu)$ an odd-symplectic manifold with a compatible Berezinian, one introduces the odd second order operator $\Delta: C^\infty(\MM)\ra C^\infty(\MM)$, the \emph{Batalin-Vilkovisky Laplacian}, defined locally, in the special Darboux charts of the Definition \ref{l21_def_SP}, as
\begin{equation}\label{l21_Delta}
\Delta=\sum_i \frac{\dd}{\dd x^i} \frac{\dd}{\dd \xi_i}
\end{equation}
Unimodularity of transition functions implies that $\Delta$ is a globally well-defined operator. Also, the BV Laplacian squares to zero,
\begin{equation}
\Delta^2=0
\end{equation}
This follows from a local computation $\Delta^2=\sum_{ij} \frac{\dd}{\dd x^i}\frac{\dd}{\dd x^j} \frac{\dd}{\dd \xi_i}\frac{\dd}{\dd \xi_j}$. Note that the summand changes sign under the transposition $(i,j)\mapsto (j,i)$, therefore the sum over $i,j$ vanishes.

Another way to define the BV Laplacian is as follows. Let $(\MM,\omega)$ be an odd-symplectic manifold and $\mu$ -- \emph{any} Berezinian on $\MM$. Define the operator $\Delta_\mu: C^\infty(\MM)\ra C^\infty(\MM)$ by setting
\begin{equation}\label{l21_Delta_mu}
\Delta_\mu(f):=\frac12 \mr{div}_\mu X_f
\end{equation} 
where $X_f\in \mathfrak{X}(\MM)$ is the \emph{Hamiltonian vector field} generated by the Hamiltonian $f$, defined by the equation 
$$\iota_{X_f}\omega=d f$$
For $f,g\in C^\infty(\MM)$, one defines the \emph{odd Poisson bracket} (also known as the \emph{anti-bracket}), similarly to Poisson bracket in ordinary symplectic geometry, as
\begin{equation}\label{l21_Poisson1}
\{f,g\}:=X_f(g)
\end{equation}
Locally, in a Darboux chart $(x^i,\xi_i)$ on $\MM$, assuming that the Berezinian has local form $\mu=\rho(x,\xi)\;d^n x \;\DD^n\xi$ with $\rho$ a local density function, we have:
\begin{equation}
\Delta_\mu (f)= \sum_i \frac{\dd}{\dd x^i}\frac{\dd}{\dd \xi_i} f+\frac12 \{\log\rho,f\}
\end{equation}
And the local form of the odd Poisson bracket is:
\begin{equation}\label{l21_Poisson2}
\{f,g\}=\sum_i f\left( \frac{\ola \dd}{\dd x^i}\frac{\ora \dd}{\dd \xi_i} - \frac{\ola \dd}{\dd \xi_i}\frac{\ora \dd}{\dd x^i} \right) g
\end{equation}

The BV Laplacian $\Delta_\mu$, as defined by (\ref{l21_Delta_mu}), does not automatically square to zero. Rather, it squares to zero, $\Delta_\mu^2=0$ if and only if the Berezinian $\mu$ is compatible with $(\MM,\omega)$, in the sense of Definition \ref{l21_def_SP}. And in this case, $\Delta_\mu$ coincides with the BV operator (\ref{l21_Delta}) defined in the special Darboux charts.

\begin{remark}\label{l21_rem_compatible_Ber} A straightforward local computation shows that the operator (\ref{l21_Delta_mu}) squares to zero iff $\sum_i\frac{\dd}{\dd x^i}\frac{\dd}{\dd \xi_i}\sqrt\rho=0$. This is also turns out to be the necessary and sufficient condition for a special local Darboux chart to exist.
\end{remark}

\subsubsection{BV integrals. Stokes' theorem for BV integrals.}

Note that, for $M$ an $n$-manifold, the Berezin line bundle of the odd cotangent bundle $\mr{Ber}(\Pi T^*M)$, as a line bundle over $M$, is canonically identified  with the tensor square of the bundle of volume forms on $M$, i.e.,  
\begin{equation}\label{l21_Ber_0}
\mr{Ber}(\Pi T^*M)\cong(\wedge^n T^*M)^{\otimes 2}
\end{equation} Similarly, for $\NN$ a supermanifold, one has 
\begin{equation}\label{l21_Ber}
\mr{Ber}(\Pi T^*\NN)|_\NN \cong \mr{Ber}(\NN)^{\otimes 2} 
\end{equation}
Here we understand $\mr{Ber}(\NN)$ as a line bundle over $\NN$ and the l.h.s. is a pullback of a line bundle over $\Pi T^* \NN$ to $\NN$. 

In particular, (\ref{l21_Ber}) implies that there is a canonical map sending Berezinians $\mu$ on $\Pi T^*\NN$ to Berezinians ``$\sqrt{\mu|_\NN}$'' on $\NN$. Locally, if $X^\alpha$ are local coordinates on $\NN$ (of even and odd parity), and $(X^i,\Xi_i)$ the respective Darboux coordinates on $\Pi T^*\NN$, a Berezinian $\mu=\rho(X,\Xi)\; \DD X\, \DD\Xi$  on $\Pi T^*\NN $ is mapped to a Berezinian $\sqrt{\mu|_\NN}:=\sqrt{\rho(X,0)}\; \DD X$ on $\NN$.

For $(\MM,\omega,\mu)$ an odd-symplectic manifold with a compatible Berezinian, a \emph{BV integral} is an integral of the form
\begin{equation}
\int_{\LL\subset \MM} f\;\sqrt{\mu|_\LL}
\end{equation}
with $\LL$ a Lagrangian submanifold of $\MM$ and $f\in C^\infty(\MM)$ a function satisfying $\Delta_\mu f=0$.

\begin{theorem}[Stokes'theorem for BV integrals, Batalin-Vilkovisky-Schwarz, \cite{Schwarz}]\label{l21_thm_BV_Stokes}
Let $(\MM,\omega,\mu)$ be an odd-symplectic manifold with compact\footnote{Compactness condition can be dropped, but then one has to request that the integrals converge.} body endowed with a compatible Berezinian.
\begin{enumerate}[(i)]
\item For any $g\in C^\infty(\MM)$ and $\LL\subset \MM$ a Lagrangian submanifold, one has
\begin{equation}\label{l21_Stokes_i}
\int_\LL \Delta_\mu g\;\sqrt{\mu|_\LL}=0
\end{equation}
\item\label{l21_thm_BV_Stokes_ii} Let $\LL$ and $\LL'$ be two Lagrangian submanifolds whose bodies are homologous cycles in the body of $\MM$ and let $f\in C^\infty(\MM)$ be a function satisfying $\Delta_\mu f=0$. Then the following holds:
\begin{equation}\label{l21_Stokes_ii}
\int_{\LL} f\;\sqrt{\mu|_\LL} = \int_{\LL'} f\;\sqrt{\mu|_{\LL'}}
\end{equation}
\end{enumerate}
\end{theorem}

\begin{proof}[Idea of proof]
By (\ref{l21_Darboux_ii}) of Theorem \ref{l21_thm_Darboux}, without loss of generality we can assume $\MM=\Pi T^* M$ for $M$ an ordinary $n$-manifold. One introduces the \emph{odd (fiberwise) Fourier transform}  
$$OFT: C^\infty(\Pi T^* M) \stackrel\sim\ra C^\infty(\Pi TM)$$
In local coordinates $(x^i,\xi_i)$ on $\Pi T^*M$ and $(x^i,\theta^i)$ on $\Pi TM$, assuming that the Berezinian $\mu$ has form $\mu=\rho(x,\xi)\, d^n x\, \DD^n\xi$, the odd Fourier transform acts as follows:
$$ f(x,\xi)\mapsto \til f(x,\theta):=\int_{\Pi T^*_x M} \sqrt{\rho(x,\xi)}\;\DD^n\xi\; e^{\lan \theta,\xi\ran} f(x,\xi)  $$
For example, in the special case when $\mu=\nu^{\otimes 2}$ for $\nu\in \Omega^n(M)$ a top form, the odd Fourier transform simply maps polyvectors $\alpha\in \mc{V}^\bt(M)\cong C^\infty(\Pi T^* M) $ to differential forms $\iota_\alpha \nu\in \Omega^{n-\bt}(M)\cong C^\infty(\Pi TM)$ via contraction with the top form $\nu$.

The odd Fourier transform maps the BV Laplacian $\Delta_\mu$ on $C^\infty(\Pi T^*M)$ to the de Rham operator $d$ on $\Omega^\bt(M)\cong C^\infty(\Pi TM)$, i.e., $OFT\circ \Delta_\mu=d\circ OFT$.

Consider $\LL=\LL_C=\Pi N^* C$ the Lagrangian of Example \ref{l21_ex_conormal}, for $C\subset M$ a closed submanifold. Then the BV integral is
$$\int_{\LL_C}f\;\sqrt{\mu|_\LL}= \int_C \til f$$
where on the r.h.s. we have an integral of a differential form $\til{f}=OFT(f)$ on $M$ over the submanifold $C\subset M$.

Restricting to Lagrangians of form $\LL_C$, (\ref{l21_Stokes_i}) and (\ref{l21_Stokes_ii}) follow from the usual Stokes' theorem on $M$:
$$\int_{\LL_C} \Delta_\mu g\; \sqrt{\mu|_\LL}=\int_C d \til{g} =0$$
and
$$\int_{\LL_{C'}}f\; \sqrt{\mu|_{\LL_{C'}}}- \int_{\LL_{C}}f\; \sqrt{\mu|_{\LL_{C}}}=\left(\int_{C'}-\int_C\right) \til f= \int_D \til f $$
where $D\subset M$ is a submanifold with boundary $\dd D=C'-C$; to apply Stokes' theorem here, we used that $\til f$ is a closed form on $M$ which follows from the assumption $\Delta_\mu f=0$.

For general Lagrangians in $\Pi T^*M$, one can reduce to the case of Lagrangians of form $\LL_C$ using (\ref{l21_Lagr_a}) of Theorem \ref{l21_thm_Lagr} for (\ref{l21_Stokes_i}). For (\ref{l21_Stokes_ii}), one reduces to Lagrangians of form $\LL_C$ using (\ref{l21_Lagr_b}) of Theorem \ref{l21_thm_Lagr} together with the following calculation. Let $\LL_t$ be a smooth family of Lagrangians in $\MM$ with $t\in [0,1]$ a parameter, such that $\LL_{t+\epsilon}=\mr{graph}(\epsilon\cdot d\Psi_t+ O(\epsilon^2)
)$ (cf. Example \ref{l21_ex_graph}) for $\Psi_t\in C^\infty(\LL_t)$. Then, for $f\in C^\infty(\MM)$ satisfying $\Delta_\mu f=0$, one has 
$$\frac{d}{dt} \int_{\LL_t} f\;\sqrt{\mu|_{\LL_t}}=\int_{\LL_t}\Delta(f\cdot \Psi_t)\;\sqrt{\mu|_{\LL_t}}\qquad =0$$
which vanishes by (\ref{l21_Stokes_i}). Thus, we can take $\LL_t$ to be a family connecting a given Lagrangian $\LL\subset \Pi T^* M$ with a Lagrangian of form $\LL_C$. Such a family exists by (\ref{l21_Lagr_b}) of Theorem \ref{l21_thm_Lagr} and the value of the BV integral is constant along this family by the calculation above.

\end{proof}

\begin{remark} In this Subsection we were focusing on the case of supermanifolds. In the setting of $\ZZ$-graded supermanifolds, the convention is that an odd-symplectic form $\omega$ has internal degree (grade) $-1$, so that the odd Poisson bracket and the BV Laplacian $\Delta$ have degree $+1$.
\end{remark}

\begin{definition}\label{l21_def_Lagr_homotopy} We say that two Lagrangians $\LL$ and $\LL'$ in an odd-symplectic manifold $(\FF,\omega)$ are \emph{homotopic as Lagrangians} (or \emph{Lagrangian-homotopic}) if there exists a smooth family $\LL_t$, $t\in[0,1]$, of Lagrangians in $(\FF,\omega)$ (the \emph{Lagrangian homotopy}) connecting $\LL$ and $\LL'$. Then we denote $\LL\sim \LL'$. 
\end{definition}

\begin{remark} More generally, since the main reason we are interested in homotopic Lagrangians is because they yield same values for the BV integral of a $\Delta$-cocycle, one can replace notion of homotopy of Lagrangians above by the (weaker) equivalence relation of (\ref{l21_thm_BV_Stokes_ii}) of Theorem \ref{l21_thm_BV_Stokes} -- the condition of having homologous bodies.
\end{remark}

\marginpar{\LARGE{Lecture 22, 11/14/2016.}}

\subsection{Algebraic picture: BV algebras. Master equation and canonical transformations of its solutions}
\subsubsection{BV algebras}
BV algebras are an algebraic counterpart of odd-symplectic manifolds with a compatible Berezinian. (And Gerstenhaber algebras are the counterpart of odd-symplectic manifolds without a distinguished Berezinian.)

\begin{definition}
A \emph{BV algebra} is a unital commutative graded algebra $\mc{A}^\bt,\cdot$ over $\RR$ (the dot stands for the graded-commutative product) endowed with 
\begin{itemize}
\item A degree $+1$ Poisson bracket $\{-,-\}:\mc{A}^j\otimes \mc{A}^k\ra \mc{A}^{j+k+1}$ satisfying
\begin{itemize}
\item skew-symmetry:\footnote{The mnemonic rule for signs is that the \emph{comma} in $\{-,-\}$ carries degree $+1$, and one accounts for that via the Koszul sign rule when pulling graded objects to the left/right slot of the Poisson bracket.} $\{x,y\}=-(-1)^{(|x|+1)\,(|y|+1)} \{y,x\}$,
\item Leibniz identity (in first and second slot): 
\begin{equation}\label{l22_Poisson_Leibniz}
\{x,yz\}=\{x,y\}z+(-1)^{(|x|+1)\,|y|} y\{x,z\},\qquad \{xy,z\}=x\{y,z\}+(-1)^{|y|\, (|z|+1)}\{x,z\}y
\end{equation}
\item Jacobi identity:
$$ \{x,\{y,z\}\}=\{\{x,y\},z\}+(-1)^{(|x|+1)\, (|y|+1)} \{y,\{x,z\}\} $$
\end{itemize}
\item In addition, $\mc{A}^\bt$ should carry a BV Laplacian -- an  $\RR$-linear map $\Delta:\mc{A}^j\ra \mc{A}^{j+1}$ satisfying
\begin{itemize}
\item $\Delta^2=0$,
\item $\Delta(1)=0$ (with $1$ the unit in $\mc{A}^\bt$),
\item second order Leibniz identity
\begin{equation}\label{l22_7-term-rel}
\Delta(xyz)\pm \Delta(xy)z\pm \Delta(xz)y\pm \Delta(yz)x\pm \Delta(x)yz\pm \Delta(y)xz\pm \Delta(z)xy=0 
\end{equation}
\item Poisson bracket arises as the ``defect'' of the first order Leibniz identity for $\Delta$:
\begin{equation}\label{l22_Delta_Leibniz}
\Delta (x y)=\Delta x\cdot y + (-1)^{|x|}x\cdot \Delta y+(-1)^{|x|}\{x,y\}
\end{equation}
\end{itemize}
\end{itemize}
\end{definition}

\begin{remark}
\begin{enumerate}
\item The  defining relations of a BV algebra given above are interdependent. E.g., the second order Leibniz identity (\ref{l22_7-term-rel}) for $\Delta$ follows from (\ref{l22_Delta_Leibniz}) and the fact that $\{,\}$ is a bi-derivation of the commutative product (\ref{l22_Poisson_Leibniz}).
\item Forgetting $\Delta$, the structure $(\mc{A},\cdot,\{,\})$ is the structure of a \emph{Gerstenhaber algebra} (or ``degree $+1$ Poisson algebra'', or  ``$\mc{P}_0$ algebra'').
\item Forgetting the commutative product and shifting the grading on $\mc{A}$ by $1$, we get a differential graded Lie algebra $\mc{A}[1],\{,\},\Delta$. The fact that $\Delta$ is a derivation of $\{,\}$, i.e. that 
\begin{equation}\label{l22_Delta_Poisson_Leibniz}
\Delta \{x,y\}=\{\Delta x,y\}+ (-1)^{|x|+1}\{x,\Delta y\}
\end{equation}
is a consequence of the relations of a BV algebra. 
\end{enumerate}
\end{remark}

\begin{example}\label{l22_ex_polyvect} Let $M$ be an $n$-manifold and $\nu$ a volume form on $M$. We construct 
\begin{itemize}
\item $\mc{A}^{-j}:=\mc{V}^j(M)=\Gamma(M,\wedge^j TM)$ -- polyvector fields on $M$ with reverse grading. The graded-commutative product on $\mc{A}$ is the wedge product of polyvectors. 
\item The Poisson bracket $\{,\}:=[,]_{NS}:\mc{V}^k\otimes \mc{V}^k\ra\mc{V}^{j+k-1}$ is the Nijenhuis-Schouten bracket of polyvectors (the Lie bracket of vector fields extended to polyvector fields via Leibniz identity). 
\item The BV Laplacian is the divergence w.r.t. top form $\nu$, $\Delta=\mr{div}_\nu: \mc{V}^j\ra \mc{V}^{j-1}$. For $j=1$ this is the ordinary divergence of a vector field, and one extends to polyvectors ($j>1$) by imposing the relation (\ref{l22_Delta_Leibniz}).
\end{itemize}
\end{example}

\begin{example}[Main example]\label{l22_ex_odd-symp} Let $(\MM,\omega,\mu)$ be an odd-symplectic $\ZZ$-graded supermanifold with a compatible Berezinian. 
\begin{itemize}
\item We set $\mc{A}^\bt:=C^\infty(\MM)_\bt$ as a commutative graded algebra. 
\item We set $\{,\}$ to be the degree $+1$ Poisson bracket (\ref{l21_Poisson1},\ref{l21_Poisson2}) induced by the odd-symplectic form $\omega$,  $\{f,g\}=X_f(g)$. 
\item The BV Laplacian is defined to be the standard BV Laplacian (\ref{l21_Delta_mu}) on an odd-symplectic manifold with a compatible Berezinian, $\Delta_\mu f=\frac12 \mr{div}_\mu X_f$.
\end{itemize}
\end{example}

Note that Example \ref{l22_ex_polyvect} is a special case of the Example \ref{l22_ex_odd-symp}, corresponding to $\MM=T^*[-1]M$ with the standard symplectic structure of the cotangent bundle, and with $\mu=\nu^{\otimes 2}$, cf. (\ref{l21_Ber_0}).

\subsubsection{Classical and quantum master equation}

Given a BV algebra $(\mc{A}^\bt,\cdot,\{,\},\Delta)$, we say that an element $S\in \mc{A}^0$ satisfies \emph{classical master equation} (CME) if
\begin{equation}\label{l22_CME}
\{S,S\}=0
\end{equation}
Note that, unlike in an ordinary Poisson algebra, this equation is not tautological: $\{S,S\}$ does not vanish automatically by skew-symmetry of the Poisson bracket $\{,\}$.

Given a solution $S$ of classical master equation, one can construct $Q:=\{S,\bt\}\;\in \mr{Der}_1 \mc{A}^\bt$ -- a degree $1$ derivation which, as a consequence of (\ref{l22_CME}), satisfies $Q^2=0$. In the case of $\mc{A}^\bt$ being the algebra of functions on an odd-symplectic manifold $\MM$ (Example \ref{l22_ex_odd-symp}), the derivation $Q\in \mathfrak{X}(\MM)_1$ is a cohomological vector field on $\MM$ arising as the Hamiltonian vector field with Hamiltonian $S\in C^\infty(\MM)_0$ solving the classical master equation.

An element $S=S^{(0)}+(-i\hbar) S^{(1)}+(-i\hbar)^2 S^{(2)}+\cdots\;\; \in \mc{A}^0[[-i\hbar]]$, with $\hbar$ a formal parameter, is said to satisfy the \emph{quantum master equation} (QME) if the following holds
\begin{equation}\label{l22_QME_MC}
\frac12 \{S,S\}-i\hbar \Delta S =0 
\end{equation}

In the case when $\hbar$ can be inverted (e.g. if $S$ as a series in $\hbar$ has nonzero convergence radius and thus $\hbar$ can be taken to be finite), quantum master equation (\ref{l22_QME_MC}) can be equivalently written\footnote{
This can be seen, e.g., from the following calculation. For $x\in \mc{A}^0$, we have $\Delta\, x^n= n x^{n-1}\Delta x+ \frac{n(n-1)}{2}x^{n-2}\{x,x\}$ (proven by induction in $n$ using (\ref{l22_Delta_Leibniz})). Therefore, $\Delta\, e^x=\Delta \left(\sum_{n=0}^\infty \frac{x^n}{n!}\right)=(\Delta x+\frac12 \{x,x\})e^x$. Substituting $x=\frac{i}{\hbar}S$, we obtain $\Delta\, e^{\frac{i}{\hbar}S}=(-i\hbar)^{-2}\left(\frac12 \{S,S\}-i\hbar\Delta\, S\right)e^{\frac{i}{\hbar}S}$. This proves equivalence of (\ref{l22_QME_MC}) and (\ref{l22_QME_exp_form}).
} as 
\begin{equation}\label{l22_QME_exp_form}
\Delta\; e^{\frac{i}{\hbar}S}=0
\end{equation}

In terms of the coefficients $S^{(0)}, S^{(1)},\ldots$ of the expansion of $S$ in powers of $-i\hbar$, the quantum master equation (\ref{l22_QME_MC}) is equivalent to a sequence of equations:
\begin{eqnarray}
\{S^{(0)},S^{(0)}\} &=& 0 \\
\{S^{(0)},S^{(1)}\} + \Delta S^{(0)} &=& 0 \label{l22_QME_e2} \\
\{S^{(0)},S^{(2)}\}+\frac12 \{S^{(1)},S^{(1)}\}+\Delta S^{(1)} &=& 0 \label{l22_QME_e3}
\end{eqnarray}
etc. In particular, the leading term $S^{(0)}$ of the $\hbar$-expansion of a solution of QME satisfies the classical master equation. 

Given a solution $S^{(0)}$ of CME one may ask whether it can be extended to a solution of QME by adding terms proportional to powers of $\hbar$. 
Then, to find the first $\hbar$-correction, we need to solve (\ref{l22_QME_e2}). It is solvable iff the class of $\Delta S^{(0)}$ in degree $1$ cohomology of $Q=\{S^{(0)},\bt \}$ vanishes (note that $\Delta S^{(0)}$ is automatically $Q$-closed, as follows from CME for $S^{(0)}$ and from (\ref{l22_Delta_Poisson_Leibniz})). If $\Delta S^{(0)}$ is indeed $Q$-exact, we can choose the primitive $-S^{(1)}$ which solves (\ref{l22_QME_e2}) and gives the first $\hbar$-correction. Next, we look for the second correction, quadratic in $\hbar$. Equation (\ref{l22_QME_e3}) is solvable for $S^{(2)}$ iff the class of $\frac12 \{S^{(1)},S^{(1)}\}+\Delta S^{(1)}$ in $H^1_Q$ vanishes (again, this element is automatically $Q$-closed). And this process goes on: at each step we have a possible obstruction in $H^1_Q$; if the obstruction vanishes, we can construct the next term in $\hbar$-expansion of the solution of QME. If the obstructions at all steps vanish, we can construct the full extension of $S^{(0)}$ to a solution of QME by incorporating the appropriate corrections in powers of $\hbar$.

\subsubsection{Canonical transformations}
\begin{definition}\label{l22_def_can_transf}
Given two solutions of QME, $S,S'\in \mc{A}^0[[-i\hbar]]$, we say that $S$ and $S'$ are \emph{equivalent} (notation: $S\sim S'$) 
if there exists a \emph{canonical BV transformation} -- a family $S_t\in \mc{A}^0[[-i\hbar]]$, $R_t\in \mc{A}^{-1}[[-i\hbar]]$ parameterized by $t\in [0,1]$, such that
$S_0=S$ and $S_1=S'$, and the following equation holds:
\begin{equation}\label{l22_can_transf}
\frac{d}{dt}S_t =\{S_t,R_t\}-i\hbar \Delta R_t
\end{equation}
$R_t$ is called the \emph{generator} of the canonical BV transformation.
\end{definition}

\begin{remark}\label{l22_rem_can_transf}
Equation (\ref{l22_can_transf}) together with the fact $S=S_0$ satisfies QME implies that $S_t$ satisfies QME 
\begin{equation}\label{l22_QME_S_t}
\frac12 \{S_t,S_t\}-i\hbar\Delta S_t=0
\end{equation}
Indeed, taking the derivative in $t$ of (\ref{l22_QME_S_t}), we get $\delta_t (\frac{d}{dt}S_t)=0$ where the $t$-dependent differential $\delta_t:=\{S_t,\bt\}-i\hbar\Delta$ squares to zero due to the QME on $S_t$. On the other hand, (\ref{l22_can_transf}) reads $\frac{d}{dt}S_t=\delta_t (R_t)$ (i.e. improves $\delta_t$-closedness of $\frac{d}{dt}S_t$ to $\delta_t$-exactness). In particular, (\ref{l22_can_transf}) implies that $\frac{d}{dt}$ of the QME (\ref{l22_QME_S_t}) vanishes at time $t$ if QME is known to hold at time $t$. Therefore QME  for $S_t$ implies that QME for $S_{t+\epsilon}$ is satisfied up to $O(\epsilon^2)$. And this implies (via subdividing shift $t\ra t+\epsilon$ into $N$ shifts of length $\epsilon/N$ and taking the limit $N\ra\infty$) that, in fact, if (\ref{l22_QME_S_t}) holds at any time $t$ and (\ref{l22_can_transf}) holds for all times, then (\ref{l22_QME_S_t}) holds for all times.
\end{remark}

\begin{remark}
Equations (\ref{l22_can_transf},\ref{l22_QME_S_t}) together imply that
$$\frac{d}{dt}\, e^{\frac{i}{\hbar}S_t}=\Delta\left(-i\hbar\; e^{\frac{i}{\hbar}S_t}R_t\right) $$
Thus, in particular, if $S\sim S'$, the difference of the exponentials is $\Delta$-exact:
$$e^{\frac{i}{\hbar}S'}- e^{\frac{i}{\hbar}S}=\Delta (\cdots) $$
where $\cdots=-i\hbar\int_0^1 dt\; R_t \, e^{\frac{i}{\hbar}S_t}$.
\end{remark}

\begin{remark}
Equation (\ref{l22_can_transf},\ref{l22_QME_S_t}) together can be packaged as a single ``extended quantum master equation'' 
$$(dt\wedge \frac{d}{dt}-i\hbar\,\Delta)\;e^{\frac{i}{\hbar}\sigma}=0$$
on an element of total degree zero $\sigma=S_t+dt\, R_t\;\in \Omega^\bt([0,1])\otimes \mc{A}^\bt[[-i\hbar]]$ in non-homogeneous forms on the interval $[0,1]$ with coefficients in $\mc{A}^\bt[[-i\hbar]]$.
\end{remark}

\subsection{Half-densities on odd-symplectic manifolds. Canonical BV Laplacian. Integral forms}
\subsubsection{Half-densities on odd-symplectic manifolds 
}

\begin{definition}
A \emph{density $\rho$ of weight $\varkappa\in \RR$} (or a \emph{$\varkappa$-density}) on a supermanifold $\MM$, covered by an atlas of coordinate charts $U_\alpha$ with local coordinates $(x^i_{(\alpha)}, \theta^a_{\alpha})$, 
is a collection of locally defined functions $\rho_{(\alpha)}(x_{(\alpha)},\theta_{(\alpha)})$ satisfying the following transformation rule on the overlaps $U_\alpha\cap U_\beta$:
\begin{equation}\label{l22_density_transformation}
\rho_{(\alpha)}(x_{(\alpha)},\theta_{(\alpha)})= \rho_{(\beta)}(x_{(\beta)},\theta_{(\beta)})\cdot \left| \mr{Sdet}\frac{\dd (x_{(\alpha)},\theta_{(\alpha)})}{\dd (x_{(\beta)},\theta_{(\beta)})}\right|^\varkappa 
\end{equation}
We denote the space of (smooth) $\varkappa$-densities on $\MM$ by $\mr{Dens}^\varkappa(\MM)$.
\end{definition}

We are interested in the case of densities of weight $\varkappa=1/2$ (or \emph{half-densities}) on an odd-symplectic manifold $(\MM,\omega)$. We assume that the body $M$ of $\MM$ is oriented (and thus the odd fiber of $\Pi T^* M\simeq \MM$ is also oriented) and the atlas agrees with the orientation, and hence the Jacobians of the transition functions are positive.

We write a half-density on $(\MM,\omega)$ locally, in a Darboux chart $(x^i,\xi_i)$ as 
$$\rho=\rho(x,\xi)\cdot d^{\frac12}x\, \DD^{\frac12}\xi$$
where $d^{\frac12}x\, \DD^{\frac12}\xi$ is a locally defined symbol (coordinate half-density associated to the local coordinates $(x^i,\xi_i)$) satisfying the transformation property
$$d^{\frac12}x\, \DD^{\frac12}\xi=\left(\mr{Sdet}\,\frac{\dd(x,\xi)}{\dd(x',\xi')}\right)^{\frac12} d^{\frac12}x'\, \DD^{\frac12}\xi'$$
This rule is equivalent to the transformation rule (\ref{l22_density_transformation}) with $\kappa=\frac12$ for the coefficient functions: $\rho(x,\xi)\mapsto \rho(x',\xi')=\rho(x,\xi)\cdot \left(\mr{Sdet}\,\frac{\dd(x,\xi)}{\dd(x',\xi')}\right)^{-\frac12} $.

One can view half-densities as sections of the (tensor) square root of the Berezin line bundle:
$$\mr{Dens}^{\frac12}\MM=\Gamma(\MM,\mr{Ber}(\MM)^{\otimes \frac12})$$

\begin{remark}[Manin, \cite{Manin}]\label{l22_rem_Manin}
For $\mc{V}$ an $(k|n-k)$-dimensional vector superspace, one can consider the space of constant (coordinate-independent) Berezinians, $\mr{BER}_\mr{const}(\mc{V})=\mr{Det}\,\Pi\mc{V}=\wedge^n \mc{V}_\even^*\otimes \wedge^m \mc{V}_\odd$. For $(\mc{W},\omega)$ an odd-symplectic $(n|n)$-dimensional vector superspace, and $\mc{V}=\LL\subset \mc{W}$ a Lagrangian subspace, the space of constant half-densities on $\mc{W}$ is canonically isomorphic to the space of constant Berezinians on $\LL$,
\begin{equation}\label{l22_HDens_via_Lagrangians}
\mr{Dens}^{\frac12}_\mr{const}(\mc{W})\cong \mr{BER}_\mr{const}(\LL)
\end{equation}
via the map  
$$\mr{BER}_\mr{const}(\LL)\ni \quad \nu\;\mapsto \; (\nu^{\otimes 2})^{\otimes \frac12}\quad \in \mr{Dens}^{\frac12}_\mr{const}(\mc{W})\cong \mr{BER}_\mr{const}^{\otimes \frac12}(\mc{W})$$ 
where $\nu^{\otimes 2}\in \mr{BER}_\mr{const}(\mc{W})\cong \mr{BER}_\mr{const}(\Pi T^* \LL)\cong \mr{BER}_\mr{const}(\LL)^{\otimes 2}$.\footnote{
The crucial linear algebra observation here, formulated in terms of determinant lines of vector superspaces, is that $\mr{Det}(\mc{V}\oplus \Pi \mc{V}^*)\cong
\mr{Det}(\mc{V})^{\otimes 2}
$, cf. (\ref{l21_Ber_0},\ref{l21_Ber}).
} Thus, one can understand constant a half-densities on an odd-symplectic space $(\mc{W},\omega)$ as a Berezinian on any Lagrangian subspace $\LL\subset \mc{W}$, or, since one has isomorphisms (\ref{l22_HDens_via_Lagrangians}), as a coherent system of Berezinians on all Lagrangian subspaces of $(\mc{W},\omega)$. Or, equivalently, as an equivalence class of pairs $(\LL,\mu_\LL)$ of a Lagrangian $\LL\subset \mc{W}$ and a constant Berezinian on $\LL$.
\end{remark}

\begin{example}
Consider odd-symplectic  $(1|1)$-superspace $\mc{W}=\Pi T^*\RR=\RR^{1|1}$ with Darboux coordinates $x,\xi$. The constant half-density $\rho=d^{\frac12}x \,\DD^{\frac12}\xi$ on $\RR^{1|1}$ induces the Berezinian (volume form) $dx$ on  the Lagrangian $\RR^1\subset \RR^{1|1}$ and the Berezinian $\DD\xi$ on the Lagrangian $\RR^{0|1}\subset \RR^{1|1}$.
\end{example}

\begin{remark}[\v{S}evera, \cite{Severa}]\label{l22_rem_Severa_1}
Given an odd-symplectic $(n|n)$-supermanifold $(\MM,\omega)$, one can consider the operator $\omega\wedge: \Omega^p(\MM)_k\ra \Omega^{p+2}(\MM)_{k-1}$ as a differential on the space of forms on $\MM$ (note that it does indeed square to zero since $\omega\wedge\omega=0$). Then the cohomology $H^\bt_{\omega\wedge}(\Omega(\MM))$ is canonically isomorphic to the space of half-densities on $\MM$. Locally, in Darboux coordinates $(x^i,\xi_i)$ on $\MM$, cohomology classes in $H^\bt_{\omega\wedge}(\Omega(\MM))$ have canonical representatives of form 
\begin{equation}\label{l22_Severa_cocycle}
\rho(x,\xi)\,dx^1\wedge\cdots\wedge dx^n \in \Omega^n(\MM)
\end{equation} 
which correspond to the half-densities $\rho(x,\xi)\, \prod_{i=1}^n d^{\frac12}x^i\,\DD^{\frac12}\xi_i$ (with the same coefficient $\rho(x,\xi)$) via Remark \ref{l22_rem_Manin}.
\end{remark}

\subsubsection{Canonical BV Laplacian on half-densities}
Let $(\MM,\omega)$ be an odd-symplectic manifold. One can define (Khudaverdian, \cite{Khudaverdian}) the \emph{canonical BV Laplacian} on half-densities,  $\Delta: \mr{Dens}^{\frac12}\MM\ra \mr{Dens}^{\frac12}\MM$, locally given in a Darboux chart by
\begin{equation}
\Delta:\quad \rho(x,\xi)\,d^{\frac12}x\,\DD^{\frac12}\xi\;\mapsto \; \left(\sum_i \frac{\dd}{\dd x^i}\frac{\dd}{\dd \xi_i}\rho(x,\xi)\right)\,d^{\frac12}x\,\DD^{\frac12}\xi 
\end{equation}
The nontrivial check \cite{Khudaverdian} is that the formula above defines a globally well-defined operator on half-densities. 

Note that the operator $\Delta$ on half-densities does not rely on a choice of a Berezinian on $\MM$, unlike the Schwarz's BV Laplacian (\ref{l21_Delta_mu}) $\Delta_\mu$ on functions on $\MM$.

Given a compatible Berezinian $\mu$ on $(\MM,\omega)$, one has the associated ``reference'' half-density $\sqrt\mu\in \mr{Dens}^{\frac12}(\MM)$, multiplication by which induces an isomorphism
$$C^\infty(\MM)\xra{\cdot \sqrt{\mu}} \mr{Dens}^{\frac12}(\MM)$$
This isomorphism intertwines the operators $\Delta_\mu$ and $\Delta$. I.e., for $f\in C^\infty(\MM)$ one has 
$$\Delta (\sqrt\mu\cdot f)=\sqrt\mu\cdot \Delta_\mu(f)$$

\begin{remark}
Note that, for $\mu$ an incompatible Berezinian, one can also introduce an operator $\til\Delta_\mu: f\mapsto \frac{1}{\sqrt\mu}\Delta (f\sqrt\mu)$ which will be, generally, different from Schwarz's BV Laplacian $\Delta_\mu$ as defined by (\ref{l21_Delta_mu}). More precisely, $\til\Delta_\mu=\Delta_\mu+\frac{1}{\sqrt\mu}\Delta (\sqrt\mu)\cdot$ (the last term is a multiplication operator). Operator $\til\Delta_\mu$ always squares to zero, but $\til \Delta_\mu (1)\neq 0$ for an incompatible Berezinian, whereas one always has $\Delta_\mu(1)=0$ but $\Delta_\mu^2\neq 0$ for an incompatible Berezinian. For a compatible Berezinian, we have $\til\Delta_\mu=\Delta_\mu$. 
Indeed, a Berezinian is compatible iff $\Delta\sqrt\mu=0$, cf. Remark \ref{l21_rem_compatible_Ber}.
\end{remark}

\begin{remark}[\v{S}evera, \cite{Severa}] One can also construct the canonical BV Laplacian $\Delta$ on half-densities by considering the spectral sequence calculating the cohomology of the total differential $d+\omega\wedge$ of the bi-complex $\Omega^\bt(\MM)$ with differentials $\omega\wedge$ and $d$ (de Rham operator on $\MM$). Cohomology of $\omega\wedge$ yields the space of half-densities on $\MM$ (cf. Remark \ref{l22_rem_Severa_1}). BV Laplacian arises on the third sheet $E_3$ of the spectral sequence as the induced differential $\Delta=d (\omega\wedge)^{-1}d$ on $H^\bt_{\omega\wedge}(\Omega(\MM))\cong \mr{Dens}^{\frac12}\MM$.\footnote{
In particular, consider the action of the operator $d(\omega\wedge)^{-1}d$ on the cocycle of form (\ref{l22_Severa_cocycle}): $\rho(x,\xi) dx^1\cdots dx^n\xra{d}\sum_i \frac{\dd}{\dd \xi_i}\rho(x,\xi)\; d\xi_i\, dx^1\cdots dx^n\xra{(\omega\wedge)^{-1}} (-1)^{|\rho|+1}\sum_i (-1)^{i-1}\frac{\dd}{\dd \xi_i}\rho(x,\xi)\; dx^1\cdots \widehat{dx^i}\cdots dx^n\xra{d}\left(\sum_i \frac{\dd}{\dd x^i}\frac{\dd}{\dd \xi_i}\rho(x,\xi)\right)\; dx^1\cdots dx^n$.

} (First sheet $E_1$ is $\Omega^\bt(\MM),\omega\wedge$ and second sheet $E_2$ is $H^\bt_{\omega\wedge}(\Omega(\MM))$ with zero differential.)
\end{remark}

For  $(\MM,\omega)$ and $\LL\subset \MM$ a Lagrangian submanifold, there is a well-defined restriction operation
$$\mr{Dens}^{\frac12}\MM\ra \BER(\LL)$$
cf. (\ref{l21_Ber}) and Remark \ref{l22_rem_Manin}. If a $(X^\alpha,\Xi_\alpha)$ is a Darboux chart on $\MM$ in which $\LL$ is given by $\Xi=0$, the map above sends $\rho(X,\Xi)\,\DD^{\frac12}X\,\DD^{\frac12}\Xi\mapsto \rho(X,0)\, \DD X$.

Thus, in terms of half-densities, a BV integral is an integral of form 
$$\int_{\LL\subset \MM}\alpha:=\int_{\LL\subset \MM}\alpha|_\LL$$
with $\LL$ a Lagrangian submanifold and $\alpha$ a $\Delta$-closed half-density. The BV-Stokes' theorem (Theorem \ref{l21_thm_BV_Stokes}) in this language states that:
\begin{enumerate}[(i)]
\item $\int_\LL \Delta\beta=0$, for any $\beta\in \mr{Dens}^{\frac12}(\MM)$
\item $\int_{\LL}\alpha= \int_{\LL'}\alpha$  for $\alpha\in \mr{Dens}^{\frac12}(\MM)$ satisfying $\Delta\alpha=0$ 
and $\LL\sim \LL'$ two Langrangians with homologous bodies.
\end{enumerate}

\begin{remark}[Canonical transformation as an action of a symplectic flow on half-densities]
In the setting of half-densities, a canonical transformation of solutions of quantum master equation (Definition \ref{l22_def_can_transf}) admits the following interpretation. Let $(\MM,\mu,\omega)$ be an odd-symplectic manifold with a compatible Berezinian. A canonical transformation (\ref{l22_can_transf},\ref{l22_QME_S_t}) can be viewed as a family of $\Delta$-closed half-densities on $\MM$ of form $\rho_t=\mu^{\frac12}e^{\frac{i}{\hbar}S_t}$ ($\Delta\rho_t=0$ is equivalent to the quantum master equation  (\ref{l22_QME_S_t})), such that for any $0\leq t_0< t_1\leq 1$, one has $\rho_{t_1}=(\Phi_{t_0,t_1})_* \rho_{t_0}$. Here $\Phi_{t_0,t_1}:\MM\stackrel\sim\ra \MM$ is the symplectomorphism arising as the flow, from time $t_0$ to time $t_1$, of the $t$-dependent Hamiltonian vector field $\{R_t,\bt\}\in \mathfrak{X}(\MM)_0$; $(\Phi_{t_0,t_1})_*$ stands for the pushforward of a half-density by the symplectomorphism. In this sense, the first term on the r.h.s. of (\ref{l22_can_transf}) corresponds to the transformation of the function $S_t$ by the Hamiltonian vector field $\{R_t,\bt\}$, whereas the second term compensates for the transformation of the reference half-density $\mu^{\frac12}$ under the infinitesimal flow by $\{R_t,\bt\}$.
\end{remark}

\subsubsection{Integral forms}

\begin{definition}[Manin, \cite{Manin}]
An \emph{integral form} on a supermanifold $\NN$ is a a half-density on $\Pi T^* \NN$ (with the standard symplectic structure of the cotangent bundle). We denote the space of integral forms on $\NN$ by $\mr{Int}(\NN):=\mr{Dens}^{\frac12}(\Pi T^*\NN)$. Given an integral form $\alpha$ on $\NN$, its integral over a submanifold $C\subset \NN$ is defined as 
\begin{equation} \label{l22_integration_of_int_forms}
\int_{C\subset \NN} \alpha := \int_{\Pi N^* C\subset \Pi T^* \NN} \til\alpha 
\end{equation}
-- the integral of the corresponding half-density $\til\alpha$ over the conormal Lagrangian $\LL_C=\Pi N^* C$ (Example \ref{l21_ex_conormal}) in the odd cotangent bundle $\Pi T^* \NN$.
\end{definition}
Integral forms on $\NN$ generalize the notion of Berezinians on $\NN$ (in particular, $\BER(\NN)\subset \mr{Int}(\NN)$). Whereas a Berezinian can be integrated over whole of $\NN$, an integral form can be integrated over an arbitrary submanifold $C\subset \NN$ (integrating a full Berezinian over a proper submanifold yields zero). Whereas $\BER(\NN)$ is a torsor over functions $C^\infty(\NN)$,  $\mr{Int}(\NN)$ is a torsor over polyvectors $\mc{V}^\bt(\NN)=C^\infty(\Pi T^* \NN)$. Put another way, one has
$$ \mr{Int}(\NN)=  \mc{V}^\bt(\NN) \otimes_{C^\infty(\NN)} \BER(\NN) $$

\begin{example} For $\NN=M$ an ordinary $n$-manifold, 
\begin{equation}\label{l22_Int(M)}
\mr{Int}(M)=
 \mc{V}^\bt(M)\otimes_{C^\infty(M)}\Omega^n(M)= \Omega^{n-\bt}(M)
\end{equation}
is the space of differential forms on $M$, where non-top degree forms arise as contractions of a top form with a polyvector. Integration of integral forms over submanifolds (\ref{l22_integration_of_int_forms}) over submanifolds yields in this case an integral of a differential form over a submanifold $C\subset M$. Canonical BV Laplacian $\Delta$ on integral forms (viewed as half-densities on $\Pi T^*M$) under the identification (\ref{l22_Int(M)}) with differential forms becomes the de Rham operator on $M$.
\end{example}

\begin{example}[Integral forms on the odd line] Consider integral forms on the odd line $\NN=\RR^{0|1}$. Let $\theta$ be the odd coordinate on $\RR^{0|1}$ and $Y$ the even fiber coordinate on $\Pi T^* \RR^{0|1}$. Then we the general integral forms on $\RR^{0|1}$ have the following form: 
\begin{equation}\label{l22_int_form_on_odd_line}
\mr{Int}(\RR^{0|1})\quad \ni \alpha=f(Y,\theta)\cdot \mu^{\frac12}=(f_0(Y)+f_1(Y)\; \theta)\cdot \mu^{\frac12} 
\end{equation}
with $f_0,f_1$ functions of $Y$. Here $\mu^{\frac12}=d^{\frac12}Y\,\DD^{\frac12}\theta$ is the standard coordinate half-density. By Remark \ref{l22_rem_Manin}, $\mu^{\frac12}$ is a class represented by equivalent pairs $(\RR^{0|1}\subset \Pi T^*\RR^{0|1}, \DD\theta)$ and $(\RR^1\subset \Pi T^*\RR^{0|1}, dY)$. Berezinians or $\RR^{0|1}$ correspond to the case $f_0(Y)=0$. An integral form (\ref{l22_int_form_on_odd_line}) is $\Delta$-closed iff $f_1(Y)$ is a constant function of $Y$. An integral form (\ref{l22_int_form_on_odd_line}) is $\Delta$-exact iff $f_1=0$ and $\int_{\RR} f_0(Y)\,dY=0$.
Supermanifold $\RR^{0|1}$ has two nonempty submanifolds: $\{0\}\subset \RR^{0|1}$ and $\RR^{0|1}\subset \RR^{0|1}$. Integral of an integral form $\alpha$ over these two submanifolds is, according to the definition (\ref{l22_integration_of_int_forms}), respectively,
$$\int_{\{0\}\subset \RR^{0|1}} \alpha = \int_{\RR^{1}\subset \Pi T^* \RR^{0|1}} f_0(Y)\, dY,\qquad
\int_{\RR^{0|1}} \alpha = \int_{\RR^{0|1}\subset \Pi T^* \RR^{0|1}} f_1(0) \theta \, \DD \theta = f_1(0)
$$
\end{example}

\marginpar{\LARGE{Lecture 23, 11/16/2016.}}

\subsection{Fiber BV integrals}\footnote{References: \cite{DiscrBF,CM,CMRpert}.}\label{ss: fiber BV int}

Let 
$(\FF',\omega')$, $(\FF'',\omega'')$ be two odd-symplectic manifolds and 
\begin{equation}\label{l23_F'_times_F''}
\FF=\FF'\times \FF''
\end{equation}
their direct product with the direct sum symplectic structure $\omega=\omega'\oplus \omega''$ (or, more pedantically, $\omega = \omega'\otimes 1+1\otimes \omega'' \in \Omega(\FF')\otimes \Omega(\FF'')\subset \Omega(\FF)$). Denote $P: \FF\ra \FF'$ the projection to the first factor in (\ref{l23_F'_times_F''}).

For $\LL\subset \FF''$ a Lagrangian submanifold, we denote
\begin{equation}\label{l23_P_*}
P_*^{(\LL)}=\int_{\LL\subset\FF''}\; :\quad \mr{Dens}^{\frac12}\FF\ra \mr{Dens}^{\frac12}\FF'
\end{equation}
the \emph{fiber BV integral} -- the fiber integral, parameterized by points $x'$ of $\FF'$, over a Lagrangian $\LL_{x'}\subset 
P^{-1}(x')$ -- a copy of $\LL\subset\FF''$ placed over $x'$.  
$$\vcenter{\hbox{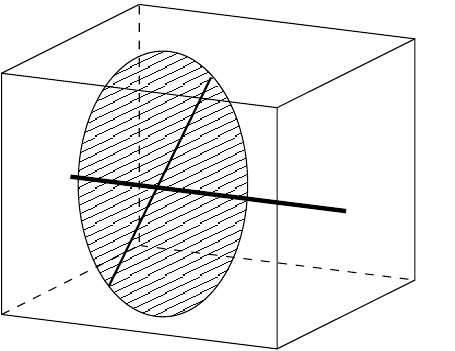}}$$
In particular, (\ref{l23_P_*}) is an $\RR$-linear map which sends $\phi=\phi'\otimes \phi''\in \HDens\FF'\otimes \HDens\FF''\subset \HDens\FF$ to $P_*^{(\LL)}\phi=\phi'\cdot\int_{\LL\subset\FF''}\phi''$. In other words, the map (\ref{l23_P_*}) is the full (ordinary) BV integral on $\FF''$ tensored with identity on $\FF'$:
$$ \PP_*^{(\LL)}:\quad \HDens\FF\cong \HDens\FF'\widehat\otimes \HDens\FF'' \xra{\mr{id}_{\FF'}\otimes\int_{\LL\subset \FF''}} \HDens\FF'  $$
We also call the map $P_*^{(\LL)}$ the \emph{BV pushforward} (of half-densities, along the odd-symplectic fibration $P:\FF\ra\FF'$).

\begin{theorem}[Stokes' theorem for fiber BV integrals
]\label{l23_thm_fiber_Stokes}

\begin{enumerate}[(i)]
\item $P_*^{(\LL)}$ is a chain map intertwining the canonical BV Laplacians $\Delta$ on $\FF$ and $\Delta'$ on $\FF'$:
\begin{equation}
\Delta' P_*^{(\LL)} = P_*^{(\LL)} \Delta
\end{equation}
\item Let $\LL\sim \til\LL$ be two homotopic Lagrangians (cf. Definition \ref{l21_def_Lagr_homotopy}) in $\FF''$, and let $\phi\in \HDens\FF$ be a half-density such that $\Delta\phi=0$. Then
\begin{equation}\label{l23_fiber_Stokes_ii}
P_*^{(\til\LL)}\phi- P_*^{(\LL)}\phi = \Delta'(\cdots)
\end{equation}
More precisely, if $\til\LL=\mr{graph}(\epsilon\cdot d\Psi
)$ is an infinitesimal Lagrangian homotopy with generator $\Psi\in C^\infty(\LL)_{-1}$ (cf. Example \ref{l21_ex_graph}), then one can write the primitive on the r.h.s. of (\ref{l23_fiber_Stokes_ii}) explicitly in terms of the generator $\Psi$: 
\begin{equation}
(\cdots)=\epsilon\cdot P_*^{(\LL)}(\Psi\cdot \phi)
\end{equation}
\end{enumerate}
\end{theorem}

Next, assume that odd-symplectic manifolds $(\FF',\omega')$, $(\FF'',\omega'')$ are equipped with compatible Berezinians $\mu',\mu''$. Then $\mu=\mu'\cdot\mu''$ is a compatible Berezinian on the direct product $\FF=\FF'\times \FF''$.

\begin{definition}
Let $S\in C^\infty(\FF)_0[[\hbar]]$ be a solution of quantum master equation on $\FF$, i.e.  $\Delta_\mu e^{\frac{i}{\hbar}S}=0 \Leftrightarrow \frac12 \{S,S\}-i\hbar \Delta_\mu S=0$. Then we call $S'\in C^\infty(\FF')_0[[\hbar]]$ the \emph{effective BV action} for $S$ induced on $\FF'$, if 
\begin{equation}\label{l23_S'_via_fiber_BV_int}
\mu'^{\frac12}\,e^{\frac{i}{\hbar}S'}=P_*^{(\LL)}\left( \mu^{\frac12}\,e^{\frac{i}{\hbar}S} \right) 
\end{equation}
By an abuse of notations, we will write $S'=P_*S$ for the effective BV action. Or, if we want to emphasize the role of the Lagrangian, $S'=P_*^{(\LL)}S$.
\end{definition}
Definition above is a realization, in the context of BV formalism, of the idea Wilson's effective action (\ref{l14_Wilson_eff_action}) of Section \ref{sss_Wilson_RG_flow}. 

The following is a corollary of Theorem \ref{l23_thm_fiber_Stokes}.
\begin{corollary}\label{l23_BV_Stokes'_cor}
\begin{enumerate}[(i)]
\item\label{l23_BV_Stokes'_cor_i} If $S$ is a solution of QME on $\FF$ then the effective action $S'$ induced on $\FF'$ via the fiber BV integral (\ref{l23_S'_via_fiber_BV_int}) satisfies QME on $\FF'$.
\item\label{l23_BV_Stokes'_cor_ii} If $S$ is a solution of QME on $\FF$ and $\LL\sim\til\LL$ are two homotopic Lagrangians in $\FF''$, the corresponding effective actions $S'=P_*^{(\LL)}S$ and $\til S'=P_*^{(\til \LL)}S$ are related by a canonical transformation, $S'\sim \til S'$.
\item\label{l23_BV_Stokes'_cor_iii} Assume that $S\sim \til{S}$ are two solutions of QME on $\FF$ are related by a canonical transformation. Then the respective effective actions (defined using the same Lagrangian $\LL\subset\FF''$) are related by a canonical transformation of $\FF'$.
\end{enumerate}
\end{corollary}

Therefore, the BV pushforward $P_*$ defines a map 
$$ \mr{SolQME}(\FF)/\sim  \qquad \xra{P_*^{[\LL]}}\qquad \mr{SolQME}(\FF')/\sim \quad $$
sending classes of solutions of QME on $\FF$ modulo canonical tranformations to classes of solutions of QME on $\FF'$ modulo canonical tranformations, and the map depends on a class $[\LL]$ of Lagrangians in $\FF''$ modulo Lagrangian homotopy.

\subsection{Batalin-Vilkovisky formalism}

\subsubsection{Classical BV formalism}\label{sec: classBV}
We call a \emph{classical BV theory} the following package of supergeometric data:
\begin{itemize}
\item A $\ZZ$-graded supermanifold $\FF$ (the space of BV fields),
\item an odd-symplectic structure $\omega\in \Omega^2(\FF)_{-1}$ (the BV 2-form),
\item a function $S\in C^\infty(\FF)_0$ (the BV action, or \emph{master action}) satisfying the classical master equation $\{S,S\}=0$.
\end{itemize}
Note that the Hamiltonian vector field on $\FF$ generated by $S$, 
$$Q:=X_S=\{S,\bt\}\in \mathfrak{X}(\FF)_1$$ (the BRST operator), squares to zero 
due to the CME. 

Also, note that $Q$ is compatible with the odd-symplectic form: 
$$L_Q\omega=0$$ 
(with $L_Q$ the Lie derivative along $Q$), which follows from $\iota_Q\omega=dS$ (the condition that $Q$ is a Hamiltonian vector field generated by $S$).

\begin{definition}\label{def_Ham_dg_mfd} A \emph{Hamiltonian dg manifold} of degree $k$ is:
\begin{itemize}
\item a dg manifold $(\MM,Q)$,
\item a symplectic form of grade (internal degree) $k$, $\omega\in \Omega^2(\MM)_k$,
\item a Hamiltonian $H\in C^\infty(\MM)_{k+1}$ satisfying $\{H,H\}_{\omega}=0$ with $\{-,-\}$ the Poisson bracket of degree $-k$ on $C^\infty(\MM)$ associated to $\omega$. 
\end{itemize} 
In particular, the Hamiltonian vector field $Q=X_H\in \mathfrak{X}(\MM)_1$ is cohomological.
\end{definition}

Case $k=-1$ of the definition above corresponds to a classical BV theory. Case $k=0$ emerges in the BFV (Batalin-Fradkin-Vilkovisky) formalism -- the Hamiltonian counterpart of the BV formalism (and also plays an important role in symplectic geometry, in the problem of describing coisotropic reductions, see \cite{Schaetz}). Case $k=n-1$ for $n\geq 0$ arises as the target structure for $n$-dimensional AKSZ sigma models \cite{AKSZ}.

\begin{example}[A BRST system in BV formalism, classically] \label{l23_ex_BRST_in_BV_cl}
Given a classical BRST package $(\FF_\BRST,Q_\BRST,S_\BRST)$, we construct the following BV package:
\begin{itemize}
\item The space of BV fields is constructed as a (shifted) cotangent bundle 
$$\FF_\BV=T^*[-1]\FF_\BRST$$ 
with $\omega_\BV$ the standard symplectic structure of the cotangent bundle.
\item The BV action is
\begin{equation}\label{l23_S_BV}
S_\BV=p^*S_\BRST+\til{Q_\BRST}
\end{equation}
Here $p: \FF_\BV\ra \FF_\BRST$ is the projection to the base of the cotangent bundle and $\til{Q_\BRST}$ is the lifting of the vector field $Q_\BRST$ on the base of the cotangent bundle to a function on the total space linear in the fibers.\footnote{Note that, generally, to $\alpha \in \mc{V}^p(M)$ a $p$-polyvector field, one can associate a function $\til\alpha\in C^\infty(T^*[-1]M)$ of degree $p$ in fiber coordinates. Here one can replace $M$ by a general $\ZZ$-graded manifold, and in particular by $\FF_\BRST$.}
\item The cohomological vector field on the total space has the form 
$$Q_\BV=X_{p^* S_\BRST}+Q_\BRST^\mr{cot.\;lift}$$
where the first term is the Hamiltonian vector field generated by the first term in (\ref{l23_S_BV}) and  $Q_\BRST^\mr{cot.\;lift}$ is the cotangent lift of a vector field $Q_\BRST$ on the base of the cotangent bundle to a vector field on the total space.
\end{itemize}
If $\Phi^\alpha$ are local coordinates on $\FF_\BRST$, then $\FF_\BV$ has corresponding Darboux coordinates $(\Phi^\alpha,\Phi^+_\alpha)$, where the fiber coordinates $\Phi^+_\alpha$ are called \emph{anti-fields} (as opposed to $\Phi^\alpha$ which are called \emph{fields}). The odd-symplectic structure is: 
$$\omega_\BV=\sum_\alpha d\Phi^\alpha\wedge d\Phi^+_\alpha$$ 
The BV action is:
\begin{equation}\label{l23_S_BV_for_BRST}
S_\BV(\Phi,\Phi^+)=S_\BRST(\Phi)+\sum_\alpha Q^\alpha_\BRST(\Phi)\cdot \Phi^+_\alpha
\end{equation}
where $Q_\BRST^\alpha=L_{Q_\BRST} \Phi^\alpha$ are the components of $Q_\BRST$ (i.e., $Q_\BRST=\sum_\alpha Q_\BRST^\alpha(\Phi)\frac{\dd}{\dd \Phi^\alpha}$). The BRST operator on the BV fields (the cohomological vector field)  is:
\begin{multline*}
Q_\BV=\sum_\alpha \left(S_\BRST(\Phi)\frac{\ola\dd }{\dd \Phi^\alpha}\right)\frac{\dd}{\dd \Phi^+_\alpha}+\\ 
+
\sum_\alpha Q^\alpha(\Phi)\frac{\dd}{\dd\Phi^\alpha}+\sum_{\alpha,\beta}\pm \Phi_\alpha^+\left(\frac{\dd }{\dd \Phi^\beta}Q^\alpha(\Phi)\right)\frac{\dd}{\dd \Phi^+_\beta}
\end{multline*}
\end{example}

\subsubsection{Quantum BV formalism}
We define a quantum (finite-dimensional) BV theory as the following package of data.
\begin{itemize}
\item A $\ZZ$-graded manifold $\FF$ of BV fields,
\item an odd-symplectic structure $\omega\in \Omega^2(\FF)_{-1}$ (the BV 2-form),
\item a Berezinian $\mu\in \BER(\FF)$ compatible with $\omega$ (the integration measure on BV fields),
\item a master action $S=S^{(0)}-i\hbar\, S^{(1)}+(-i\hbar)^2S^{(2)}+\cdots \in C^\infty(\FF)_0[[-i\hbar]]$ satisfying the quantum master equation
$$\frac12 \{S,S\}-i\hbar\Delta_\mu S=0\quad\Leftrightarrow \quad \Delta_\mu e^{\frac{i}{\hbar}S}=0$$ 
\end{itemize}

\begin{remark} 
Unlike in the classical case, the vector field $X_S$ does not automatically square to zero (since $S$ satisfies QME rather than CME). However, one can define the second order operator 
$$\delta_S=\{S,\bt\}-i\hbar\,\Delta = -i\hbar e^{-\frac{i}{\hbar}S}\Delta\left(e^{\frac{i}{\hbar}S}\cdot \bt\right)$$ 
which squares to zero due to QME (also note that the second equality above uses QME) and serves as a quantum replacement for the BRST operator in BV formalism. (We have encountered this operator before, in Remark \ref{l22_rem_can_transf}.) Note also that $\delta_S\bmod\hbar = X_{S^{(0)}}=:Q$ is the classical BRST operator associated to the classical part of the master action $S$, and it does square to zero.
\end{remark}

\textbf{Idea of gauge-fixing in BV formalism.} The partition function, as defined by a BV integral over a Lagrangian $\LL\subset \FF$
$$Z=\int_{\LL\subset \FF}\sqrt\mu\;e^{\frac{i}{\hbar}S}$$
does not change under the Lagrangian homotopy $\LL_0\sim \LL_1$ (smooth deformation staying in the class of Lagrangians, cf. Definition \ref{l21_def_Lagr_homotopy}) by Theorem \ref{l21_thm_BV_Stokes}, since the integrand is $\Delta$-closed. If it happens that $S$ has degenerate critical points on a Lagrangian $\LL_0$, we use the freedom to deform $\LL_0$ to another Lagrangian $\LL_1$ in such a way that $S$ has non-degenerate critical points on $\LL_1$ and the integral can be calculated by the stationary phase formula. Thus, the gauge-fixing in BV formalism is the choice of the Lagrangian submanifold in $\FF$.

One can also study observables in BV formalism. One says that $\OO\in C^\infty(\FF)[[\hbar]]$ is a (quantum) BV observable, if 
$\delta_S \OO=0$
is satisfied. The expectation value of an observable is the BV integral of form 
$$\lan \OO \ran = \frac{1}{Z}\int_{\LL\subset \FF}\sqrt\mu\; e^{\frac{i}{\hbar}S}\OO$$
Equation $\delta_S\OO=0$ is a way to express gauge-invariance of the observable in BV formalism, and guarantees that the integrand above is $\Delta$-closed and hence one can deform $\LL$ in the class of Lagrangians, thereby applying the gauge-fixing strategy as above and converting the integral to the form where it can be calculated by the stationary phase formula.

\begin{remark}
Note that, since $\delta_S$ is not a derivation, a product of observables in BV formalism is not necessarily an observable. (Though, one can 
correct the naive product to a $\delta_S$-cocycle using homological perturbation theory.) However, in the context of local field theory, a product of observables with disjoint support is indeed an observable (e.g. the product of Wilson loop observables in Chern-Simons theory for several non-intersecting loops is an observable).
\end{remark}

\begin{example}[A quantum BRST system in BV formalism]\label{l23_ex_BRST_in_BV_q}
Let $(\FF_\BRST,Q_\BRST,S_\BRST,\mu_\BRST)$ be a quantum BRST package (cf. Section \ref{sss: qBRST}). We define $\FF_\BV,\omega_\BV,S_\BV$ as in the Example \ref{l23_ex_BRST_in_BV_cl}. For the Berezinian on the cotangent bundle we set $\mu_\BV=\mu_\BRST^{\otimes 2}$ (using (\ref{l21_Ber_0})). Note that, since the BV action (\ref{l23_S_BV}) does not depend on $\hbar$, the quantum master equations splits into two equations: $\{S_\BV,S_\BV\}=0$ (the CME) and $\Delta_{\mu_\BV}S_\BV=0$. The CME is satisfied due to the classical BRST relations $Q_\BRST^2=0$, $Q_\BRST (S_\BRST)=0$, while equation $\Delta_{\mu_\BV}S_\BV=0$ follows from the relation $\mr{div}\;Q_\BRST=0$ for the quantum BRST package.

Consider the gauge-fixing, within BV framework, for such a system coming from a BRST package. Denote $\LL_0$ the zero-section of $\FF_\BV=T^*[-1]\FF_\BRST$ and let $\LL_\Psi=\mr{graph}(d\Psi)\subset T^*[-1]\FF_\BRST$ be the graph Lagrangian, for $\Psi=\Psi(\Phi)\in C^\infty(\FF_\BRST)_{-1}$. 
$$\vcenter{\hbox{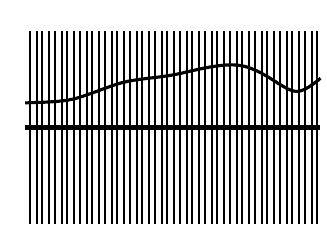}}$$
We use $\Phi^\alpha$ for local coordinates on $\FF_\BRST$ (and we assume for simplicity that $\mu_\BRST=\DD\Phi$ locally) and $\Phi^+_\alpha$ for the corresponding fiber coordinates on $T^*[-1]\FF_\BRST$. Then gauge-fixing consists in the replacement
\begin{equation}\label{l23_BRST_BV_gf}
\int_{\LL_0\subset T^*[-1]\FF_\BRST} \sqrt{\mu_\BV}\;e^{\frac{i}{\hbar}S_\BV(\Phi,\Phi^+)}  \qquad \mapsto \qquad \int_{\LL_\Psi\subset T^*[-1]\FF_\BRST} \sqrt{\mu_\BV}\;e^{\frac{i}{\hbar}S_\BV(\Phi,\Phi^+)} 
\end{equation}
Since $S_\BV$ on the zero-section reduces to $S_\BRST$, the l.h.s. of (\ref{l23_BRST_BV_gf}) reduces to $\int_{\FF_\BRST}\DD\Phi\;e^{\frac{i}{\hbar}S_\BRST(\Phi)}$. On the other hand, to evaluate the r.h.s. of (\ref{l23_BRST_BV_gf}), we note that $S_\BV$ restricted to the Lagrangian $\LL_\Psi$ is $S_\BV(\Phi^\alpha,\Phi^+_\alpha=\frac{\dd}{\dd \Phi^\alpha}\Psi)=S_\BRST+Q_\BRST(\Psi)$. Therefore, r.h.s. of (\ref{l23_BRST_BV_gf}) reads $\int_{\FF_\BRST}\DD\Phi\; e^{\frac{i}{\hbar}(S_\BRST(\Phi)+Q_\BRST(\Psi))}$.
Thus, BV gauge-fixing, performing the Lagrangian homotopy $\LL_0\mapsto \LL_\Psi$ precisely corresponds to the gauge-fixing procedure of BRST formalism (\ref{l20_shift_by_Q(Psi)}), shifting the BRST action by a $Q_\BRST$-coboundary.
\end{example}

\marginpar{\LARGE{Lecture 24, 11/21/2016.}}

\subsubsection{Faddeev-Popov via BV}\label{sss: FP via BV}
Starting with Faddeev-Popov data -- an $n$-manifold $X$ acted on by a compact group $G$, an invariant function\footnote{We put a subscript $cl$ now to distinguish the classical action from its BV counterpart.} 
$S_{cl}\in C^\infty(X)^G$ and the invariant integration measure $\mu_X\in \Omega^n(X)^G$ -- we construct the associated BRST package as in Section \ref{sss: FP via BRST} and then construct the BV package as the $-1$-shifted cotangent bundle (Examples \ref{l23_ex_BRST_in_BV_cl}, \ref{l23_ex_BRST_in_BV_q}).  The minimal BRST package (cf. Section \ref{sss: FP via BRST}) associated to the Faddeed-Popov data yields, by applying the construction of Example \ref{l23_ex_BRST_in_BV_cl}, the ``minimal'' space of BV fields 
\begin{equation*}
\FF_\mr{min}=T^*[-1](X\times \g[1])=T^*[-1]X\times \g[1]\times \g^*[-2]
\end{equation*}
If $x^i$ are the local coordinates on $X$ and assuming we chose a basis $\{T_a\}$ in $\g$, we have the following coordinates on $\FF_\mr{min}$ (we also indicate the ghost number defining the $\ZZ$-grading on $\FF_\mr{min}$):\\ 
\begin{tabular}{c|c|c}
coordinates & name & degree (ghost number) \\ \hline
$x^i$ on $X$ & classical fields & 0 \\
$c^a$ on $\g[1]$ & ghosts &  1 \\
$x^+_i$ on fibers of $T^*[-1]X$ & anti-fields (for classical fields) & -1 \\
$c^+_a$ on $\g^*[-2]$ & anti-ghosts & -2 \\ \hline
\end{tabular}

Applying the construction of Example \ref{l23_ex_BRST_in_BV_cl}, we get:
\begin{itemize}
\item The BV 2-form -- the canonical odd-symplectic form of the shifted cotangent bundle:
$$\omega_\mr{min}=\sum_i dx^i\wedge dx^+_i+\sum_a dc^a\wedge dc^+_a$$
\item The BV action (\ref{l23_S_BV}):
\begin{equation}\label{l24_S_FP_BV_min}
S_\mr{min}=  S_{cl}(x)+\sum_{i,a} c^a v^i_a(x)x^+_i + \frac12 \sum_{a,b,c} f^c_{ab}\; c^b c^c c^+_a
\end{equation}
\item Assuming that the volume form on $\mu_X$ locally has the form $\mu_X=\rho(x)\,d^nx$ with $\rho$ a local density function, the Berezinian on $\FF_\mr{min}$ obtained by the construction of Example \ref{l23_ex_BRST_in_BV_q} takes the local form
$$\mu_\mr{min}=\rho(x)^2\cdot d^nx\cdot\DD^n x^+\cdot\DD^m c\cdot d^m c^+ $$
\end{itemize}

Passing to the non-minimal BRST model for the Faddeev-Popov data (cf. Section \ref{sss: FP via BRST}), on the level of the BV package corresponds to extending the space of minimal BV fields by auxiliary fields:
\begin{equation}\label{l24_min_to_full_BV_for_FP}
\FF_\mr{min}\quad\mapsto\quad \FF=\FF_\mr{min}\times T^*[-1](\g^*[-1]\oplus \g^*)
\end{equation}
The auxiliary fields (coordinates on the second term in the r.h.s. above) are:
\begin{tabular}{c|c}
coordinates & degree (ghost number) \\ \hline
$\lambda_a$ on $\g^*$ & $0$ \\
$\bar{c}_a$ on $\g^*[-1]$ & $-1$ \\
$\lambda^{+a}$ on the fiber of $T^*[-1]\g^*$ & $-1$ \\
$\bar{c}^{+a}$ on the fiber of $T^*[-1]\g^*[-1]$ & $0$\\ \hline
\end{tabular}

In the non-minimal version, the minimal BV package gets extended as follows:
\begin{itemize}
\item The BV 2-form:
$$\omega=\omega_\mr{min}+\sum_a d\lambda_a\wedge d\lambda^{+a}+\sum_a d\bar{c}_a\wedge d\bar{c}^{+a}$$
\item 
The BV action: 
\begin{equation}\label{l24_S_FP_BV_full}
S=S_\mr{min}+\sum_a \lambda_a \bar{c}^{+a} 
\end{equation}
\item The Berezinian:
$$\mu= \mu_\mr{min}\cdot d^m\lambda\cdot \DD^m\lambda^+\cdot \DD^m \bar{c}\cdot d^m \bar{c}^+$$
\end{itemize}

\textbf{Faddeev-Popov gauge-fixing}, corresponding to the zero-section of a map $\phi:X\ra\g$, can be described within BV formalism in the following two equivalent ways:
\begin{enumerate}[I.]
\item In non-minimal BV package, we can choose the gauge-fixing Lagrangian of graph type
$$ \LL=\mr{graph}(d\Psi) \;\;\subset\;\; \FF $$
with $$\Psi=\lan \bar{c},\phi(x) \ran$$
Here $\LL$ is locally given by
$$\LL:\qquad \left\{\begin{array}{l}
x,c,\lambda,\bar{c} \;\;\mbox{are free}\\
x^+=(d_x\phi)^T(\bar{c}) \\
c^+=0 \\
\lambda^+=0\\
\bar{c}^+=\phi(x)
\end{array}\right.
$$
Note that the restriction of the full BV action (\ref{l24_S_FP_BV_full}) to $\LL$ yields precisely the Faddeev-Popov action (\ref{l15_S_FP}), and thus the BV integral 
\begin{equation}\label{l24_FP_int_via_BV_non-min}
\int_{\LL\subset\FF} \sqrt\mu \;e^{\frac{i}{\hbar}S}
\end{equation} 
yields the r.h.s. of (\ref{l15_FP2}).
\item In the minimal BV package, we can take the gauge-fixing Lagrangian of conormal type $$\LL_\mr{min}=N^*[-1]C\;\;\subset\FF_\mr{min}$$ 
with $$C=\phi^{-1}(0)\times \g[1]\;\;\subset X\times \g[1] $$
Locally, this Lagrangian 
is described as follows:
$$\LL_\mr{min}:\qquad \left\{\begin{array}{l}  
c \;\; \mbox{free} \\
x\;\;\mbox{such that}\;\phi(x)=0 \\
x^+=(d_x\phi)^T(\bar{c})\;\;\mbox{for some}\; \bar{c} \\
c^+=0
\end{array}\right.$$
Note that here the second Faddeev-Popov ghost $\bar{c}$ emerges as a coordinate on the Lagrangian $\LL_\mr{min}$. Classical field $x$ on $\LL_\mr{min}$ is subject to the gauge-fixing constraint $\phi(x)=0$; one way to impose it for describing the integral over $\LL_\mr{min}$ is by introducing the Lagrange multiplier field $\lambda$, imposing the constraint. In this way, one can see that the BV integral 
\begin{equation}\label{l24_FP_int_via_BV_min}
\int_{\LL_\mr{min}\subset \FF_\mr{min}}\sqrt{\mu_\mr{min}}\; e^{\frac{i}{\hbar}S_\mr{min}}
\end{equation}
does indeed yield the r.h.s. of (\ref{l15_FP2}). Put another way, the connection between non-minimal and minimal realization here is that, integrating the Lagrange multiplier field $\lambda$ in (\ref{l24_FP_int_via_BV_non-min}) yields (\ref{l24_FP_int_via_BV_min}).
\end{enumerate}


\subsubsection{BV for gauge symmetry given by a non-integrable distribution}\label{sss: BV for non-integrable distribution}
The power of BV formalism, setting it aside from Faddeev-Popov construction and BRST formalism, is that it can be used to treat gauge symmetry not given by a group action on the space of classical fields (and, more generally, not given by a Lie algebroid -- this case can be treated by BRST, cf. Remark \ref{l20_rem_Lie_algbd_symmetry_via_BRST}), but rather given by a possibly non-integrable distribution on the space of classical fields. 

In the most general setting, the data of classical gauge system is: 
\begin{itemize}
\item A manifold $X$ (the space of classical fields).
\item A (generally, non-integrable) tangential distribution $E$ on $X$, defining the infinitesimal gauge transformations of classical fields. I.e., we have a subbundle of the tangent bundle, $E\subset TX$, which we can think of as a standalone bundle $E\ra X$ together with an injective bundle map $\rho: E\hra TX$. 
\item An $E$-invariant function $S_{cl}$ on $X$ -- the classical action. An additional assumption is that $E$ is integrable on the critical locus of $S_{cl}$, i.e., on zero locus $X_\mr{crit}\subset X$ of $dS_{cl}$. Note that $E$ is automatically tangent to $X_\mr{crit}$, due to $E$-invariance of $S_{cl}$. (And $E$ is allowed to be non-integrable away from $X_\mr{crit}$.)
\end{itemize}

Then one tries to construct the associated BV package (here we only discuss the classical part of the BV package), with the space of BV fields given as 
\begin{equation}\label{l24_F=T^*[-1]E[1]}
\FF=T^*[-1] E[1]
\end{equation}
-- the shifted cotangent bundle of $E[1]$ (the graded manifold obtained by degree-shifting the fibers of the vector bundle $E\ra X$).
Then one tries to construct a function $S\in C^\infty(\FF)_0$ -- the BV action, satisfying the classical master equation
\begin{equation}\label{l24_CME}
\{S,S\}=0
\end{equation}
and two compatibility properties with the classical gauge system we started with:
\begin{enumerate}
\item The restriction of $S$ to $E[1]$ (the zero-section of the shifted cotangent bundle \ref{l24_F=T^*[-1]E[1]}) is $S_{cl}$.\footnote{More pedantically, we should say that the restriction of $S$ to the zero-section is $\pi^* S_{cl}$ where $\pi: E[1]\ra X$ is the mapping of graded manifolds induced from the bundle map $\pi:E\ra X$.}
\item $S$ is compatible with the inclusion $\rho: E\hra TX$ in the following sense. If $\{e_a\}$ is a basis of local sections of $E\ra X$ and $v_a=\rho(e_a)=\sum_i v_a^i(x)\frac{\dd}{\dd x^i}$ are the corresponding local vector fields on $X$, then the 
term linear in ghosts $c^a$ (fiber coordinates of $E[1]$) in $S$ is $\sum_{a,i} c^a v^i_a(x)x^+_i=\lan \rho(c),x^+ \ran$. Here $x^+_i$ are the fiber coordinates on $T^*[-1]X$ corresponding to local coordinates $x^i$ on the base. This compatibility condition can be spelled without resorting to local coordinates: $S$ must have the form $S=S_{cl}+Y$ for some function $Y$ on $\FF$, with the Hamiltonian vector field $\{Y,\bt\}$ tangential to the zero-section $E[1]$ of (\ref{l24_F=T^*[-1]E[1]}). Denoting by $q\in \mathfrak{X}(E[1])_1$ the induced vector field on the zero-section, the requirement is that the Lie brackets $[q,u]$ of $q$ with vertical (w.r.t. the bundle projection $\pi:E[1]\ra X$) vector fields $u\in \mathfrak{X}(E[1])_{-1}$ are projectable to $X$ and project precisely to the sections of the distribution $E$ (i.e. to the image of $\rho_*:\Gamma(E)\ra \mathfrak{X}(X)$).
\end{enumerate}

Thus, $S$ necessarily has the form 
\begin{equation}\label{l24_S=S_cl+v}
S=S_{cl}(x)+\sum_{a,i}c^a v^i_a(x) x^+_i+\cdots
\end{equation}
where $\cdots$ are terms of degree at least $2$ in the ghosts $c^a$ and also depending on anti-fields $x^+_i, c^+_a$. These terms are to be added in such a way that $S$ satisfies the classical master equation (\ref{l24_CME}).

\begin{remark}\label{l24_rem_crit}
Note that the zero locus of the cohomological vector field $Q=\{S,\bt\}$ restricted to $X$ (viewed as the degree zero part of $\FF$) is precisely the critical locus $X_\mr{crit}$ of $S_{cl}$, which is in turn the same as the space of solutions of Euler-Lagrange equations associated with the classical action $S_{cl}$.
\end{remark}

\begin{remark}
Restricting to a tubular neighborhood $U$ of the critical locus $X_\mr{crit}$ in $X$, one can be more specific about the ansatz (\ref{l24_S=S_cl+v}) and write
\begin{equation}\label{l24_S_tub_nbhd}
S=S_{cl}(x)+\sum_{a,i}c^a v^i_a(x) x^+_i+\sum_{a,b,c}\frac12 f^a_{bc}(\s{p}(x))\;c^b c^c c^+_a  +\cdots
\end{equation}
Here $\s{p}:U\proj X_\mr{crit}$ is the tubular neighborhood projection, $f^a_{bc}$ are the structure constants of the Lie algebroid given by the restriction of $E$ to $X_\mr{crit}$. The correction terms $\cdots$ in (\ref{l24_S_tub_nbhd}) are terms of degree $\geq 2$ in anti-fields $x^+_i,c^+_a$. Deforming the projection $\s{p}:U\proj X_\mr{crit}$ is equivalent to a canonical BV transformation\footnote{Here we mean the \textit{classical}, i.e. $\bmod\; \hbar$, part of the full (quantum) canonical BV transformation (\ref{l22_can_transf}).} of $S$ (the pullback of $S$ by a specific symplectomorphism of $\FF$).
\end{remark}

\begin{remark}
Throughout this subsection, we are using a simplifying assumption that the gauge symmetry is given by an injective map of vector bundles $\rho:E\hra TX$. There are interesting cases when it is natural to parameterize the gauge transformations by a vector bundle $E$ mapping into $TX$ non-injectively. In this case, the space of BV fields will incorporate higher ghosts corresponding to the kernel of $E\ra TX$ (or even further terms of a resolution of $E\ra X$ by a complex of vector bundles $\cdots \ra E\ra TX$), cf. Section \ref{sss: higher ghosts}, and the anti-fields for the higher ghosts.
\end{remark}

\begin{example}\label{l24_ex_action_up_to_homotopy} Let $X$ be a manifold with a function $S_{cl}\in C^\infty(X)$ and let  
\begin{equation}\label{l24_g_to_Vect(X)}
v:\g\ra \mathfrak{X}(X)
\end{equation}
be a linear map from a Lie algebra $\g$ to the vector fields on $X$, mapping the generators $T_a$ of $\g$ to vector fields $v_a$ on $X$. Assume that $S_{cl}$ satisfies 
\begin{equation}\label{l24_v(S)=0}
v_a(S_{cl})=0
\end{equation} 
for all $a$ (in particular, this implies that $v_a$ are tangent to the critical locus $X_\mr{crit}$ of $S_{cl}$). Also, assume that (\ref{l24_g_to_Vect(X)}) gives a strict action of $\g$ on $X_\mr{crit}$ but only an action-up-to-homotopy away from $X_\mr{crit}$. Explicitly, we request that
\begin{equation}\label{l24_Lie_hom_up_to_homotopy}
[v(\alpha),v(\beta)]=v([\alpha,\beta]_\g)+[A(\alpha,\beta),S_{cl}]_{NS}
\end{equation}
for any $\alpha,\beta\in\g$; $[,]_\g$ is the Lie bracket in $\g$ and $[,]_{NS}$ is the Nijenhuis-Schouten bracket of polyvectors on $X$. In (\ref{l24_Lie_hom_up_to_homotopy}), the rightmost term is the ``defect'' of the Lie algebra homomorphism property of the map (\ref{l24_g_to_Vect(X)}), with 
\begin{equation}\label{l24_A}
A: \wedge^2\g\ra \mc{V}^2(X)
\end{equation} 
the homotopy (given by some bivector on $X$, depending skew-symmetrically on a pair of Lie algebra elements).
In this case, the BV package is given by $\FF=T^*[-1](X\times \g[1])$ with the BV action
\begin{equation}\label{l24_S_Lie_up_to_hom}
S=S_{cl}(x)+\sum_{a,i}c^a v_a^i(x) x^+_i+\sum_{a,b,c}\frac12 f^c_{ab}c^a c^b c^+_c + \sum_{a,b,i,j}\frac14 A^{ij}_{ab}(x)c^a c^b x^+_i x^+_j
\end{equation}
Here $f^c_{ab}$ are the structure constants of $\g$ in the basis $\{T_a\}$ and $A^{ij}_{ab}(x)$ are the components of the homotopy (\ref{l24_A}) in the local coordinates $x^i$ on $X$, i.e. $A(T_a,T_b)=\sum_{i,j}A^{ij}_{ab}\frac{\dd}{\dd x^i}\wedge \frac{\dd}{\dd x^j}$. 
Note that, due to (\ref{l24_v(S)=0},\ref{l24_Lie_hom_up_to_homotopy}), in the expression $\{S,S\}$, terms proportional to $c$ and $ccx^+$ vanish. If certain higher coherence relations are fulfilled for the data $(S,v,A)$, then we have the classical master equation $\{S,S\}=0$. Otherwise, we may get terms proportional to $cccc^+,cccx^+x^+,ccccx^+x^+x^+$ in $\{S,S\}$ and then we would need to add higher degree terms in ghosts to (\ref{l24_S_Lie_up_to_hom}) in order to correct for that error and have the classical master equation.
\end{example}

\begin{remark}\label{l24_rem_antifield_grading}
Generally, terms in the BV action $S$ have a natural grading by polynomial degree in anti-fields:
\begin{itemize}
\item The term independent of anti-fields is the classical action.
\item The terms linear in anti-fields are the data of gauge symmetry (the action of gauge transformations on classical fields and the algebra of gauge transformations, which is closed under commutators on $X_\mr{crit}$).
\item The terms of degree $\geq 2$ in anti-fields correspond to non-integrability of gauge symmetry (and are the homotopies correcting for this non-integrability).
\end{itemize}
\end{remark}


\begin{example}
Assume that we have a graded manifold $f$ with a function $s\in C^\infty(f)_0$ and a vector field $q\in\mathfrak{X}(f)_1$, such that $q(s)=0$. Instead of asking that $q^2=0$ (then the triple $\FF_\BRST=f,S_\BRST=s,Q_\BRST=q$ would have been a classical BRST package), we ask that $q^2=-[s,a]$ for some bivector field $a\in \mc{V}^2(f)_2$. (The bracket $[,]$ is the Nijenhuis-Schouten bracket of polyvectors on $f$.) Then we construct the BV package by deforming the construction of Example \ref{l23_ex_BRST_in_BV_cl} as follows:
\begin{eqnarray}
\FF &=& T^*[-1] f \nonumber \\
S(\Phi,\Phi^+) &=& s(\Phi)+\widetilde{q}+\til{a} \label{l24_s+q+a}
\end{eqnarray}
Where $\Phi$ are the coordinates on $f$ (fields) and $\Phi^+$ are the fiber coordinates on $T^*[-1]f$ (anti-fields); $\til{q}$ is the lifting of $q$ to a function on $T^*[-1]f$, linear in $\Phi^+$; $\til{a}$ is the lifting of $a$ to a function on $T^*[-1]f$ quadratic in $\Phi^+$. If the ``higher coherencies'' $[q,a]=0, [a,a]=0$ are observed, then (\ref{l24_s+q+a}) satisfies the classical master equation $\{S,S\}=0$. Otherwise, one should add corrections to (\ref{l24_s+q+a}) corresponding to polyvectors of degree $\geq 3$ on $f$, which correct for the error in the classical master equation. E.g., if $[q,a]=-[s,b]$ for $b\in\mc{V}^3(f)_3$, one should add $\til{b}$ (a term cubic in anti-fields) to the r.h.s. of (\ref{l24_g_to_Vect(X)}), and so on.
\end{example}

\begin{example}[System with no gauge symmetry cast in BV and the Koszul complex]
Consider the case of a classical system defined by an $n$-manifold $X$ endowed with a function $S_{cl}\in C^\infty(X)$ with no gauge symmetry. In this case, the space of BV fields (\ref{l24_F=T^*[-1]E[1]}) is simply $\FF=T^*[-1]X$, the BV action is $S=S_{cl}$ (the classical action pulled back from $X$ to $T^*[-1]X$), classical master equation $\{S,S\}=0$ holds trivially (since $S$ is constant in the fiber direction of $T^*[-1]X$). If $x^i$ are the local coordinates on $X$ and $x^+_i$ the corresponding fiber coordinates on $T^*[-1]X$, the cohomological vector field $Q=\{S,\bt\}$ on $\FF$ locally has the form 
\begin{equation}\label{l24_Q_Koszul}
Q=\sum_i \frac{\dd S_{cl}}{\dd x^i}\;\frac{\dd}{\dd x^+_i}
\end{equation}
The zero locus of $Q$ is the space of solutions of the Euler-Lagrange equation $dS_{cl}=0$ which is the same as the critical locus $X_\mr{crit}\subset X$ of $S_{cl}$ (this is a trivial 
case of the Remark \ref{l24_rem_crit}). The complex $C^\infty(\FF)$ with differential $Q$ is the Koszul complex\footnote{
Recall that, generally, given a ring $R$, a rank $n$ free module $A$ and an $R$-module morphism $e:A\ra R$, one can construct the Koszul chain complex $K_n\ra\cdots \ra K_2\ra K_1=A\xra{e} R$ with $K_i=\wedge^i A$ and with the boundary operator $\delta_\mr{Koszul}:K_i\ra K_{i-1}$ mapping $a_1\wedge\cdots\wedge a_i\mapsto \sum_{j=1}^i (-1)^{j-1} e(a_j)\cdot a_1\wedge\cdots \widehat{a_j}\cdots\wedge a_i$. Here we prefer to use tho cohomological grading and thus denote $K_i$ as $K^{-i}$.
}
\begin{equation}\label{l24_Koszul_complex}
\underbrace{\Gamma(\wedge^n TX)}_{K^{-n}}\ra  \cdots \ra \underbrace{\Gamma(\wedge^2 TX)}_{K^{-2}} \ra \underbrace{\Gamma(TX)}_{K^{-1}} \xra{\lan dS_{cl},\bt \ran} \underbrace{C^\infty(X)}_{K^0} 
\end{equation}
Locally, in a coordinate neighborhood $U\subset X$, it has the form $K^{-\bt}_U=C^\infty(U)\otimes \wedge^\bt \RR^n=C^\infty(U)[x^+_1,\ldots,x^+_n]$ with differential (\ref{l24_Q_Koszul}). In particular, $K^{-\bt}_U$ is a free supercommutative algebra, obtained from $C^\infty(U)$ by adjoining the free anti-commuting variables $x^+_1,\ldots,x^+_n$ (the anti-fields or the Koszul generators, cf. \cite{Stasheff}), placed in degree $-1$. The value of the Koszul differential on the generator $x^+_i$ is the respective partial derivative of the action, $x^+_i\mapsto \frac{\dd S_{cl}}{\dd x^i}$ (and then the differential is extended to the entire Koszul complex as a derivation, by Leibniz identity). The cohomology of the complex (\ref{l24_Koszul_complex}) has the form
\begin{equation}
H^{-k}_Q=\left\{\begin{array}{cl}
C^\infty(X_\mr{crit}) & \mr{for}\; k=0, \\
0 & \mr{for}\; k<0
\end{array}
\right.
\end{equation}
under the assumption that $dS_{cl}$ has maximal rank everywhere on $X_\mr{crit}$, i.e., if the intersection of $\mr{graph}(dS_{cl})\subset T^*X$ with the zero-section of $T^*X$ is transversal (which is related to the assumption that $S_{cl}$ does not have gauge symmetry).
\end{example}

\begin{example}[P. M., \cite{M13}]\label{l24_exa_contact} One idea of constructing a small-dimensional example of gauge symmetry given by a non-integrable distribution $E$ on $X$ with a (non-constant) invariant function $S_{cl}$ is as follows. Take $X$ to be a bundle over the base $\RR$ (parameterized by a coordinate $t$), with fiber $M$ and consider a family of contact structures $\alpha_t$ on $M$ depending on $t$, such that for $t=0$ the maximally non-integrable distribution $E_t$ on $M$ corresponding to the contact structure $\alpha_t$ degenerates into  an integrable distribution $E_0$. Then the family $\{E_t\}$ yields a distribution on $X$ (a subbundle of the vertical tangent bundle on $X\ra \RR$), and for $S_{cl}$ one can take any function of $t$ whose sole extremum is at $t=0$. 

A simple exmplicit example of this situation is the following: $M=\RR^3$ with coordinates $x,y,z$ with a $t$-dependent contact 1-form $\alpha_t=dz-ty\,dx$ which fails the contact property $\alpha\wedge d\alpha\neq 0$ if and only if $t=0$. The corresponding rank $2$ distribution $E_t=\ker\alpha_t$ on $\RR^3$ is spanned by vector fields $\dd_y,\dd_x+ty\, \dd_z$. Thus, we obtain a gauge system with $X=\RR^3\times \RR$ (with coordinates $x,y,z,t$), with distribution 
$$E=\mr{Span}(\dd_y,\dd_x+ty\, \dd_z)\;\;\subset TX$$
We can take $S_{cl}=\frac{t^2}{2}$ as a simplest choice. The critical locus is $X_\mr{crit}=
\RR^3\times \{0\}\subset X$ -- the fiber of $X$ over $t=0$. In the BV formulation, we have $\FF=T^*[-1]E[1]$ with the following coordinates.\\
\begin{tabular}{c|c}
coordinates & degree (ghost number) \\ \hline
$x,y,z,t$ on $X$ & $0$ \\
$c^1,c^2$ on the fiber of $E[1]$ & $1$ \\
$x^+,y^+,z^+,t^+$ on the fiber of $T^*[-1]X$ & $-1$ \\
$c^{+}_1,c^+_2$ on the cotangent fiber of the fiber of $E[1]$ & $-2$\\ \hline
\end{tabular}

The corresponding BV action is:
\begin{equation}\label{l24_S_contact}
S=\frac{t^2}{2}+c^1 \, y^++c^2 (x^++ty\, z^+)+ c^1 c^2\, t^+ z^+
\end{equation}
Note the last term here, quadratic in anti-fields, which corresponds to non-integrability of the distribution $E$, as per Remark \ref{l24_rem_antifield_grading}.  If we choose $S_{cl}$ to be given by some other function $f(t)$ whose sole extremum is at $t=0$, the last term in (\ref{l24_S_contact}) will get rescaled with the factor $\frac{1}{f''(0)}$.
\end{example}

\begin{remark} The Example \ref{l24_exa_contact} can be pushed further, to the quantum BV setting. Then in (\ref{l24_S_contact}) it is convenient to make the coordinate $z$ periodic, i.e., $M=\RR^2\times S^1$ instead of $M=\RR^3$. Then, choosing the gauge-fixing Lagrangian $\LL\subset \FF$ given by $t^+=z^+=c_1^+=c_2^+=0,x=x_0,y=y_0$ (for some fixed $x_0,y_0$), the BV integral $\int_{\LL\subset\FF} e^{\frac{i}{\hbar}S}$ is convergent and can be regarded, morally, as a way to make sense of the volume of the non-Hausdorff quotient $X/E$ (cf. Remark \ref{l20_rem_Lie_algbd_symmetry_via_BRST}).\footnote{Note that $E$ integrates to a rank $2$ foliation on $X_\mr{crit}=M\times\{0\}$, with leaf space $S^1$. On the other hand, $E$ is a maximally non-integrable distibution on each $M\times \{t\}$ for each $t\neq 0$, i.e. one can connect any two points on $M\times\{t\}$ by a path tangential to $E$. Thus, each $M\times\{t\}$ for each $t\neq 0$ constitutes a rank $3$ ``leaf'' of $E$. In this sense, $X/E$, viewed as points of $X$ modulo the equivalence relation given by the possibility to connect two points by a path tangent to $E$, is $\RR$ (parameterized by $t$), with the point $t=0$ thickened to $X_\mr{crit}/E\simeq S^1$. It is a non-Hausdorff space.}
\end{remark}

\begin{example}[Felder-Kazhdan, Example 6.7 in \cite{FelderKazhdan}]
Let $\pi: X\ra Y$ be a vector bundle with base $Y$ endowed with inner product $g$ on the fibers and with a (possibly non-flat) connection $\nabla$ preserving the inner product (i.e. $d\, g(u,v)=g(\nabla u,v)+g(u,\nabla v)$ for any sections $u,v\in\Gamma(Y,X)$). Then we define a classical gauge system on $X$ by setting $S_{cl}(x)=\frac12 g(x,x)$ (the quadratic form on fibers determined by $g$), and setting $E$ to be the horizontal distribution on the total space of the bundle $X\ra Y$ determined by $\nabla$ (viewed as an Ehresmann connection). Note that $S_{cl}$ is $E$-invariant, since $\nabla$ is compatible with the inner product on fibers, and 
$E$ is integrable if and only if $\nabla$ is flat.

The space of BV fields here is 
\begin{equation}\label{l24_FK_example_F}
\FF=T^*[-1](\pi^*T[1]Y)
\end{equation} 
To write down the BV action, let us introduce the local coordinates $y^\mu$ on $Y$ and $v^i$ in the fiber of $\pi:X\ra Y$, corresponding to a basis of sections $\{e_i\}$. Then $g$ has the local components $g_{ij}(y)=g(e_i,e_j)$ and we have $S_{cl}=\sum_{i,j}\frac12 g_{ij}(y)v^i v^j$. Let 
$$A=\sum_\mu A_{\mu\;j}^{\;\,i}(y)\,e_i\otimes e^j\cdot dy^\mu\quad \in \Omega^1(Y,\mr{End}(X))$$ 
be the local connection 1-form of $\nabla$ (i.e., $\nabla$ acts on sections by $\sum_i v^i e_i\mapsto \sum_{i,\mu}(\frac{\dd}{\dd y^\mu} v^i+\sum_j A_{\mu\; j}^{\;\,i}(y)\, v^j)\,e_i\, dy^\mu$
and the respective  horizontal distribution is $E=\mr{Span}\left\{\frac{\dd}{\dd y^\mu}+\sum_{i,j} A_{\mu\;j}^{\;\, i}(y)\, v^j\frac{\dd}{\dd v^i}\right\}$), with curvature $2$-form 
$$
F_\nabla=dA+\frac12[A,A]=\sum_{i,j,\mu,\nu} F_{\mu\nu\;j}^{\;\;\; i}(y)\,e_i\otimes e^j\cdot  dy^\mu\wedge dy^\nu\quad \in \Omega^2(Y,\mr{End}(X))
$$
Then the BV action on $\FF$ takes the form $S=S_{cl}+S_\nabla+S_F$ -- the classical action plus the term associated to the horizontal distribution corresponding to $\nabla$ plus the term associated to the curvature. More explicitly:
\begin{equation}\label{l24_FK_example_SS}
S=\sum_{i,j}\frac12 g_{ij}(y)v^i v^j+\sum_\mu c^\mu \left(y^+_\mu +\sum_{i,j }A_{\mu\;j}^{\;\, i}(y)\, v^j v^+_i\right)+\frac14 \sum_{i,j,k,\mu,\nu} (g^{-1}(y))^{jk}F_{\mu\nu\; k}^{\;\;\; i}(y)\, c^\mu c^\nu v^+_i v^+_j 
\end{equation}
Here the ghosts $c^\mu$ are the fiber coordinates on $T[1]Y$; $y^+_\mu,v^+_i,c^+_\mu$ are the fiber coordinates on the shifted cotangent bundle (\ref{l24_FK_example_F}) corresponding to the coordinates $y^\mu,v^i,c^\mu$ on $\pi^*T[1]Y$. First two terms in the BV action (\ref{l24_FK_example_SS}) correspond to the first two terms in (\ref{l24_S=S_cl+v}) and the third term in (\ref{l24_FK_example_SS}) is quadratic in the anti-fields and is the correction necessary for $S$ to satisfy $\{S,S\}=0$. Note that in the case when $\nabla$ is flat, the last term in (\ref{l24_FK_example_SS}) vanishes and then $S$ satisfies the ansatz (\ref{l23_S_BV_for_BRST}) and thus corresponds to a BRST theory via the construction of Example \ref{l23_ex_BRST_in_BV_cl}. On the other hand, the case of $\nabla$ non-flat cannot be treated without using the BV formalism. Finally, note that in this example, the critical locus $X_\mr{crit}=Y$ -- the zero-section of $\pi:X\ra Y$ and the ``Euler-Lagrange moduli space'' -- the quotient of $X_\mr{crit}$ by the gauge transformations $E$ -- is the discrete set $\pi_0(Y)$.
\end{example}

\subsubsection{Felder-Kazhdan existence-uniqueness result for solutions of the classical master equation}
In \cite{FelderKazhdan}, Felder and Kazhdan considered the following setup for the BV formalism, within the context of algebraic (rather than differential) geometry. For $X$ a non-singular affine variety over a field $\s{k}$ of characteristic zero, endowed with a regular function $S_{cl}$ on $X$, Felder and Kazhdan consider the following version of the ``BV problem''. One wants to construct, starting from the data $(X,S_{cl})$, a $-1$-symplectic $\ZZ$-graded variety $\FF$ of form $T^*[-1]\s{f}$ for some non-negatively graded variety $\s{f}=\FF_\BRST$, together with a regular function $S$ on $\FF$, such that:
\begin{enumerate}[(a)]
\item The support of $\FF$ is $X$, i.e., $\FF$ is a graded vector bundle over $X$.
\item The restriction to the support $S|_X$ is $S_{cl}$.
\item\label{l24_KF_CME} $S$ solves the classical master equation $\{S,S\}=0$.
\item\label{l24_KF_neg_degree_cohom} The cohomology sheaf of the sheaf of non-positively graded complexes $\OO_\FF/I_\FF, \underline{Q} 
$ \emph{vanishes in negative degrees}.
Here $\OO_\FF$ is the structure sheaf of $\FF$ over the base $X$, $I_\FF$ is the ideal generated by elements of positive degree in $\OO_\FF$. The derivation $Q=\{S,\bt\}$ preserves the ideal $I_\FF$ and thus induces a derivation $\underline{Q}$ on the quotient $\OO_\FF/I_\FF$ which still squares to zero, due to (\ref{l24_KF_CME}).
\end{enumerate}
Then, in the terminology of \cite{FelderKazhdan}, $(\FF,S)$ is a \emph{BV variety with support $(X,S_{cl})$}.

Two BV varieties $(\FF,S)$ and $(\FF',S')$ with the same support $(X,S_{cl})$ are said to be \emph{equivalent} if there exists a Poisson isomorphism (i.e. an isomorphism of sheaves of graded commutative algebras preserving the degree $1$ Poisson bracket) $\phi: \FF\ra\FF'$ inducing identity on the support $X\xra{\mr{id}}X$ and such that $S=\pi^*S'$.

Two BV varieties $(\FF,S)$ and $(\FF',S')$ with the same support $(X,S_{cl})$ are said to be \emph{stably equivalent} if they become equivalent after taking a direct product (possibly on both sides) with BV varieties with trivial support $(X=\mr{point},S_{cl}=0)$.\footnote{
Note that 
the transition $(\FF_\mr{min},S_\mr{min})\mapsto (\FF,S)$ between the minimal and non-minimal BV packages associated to Faddeev-Popov data in Subsection \ref{sss: FP via BV} is
an example of stable equivalence. In this example we think of the space of auxiliary fields $T^*[-1](\g^*[-1]\oplus \g^*)$, appearing in (\ref{l24_min_to_full_BV_for_FP}), as a graded variety with support a point (despite the fact that there are auxiliary fields of degree zero -- we consider them as vertical coordinates on the graded vector bundle over the support rather than geometric coordinates on the support).
}

\begin{theorem}[Felder-Kazhdan, \cite{FelderKazhdan}]
Let $S_{cl}$ be a regular function on a non-singular affine variety $X$. Then:
\begin{enumerate}
\item\label{l24_KFthm_a} There exists a BV variety $(\FF,S)$ with support $(X,S_0)$, such that $\FF=T^*[-1]\s{f}$ for some non-negatively graded variety $\s{f}$. 
\item\label{l24_KFthm_b} The BV variety $(\FF,S)$ satisfying the condition of (\ref{l24_KFthm_a}) is unique up to stable equivalence.
\item The BRST cohomology $H^\bt_Q(\OO_\FF)$, as a sheaf over $X$, is uniquely determined by the data $(X,S_{cl})$.
The BRST cohomology sheaf is supported on the critical locus of $S_{cl}$ and vanishes in negative degrees.
\item Zeroth BRST cohomology has the form
\begin{equation}\label{l24_KF_H0}
H^0_Q\simeq  J(S_{cl})^{L(S_{cl})}
\end{equation}
where $J(S_{cl})$ is the Jacobian ring (viewed as a sheaf of local rings over $X$) -- $\OO_X$ modulo the ideal generated by partial derivatives of $S_{cl}$, or equivalently the cokernel of $TX\xra{dS_{cl}}\OO_X$. $L(S_{cl})$ in (\ref{l24_KF_H0}) are the ``infinitesimal symmetries'' -- the kernel of $TX\xra{dS_{cl}}\OO_X$.
\end{enumerate}
\end{theorem}
The r.h.s. of (\ref{l24_KF_H0}) should be understood 
as a GIT replacement for the space of functions on the quotient of the critical locus of $S_{cl}$ by the infinitesimal gauge symmetries.

The idea of construction of the BV variety $(\FF,S)$ with given support $(X,S_{cl})$ is as follows:
\begin{enumerate}[(i)]
\item\label{l24_KF_construction_1} First, one constructs the complex  
\begin{equation}\label{l24_Koszul-Tate}
(\OO_\FF/I_\FF,\underline{Q}) \quad =\quad \cdots\ra \underbrace{TX}_{\deg=-1}\xra{dS_{cl}}\underbrace{\OO_X}_{\deg=0}
\end{equation}
as a Koszul-Tate resolution of the sheaf map $TX\xra{dS_{cl}}\OO_X$. I.e., locally on $X$, using local coordinates $x^i$, one does the following:
\begin{itemize}
\item One adjoins to $O_X$ the Koszul generators $x^+_i$ of degree (ghost number $-1$) with Koszul differential $\underline{Q}^{(1)}: x^+_i\mapsto \frac{\dd S_{cl}}{\dd x^i}$ (extended as a derivation on $\OO_X[x^+_1,\ldots,x^+_n]$).
\item One adjoins Tate generators $c^+_a$ in degree $-2$ in order to kill the Koszul cohomology classes in degree $-1$ which arose on the previous stage. 
\item One continues the process, introducing generators of degree $-k-1$ in order to kill the cohomology classes appearing in the previous stage in degree $-k$.
\end{itemize}
Eventually, we obtain 
the complex (\ref{l24_Koszul-Tate}) whose cohomology in negative degree vanishes by construction.
\item\label{l24_KF_construction_2} To construct the entire $-1$-symplectic variety $\FF$, we adjoin the ``conjugates'' of the Tate generators: for each Tate generator of degree $-k\leq -2$ we adjoin its Darboux conjugate (in the sense of the symplectic structure of $\FF$) in degree $k-1\geq 1$. 
\item An auxiliary theorem of Felder-Kazhdan establishes that one can indeed construct a function $S$ on $\FF$ with the property that it induces, via $Q=\{S,\bt\}$, the Koszul-Tate coboundary map $\underline{Q}$ on $\OO_\FF/I_\FF$ in (\ref{l24_Koszul-Tate}) that we constructed in (\ref{l24_KF_construction_1}) by successive elimination of negative-degree cocycles.
\end{enumerate}

\begin{remark}
Note that the Felder-Kazhdan's approach is quite different from our setup in Section \ref{sss: BV for non-integrable distribution}: in the latter we used the classical action \emph{and the data of infinitesimal gauge symmetry} as input. In Felder-Kazhdan approach, one uses only the classical action as the input and the symmetries are recovered indirectly -- as a ``reflection'' (conjugation) of degree $-2$ Tate generators, as in  (\ref{l24_KF_construction_2}) above.
\end{remark}

\begin{remark}
Felder-Kazhdan's condition (\ref{l24_KF_neg_degree_cohom}) of vanishing of the cohomology of $\underline{Q}$ in negative degree is 
an extremely strong condition. It leads to having to add an infinite tower of higher Tate generators/higher ghost fields even for certain simple finite-dimensional examples of $(X,S_{cl})$.\footnote{See, e.g., Example 6.8 in \cite{FelderKazhdan}, with $X$ the affine plane with coordinates $x,y$ and with $S_{cl}=\frac{1}{4}(x^2+y^2-1)^2$.} This acyclicity condition is in fact typicaly spoiled in field theoretic examples: there one has an additional requirement that the resolution (\ref{l24_Koszul-Tate}) should be compatible with locality on the underlying space-time manifold $M$ (in this setting, $X$ is itself the space of sections of a sheaf $\mathbb{F}^0$ over $M$ and one would like the space of Koszul-Tate generators in any degree  $-k$ to be the space of sections of some sheaf  $\mathbb{F}^{-k}$ over $M$). This requirement is typically incompatible with acyclicity: one should rather allow the Koszul-Tate complex (\ref{l24_Koszul-Tate}) to have \emph{small} (but possibly nonzero) cohomology in negative degrees. See also \cite{Getzler_spinning} for a 1-dimensional field theory example where one gets 
 large (infinite-dimensional) BRST cohomology spaces in negative degrees.
\end{remark}

\marginpar{\LARGE{Lecture 25, 11/28/2016.}}
\subsection{AKSZ sigma models}
in \cite{AKSZ}, Alexandrov, Kontsevich, Schwarz and Zaboronsky proposed a construction of a solution of the classical master equation on the mapping space between two supermanifolds (the source and the target) endowed with appropriate supergeometric data. This construction, 
referred to as the AKSZ construction, 
produces (for appropriate choices of the source/target data) the BV action functionals for many known topological field theories, including, in particular, the Chern-Simons theory, Poisson sigma model, $BF$ theory and more. Restrictions to special gauge-fixing Lagrangians also yield the  Witten's $A$ and $B$ models, Rozansky-Witten theory and more. 

One virtue of the AKSZ construction is that it is a very natural geometric construction, giving a new supergeometric perspective to known topological field theories (with Chern-Simons theory being the main motivating example), as well as producing new examples. 
Another virtue is that it produces right away the BV action for a topological field theory, thus 
solving/circumventing the potentially 
hard problem of constructing the BV action from the classical action and its symmetries.

\subsubsection{AKSZ construction} The input of the construction is the source and target data.

\begin{itemize}
\item \textbf{Source data.} A closed oriented $n$-dimensional manifold $M$ -- the source manifold (or the spacetime).

\item \textbf{Target data.} A Hamiltonian dg manifold (cf. Definition \ref{def_Ham_dg_mfd}) of degree $n-1$,
i.e., the data of:
\begin{itemize}
\item a dg manifold $(\NN,Q_\NN)$,
\item a symplectic form of internal degree $n-1$, $\omega_\NN\in \Omega^2(\NN)_{n-1}$,
\item a Hamiltonian of internal degree $n$, $\Theta_\NN\in C^\infty(\NN)_n$, generating the cohomological vector field $Q_\NN$ as its associated Hamiltonian vector field, and satisfying $\{\Theta_\NN,\Theta_\NN\}_{\omega_\NN}=0$.
\end{itemize}
Additionally, we require that $\omega_\NN$ is exact, with $\alpha_\NN\in \Omega^1(\NN)_{n-1}$ a fixed primitive $1$-form, i.e., $\omega_\NN=\delta\alpha_\NN$.\footnote{
Note that (cf. \cite{Roytenberg}), since $\omega_\NN$ is closed (being a symplectic form), there is, assuming $n\neq 1$, a distinguished primitive coming from the contraction with the Euler vector field on $\NN$: $\delta\iota_\mathbb{E} \omega_\NN=L_\mathbb{E} \omega_\NN=(n-1) \omega_\NN$. Thus, $\omega_\NN=\delta\left(\frac{1}{n-1}\iota_\mathbb{E} \omega_\NN\right)$ and so we may set $\alpha_\NN=\frac{1}{n-1}\iota_\mathbb{E} \omega_\NN$.
}
\end{itemize}

\textbf{A notation/terminology convention.} We will be using $\delta$ for de Rham operator on (forms on) the target and on the space of fields and will reserve $d$ for the de Rham operator on (forms on) the source. Also, we will be using the terms degree, internal degree or the ghost number interchangeably for the $\ZZ$-grading on functions (and forms etc.) on the target and on the mapping space, and will denote it by $\left|- \right|$. 

One constructs the space of fields of the model as the mapping space between graded supermanifolds:\footnote{\label{l25_footnote_mappings_morphisms}
For details on the mapping spaces in the category of graded supermanifolds, we refer to \cite{CattaneoSchaetz} and \cite{CMR}, Appendix B.2.
In particular, for $\MM,\NN$ two graded supermanifolds, one first introduces the space of \textit{morphisms} (or \textit{degree preserving mappings}):
a morphism $\phi: \MM\ra\NN$ is a morphism of sheaves of graded commutative algebras of functions (in particular, we have a map of bodies as smooth manifolds $\phi_0:M\ra N$ and for each open set in the body $U\subset N$ one has a morphism of commutative graded algebras $\phi_U^*:\left.\OO_\NN\right|_U\ra \left.\OO_\MM\right|_{\phi_0^{-1}(U)}$). The space of morphisms, denoted $\mr{Mor}(\MM,\NN)$ or $\mr{Map}_0(\MM,\NN)$, can be endowed with the structure of an infinite-dimensional (Fr\'echet) manifold. In the category of graded manifolds, the space of morphisms fails the adjunction property: $\mr{Mor}(\LL,\mr{Mor}(\MM,\NN))\neq \mr{Mor}(\LL\times \MM,\NN)$ for $\LL,\MM,\NN$ a general triple of graded manifolds (in fact, it fails even in the category of graded vector spaces).
However, if one extends from morphisms to all, possibly non-degree preserving, mappings, the corresponding $\mr{Hom}$ space -- the full mapping space -- satisfies the adjunction property: $\mr{Map}(\LL,\mr{Map}(\MM,\NN))=\mr{Map}(\LL\times\MM,\NN)$. Restricting to degree zero, we have the property $\mr{Mor}(\LL,\mr{Map}(\MM,\NN))=\mr{Mor}(\LL\times\MM,\NN)$. The latter property characterizes uniquely the mapping space $\mr{Map}(\MM,\NN)$ by probing it with 
morphisms of any test graded manifold $\LL$ into it. In other words, it is a description of the mapping space via the functor of points (a.k.a. the \textit{$S$-point formalism}). In the special case when $\NN$ is a graded vector space, one has $\mr{Map}(\MM,\NN)=C^\infty(\MM)\otimes \NN$ -- a graded infinite-dimensional manifold, whereas its degree zero part $\mr{Mor}(\MM,\NN)=[C^\infty(\MM)\otimes\NN]_0$  is an ordinary (non-graded) infinite-dimensional manifold. For $\NN$ a general graded manifold, mappings whose range fits in a single coordinate chart can be described by this construction (in particular, this applies to mappings which are infinitesimally close to constant mappings, which are of particular importance in perturbative quantization). 
More care (correct gluing of local pictures) is required to descibe mappings whose body is a non-trivial morphism $\phi$, whose range does not fit in a single coordinate chart on $\NN$ -- this is of particular importance for evaluating the perturbative contributions of instantons in AKSZ theories. 
%
}
\begin{equation}\label{l25_AKSZ_mapping_space}
\FF=\mr{Map}(T[1]M,\NN)
\end{equation}
We first focus on the local description of the AKSZ construction and then give a global (coordinate-invariant) description.

\textbf{Local description of the construction.}
Locally, if $x^a$ are local coordinates on $\NN$ and $u^i$ are local coordinates on $M$, with $\theta^i=du^i$ the corresponding degree $1$ fiber coordinates on $T[1]M$, one introduces the \textit{superfield} $X^a$ associated to the target coordinate $x^a$ as 
\begin{equation}\label{l25_superfield}
X^a=\sum_{k=0}^n\;\; \sum_{1\leq i_1<\cdots < i_k \leq n} X^a_{i_1\cdots i_k}(u)\;\theta^{i_1}\cdots \theta^{i_k}
\end{equation}
Here the 
coefficient functions $X^a_{i_1\cdots i_k}\in C^\infty(M)$ can be seen as \textit{coordinates} on the mapping space (\ref{l25_AKSZ_mapping_space}) of degree $|X^a_{i_1\cdots i_k}|=|x^a|-k$, where $|x^a|$ is the degree of $x^a$ as a coordinate on the target. Functions $X^a_{i_1\cdots i_k}(u)$ have to satisfy certain natural transition rule on the overlaps of coordinate charts on the target (and transform as local coefficients of a $k$-form on the source on the overlaps of source coordinate charts).

In the special case when the target $\NN$ is a graded vector space, one can cover it by a single coordinate chart and identify the mapping space with the space of $\NN$-valued nonhomogeneous forms on the source, 
$$\FF=\Omega^\bt(M)\otimes \NN$$

Assume that the target symplectic form is locally written as 
$$\omega_\NN=\frac12 \sum_{a,b} \omega_{ab}(x)\, \delta x^a \wedge \delta x^b = \delta\underbrace{(\sum_a \alpha_a(x)\, \delta x^a)}_{\alpha_\NN}$$
Then one constructs the odd-symplectic form (the BV 2-form) on the mapping space as
\begin{equation}\label{l25_omega}
\omega=\frac12 \sum_{a,b} \int_M \omega_{ab}(X)\cdot \delta X^a\wedge \delta X^b \qquad \in \Omega^2(\FF)_{-1}
\end{equation}
Here one understands the superfield $X^a$ as a coordinate function on $\FF$ taking values in forms on $M$. Thus, in (\ref{l25_omega}) the integrand is a $2$-form on $\FF$ valued in forms on $M$; we pick the top-degree part of the integrand w.r.t. the form degree on $M$ and integrate. Expression $\omega_{ab}(X)$ means, informally, ``substitute the superfield $X^c$ instead of the target coordinate $x^c$ into the expression for $\omega_{ab}(x)$''. (We will discuss a more invariant description below.)

The BV action of the model is defined as follows:
\begin{equation}\label{l25_AKSZ_action}
S
=\int_M \;\; \sum_a \alpha_a(X)\, dX^a + 
 \Theta_\NN (X) \qquad \in C^\infty(\FF)_0
\end{equation}
Here $d=\sum_i \theta^i \frac{\dd}{\dd u^i}$ is the de Rham differential on the source manifold $M$.
Note that the action (\ref{l25_AKSZ_action} is manifestly split into two parts -- the source (or ``kinetic'') term and the target (or ``interaction'') term.

The cohomological vector field on $\FF$ can be written in terms of variational derivatives w.r.t. the superfield components:
$$Q=\int_M (dX^a)_{i_1\cdots i_k}
(u)\frac{\delta}{\delta X^a_{i_1\cdots i_k}
(u)}+Q_\NN(X(u))^a\frac{\delta}{\delta X^a(u)}\qquad \in \mathfrak{X}(\FF)_1$$

\textbf{Global (coordinate-free) description of the construction.}
In the mapping space (\ref{l25_AKSZ_mapping_space}) the target $\NN$ is a dg manifold with a cohomological vector field $Q_\NN$; the source $T[1]M$ is also a dg manifold, with cohomological vector field defined by the de Rham operator $d=d_M$ on $M$. They can be viewed as infinitesimal diffeomorphisms (parameterized by an odd flow parameter) of the target and source of the mapping space, and can be lifted, via left and right action on mappings, respectively, to two infinitesimal diffeomorphisms of the mapping space. The latter correspond to two cohomological vector fields, $d_M^\mr{lifted}, Q_\NN^\mr{lifted} \;\;\in \mathfrak{X}(\FF)_1$ -- the lifts of $d_M$ and $Q_\NN$ to the mapping space via the left/right action on mappings. We set the total cohomological vector field on $\FF$ to be the sum of those, 
\begin{equation}\label{l25_Q_invar}
Q=d_M^\mr{lifted}+Q_\NN^\mr{lifted}
\end{equation}
It is manifestly split into a source part and a target part.

Consider the following diagram:
\begin{equation} \label{l25_transgression_CD}
\begin{CD}
T[1]M\times \FF @>\mr{ev}>> \NN \\
@V\pi VV \\
\FF
\end{CD}
\end{equation}
where $\mr{ev}:\; T[1]M\times \mr{Map}(T[1]M,\NN)\ra \NN$ is the evaluation mapping and $\pi: T[1]M\times \FF\ra \FF$ is the projection to the second factor. Given a $p$-form on the target $\phi\in \Omega^p(\NN)_j$, we can pull it back to a form $\mr{ev}^*\phi$ on $T[1]M \times \FF$, of which we pick the bi-degree $(0,p)$ component according to the bi-grading of the de Rham complex of a direct product (i.e. the component of $\mr{ev}^*\phi$ in $\Omega^0(T[1]M)\widehat\otimes \Omega^p(\FF)$). Then we integrate the result against the canonical Berezinian on $T[1]M$ (in other words, we use the identification $C^\infty(T[1]M)\cong \Omega^\bt(M)$ and integrate the result as a differential form on $M$, picking its top-degree component); we denote this integration procedure by $\pi_*$ (viewed as a pushforward along the projection in (\ref{l25_transgression_CD})). Thus, we have a mapping (the \textbf{transgression map}) sending forms on the target to forms on the mapping space:
\begin{equation}
\mathbb{T}=\pi_* \mr{ev}^*:\quad \Omega^p(\NN)_j \ra \Omega^p(\FF)_{j-n}
\end{equation}
Note that the trangression $\mathbb{T}$ preserves the de Rham degree of a form but changes its internal grading by $-n$ (since the integration $\pi_*$ is nontrivial on top forms on $M$, identified with elements of $C^\infty(T[1]M)_n$, and maps them to numbers, thus dropping the internal degree by $n$).  

Note that the superfield (\ref{l25_superfield}) is simply the pullback of a target coordinate function $x^a$ by the evaluation: 
$$X^a=\mr{ev}^*x^a \qquad \in\quad  C^\infty(T[1]M)\widehat\otimes C^\infty(\FF)
\cong \Omega^\bt(M)\widehat\otimes C^\infty(\FF)$$
Written in local coordinates, the  transgression acts on a form $\phi$ on the target as follows: 
$$\phi=\sum_{a_1,\ldots,a_p}\phi_{a_1\cdots a_p}(x) \, \delta x^{a_1}\wedge\cdots \wedge \delta x^{a_p}\quad 
\stackrel{\mathbb{T}}{\mapsto} 
\underbrace{\int_M}_{p_*} \underbrace{\sum_{a_1,\ldots, a_p} \phi_{a_1\cdots a_p}(X)\, \delta X^{a_1}\wedge\cdots \wedge \delta X^{a_p}}_{\mr{ev}^* \phi}
$$
In particular, the coefficients $\phi_{\cdots }(X)$ on the right hand side are the pullbacks of the coefficients $\phi_{\cdots }(x)$ by the evaluation, which can be informally understood as the substitution of the superfields $X^a$ instead of target coordinate functions $x^a$ in $\phi_{\cdots}(x)$. Note that $\mr{ev}^*$ is a homomorphism of commutative dg algebras and thus is defined by its action on generators $\mr{ev}^*:\; x^a\mapsto X^a$.

Using the transgression map, one defines the BV symplectic structure on $\FF$ as the transgression of the target symplectic form:
\begin{equation}\label{l25_omega_invar}
\omega=\mathbb{T}\omega_\NN\quad \in \Omega^2(\FF)_{-1}
\end{equation}
Note that the internal degree of $\omega_{\NN}$ was $(n-1)$ but became $-1$ after the transgression to the mapping space. The action is defined as follows:
\begin{equation}\label{l25_S_invar}
S=\iota_{d_M^\mr{lifted}}\mathbb{T}\alpha_\NN + \mathbb{T}\Theta_\NN\qquad \in C^\infty(\FF)_0
\end{equation}
The first and second terms here correspond to the first and second terms in the local expression (\ref{l25_AKSZ_action}). Note that $\mathbb{T}\alpha_\NN= \int_M\sum_a \alpha_a(X)\,\delta X^a \;\in \Omega^1(\FF)_{-1}$. Contracting this expression with the vector field $d_M^\mr{lifted}$ amounts, in practice, to replacing $\delta X^a$ with $d X^a$, thus leading to the first term in (\ref{l25_AKSZ_action}).

\begin{theorem}
The data $(\FF,\omega,Q,S)$ as defined above by (\ref{l25_AKSZ_mapping_space},\ref{l25_omega},\ref{l25_Q_invar},\ref{l25_S_invar}) satisfies the axioms of a classical BV theory (cf. Section \ref{sec: classBV}). In particular:
\begin{enumerate}[(i)]
\item\label{l25_i} $\omega$ is odd-symplectic,
\item\label{l25_ii} $S$ satisfies the classical master equation $\{S,S\}=0$,
\item\label{l25_iii} $Q$ is the Hamiltonian vector field for $S$,
\item\label{l25_iv} $Q^2=0$.
\end{enumerate}
\end{theorem}

\begin{proof} Note that, using Stokes' theorem on $M$, we obtain that $\mathbb{T}$ is a chain map, i.e., $\mathbb{T}\delta\phi=\delta \mathbb{T}\phi$ for any form $\phi$ on the target. Thus, e.g., $\delta \omega =\delta\mathbb{T}\omega_\NN=\mathbb{T}\delta\omega_\NN=0$, i.e., $\omega$ is closed. The fact that $\omega$ is weakly non-degenerate follows by inspection of (\ref{l25_omega}) from non-degeneracy of $\omega_\NN$. 

The fact $Q$ squares to zero follows from the construction (\ref{l25_Q_invar}): it is built out of the liftings of two cohomological vector fields on the source and the target -- the liftings automatically square to zero individually and also they automatically anti-commute with each other.

Let us check that $Q$ is the Hamiltonian vector field for $S$:
\begin{multline}
\iota_Q \omega = \iota_{d_M^\mr{lifted}}\mathbb{T}\omega_\NN + \iota_{Q_\NN^\mr{lifted}}\mathbb{T} \omega_\NN=
\iota_{d_M^\mr{lifted}}\delta\mathbb{T} \alpha_\NN + \mathbb{T}\iota_{Q_\NN}\omega_\NN=\\
=
\underbrace{L_{d_M^\mr{lifted}} \mathbb{T} \alpha_\NN}_0+\delta \iota_{d_M^\mr{lifted}}\alpha_\NN+\delta \mathbb{T}\Theta_\NN = \delta S
\end{multline}
Here we used the observation that the Lie derivative along $d_M^\mr{lifted}$ of any transgressed object $\mathbb{T}(\cdots)$ is the integral over $M$ of a total derivative and thus vanishes by Stokes' theorem on $M$.

The classical master equation follows from $Q^2=0$ and the fact that $S$ is the Hamiltonian for $Q$.
\end{proof}

\begin{remark}
A more general form of the AKSZ construction allows the source to be any dg manifold $\MM$ with a cohomological vector field $Q_\MM$, endowed with a compatible Berezinian $\mu_\MM$ of degree $-n$ (linked to the degrees of the target symplectic form and Hamiltonian), i.e. defining an integration map $\int_\MM \mu_\MM \cdot (\cdots): \; C^\infty(\MM)_n\ra \RR$. The discussion above goes through, changing everywhere $T[1]M$ to $\MM$ (in particular, the space of fields becomes the mapping space $\FF=\mr{Map}(\MM,\NN)$), changing $d_M$ to $Q_\MM$ and replacing the integration of forms on $M$ with integration of functions on $\MM$ w.r.t. $\mu_\MM$. Source of the form $T[1]M, d_M$ with canonical Berezinian corresponding to integration of forms on $M$ is the main class of examples and leads to topological (diffeomorphism invariant) AKSZ sigma models. However, one can study other interesting examples. E.g., 
for $M$ a closed complex manifold of (complex) dimension $n$ endowed with a holomorphic volume form $v\in \Omega^{n,0}(M)$, 
one can set $\MM=T^{0,1}[1]M$ -- the shifted anti-holomorphic tangent bundle of $M$. 
Here $C^\infty(\MM)$ is identified with the Dolbeaux complex $\Omega^{0,\bt}(M)$ of $M$ and the cohomological vector field on $\MM$ corresponds to the Dolbeault differential $\bar\partial$ on $M$. The integration on $\MM$ in this case is defined by pairing with the holomorphic volume form $v$. The AKSZ construction with such a source leads to ``holomorphic'' models depending on the complex structure on $M$.
\end{remark}

\begin{remark}\label{l25_rem_AKSZ_to_cl}
An AKSZ theory can be regarded as a BV extension of a certain classical gauge system. The data of the latter can be read off from the AKSZ theory, by expanding the BV action in the homogeneous components of fields, according to ghost degree. In particular, the classical fields are the \textit{degree preserving} mappings (or \textit{morphisms}, see footnote \ref{l25_footnote_mappings_morphisms}),
$F_{cl}=\mr{Map}_0(T[1]M,\NN)$ inside the space of all mappings (this is the body, or the ghost number zero part of the entire mapping space $\FF=\mr{Map}(T[1]M,\NN)$). The classical action $S_{cl}$ is the BV action $S$, as given by the AKSZ construction (\ref{l25_AKSZ_action},\ref{l25_S_invar}), restricted to ghost number zero fields, $S_{cl}=S|_{X^a\ra X_{cl}^a}$ where $X_{cl}^a$ is the $\gh=0$ part of the superfield $X^a$ which is the $k=|x^a|$ term in the sum (\ref{l25_superfield}) if $0\leq |x^a|\leq n$ and $X^a_{cl}=0$ otherwise (thus the classical field corresponding to a coordinate $x^a$ on the target is a $k$-form field on the source $M$, with $k=|x^a|$, whereas the whole superfield is a nonhomogeneous form on the source).\footnote{\label{l25_footnote_AKSZ_BV_from_cl}
Miraculously, under the assumption that 
$\NN$ is nonnegatively graded (which implies that degrees of coordinates on $\NN$ are in the  range $0\leq |x^a|\leq n-1$ due to the existence of an $(n-1)$-symplectic structure $\omega_\NN$)
the converse is true: one can obtain the BV action of the theory $S$ from the classical action $S_{cl}$ by substituting superfields instead of the classical fields, $S=\left. S_{cl}\right|_{X_{cl}^a\ra X^a}$, i.e., substituting nonhomogeneous forms instead of fields given by forms of fixed degree. This is a special feature of AKSZ theories and doesn't hold for general gauge theories.} Generators of infinitesimal gauge transformations correspond to degree $1$ part of the mapping space (mappings shifting degree by $1$), $\mr{Map}_1(T[1]M,\NN)$. Gauge symmetry transformations are read off of the terms in the BV action which 
are of form 
$\int_M X_{\gh=-1}X_{\gh=1}(\cdots)$
i.e. 
are bilinear in $\gh=1$ fields (ghosts) and  $\gh=-1$ fields (anti-fields for classical fields), 
with $(\cdots)$ depending only on $X_{cl}=X_{\gh=0}$. Structure constants of the gauge symmetry algebra are read off from the terms of form $\int_M X_{\gh=-2}X_{\gh=1}X_{\gh=1}(\cdots)$ in the BV action.
\end{remark}

\begin{remark}
The ghost degree zero part of the mapping space $\mr{Map}(T[1]M,\NN)$ contains an important submanifold $\mr{Map}_{0,dg}(T[1]M,\NN)$ given by \textit{morphisms of dg manifolds} from $T[1]M$ to $\NN$, i.e. degree zero maps intertwining the source and target cohomological vector fields (in other words, morphisms of sheaves of commutative \textit{differential} graded algebras $\phi_0:M\ra N$, $\phi^*:\left.\OO_\NN\right|_U\ra \left. \OO_{T[1]M}\right|_{\phi_0^{-1}(U)}$). It is the degree zero part of the vanishing locus of the total cohomological vector field $Q$ on the mapping space (\ref{l25_Q_invar}).\footnote{More precisely, with our sign conventions, the zero locus of $Q$ in degree zero is given by \textit{anti-dg morphisms}, i.e., degree-preserving maps $\phi:T[1]M\ra \NN$ satisfying $\phi^*Q_\NN=-d_M\phi^* :\;\; C^\infty(\NN)\ra C^\infty(T[1]M)$} In terms of the corresponding classical gauge system, it is the set of solutions of Euler-Lagrange equations defined by the classical action $S_{cl}$ of Remark \ref{l25_rem_AKSZ_to_cl}.
\end{remark}

\subsubsection{Example: Chern-Simons theory}
Chern-Simons theory in BV formalism is an 
instance of an AKSZ theory (and the original motivating example), with $n=3$, see \cite{AKSZ}. 

Fix $M$ a closed oriented $3$-manifold and $\g$ a Lie algebra with a nondegenerate ad-invariant pairing $\lan,\ran$. For simplicity, we assume that 
the pairing is the Killing form
$\lan x,y \ran=\tr xy$ with the trace taken in the adjoint representation of $\g$. 

We define the AKSZ target as the degree-shifted Lie algebra $\NN=\g[1]$. Let $\{T^a\}$ be an orthonormal basis in $\g$. Associated to it are the degree (ghost number) $1$ coordinate functions $\psi^a$ on $\g[1]$; it is also convenient to introduce an element 
\begin{equation}\label{l25_psi}
\psi=\sum_a \psi^a T^a \;\;\in C^\infty(\g[1])\otimes\g
\end{equation} -- a generating function for coordinates on $\g[1]$, or a universal $\g$-valued coordinate on $\g[1]$. The structure of an exact Hamiltonian dg manifold of degree $2$ on $\NN=\g[1]$ is defined as follows.
\begin{itemize}
\item The cohomological vector field is 
$$Q_\NN=\frac12\sum_{a,b,c} f^{abc} \psi^b\psi^c \frac{\dd}{\dd \psi^a} = \frac12 \lan [\psi,\psi] ,\frac{\dd}{\dd \psi}\ran$$
Here $f^{abc}$ are the structure constants of the Lie bracket in $\g$ in the basis $T^a$. Note that $C^\infty(\NN)=\wedge^\bt \g^*=C^\bt_{CE}(\g)$ is the Chevalley-Eilenberg complex of the Lie algebra $\g$ and $Q_\NN=d_{CE}$ is the Chevalley-Eilenberg differential.
\item The degree $2$ symplectic structure is 
$$\omega_\NN=\frac12 \sum_{a} \delta \psi^a\wedge \delta \psi^a =\frac12 \lan\delta\psi,\delta\psi\ran\quad\in\Omega^2(\NN)_2$$
It is exact, with a distinguished primitive
$$\alpha_\NN=\frac12 \sum_{a}  \psi^a\wedge \delta \psi^a =\frac12 \lan \psi,\delta\psi\ran \quad \in\Omega^1(\NN)_2$$
\item The degree $3$ Hamiltonian is
$$\Theta_\NN=\frac16 \sum_{a,b,c} f^{abc}\psi^a\psi^b\psi^c=\frac16\lan\psi,[\psi,\psi]\ran\quad \in C^\infty(\NN)_3$$
\end{itemize}
Note that the equation $\{\Theta_\NN,\Theta_\NN\}=0$ follows from the fact that $\Theta$ is a Chevalley-Eilenberg cocycle.

The space of AKSZ fields (\ref{l25_AKSZ_mapping_space}) is the mapping space
$$\FF=\mr{Map}(T[1]M,\g[1]) 
$$
which can be conveniently identified with the graded vector space of $\g$-valued nonhomogeneous differential forms on $M$ with degree shifted by $1$:
$$\FF= \Omega^\bt(M)\otimes \g[1]$$
We parameterize fields by the total superfield (the superfield corresponding to the universal coordinate $\psi$, as in (\ref{l25_psi}), on the target)
$$\mc{A}=\sum_a T^a \mc{A}^a = \mc{A}^{(0)}+\mc{A}^{(1)}+\mc{A}^{(2)}+\mc{A}^{(3)}$$
where $\mc{A}^a$ is the superfield corresponding to the coordinate $\psi^a$ on the target; $\mc{A}^{(k)}$ is the component of $\mc{A}$ which is a $k$-form on $M$. Comparing to the notations of Section \ref{sss: FP via BV}, we have the following:
\begin{itemize}
\item The de Rham degree $0$, ghost number $1$ component $\mc{A}^{(0)}=c\;\;\in \Omega^0(M,\g)$ is the ghost, corresponding to the generators of the infinitesimal gauge symmetry in Chern-Simons theory (the infinitesimal automorphisms of the principal bundle over $M$). 
\item The de Rham degree $1$, ghost number $0$ component $\mc{A}^{(1)}=A\;\;\in \Omega^1(M,\g)$ is the classical connection field.
\item  The de Rham degree $2$, ghost number $-1$ component $\mc{A}^{(2)}=A^+\;\;\in \Omega^2(M,\g)$ is the anti-field for the classical field.
\item The de Rham degree $3$, ghost number $-2$ component $\mc{A}^{(3)}=c^+\;\;\in \Omega^3(M,\g)$ is the anti-field for the ghost.
\end{itemize}

The BV symplectic structure on fields is
$$\omega=\frac12\int\tr \delta \mc{A}\wedge \delta \mc{A}\;\;= \int\tr \left(\delta A\wedge \delta A^++\delta c\wedge \delta c^+\right)  \quad\in \Omega^2(\FF)_{-1}$$
-- we are writing both the expression in terms of the superfield, as given by the AKSZ construction, and its expansion in terms of de Rham/ghost number homogeneous components.
The action (\ref{l25_AKSZ_action}) becomes
\begin{equation}\label{l25_AKSZ_CS}
S=\int_M\tr \left(\frac12 \mc{A}\wedge d\mc{A}+\frac16 \mc{A}\wedge [\mc{A},\mc{A}]\right) 
\end{equation}
Expanded in homogenenous components, it yields:
\begin{equation}
S=\int_M \tr \left(\frac12 A\wedge dA +\frac16 A\wedge [A,A] \right)+\int_M\tr A^+ \wedge (dc+[A,c]) + \int_M\tr \frac12 c^+[c,c]
\end{equation}
Note that the three terms on the r.h.s. here correspond to the three terms in (\ref{l24_S_FP_BV_min}) -- the classical action, the term corrseponding to the infinitesimal gauge transformations of classical fields and the term associated to the Lie algebra structure of infinitesimal gauge symmetry.

Note also that the AKSZ construction gives the BV action (\ref{l25_AKSZ_CS}) for Chern-Simons theory which has the same form as the classical action of Chern-Simons theory where one substitutes the superfield in place of the classical field. In fact, this property holds universally for  AKSZ theories.

The cohomological vector field (the BRST operator) is:
\begin{multline*}
Q=\int_M \lan d\mc{A}+\frac12 [\mc{A},\mc{A}],\frac{\delta}{\delta\mc{A}}\ran \\
=\int_M \lan \frac12 [c,c],\frac{\delta}{\delta c} \ran +
 \lan d_A c,\frac{\delta}{\delta A} \ran + \lan F_A+[c,A^+],\frac{\delta}{\delta A^+} \ran
 + \lan d_A A^+ + [c,c^+] ,\frac{\delta}{\delta c^+} \ran 
\end{multline*}
where $F_A=dA+\frac12[A,A]$ is the curvature and $d_A=d+[A,-]$ is the covariant derivative.
I.e., $Q$ acts on fields as
$$Qc=\frac12[c,c],\quad QA=d_A c,\quad QA^+=F_A+[c,A^+],\quad Qc^+=d_A A^++[c,c^+]$$

Lorentz gauge-fixing in Chern-Simons theory corresponds to fixing an arbitrary Riemannian metric $g$ on $M$ and choosing the Lagrangian\footnote{
To be precise, for $\LL_g$  defined by the condition $d^*\mc{A}=0$ to be Lagrangian, we must consider forms on $M$ twisted by an acyclic flat connection, as in \cite{AS1,AS2}. Otherwise, for a non-acyclic (e.g. trivial) flat connection, $\LL_g$ is not Lagrangian in the entire $\FF$, however it is Lagrangian in the complement of the harmonic forms and can be used as a gauge-fixing Lagrangian for the fiber BV integral, yielding the effective BV action on harmonic forms, see \cite{CM}.
} $\LL_g\subset \FF$ in the space of fields given by the condition $d^*\mc{A}=0$ (with $d^*=-*d*$ the Hodge conjugate of de Rham operator)  -- i.e. all homogeneous components of the superfield $\mc{A}$ must be coclosed. One can present the space of fields as $\FF\cong T^*[-1] (\Omega^0(M,\g)[1]\oplus \Omega^1(M,\g))$. Then $\LL_g$ can be identified with the conormal bundle $N^*[-1]C$ for 
$$C=\{c+A\;|\; c\;\mbox{any},\; A\;\;\mbox{s.t.}\; d^*A=0\}\;\;\subset \Omega^0(M,\g)[1]\oplus \Omega^1(M,\g)$$

\subsubsection{Example: Poisson sigma model}
This is an example of an AKSZ theory with $n=2$. Let 
$\Sigma$ be a closed oriented surface (the source manifold). 
Let $(N,\pi)$ be a Poisson manifold, i.e. a manifold $N$ endowed with a bivector $\pi\in \Gamma(N,\wedge^2 TN)
$ satisfying the Poisson property $[\pi,\pi]_{NS}=0$ where $[-,-]_{NS}$ is the Nijunhuis-Schouten bracket of polyvectors.

\textbf{Classical picture.} Classically, fields of Poisson sigma model (originally introduced in \cite{SchallerStrobl, Ikeda}) are budle maps from $T\Sigma$ to $T^*N$; we denote the base map by $X:\Sigma\ra N$ and the fiber map by $\eta$:
\begin{equation}\label{l25_X_eta_comm_diag}
\begin{CD}\\
T\Sigma @>\eta>> T^*N \\
@VVV @VVV \\
\Sigma @>X>> N
\end{CD}
\end{equation}
Here the vertical maps are the bundle projections. Note that $\eta$ can be viewed as a section, over $\Sigma$, of the bundle $T^*\Sigma\otimes X^*T^*N$. The classical action of the Poisson sigma model is:
\begin{equation}\label{l25_S_PSM_cl}
S_{cl}(X,\eta)=\int_\Sigma \lan \eta\stackrel{\wedge}{,}dX \ran + \frac12 \lan \pi(X), \eta\wedge\eta \ran
\end{equation}
Here we understand that $\eta\in \Omega^1(\Sigma,X^*T^*N)$, $dX\in \Omega^1(\Sigma, X^*TN)$, thus in the first term we canonically pair the tangent and cotangent fibers to $N$ and wedge two $1$-forms on $\Sigma$, to obtain a number-valued $2$-form on $\Sigma$ which can be integrated. In the second term we understand that $\pi(X)=X^*\pi\in\Gamma(\Sigma,\wedge^2X^*TN)$ which can be paired to $\eta\wedge\eta \in \Omega^2(\Sigma,\wedge^2 X^* T^*N)$ to a number-valued $2$-form on $\Sigma$ which again can be integrated over $\Sigma$.

If $x^i$ are local coordinates on $N$ and $\pi=\sum_{i,j} \pi^{ij}(x) \dd_i\wedge \dd_j$ locally with $\pi^{ij}(x)$ the coefficient functions, then the action (\ref{l25_S_PSM_cl}) can be written as 
$$S_{cl}(X,\eta)=\int_\Sigma \sum_i\eta_i \wedge dX^i+\sum_{i,j}\frac12 \pi^{ij}(X)\eta_i\wedge \eta_j$$
Where  $\{\eta_i\}$ is a collection of $1$-forms on $\Sigma$ -- components of $\eta=\sum_{i}\eta_i dx^i$.

Euler-Lagrange equations corresponding to the action (\ref{l25_S_PSM_cl}) read:
\begin{equation}\label{l25_PSM_EL}
dX+\lan \pi(X),\eta \ran =0, \qquad d\eta+\frac12 \lan \dd \pi(X),\eta\wedge \eta\ran=0
\end{equation}
Or, more explicitly, 
$dX^i + \sum_j\pi^{ij}(X)\eta_j=0 ,\;\; d\eta_i + \sum_{j,k}\frac12 \dd_i \pi^{jk}(X)\eta_j\wedge \eta_k=0$.
These equations are equivalent to the condition that the bundle map (\ref{l25_X_eta_comm_diag}) is a morphism of Lie algebroids, where $T\Sigma$ carries the tautological Lie algebroid structure and the Lie algebroid structure on $T^*N$ is defined by the Poisson bivector $\pi$.

The action (\ref{l25_S_PSM_cl}) is invariant under the infinitesimal gauge transformations
{
\begin{equation} \label{l25_PSM_gauge_transf}
X\mapsto X+\epsilon \underbrace{\lan\pi(X), b\ran}_{\delta_b X}, 
\qquad \eta \mapsto \eta+\epsilon \underbrace{(db+\lan \dd\pi(X),\eta\wedge b \ran)}_{\delta_b \eta}
\end{equation}
}
or, more explicitly, 
$X^i\mapsto X^i+ \epsilon\sum_j\pi^{ij}(X)b_j,\quad \eta_i\mapsto \eta_i+\epsilon \left(db_i+\sum_{j,k}\dd_i\pi^{jk}(X)\eta_j b_k\right)$.
Here $b=\sum_i b_i dx^i\in \Gamma(\Sigma, X^*T^*N)$ is the generator of the gauge transformation.

\begin{remark} The commutator of two transformations of type (\ref{l25_PSM_gauge_transf}) is not a transformation of the same type, but it is one ``on shell'', i.e., modulo Eular-Lagrange equations (\ref{l25_PSM_EL}). Explicitly (see \cite{CF}):
\begin{align}
[\delta_b,\delta_{b'}]X^i& = \delta_{[b,b']}X^i \\
[\delta_b,\delta_{b'}]\eta_i &=  
\delta_{[b,b']}\eta_i+ \sum_{k,r,s}\dd_i\dd_k \pi^{rs}(X)b_r b'_s (dX^k+\sum_j\pi^{kj}(X)\eta_j) \label{l25_PSM_non-integrability}
\end{align}
with $\delta_bX$, $\delta_b\eta$ as in  (\ref{l25_PSM_gauge_transf}). Here the Lie bracket of two generators is defined by 
\begin{equation}\label{l25_PSM_gauge_alg}
[b,b'] =  \lan \dd\pi(X), b\wedge b' \ran
\end{equation}
or, explicitly, $[b,b']_i=\sum_{j,k}\dd_i\pi^{jk}(X)b_jb'_k$. 
Thus, the gauge symmetry here is given by a non-integrable distribution on the space of classical fields, which is only integrable on the space of solutions of Euler-Lagrange equations.\footnote{Unless components of $\pi$ are at most linear in coordinates; if they are, the distribution defined by (\ref{l25_PSM_gauge_transf}) is integrable everywhere, since the defect of integrability in (\ref{l25_PSM_non-integrability}) is proportional to the second derivative of the Poisson bivector.} The gauge symmetry is given by a Lie algebroid $E$ over $\mr{Map}(\Sigma,N)$ whose fiber over a map $X:\Sigma\ra N$ is the space of sections $\Gamma(\Sigma,X^*T^*N)$, with the anchor given by $b\mapsto \delta_b X\in T_X \mr{Map}(\Sigma,N)$ (as in the first formula in (\ref{l25_PSM_gauge_transf})) and with the bracket of sections defined by (\ref{l25_PSM_gauge_alg}). The algebroid $E$ acts on the entire space of classical fields; however, this is not a strict action, but rather an action up-to-homotopy. In this regard, Poisson sigma model is structurally similar to the Example \ref{l24_ex_action_up_to_homotopy}, replacing a Lie algebra by a Lie algebroid.
\end{remark}

\textbf{Poisson sigma model as an AKSZ theory.} 
BV extension of the Poisson sigma model can be constructed as an AKSZ theory with $n=2$, with the target $$\NN=T^*[1]N$$
-- the degree-shifted cotangent bundle of the Poisson manifold. To write explicit formulas, we use local coordinates $x^i$ on $N$ and corresponding degree $1$ coordinates $p_i$ on the odd cotangent fibers of $T^*[1]N$. The structure of an exact degree $1$ Hamiltonian dg manifold on $\NN$ is defined as follows.  
\begin{itemize}
\item The symplectic structure on $\NN$ is the canonical symplectic structure of the cotangent bundle, and for the distinguished primitive we choose the canonical Liouvulle $1$-form of the cotangent bundle. Locally: 
$$\omega_\NN=\sum_i \delta p_i\wedge \delta x^i\quad \in \Omega^2(\NN)_1,\qquad
\alpha_\NN=\sum_i p_i \delta x^i \quad \in \Omega^1(\NN)_1$$
\item The degree $2$ Hamiltonian is the lifting of $\pi$ to a function $\Theta=\widetilde\pi$ on $T^*[1]N$ quadratic in fibers. Explicitly: 
$$\Theta_\NN=\frac12 \sum_{i,j} \pi^{ij}(x)p_i p_j\quad \in C^\infty(\NN)_2$$
\item The cohomological vector field $Q_\NN$ corresponds to the Poisson-Lichnerowicz differential $[\pi,-]_{NS}$ on polyvector fields $\mc{V}^\bt(N)$ under the identification $C^\infty(T^*[1]N)\cong \mc{V}^\bt(N)$. Explicitly,
$$Q_\NN=\sum_{i,j} \pi^{ij}(x)p_i \frac{\dd}{\dd x^j}+\sum_{i,j,k} \frac12 \dd_k \pi^{ij}(x) p_i p_j \frac{\dd}{\dd p_k}$$
\end{itemize}

The AKSZ space of fields of the model is the mapping space
$$\FF=\mr{Map}(T[1]\Sigma,T^*[1]N)$$
The mappings are parameterized by two superfields
\begin{equation}\label{l25_PSM_superfields}
\til X=X^{(0)}+X^{(1)}+X^{(2)},\qquad  \til\eta = \eta^{(0)}+ \eta^{(1)}+ \eta^{(2)}
\end{equation}
where $X^{(0)}=X: \Sigma\ra N$ is the base map (assigned ghost number $0$) and $X^{(k)}\in \Omega^k(\Sigma, X^*TN)$ for $k=1,2$, with ghost number $\gh\, X^{(k)}=-k$. Field component $\eta^{(k)}\in \Omega^k(\Sigma,X^*T^*N)$ for $k=0,1,2$, with ghost number $\gh\, \eta^{(k)}=1-k$. In particular, $\eta^{(1)}=\eta$ is the classical $\eta$-fields, the fiber map $T\Sigma\ra T^*N$ covering the base map $X$.

In the notations of \cite{CF}, one denotes the component $\eta^{(0)}$ by $\beta$ (it is the ghost corresponding to the generator $b$ of the infinitesimal gauge symmetry (\ref{l25_PSM_gauge_transf})) and denotes the anti-field (the conjugate field, w.r.t. the BV 2-form)  for the field $\Phi\in\{X,\eta,\beta\}$ by $\Phi^+$. Thus, one has the following notations for the homogeneous field components:
$$\widetilde{X}=X+\eta^++\beta^+,\qquad \til\eta = \beta+\eta + X^+$$
(the terms here correspond one-to-one to terms in 
(\ref{l25_PSM_superfields}); the order is preserved).
Or, arranged against de Rham degree and ghost number:\\
\begin{tabular}{c|ccc}
& $\mr{deg}=0$ & $\mr{deg}=1$ & $\mr{deg} =2$ \\
\hline
$\gh=-2$ & & & $\beta^+$ \\
$\gh=-1$ & & $\eta^+$ & $X^+$ \\
$\gh=0$ & $X$ & $\eta$ & \\
$\gh=1$ & $\beta$  & &
\end{tabular}

The odd-symplectic form (\ref{l25_omega},\ref{l25_omega_invar}) on the space of BV fields is:
\begin{multline*}
\omega=\sum_i\int_\Sigma \delta \til\eta_i \wedge \delta \til X^i  = \\ 
= \sum_i\int_\Sigma \delta X^i\wedge\delta X^+_i+ \delta \eta_i\wedge \delta \eta^+_i + \delta \beta_i\wedge \delta \beta^{+i}\qquad \in \Omega^{2}(\FF)_{-1}
\end{multline*}
The BV action (\ref{l25_AKSZ_action}) is:
\begin{equation}\label{l25_S_PSM_AKSZ}
S
= \int_\Sigma \sum_i \til\eta_i \wedge d \til X^i + \sum_{i,j}\frac12 \pi^{ij}(\til X)\til\eta_i\wedge \til\eta_j
\end{equation}
Expanding this expression in homogeneous field components, one obtains:
\begin{multline}\label{l25_S_PSM_AKSZ_expanded}
S=\int_\Sigma \left(\sum_i \eta_i\wedge dX^i +\sum_{i,j}\frac12\pi^{ij}(X)\eta_j\wedge \eta_k \right) +\\
+
\int_\Sigma \sum_{i,j} X^+_i \pi^{ij}(X)\beta_j+
\int_\Sigma \sum_i \eta^{+i} \wedge \left(d\beta_i+\sum_{j,k}\dd_i\pi^{jk}(X)\eta_j\beta_k\right)
+\int_\Sigma \sum_{i,j,k} \frac12 \beta^+_i \dd_i \pi^{jk}(X)\beta_j\beta_k+\\
+\int_\Sigma \sum_{i,j,k,l} \frac14 \eta^{+i}\wedge\eta^{+j}\;\dd_i\dd_j\pi^{kl}(X)\;\beta_k\beta_l
\end{multline}
Here we have the following five terms on the r.h.s.:
\begin{itemize}
\item The first term is the classical action (\ref{l25_S_PSM_cl}).
\item Second and third term correspond to the gauge transformations of classical fields (\ref{l25_PSM_gauge_transf}). 
\item Fourth term corresponds to the Lie brackets (\ref{l25_PSM_gauge_alg}). 
\item Fifth term is the homotopy for the defect of integrability of the gauge symmetry as a distribution on the space of classical fields (\ref{l25_PSM_non-integrability}). 
\end{itemize}
Thus, the entire BV action has the structure similar to  the ansatz 
(\ref{l24_S_Lie_up_to_hom}) of the toy Example \ref{l24_ex_action_up_to_homotopy}. By a miracle of AKSZ construction, this whole structrure is developed from the BV action (\ref{l25_S_PSM_AKSZ}) which coincides, formally, with the original classical action, but with the classical fields $X,\eta$ replaced by the AKSZ superfields $\til X,\til\eta$ (cf. footnote \ref{l25_footnote_AKSZ_BV_from_cl}).

The appearance of a term quadratic in antifields in the BV action (\ref{l25_S_PSM_AKSZ_expanded}) is a sign that we are dealing with a theory where the gauge symmetry is given by a non-integrable distribution (cf. Remark \ref{l24_rem_antifield_grading}).

\marginpar{\LARGE{Lecture 26, 11/30/2016.}}

\subsubsection{Example: $BF$ theory}
The $BF$ theory is a close relative of Chern-Simons theory.\footnote{In particular, $2$-dimensional $BF$ theory arises, on one hand, as the dimensional reduction of Chern-Simons theory on a manifold of type $\Sigma\times S^1$ with $\Sigma$ a surface. On the other hand, it is the zero area limit of $2$-dimensional Yang-Mills theory \cite{Migdal, Witten_2Dgauge}. $BF$ theory in dimension $3$ is the same as Chern-Simons theory with structure Lie algebra $\g\ltimes \g^*$ (where one equips $\g^*$ with zero bracket and brackets between $\g$ and $\g^*$ are given by coadjoint action; the inner product is 
built out of the
canonical pairing between $\g$ and $\g^*$ summands).  Abelian $BF$ theory in arbitrary dimension was first introduced in \cite{SchwarzBF}; the main result of this paper is that its path integral quantization yields the Ray-Singer torsion.} It is a rare example of a topological field theory which exists in any dimension.
Let $\g$ be a finite-dimensional Lie algebra\footnote{Unlike Chern-Simons theory, one does not need a non-degenerate ad-invariant pairing on $\g$ to formulate the $BF$  theory. However, for the BV quantization to go through (for the quantum master equation), one needs to require that the Lie algebra $\g$ is \textit{unimodular}, i.e., $\tr_\g [x,-]=0$ for any $x\in\g$. In particular, this property is satisfied automatically if $\g$ is the Lie algebra of a \textit{compact} group $G$.} corresponding to a Lie group $G$ and $M$ a closed oriented $n$-manifold for any $n\geq 2$.

\textbf{Classical picture.} Classically, the theory is defined by the action
\begin{equation}\label{l26_BF_S_cl}
S_{cl}=\int_M \lan B \stackrel{\wedge}{,} F_A 
\ran
\end{equation}
(And hence the name ``$BF$'' theory.)
Here the classical fields are: 
\begin{itemize}
\item The connection $A$ in the trivial $G$-bundle on $M$, i.e. $A\in \Omega^1(M,\g)$.
\item An $(n-2)$-form  $B\in \Omega^{n-2}(M,\g^*)$.
\end{itemize}
In (\ref{l26_BF_S_cl}), $F_A=dA+\frac12 [A,A]\in \Omega^2(M,\g)$ is the curvature of the connection $A$; the brackets $\lan,\ran$ stand for the canonical pairing between $\g$ and $\g^*$.

\begin{remark}
In a more general setup, one fixes a (possibly, non-trivial) $G$-bundle $\PP$ over $M$.
Then the field $A$ is a principal connection,  $A\in \mr{Conn}(\PP)$ (note that its curvature $F_A\in \Omega^2(M,\mr{ad}(\PP))$ is a $2$-form with coefficients in the adjoint bundle $\mr{ad}(\PP)$)  and $B\in \Omega^{n-2}(M,\mr{ad}^*(\PP))$. Thus, the expression (\ref{l26_BF_S_cl}) still makes sense.
\end{remark}

Euler-Lagrange equations read:
\begin{eqnarray}
F_A &=& 0 \\
d_A B &=& 0
\end{eqnarray}
Here $d_A B= (d+\mr{ad}^*_A)B 
$ is the 
de Rham differential of $B$ twisted by the connection $A$. 
Thus, a solution of Euler-Lagrange equations is a pair $(A,B)$ with $A$ a \textit{flat} connection and $B$ a horizontal (covariantly constant) $(n-2)$-form.

Action (\ref{l26_BF_S_cl}) is invariant under the following two types of gauge transformations:
\begin{align}\label{l26_BF_gauge_transf_1}
(A,B)&\mapsto& (A^g,B^g)=&\left(\;\; g^{-1}A g+g^{-1}dg\;\; ,\;\; \mr{Ad}^*_{g^{-1}} (B)\;\; \right) \\
\label{l26_BF_gauge_transf_2}
(A,B)&\mapsto& (A^t,B^t)=&\left(\;\;  A \;\; ,\;\; B+d_A t   \;\;\right)
\end{align}
with $g:M\ra G$ and $t\in \Omega^{n-3}(M,\g^*)$ the generators of the gauge transformations. Thus, the gauge symmetry is given by the action of the group $\mc{G}=\mr{Map}(M,G)\ltimes \Omega^{n-3}(M,\g^*)$ on the space of classical fields $F_{cl}=\Omega^1(M,\g)\oplus \Omega^{n-2}(M,\g^*)$. The group $\mc{G}$ acts on the classical fields with nontrivial stabilizers. In particular, if $(A,B)$ is a solution of Euler-Lagrange equations (in fact, we just need $A$ to be flat for this), 
elements $(g,t)$ and $(\til g,\til t)$ of $\mc{G}$ act in the same way on $(A,B)$ if 
\begin{equation}\label{l26_BF_gauge_step2}
(\til g,\til t)=(g,t+d_A t')
\end{equation}
with generator $t'\in \Omega^{n-4}(M,\g^*)=: \mc{G}_{2}$ (we view $\mc{G}_2$ as an abelian Lie group). In particular, the stabilizer of a solution of Euler-Lagrange equations $(A,B)\in F_{cl}$ under $\mc{G}$-action contains the orbit of the unit of $\mc{G}$ under the shifts $\mc{G}_{2}\ra \mc{G}$  given by  (\ref{l26_BF_gauge_step2}). Furthermore, shifts (\ref{l26_BF_gauge_step2}) are the same if $\til t'=t'+d_A t''$ with $t''\in \Omega^{n-5}(M,\g^*)=: \mc{G}_{3}$, etc. Thus, over the space $EL\subset F_{cl}$ of solutions of Euler-Lagrange equations, one has the ``tower of reducibility'' of gauge symmetry (cf. Section \ref{sss: higher ghosts}):
\begin{equation}\label{l26_BF_tower}
\mc{G}_{n-2}\acts \cdots \acts\mc{G}_2\acts \mc{G}=\mc{G}_1\acts EL 
\end{equation}
where,  for $k\geq 2$, groups  $\mc{G}_k=\Omega^{n-2-k}(M,\g^*)$ are abelian and act on  $\mc{G}_{k-1}$ via shifts 
\begin{equation}\label{l26_BF_shifts}
t_{k-1}\mapsto t_{k-1}+d_A t_k
\end{equation} 
The group $\mc{G}_1:=\mc{G}=\mr{Map}(M,G)\ltimes\Omega^{n-3}(M,\g^*)$ is non-abelian.

The infinitesimal picture is that the gauge symmetry is given by a Lie algebroid 
\begin{equation}\label{l26_BF_gauge_sym_infin}
\underbrace{\Omega^{0}(M,\g)\oplus \Omega^{n-3}(M,\g^*) }_{\mr{Lie}(\mc{G})}\ra T F_{cl}
\end{equation}
over $F_{cl}$. The anchor corresponds to the infinitesimal version of gauge transformations (\ref{l26_BF_gauge_transf_1},\ref{l26_BF_gauge_transf_2}):
\begin{align}
(A,B)&\mapsto&( A+\epsilon\; d_A \gamma, B-\epsilon\; \mr{ad}^*_\gamma(B) ) \label{l26_BF_infin_gauge_transf_1}\\
(A,B)&\mapsto &(A,B+\epsilon\; d_A t) \label{l26_BF_infin_gauge_transf_2}
\end{align}
with infinitesimal generators $(\gamma,t)\in \Omega^0(M,\g)\oplus \Omega^{n-3}(M,\g^*)$.
The anchor has large stabilizers over points of $EL\subset F_{cl}$. 
Over $EL$, one has a resolution of (\ref{l26_BF_gauge_sym_infin}) by an exact sequence of vector bundles over $EL$:
\begin{equation}\label{l26_BF_tower_infin}
\underbrace{\Omega^0(M,g^*)}_{\mr{Lie}(\mc{G}_{n-2})}\xra{d_A} 
\cdots \xra{d_A} \underbrace{\Omega^{n-3}(M,g^*)}_{\mr{Lie}(\mc{G}_{2})} \xra{(0,d_A)} \underbrace{\Omega^{0}(M,\g)\oplus \Omega^{n-3}(M,\g^*) }_{\mr{Lie}(\mc{G})}\ra T\;EL
\end{equation}

\textbf{$BF$ as an AKSZ theory.} We set the AKSZ target to be $\NN=\g[1]\oplus \g^*[n-2]$. If $\{T_a\}$ is a basis in $\g$, and $\{T^a\}$ the dual basis in $g^*$, we have the corresponding coordinates $\psi^a$ and $\xi_a$ on $\NN$ of degrees $|\psi^a|=1$, $|\xi_a|=n-2$; we also denote 
$$\psi=\sum_a \psi^a T_a\in \mr{Fun}_1(\NN,\g),\qquad \xi=\sum_a \xi_a T^a\in \mr{Fun}_{n-2}(\NN,\g^*)$$ 
the generating functions for coordinates on $\NN$. We denote the structure constants of $\g$ in the basis $\{T_a\}$ by $f^a_{bc}$, i.e., $[T_b,T_c]=\sum_a f^a_{bc}T_a$. We fix the following target data:
\begin{itemize}
\item The target symplectic form 
$$\omega_\NN=\lan \delta \xi\stackrel{\wedge}{,}\delta\psi \ran=\sum_{a}\delta\xi_a\wedge \delta\psi^a\qquad \in \Omega^2(\NN)_{n-1}$$
with the distinguished primitive $1$-form
$$\alpha_\NN=\lan \xi,\delta \psi\ran =\sum_a \xi_a \delta\psi^a\qquad \in \Omega^1(\NN)_{n-1}$$
\item The target cohomological vector field $Q_\NN$ is the Chevalley-Eilenberg differential in $C^\infty(\NN)\cong C^\bt_{CE}(\g,\mr{Sym}^\bt(\g^*[n-2]))$ -- the cochains of the Lie algebra $\g$ with coefficients in the module given by the symmetric powers of the (degree shifted) coadjoint module. Explicitly:
\begin{multline*}
Q_\NN=\lan \frac12 [\psi,\psi],\frac{\dd}{\dd \psi} \ran+\lan \mr{ad}^*_\psi(\xi),\frac{\dd}{\dd\xi} \ran=\\
=\sum_{a,b,c} \frac12 f^a_{bc} \psi^b\psi^c\frac{\dd}{\dd\psi^a}-\sum_{a,b,c}f^a_{bc} \psi^b\xi_a\frac{\dd}{\dd\xi_c}\qquad \in \mathfrak{X}(\NN)_1
\end{multline*}
\item The target Hamiltonian is:
\begin{equation}
\Theta_\NN=\frac12 \lan \xi,[\psi,\psi] \ran = \sum_{a,b,c}\frac12 f^a_{bc}\; \xi_a \psi^b\psi^c\qquad \in C^\infty(\NN)_{n}
\end{equation}
The equation $\{\Theta_\NN,\Theta_\NN\}_{\omega_\NN}=0$ holds by virtue of the Jacobi identity for the structure constants of $\g$.
\end{itemize}

The AKSZ space of fields is the mapping space
\begin{equation}
\FF=\mr{Map}(T[1]M,\g[1]\oplus \g^*[n-2]) \cong \Omega^\bt(M,\g)[1]\oplus \Omega^\bt(M,\g^*)[n-2]
\end{equation}
It is parameterized by two superfields, $\mc{A}$ and $\mc{B}$, corresponding to the target universal coordinates $\psi$ and $\xi$ via pullback by the evaluation map $\mr{ev}: T[1]M\times\FF\ra \NN$:
\begin{equation}\label{l26_BF_superfields}
\mc{A}=\mr{ev}^*\psi= \mc{A}^{(0)}+\cdots + \mc{A}^{(n)},\qquad \mc{B}=\mr{ev}^*\xi=\mc{B}^{(0)}+\cdots + \mc{B}^{(n)}
\end{equation}
Here $\mc{A}^{(k)}$ is a $k$-form on $M$ valued in $\g$ and has the ghost number $\mr{gh}(\mc{A}^{(k)})=1-k$; $\mc{B}^{(k)}$ is a $k$-form on $M$ valued in $\g^*$ and has the ghost number $\mr{gh}(\mc{B}^{(k)})=n-2-k$. 

Comparing to the classical gauge theory picture of $BF$ theory discussed above, we have the following interpretation of the components of the superfields $\mc{A},\mc{B}$:
\begin{itemize}
\item Fields of ghost number zero, $A=\mc{A}^{(1)}\in \Omega^1(M,\g)$ and $B=\mc{B}^{(n-2)}\in \Omega^{n-2}(M,\g^*)$ are the classical fields of the theory, 
as in (\ref{l26_BF_S_cl}).
\item Fields of ghost number one, $c=\mc{A}^{(0)}\in\Omega^{0}(M,\g)$ and $\tau_1=\mc{B}^{(n-3)}\in\Omega^{(n-3)}(M,\g^*)$ are the ghosts for the gauge transformations (\ref{l26_BF_infin_gauge_transf_1},\ref{l26_BF_infin_gauge_transf_2}).
\item Field of ghost number two, $\tau_2=\mc{B}^{(n-4)}\in \Omega^{n-4}(M,\g^*)$ -- the higher ghost corresponding to the transformations (\ref{l26_BF_gauge_step2}), etc. For every $j$ up to $(n-2)$ we have a $j$-th higher ghost $\tau_j=\mc{B}^{(n-2-j)}\in \Omega^{n-2-j}(M,\g^*)$, of ghost degree $j$, corresponding to the $j$-th step, $\mc{G}_j$, of the tower (\ref{l26_BF_tower},\ref{l26_BF_tower_infin}).
\item For each $\Phi$ a classical field or a (higher) ghost, of de Rham degree $k$ and ghost number $j$, we have a corresponding anti-field $\Phi^+$ of de Rham degree $n-k$ and ghost number $-1-j$.
\end{itemize}
According to this interpretation, 
the homogeneous components of the superfields (\ref{l26_BF_superfields}) are:
\begin{equation}
\mc{A}=c+A+B^++\tau_1^++\cdots + \tau_{n-2}^+, \qquad \mc{B}=\tau_{n-2}	+\cdots+\tau_1+B+A^++c^+
\end{equation}
The respective de Rham degrees and ghost numbers are:

\begin{tabular}{|l|cccccc|}\hline
components of $\mc{A}$: & $c$ & $A$ & $B^+$ & $\tau_1^+$ & $\cdots$ & $\tau_{n-2}^+$\\
de Rham degree: & $0$ & $1$ & $2$ & $3$ & $\cdots$ & $n$ \\
ghost number: & $1$ & $0$ & $-1$ & $-2$ & $\cdots$ & $1-n$
\end{tabular}\\ \\
\begin{tabular}{|l|cccccc|}\hline \hline
components of $\mc{B}$: & $\tau_{n-2}$ & $\cdots$ & $\tau_1$ & $B$ & $A^+$ & $c^+$ \\
de Rham degree: & $0$ & $\cdots$ & $n-3$ & $n-2$ & $n-1$ & $n$ \\
ghost number: & $n-2$ & $\cdots$ & $1$ & $0$ & $-1$ & $-2$ \\ \hline
\end{tabular}

Note that for all components of $\mc{A}$, the sum of the de Rham degree and the ghost number is $\deg+\gh=1$. For the components of 
$\mc{B}$ we have $\deg+\gh=n-2$.

The odd-symplectic form on the space of BV fields is:
\begin{multline*} \omega=\int_M \lan \delta \mc{B}\stackrel{\wedge}{,}\delta\mc{A}\ran \\
=\int_M \lan \delta A \stackrel{\wedge}{,} \delta A^+\ran + 
\lan \delta B \stackrel{\wedge}{,} \delta B^+\ran+
\lan \delta c \stackrel{\wedge}{,} \delta c^+\ran+
\lan \delta \tau_1 \stackrel{\wedge}{,} \delta \tau_1^+\ran+ \cdots+
\lan \delta \tau_n \stackrel{\wedge}{,} \delta \tau_n^+\ran
\end{multline*}
And the BV action (\ref{l25_AKSZ_action}) is:
\begin{equation}\label{l26_S_BF_AKSZ}
S=\int_M \lan \mc{B} \stackrel{\wedge}{,} d\mc{A}+\frac12 [\mc{A},\mc{A}]\ran
\end{equation}
Expanding it in the homogeneous components of the AKSZ superfields, we obtain
\begin{multline}\label{l26_BF_AKSZ_S_expanded}
S
= \int_M \lan B \stackrel{\wedge}{,} F_A 
 \ran + 
\lan A^+ \stackrel{\wedge}{,} d_A c \ran + \lan B^+ \stackrel{\wedge}{,}(-1)^n\mr{ad}^*_c (B)+d_A\tau_1  \ran+ \\ +
\lan \tau_1^+\stackrel{\wedge}{,} d_A \tau_2 \ran+ \cdots \lan \tau_{n-3}^+\stackrel{\wedge}{,}d_A \tau_{n-2} \ran +  \\ +
\lan c^+\stackrel{\wedge}{,} \frac12 [c,c] \ran + (-1)^n \lan \tau_1^+ \stackrel{\wedge}{,} \mr{ad}^*_c \tau_1 \ran + \cdots 
+(-1)^n \lan \tau_{n-2}^+ \stackrel{\wedge}{,} \mr{ad}^*_c \tau_{n-2} \ran+\\
+\mbox{terms quadratic in anti-fields}
\end{multline}
Here we see the classical action, terms corresponding to infinitesimal gauge transformation of fields (\ref{l26_BF_infin_gauge_transf_1},\ref{l26_BF_infin_gauge_transf_2}) and the ``higher gauge transformations'' -- the shifts  (\ref{l26_BF_shifts}); terms corresponding to the Lie algebra structure on $\mr{Lie}(\mc{G})$ and to the action of $\mr{Lie}(\mc{G})$ on $\mr{Lie}(\mc{G}_j)$. 
Terms quadratic in anti-fields (e.g. the term $\int_M \lan \tau_2 \stackrel{\wedge}{,} [B^+,B^+]\ran$) appear 
for the $BF$ theory in dimension $n\geq 4$. Their appearance is related to the fact that the stabilizers of the gauge group action become large over solutions of Euler-Lagrange equations and are small or trivial over other classical field configurations. 

It is instructive to write out the action (\ref{l26_BF_AKSZ_S_expanded}) in low dimensions explicitly:
\begin{itemize}
\item In dimension $n=2$,  we have 
$$S=\int_M \lan B, F_A \ran + \lan A^+ \stackrel{\wedge}{,} d_A c \ran + \lan B^+, \mr{ad}^*_c(B) \ran+\frac12 \lan c^+,[c,c]\ran$$
This is the case arising as the zero area limit of $2$-dimensional Yang-Mills theory. In this dimension, classical $B$ field is a scalar, the superfields are  $\mc{A}=c+A+B^+$, $\mc{B}=B+A^++c^+$. Here higher ghosts $\tau_{\geq 2}$ are absent, and even the first ghost $\tau_1$ is absent, for degree reasons.  
The BV action satisfies the ansatz (\ref{l24_S_FP_BV_min}) for the gauge Lie algebra $\mr{Lie}(\mc{G})=\Omega^0(M,\g)$ acting on $F_{cl}=\Omega^1(M,\g)\oplus \Omega^0(M,\g^*)$.
\item In dimension $n=3$, we have
$$ 
S=\int_M \lan B\stackrel{\wedge}{,} F_A\ran + \lan A^+ \stackrel{\wedge}{,} d_A c \ran + \lan B^+\stackrel{\wedge}{,} -\mr{ad}^*_c(B)+d_A\tau_1 \ran+\frac12 \lan c^+,[c,c]\ran - \lan \tau_1^+,\mr{ad}^*_c\tau_1 \ran
$$
Here the superfields are $\mc{A}=c+A+B^++\tau_1^+$, $\mc{B}=\tau_1+B+A^++c^+$. In comparison to the 2-dimensional case, the ghost $\tau_1$ appears, but we still don't have higher ghosts $\tau_{\geq 2}$. This BV action again satisfies the ansatz (\ref{l24_S_FP_BV_min}) for the gauge Lie algebra $\mr{Lie}(\mc{G})=\Omega^0(M,\g)\oplus \Omega^0(M,\g^*)$ acting on $F_{cl}=\Omega^1(M,\g)\oplus \Omega^1(M,\g^*)$. Three-dimensional $BF$ theory is special in that it is the Chern-Simons theory with structure group $\g\oplus \g^*$. Also, for $\g$ a quadratic Lie algebra, it possesses a very interesting deformation by including a term $B\wedge B\wedge B$ in the action (see \cite{CCFM,CMR}) -- it 
can be constructed as the deformation of the target AKSZ Hamiltonian $\Theta$ by the term $\pm \frac16 (\xi,[\xi,\xi])$ where one uses the inner product on $\g$ to identify $\g^*\simeq \g$. The resulting theory, and especially its quantization, depends strongly on the sign of the $B^3$ term.
\item In dimension $n=4$, we have
\begin{multline*}
S=\int_M \lan B\stackrel{\wedge}{,} F_A\ran + \lan A^+\stackrel{\wedge}{,} d_A c \ran + \lan B^+\stackrel{\wedge}{,} \mr{ad}^*_c(B)+d_A\tau_1 \ran + \lan \tau_1^+ \stackrel{\wedge}{,} d_A \tau_2 \ran +\\
+
\frac12 \lan c^+,[c,c]\ran + \lan \tau_1^+\stackrel{\wedge}{,}\mr{ad}^*_c\tau_1 \ran + \lan \tau_2^+,\mr{ad}^*_c\tau_2 \ran
+ \frac12 \lan \tau_2 ,[B^+,B^+] \ran
\end{multline*}
In this dimension, the  gauge symmetry tower (\ref{l26_BF_tower}) attains the second stage $\mc{G}_2$. Thus, in comparison with the three-dimensional situation,  wenow have the higher ghost $\tau_2$. Also, the BV action contains a term quadratic in anti-fields and is not anymore of the form (\ref{l24_S_FP_BV_min}). 
\end{itemize}
In particular, we see that the $BF$ theory can be treated by Faddeev-Popov in dimensions $n=2,3$, but Batalin-Vilkovisky formalism becomes essential for its treatment in dimensions $n\geq 4$, due to the fact that the gauge symmetry -- the tower (\ref{l26_BF_tower} -- becomes more complicated.

We required from the beginning that $n\geq 2$, because otherwise one or both classical fields $A,B$ vanish by degree reasons and $S_{cl}$ is identically zero; a related point is that for $n<2$, the AKSZ target $\NN=\g[1]\oplus \g[n-2]$ is not a nonnegatively graded manifold (cf. footnote \ref{l25_footnote_AKSZ_BV_from_cl}). Nevertheless, one can consider $BF$ theory in its AKSZ formulation in dimensions $n=0,1$:
\begin{itemize}
\item For $n=1$, the superfields are: $\mc{A}=c+A$, $\mc{B}=A^++c^+$. In particular, the classical $B$ field is absent (i.e. vanishes identically) for degree reason -- it should have beern $(-1)$-form. The BV action (\ref{l26_BF_AKSZ_S_expanded}) is:
$$
S=\int_M \lan A^+,d_A c \ran + \frac12 \lan c^+,[c,c] \ran
$$
In this case we have the space of classical fields $F_{cl}=\Omega^1(M,\g)$ with zero classical action $S_{cl}=0$, acted on by the Lie algebra $\Omega^0(M,\g)$. So, the BV action does satisfy the ansatz (\ref{l24_S_FP_BV_min}) with first term vanishing (i.e. the two terms we do have in $S$ correspond to the action of gauge transformations on connection fields and to the Lie algebra of gauge transformations).
\item In the most degenerate case $n=0$ (assume for simplicity that $M$ is a single point), we have $\mc{A}=c$, $\mc{B}=c^+$. Both classical fields $A,B$ are absent (vanish) for degree reasons. Thus, classically, $F_{cl}=\{0\}$ is a point endowed with the action of $\Omega^0(M,\g)=\g$. The action (\ref{l26_BF_AKSZ_S_expanded}) is:
$$S= \frac12 \lan c^+,[c,c]\ran$$
It again satisfies the ansatz (\ref{l24_S_FP_BV_min}), now with first two terms vanishing and only the term corresponding to the Lie algebra of gauge symmetry surviving.
\end{itemize}

\begin{remark} 
Alongside the version of the $BF$ theory presented above, where the fields are differential forms on $M$ (of fixed degrees in the classical setting and non-homogeneous forms in AKSZ/BV setting), one can consider the ``canonical'' variant of the $BF$ theory (cf. e.g. \cite{CR,CMRcell}). Here one replaces the form-valued field $B$ (respectively, superfield $\mc{B}$) with the de Rham current (i.e. a linear functional on differential forms, or a distributional form) 
$$\mc{B}^\#=\int_M\lan \mc{B}\stackrel{\wedge}{,} - \ran\quad \in (\Omega^\bt(M))^*\otimes \g^* $$
In particular, the classical $B$ field becomes a de Rham $2$-current $B^\#$ (so that it pairs naturally to the curvature $2$-form). The action of the theory in this setting is $S_{cl}=\lan B^\#,F_A \ran$ (for the classical action) or 
\begin{equation}\label{l26_S_canonical}
S=\lan \mc{B}^\#,d\mc{A}+\frac12[\mc{A},\mc{A}]\ran
\end{equation} (the BV action), with $\lan,\ran$ the canonical pairing between the currents and differential forms. Note that, since the action is linear in $\mc{B}^\#$, one doesn't run into the problem of regularizing products of currents. 
Also note that in the canonical setting the space of fields is, by construction, canonically identified with the shifted cotangent bundle 
$$\FF^\mr{can}=T^*[-1]\left(\Omega^\bt(M,\g)[1]\right) \quad = \Omega^\bt(M,\g)[1]\oplus (\Omega^\bt(M,\g))^*[-2]$$
with BV $2$-form being the standard symplectic structure of the cotangent bundle. The bizarre feature of this setup is that the field $\mc{B}^\#$ has the ``wrong'' functoriality: instead of being able to restrict it to submanifolds of $M$, one can extend it from submanifolds.
\end{remark}

\section{Applications}

\subsection{Cellular $BF$ theory}
The combinatorial construction of (non-abelian) $BF$ theory on cochains of a triangulation of a manifold, via BV pushforward from differential forms, was developed in \cite{SimpBF,DiscrBF} and developed further in \cite{CMRcell} (in particular, in the latter work, the theory is extended to general CW decompositions). 

This construction
associates to any triangulation\footnote{Or a more general CW decomposition; for simplicity, we talk about triangulations here,} $X$ of a manifold $M$ a BV package $(\FF_X,S_X)$, where fields are pairs of a cochain and a chain of $X$, 
and the action $S_X$ satisfies the quantum master equation, in such a way that if $X'$ is a subdivision of $X$ (then we say that $X$ is an \textit{aggregation} of $X'$), then the action $S_X$ is the BV pushforward of $S_{X'}$ along the odd-symplectic fibration $\FF_{X'}\ra \FF_X$. In other words, the action on a sparser complex $X$ can be obtained from the action on a denser complex $X'$ by integrating out (in the BV integral sense) the redundant fields. Also, the cellular action $S_X$ converges, in an appropriate sense, in the limit of a dense refinement of $X$, to the usual BV action of the $BF$ theory on differential forms on $M$ (\ref{l26_S_BF_AKSZ}).  

The idea here is that, having such a consistent system of BV packages for different triangulations of $M$, one can replace the problem of calculating the partition function or expectation values of observables of the $BF$ theory via path integral by calculating respective quantities within the cellular theory, by a finite-dimensional integral. By consistency with cellular aggregations, the result will be independent of $X$ and can serve as a consistent replacement (in a sense, a measure-theoretic definition) for the path integral.

\subsubsection{Abstract $BF$ theory associated to a dgLa}\label{sss: abstr BF theory}
One can associate to a differential graded Lie algebra $V^\bt,d,[-,-]$, subject to the \textit{unimodularity} condition $\mr{Str}_V[x,-]=0$ for any $x\in V$, the following BV package (the ``abstract $BF$ theory'' in the terminology of \cite{SimpBF,DiscrBF}):
\begin{itemize}
\item The space of fields: 
\begin{equation}\label{l26_abstr_BF_space_of_fields}
\FF=T^*[-1]V[1]=V[1]\oplus V^*[-2]
\end{equation} 
Let $e_a$ be a basis in $V$ and $e^a$ the dual basis in $V^*$. The superfields  parameterizing $\FF$ are: 
\begin{equation}\label{l26_abstr_BF_superfields}
A=\sum_a e_a A^a
, \qquad B=\sum_a B_a e^a
\end{equation}
Here $A^a$ is a coordinate on $V[1]$ corresponding to $e_a$ and $B_a$ is a coordinate on $V^*[-2]$ corresponding to $e_a$. The ghost numbers are $\gh (A^a)=1-|a|$, $\gh (B_a)=-2+|a|$ where $|a|$ is the degree (according to the grading on $V^\bt$) of the basis vector $e_a$. Note that $A\in V\otimes (V[1])^*\subset \mr{Fun}(\FF,V)$ corresponds to the degree-shifted identity map $\mr{id}: V[1]\ra V$. Similarly, $B\in (V^*[-2])^*\otimes V^*\subset \mr{Fun}(\FF,V^*)$ corresponds to the degree shifted identity $\mr{id}:V^*[-2]\ra V^*$. We understand the superfields $A,B$ as the generating functions for coordinates $A_a,B^a$ on $\FF$, valued in $V$ and $V^*$,  respectively: $(A,B):\FF\ra V\oplus V^*$.
\item The odd-symplectic form $\omega$ on $\FF$ is constructed out of the canonical pairing $\lan,\ran$ between $V$ and $V^*$:
\begin{equation}\label{l26_abstr_BF_omega}
\omega=\lan \delta B \stackrel{\wedge}{,}\delta A\ran = \sum_a (-1)^{|a|}\delta B^a \wedge \delta A_a\qquad \in \Omega^2(\FF)_{-1}
\end{equation}
\item The BV action is built out of the dgLa operations $d, [-,-]$ on $V$:
\begin{equation} \label{l26_S_abstr_BF}
S=\lan B, dA+\frac12 [A,A]\ran \quad = \sum_{a,b}  d^a_{b} B_a  A^b + \sum_{a,b,c} (-1)^{(|b|+1)|c|} \frac12 f^a_{bc} B_a A^b A^c \qquad \in C^\infty(\FF)_0
\end{equation}
where $d^a_b\lan e^a, d e_b \ran$ are the structure constants of the differential $d$ and $f^a_{bc}=\lan e^a,[e_b,e_c] \ran$ are the structure constants of the Lie bracket $[-,-]$ on $V$.
\item We fix the Berezinian on $\FF$ to be the coordinate Berezinian 
\begin{equation}\label{l26_abstr_BF_mu}
\mu=c\cdot \prod_a \DD A^a \DD B_a
\end{equation}
Here $c\in \mathbb{C}$ is a 
scaling factor which we leave unspecified for the moment.\footnote{In the geometric context of CW decompositions, consistency with CW aggregations impose certain conditions on the scaling factor $c$ and ultimately lead to the appearance of a certain power $\hbar$ and a $\bmod 8$ complex phase in $c$ (which leads to a $\bmod 16$ phase in the expression $\mu^{1/2}e^{\frac{i}{\hbar}S}$ and ultimately in the partition function), see \cite{CMRcell} and the Remark \ref{l27_rem_xi} below.} This Berezinian is compatible with the odd-symplectic structure $\omega$.
\end{itemize}
The pair $\omega,\mu$ induces the BV Laplacian 
\begin{equation}
\Delta=\lan \frac{\delta}{\delta A},\frac{\delta}{\delta B} \ran \quad = \sum_a (-1)^{|a|}\lan \frac{\delta}{\delta A^a},\frac{\delta}{\delta B_a} \ran
\end{equation}
acting on functions on $\FF$. Action (\ref{l26_S_abstr_BF}) satisfies the quantum master equation
\begin{equation}\label{l26_abstr_BF_QME}
\Delta e^{\frac{i}{\hbar}S} = 0 \quad \Leftrightarrow\quad \frac12\{S,S\}-i\hbar \Delta S=0
\end{equation}
Note that 
$$ \frac12\{S,S\}-i\hbar \Delta S= \lan B, d(dA)+ \left(\frac12 d[A,A]+[A,dA]\right)+\frac12 [A,[A,A]] \ran -i\hbar (\underbrace{\mr{Str}_V d}_{=0} +\mr{Str}_V [A,-]) $$
Thus, the quantum master equation, by inspecting the terms of orders $BA, BAA, BAAA$ and $\hbar A$ in the r.h.s. above, is equivalent to the four identities on the structure operations $d,[-,-]$ in $V$:
\begin{enumerate}[(i)]
\item $d^2=0$,
\item Leibniz identity
\item Jacobi identity
\item unimodularity property $\mr{Str}[x,-]=0$. (Note that $\mr{Str}d=0$ is automatic for degree reasons.)
\end{enumerate}
Thus one can say that the action $S$ as defined by (\ref{l26_S_abstr_BF}) is the generating structure for the structure constants of the structure operations in $V$, whereas the quantum master equation (\ref{l26_abstr_BF_QME}) is the generating function for the structure relations in $V$.

\subsubsection{Effective action induced on a subcomplex}\label{sss: eff action for abstr BF theory}
Let $V$ be a dgLa as before and let $\iota:V'\hra V$ be a subcomplex of the complex $V^\bt,d$ (with no compatibility condition with the bracket $[-,-]$), such that the inclusion induces an isomorphism on cohomology $\iota_*: H^\bt(V')\stackrel{\sim}{\ra} H^\bt(V)$
-- in this case we call $V'$ a deformation retract of $V$. Next, promote $\iota $ to a package of ``induction data'' $(\iota,p,K)$ from $V$ to $V'$.

\begin{definition}\label{def: ind data}
For $V$, $V'$ two cochain complexes, we call a triple of maps $\iota: V'\hra V$, $p:V\proj V'$, $K:V^\bt\ra V^{\bt-1}$ a \emph{(set of) induction data from $V$ to $V'$}, or the \emph{HPT\footnote{For ``homological perturbation theory''} data} if the following properties hold:
\begin{enumerate}
\item $\iota$ and $p$ are chain maps, $p$ is a left inverse for $\iota$, i.e., $p\circ\iota=\mr{id}_{V'}$;
\item $K$ is a chain homotopy between the identity on $V$ and the projection to $V'$, i.e., $dK+Kd=\mr{id}-\iota\circ  p$.
\item $K$ satisfies additionally $K\iota=0$, $pK=0$ , $K^2=0$.
\end{enumerate}
We denote such a triple by $V\stackrel{(\iota,p,K)}{\rightsquigarrow} V'$.
\end{definition}

For $V,V'$ two cochain complexes, the induction data $V\stackrel{(\iota,p,K)}{\rightsquigarrow} V'$ exists if and only if $V'$ is a deformation retract of $V$. If it is, the space of all possible induction data, inducing a fixed isomorphism on cohomology $H^\bt(V')\stackrel{\sim}{\ra}H^\bt(V)$, is \textit{contractible}.

Induction data $V\stackrel{(\iota,p,K)}{\rightsquigarrow} V'$ induces the splitting 
\begin{equation}\label{l26_V=V'+V''}
V=\iota(V')\oplus V''
\end{equation}
with $V''=\ker p$, -- a splitting of $V$ into subcomplexes -- the image of the retract $V'$ and an \emph{acyclic}\footnote{Acyclicity follows from existence of the chain homotopy $K$: for any $x''\in V''$ a cocycle, we have $x''=(dK+Kd)x''=d(Kx'')$. Thus, a cocycle in $V''$ is automatically a coboundary.} complement $V''$. Moreover, we have the following refinement of this splitting -- the (abstract) Hodge decomposition:
\begin{equation}\label{l26_Hodge}
V=\iota(V')\oplus \underbrace{d(V'') \oplus K(V'')}_{V''}
\end{equation}
I.e. any element of $V$ splits into a part coming from the retract $V'$, a $d$-exact element in $V''$ and a $K$-exact element of $V''$. 
A special case of this is the ordinary Hodge-de Rham decomposition, splitting the space of differential forms on a compact manifold into harmonic forms, exact forms and coexact forms.

We build the spaces of fields $\FF,\FF',\FF''$ associated to the complexes $V,V',V''$ by the doubling construction (\ref{l26_abstr_BF_space_of_fields}), i.e.,
$$\FF=T^*[-1]V[1],\quad \FF'=T^*[-1]V',\quad \FF''=T^*[-1]V''[1]$$ 
Then the splitting (\ref{l26_V=V'+V''}) induces a splitting of the spaces of fields,\footnote{More pedantically, we should write $\FF=(\iota\oplus p^*)\FF'\;\oplus \FF''$, taking care of the way $\FF'$ is embedded into $\FF$.}
\begin{equation}\label{l26_F=F'+F''}
\FF=\FF'\oplus \FF''
\end{equation}
This splitting is compatible with the odd-symplectic forms: $\omega=\omega'\oplus \omega''$. Furthermore, we have a distinguished Lagrangian subspace $\LL\subset \FF''$, constructed out of the data of the Hodge decomposition (\ref{l26_Hodge}) as the conormal bundle to the last term of the decomposition in $\FF''=T^*[-1]V''[1]$:
\begin{equation}\label{l26_L}
\LL=N^*[-1]\left(\mr{im}(K)[1]\right)=\mr{im}(K)[1]\oplus \mr{im}(K^*)[-2]\qquad \subset\quad  T^*[-1]V''[1]=\FF''
\end{equation}
where $K^*:(V')^*\ra (V')^*$ is the linear dual of $K$. 

Note that the splitting (\ref{l26_F=F'+F''}) and the Lagrangian (\ref{l26_L}) are precisely the gauge-fixing data needed to define the BV pushforward from $\FF$ to $\FF'$ (see Section \ref{ss: fiber BV int}). We construct the effective theory (the induced BV package\footnote{ The BV package we have in mind is
$(\FF',\omega',S',\mu')$ with standard $\omega'$ and $\mu'$, as  in (\ref{l26_abstr_BF_omega},\ref{l26_abstr_BF_mu}), with $V$ replaced by $V'$, and with nontrivial $S'$, constructed by the BV pushforward from $V$.
})
on $\FF'$ with the action $S'\in C^\infty(\FF')_0[[\hbar]]$ defined by the fiber BV integral
\begin{equation}\label{l26_fiber_BV_int}
e^{\frac{i}{\hbar}S'}\sqrt{\mu'}=\int_{\LL\subset \FF''} e^{\frac{i}{\hbar}S}\sqrt\mu
\end{equation}

Calculating this integral by the stationary phase formula yields the following Feynman diagram expansion for $S'$:
\begin{equation}\label{l26_S'}
S'(A',B';\hbar)=\sum_{\Gamma}\frac{(-i\hbar)^l}{|\mr{Aut}(\Gamma)|}\;\Phi(\Gamma)
\end{equation}
Here $\Gamma$ runs over connected oriented graphs with leaves, where all internal vertices are trivalent, with $2$ incoming and $1$ outgoing half-edges; $l\in\{0,1\}$ is the number of loops in $\Gamma$. Note that the restrictions on orientations at vertices of $\Gamma$ imply that a connected $\Gamma$ has either zero loops (i.e. $\Gamma$ is a binary rooted tree) or one loop (i.e. $\Gamma$ is an oriented cycle with several binary trees with roots placed on the cycle). Feynman weights $\Phi(\Gamma)$ are polynomials in $A',B'$ (the superfields for $\FF'$, constructed as in (\ref{l26_abstr_BF_superfields})) and are constructed by the following Feynman rules.
\begin{itemize}
\item For $\Gamma$ a binary rooted tree $\vcenter{\hbox{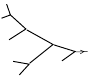}}$, we decorate the leaves with $\iota(A')$, internal vertices with $[-,-]$, internal edges with $-K$ and the root with $\lan B',p(-)\ran$. Thus, we read the rooted tree as an iterated operation with inputs on the leaves, proceeding from the leaves to the root we compute the commutators in the vertices and minus the chain homotopy on the edges, and in the very end we pair the result read off at the root with $B'$. For example, for the graph
$$ \Gamma=\vcenter{\hbox{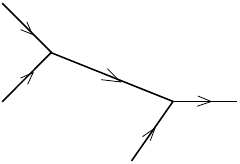}}$$ 
we decorate it as follows
$$\vcenter{\hbox{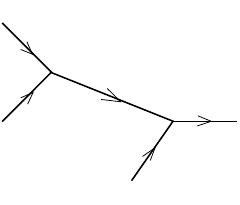}}$$
Thus, for the Feynman weight we have
$$ \Phi\left(\vcenter{\hbox{\input{l26_tree.pdftex_t}}}\right) = \lan B', p[\iota(A'),-K[\iota(A'),\iota(A')]] \ran $$
\item For $\Gamma$ a one-loop graph $\vcenter{\hbox{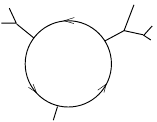}}$, we put the decorations as above on leaves, edges and vertices, and we cut the loop at an arbitrary point and compute the supertrace over $V$ of the endomorphism of $V$, depending parametrically on $A'$, as constructed using the Feynman rules for the trees. E.g., for 
$$\Gamma=\vcenter{\hbox{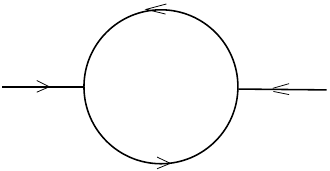}} $$
we decorate it as 
$$ \vcenter{\hbox{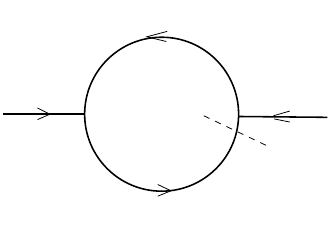}} $$
where the dashed line indicates the place where we cut the loop (the resulting expression does not depend on the place of the cut by the cyclic property of the supertrace). The corresponding Feynman weight is:
$$\Phi\left(  \vcenter{\hbox{\input{l26_loop.pdftex_t}}} \right)= \mr{Str}\left( -K[\iota(A'),-K[\iota(A'),\bullet]]\right)$$
where the bullet stands for the argument of the endomorphism we are computing the supertrace of.
\end{itemize}

\begin{remark}
Note that if $V'\hra V$ happens to be a dg Lie subalgebra (not just a subcomplex), then the Feynman weights of all the tree diagrams vanish (since $K$ of a commutator of images of two elements of $V'$ vanishes for $V'$ a subalgebra), except for the only two diagrams not containing any internal edges, 
\begin{equation}\label{l26_simple_diagrams}
\vcenter{\hbox{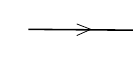}}\qquad \qquad,\qquad  \vcenter{\hbox{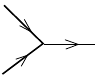}}
\end{equation}
The first diagram here requires a special decoration $d$ for the edge connecting the two leaves; thus its Feynman weight is $\lan B', dA'\ran$; the weight of the second diagram, by the rules above, is $\lan B',[A',A']\ran$. Thus, in the case when $V'$ is a dg Lie subalgebra, the $\bmod\hbar$ part of the effective action (\ref{l26_S'}) coincides with the action of abstract $BF$ theory (\ref{l26_S_abstr_BF}) associated to the dgLa $(V',d,[-,-])$. When $V'$ is not a subalgebra, contributions of higher trees to $S'$ are generally nontrivial and measure, in a sense, the failure of the subcomplex $\iota(V')\subset V$ to be closed under commutators.
\end{remark}

The action (\ref{l26_S'}) satisfies the following ansatz (called a $BF_\infty$ theory in \cite{DiscrBF}), generalizing the abstract $BF$ ansatz (\ref{l26_S_abstr_BF}):
\begin{equation}\label{l26_S_BF_infty}
S'(A',B';\hbar)=\sum_{n\geq 1} \frac{1}{n!} \lan B',l_n(A',\ldots,A')\ran -i\hbar\sum_{n\geq 2} \frac{1}{n!}q_n(A',\ldots,A')
\end{equation}
where $l_n:\wedge^n V'\ra V'$ are graded anti-symmetric $n$-linear operations of degree $2-n$ (the $L_\infty$ operations on $V'$) and $q_n:\wedge^n V'\ra \RR$ are graded anti-symmetric $n$-linear operations of degree $-n$ valued in numbers (the ``unimodular'' or ``quantum'' operations on $V'$). Quantum master equation for the action (\ref{l26_S_BF_infty}), which is satisfied automatically by construction (\ref{l26_fiber_BV_int}), by virtue of the BV-Stokes' theorem for the BV pushforward (item (\ref{l23_BV_Stokes'_cor_i}) of Corollary \ref{l23_BV_Stokes'_cor}), 
is equivalent to a collection of (nonhomogeneous) quadratic relations on the operations $\{l_n,q_n\}$:
\begin{align}\label{l26_uL_inf_relations_1}
\sum_{r+s=n}\frac{1}{r!s!}l_{r+1}(\underbrace{A',\ldots,A'}_r,l_s(\underbrace{A',\ldots,A'}_s))=0,\\
\label{l26_uL_inf_relations_2}
\frac{1}{n!}\mr{Str}\;l_{n+1}(\underbrace{A',\ldots,A'}_n,\bt)+\sum_{r+s=n}\frac{1}{r!s!}q_{r+1}(\underbrace{A',\ldots,A'}_r,l_s(\underbrace{A',\ldots,A'}_s))=0
\end{align}
For each $n\geq 1$. 
Or, pictorially, the relations are: 
$$
\sum_{r+s=n} 
\vcenter{\hbox{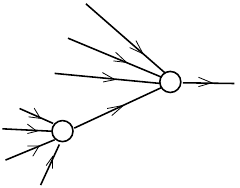}} =0\quad ,
\qquad
\vcenter{\hbox{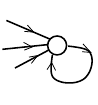}} \;\;+\sum_{r+s=n}\vcenter{\hbox{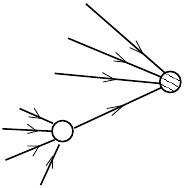}} =0
$$
where we put the elements $x_1,\ldots,x_n\in V'$ as inputs on the incoming leaves and skew-symmetrize over the ways to distribute them between the leaves.

\begin{definition} A graded vector space $V'$ endowed with polylinear operations $\{l_n,q_n\}$ satisfying the relations (\ref{l26_uL_inf_relations_1},\ref{l26_uL_inf_relations_2}) is called a \emph{quantum $L_\infty$ algebra} (terminology of \cite{SimpBF,DiscrBF}) or \emph{unimodular $L_\infty$ algebra} (terminology introduced in \cite{Granaker}).\footnote{Equivalently, the structure of a unimodular $L_\infty$ algebra on a graded vector space $V'$ can be summarized by saying that the shifted space $V'[1]$ is endowed with a cohomological vector $Q$ vanishing at the origin and a Berezinian $\mu=\rho\cdot \mu_\mr{coord}$ preserved by the Lie derivative along $Q$ (equivalently, the divergence of $Q$ w.r.t. $\mu$ vanishes); $Q$ and $\rho$ can be given by formal power series in coordinates. The relation to the previous definition is that the Taylor expansion of $Q$ at the origin yields operations $l_n$ and the density of the Berezinian $\mu$ is $\rho=\exp(\sum_n\frac{1}{n!}q_n)\in \mr{Fun}(V'[1])$. Relations $Q^2=0$, $\mr{div}_\mu Q=0$ expand into the relations on operations (\ref{l26_uL_inf_relations_1},\ref{l26_uL_inf_relations_2}).}
\end{definition}
A unimodular $L_\infty$ algebra is in particular an $L_\infty $ algebra (by forgetting the quantum operations $q_n$). Thus a unimodular $L_\infty$ algebra is a certain enrichment of the classical structure of an $L_\infty$ algebra. Also, a unimodular dg Lie algebra $(V,d,[-,-])$ is a special case of a unimodular $L_\infty$ algebra (with all operations except $l_1,l_2$ vanishing).

Relations (\ref{l26_uL_inf_relations_1}) are the homotopy Jacobi identities of the $L_\infty$ algebra $V',\{l_n\}$ and the non-homogeneous relations (\ref{l26_uL_inf_relations_2}) are the homotopy unimodularity identities.

Notice that, by BV-Stokes' theorem (item (\ref{l23_BV_Stokes'_cor_iii}) of Corollary \ref{l23_BV_Stokes'_cor}), a change of the induction data $(\iota,p,K)$ and the associated change of gauge-fixing for the BV pushforward (\ref{l26_fiber_BV_int}) induce a change of the action $S'$ by a canonical transformation $S'\mapsto S'+\{S',R'\}-i\hbar R'$ with $R'$ satisfying the ansatz similar to (\ref{l26_S_BF_infty}).\footnote{More explicitly, $R'$ has the form $R'=\sum_{n}\frac{1}{n!}\lan B',\lambda_n(A',\ldots,A') \ran-i\hbar \sum_n \frac{1}{n!}\tau_n(A',\ldots,A')$ where the coefficient polylinear maps $\lambda_n:\wedge^n V'\ra V'$, $\tau_n:\wedge^n V'\ra\RR$ are of degrees $1-n$ and $-1-n$, respectively.} 
This canonical transformation can be interpreted, via the correspondence (\ref{l26_S_BF_infty}) of solutions of QME of $BF_\infty$ type and unimodular $L_\infty$ algebras, as an isomorphism (in the appropriate sense) of unimodular $L_\infty$ algebras $(V',\{l_n,q_n\})\sim (V',\{\til l_n,\til q_n\})$.

We have the following diagram:
$$
\begin{CD}
\mbox{u. dgLa }(V,d,[-,-]) @>>> \left(\FF=T^*[-1]V[1],\;\; S\mbox{ built out of }d,[-,-] \right)\\
@V\mbox{homotopy transfer}VV @VV\mbox{BV pushforward}V \\
uL_\infty\mbox{ algebra }(V',\{l_n\},\{q_n\}) @<<< \left(\FF'=T^*[-1]V'[1],\;\; S'\mbox{ built out of }\{l_n\},\{q_n\}\right)
\end{CD}
$$
Here in the left column we have algebraic structures and in the right column we have BV packages. Horizontally, going from left to right, we associate a BV package to an algebraic structure or, going right-to-left, we read off the operations of the algebraic structure from the Taylor expansion of the action. The vertical arrow on the right is the BV pushforward/calculation of the effective BV action. Starting from a unimodular dgLA $V,d,[-,-]$, associating to it an abstract $BF$ theory, calculating the effective BV action on a subcomplex $V'$ and then reading off the operations of the unimodular $L_\infty$ structure on $V'$ from the Taylor coefficients of the effective action, we obtain an appropriate notion of homotopy transfer of a unimodular dgLa to a subcomplex. The transferred operations are given by the Feynman diagrams appearing in the stationary phase expansion of $S'$ (understood as polylinear operations on $V'$). In particular, operations $l_n$ are given by sums over binary rooted trees with $n$ leaves, and this formula coincides with the Lie version of the Kontsevich-Soibelman homotopy transfer formula for $A_\infty$ algebras \cite{KS}.

\begin{remark}\label{l26_rem_induction_for_BF_infty}
Note that instead of starting with a unimodular dgLa structure on $V$, we could start with a unimodular $L_\infty$ structure $(V,\{l_n^V\},\{q_n^V\})$. Then we associate a BV package to it with the action given by (\ref{l26_S_BF_infty}) and induce the effective action on a subcomplex via a BV pushforward. In this case, we have more Feynman diagrams for $S'$ -- we have vertices with $n\geq2$ inputs and one output decorated by $l_n^V$ and vertices with $n\geq 2$ inputs and no outputs decorated by $q_n^V$. This still results in the connected diagrams being at most one-loop. More precisely, now there are three types of diagrams for $S'$: 
\begin{enumerate}
\item rooted trees (not necessarily binary), 
\item diagrams consisting of a cycle with several trees rooted in it, 
\item trees with root replaced by a vertex decorated with a quantum operation $q_n^V$.
\end{enumerate}
All the vertices, except for the root in the last case are decorated with the classical operations $l_n^V$.
The resulting effective action $S'$ again satisfies the $BF_\infty$ ansatz (thus, the $BF_\infty$ ansatz reproduces itself under BV pushforwards); one can read off from it the induced (homotopy transferred) structure of a unimodular $L_\infty$ algebra on $V'$.
\end{remark}

\marginpar{\LARGE{Lecture 27, 12/05/2016.}}
\subsubsection{Geometric situation}
Fix a manifold $M$ with endowed a triangulation $X$ and fix a unimodular Lie algebra of coefficients $\g$. 

Non-abelian $BF$ theory on $M$ (most directly, in its canonical version (\ref{l26_S_canonical})) can be viewed as the abstract $BF$ theory associated to the dg Lie algebra of $\g$-valued differential forms on $M$, $V=\Omega^\bt(M,\g)=\Omega^\bt(M)\otimes\g$, with differential $d: \alpha\otimes x\mapsto d\alpha\otimes x$ and the Lie bracket $[\alpha\otimes x,\beta\otimes y]=(\alpha\wedge\beta)\otimes [x,y]_\g$ for $\alpha,\beta\in \Omega^\bt(M)$ and $x,y\in\g$. 

We want to apply the construction of effective action of Section \ref{sss: eff action for abstr BF theory} for the complex of $\g$-valued cellular cochains of the triangulation $V'=C^\bt(X,\g)$, viewed as a subcomplex of $V$.

\textbf{Induction data.}
The inclusion $\iota:V'\hra V$ is the extension by $\g$-linearity of the realization of the cochains of triangulation by Whitney elementary forms \cite{Whitney}, $\iota: C^\bt(X)\xra{\sim}\Omega_\mr{Whitney}(M,X)\subset \Omega^\bt(M)$. Explicitly, for $\Delta^N$ the standard $N$-simplex with barycentric coordinates $t_0,\ldots,t_N\geq 0$ subject to $t_0+\cdots+t_N=1$, one assigns to the $k$-dimensional face $[i_0\cdots i_k]$ (the convex hull of 
vertices $\{v_{i_0},\ldots,v_{i_k}\}$; 
here $v_i=(0,\ldots,t_i=1,\ldots,0)$ is the $i$-th vertex of the simplex) the elementary Whitney form 
\begin{equation}\label{l27_Whitney}
\chi_{i_0\cdots i_k}: = \sum_{j=0}^k (-1)^j t_{i_j} dt_{i_0}\wedge\cdots \widehat{dt_{i_j}}\cdots \wedge dt_{i_k}\quad \in \Omega^k(\Delta^N)
\end{equation}
Then we glue elementary Whitney forms over the simplices of $X$: for $\sigma$ a $k$-simplex of $X$, we construct a piecewise-linear $k$-form $\chi_\sigma$ on $M$, supported on the star of $\sigma$. It is defined by $\left.\chi_\sigma\right|_{\sigma'}=\chi_ \sigma^{(\sigma')}$ for any simplex $\sigma'$ of $X$ containing $\sigma$ as a sub-simplex (a face of arbitrary codimension). Here $\chi_ \sigma^{(\sigma')}$ is the elementary form on $\sigma'$ associated to its face $\sigma$ by formula  (\ref{l27_Whitney}). The $\RR$-span of the forms $\chi_\sigma$ is a subcomplex in the space of piecewise-linear continuous\footnote{They are continuous in the sense that the pullbacks to the simplices $\sigma$ are well-defined, however the components of the form in the normal direction to a simplex may have a discontinuity.} forms on $M$, which we call $\Omega^\bt_\mr{Whitney}(M,X)$. This complex is isomorphic to $C^\bt(X)$, with the isomorphism given by the map 
\begin{equation}\label{l27_Whitney_iota}
\iota: e_\sigma \mapsto \chi_\sigma
\end{equation}
where $e_\sigma$ is the basis cochain associated to the simplex $\sigma$. In particular, the map (\ref{l27_Whitney_iota}) intertwines the cellular coboundary operator in $C^\bt(X)$ and the de Rham operator on forms.

We also have a projection $p:V\proj V'$ coming from the extension by $\g$-linearity of the Poinar\'e integration map $p:\Omega^\bt(M)\ra C^\bt(X)$ which integrates a form over the simplices of $X$:
\begin{equation}
p: \alpha\mapsto \sum_{\sigma\subset X} \left(\int_\sigma \alpha\right) \cdot e_\sigma
\end{equation}
for any $\alpha\in \Omega^\bt(M)$. It is easy to see that $p$ is a chain map and that $p\circ \iota=\mr{id}_{C^\bt(X)}$. The latter property amounts to the fact that the form (\ref{l27_Whitney}) is a volume form of total volume $1$ on the face $[i_0\cdots i_k]$ and integrates to zero on all other faces of $\Delta^N$ (in fact, it vanishes identically on all faces not containing the face $[i_0\cdots i_k]$).

Finally, we have a chain homotopy $K:V^\bt\ra V^{\bt-1}$ coming by extension by $\g$-linearity of the Dupont's chain homotopy $\Omega^\bt(M)\ra \Omega^{\bt-1}(M)$ \cite{Dupont,Getzler}. The latter is glued out of Dupont's operators on standard simplices. On $\Delta^N$, one first defines the dilation (or ``homothety'') maps
\begin{align}
h_j:\quad  &[0,1]\times \Delta^N &\ra & \qquad\Delta^N
\\ \nonumber & (u;t_0,\ldots,t_N) &\mapsto &\qquad (ut_0,\ldots,1-u+u t_j,\ldots,u t_N)
\end{align}
Map $h_j$ 
pulls the points of $\Delta^N$ towards $j$-th vertex $v_j$ by the factor $u$: if $u=0$, all points go to $v_j$, whereas if $u=1$, $h_j$ does not move the points. Next, one defines the maps $$\phi_j=\pi_*h_j^*:\Omega^\bt(\Delta^N)\ra \Omega^{\bt-1}(\Delta^N)$$ 
Here $\pi: [0,1]\times \Delta^N \ra \Delta^N$ is the projection to the second factor and $\pi_*$ stands for fiber integration. Map $\phi_j$ is the chain homotopy between the identity and evaluation at $j$-th vertex for the forms on $\Delta^N$.
Finally, one defines the operator
\begin{align}
K^{(\Delta^N)}:\quad & \Omega^\bt(\Delta^N) &\ra & \qquad\Omega^{\bt-1}(\Delta^N) \\ \nonumber
& \alpha & \mapsto & \qquad \sum_{k=0}^{N-1}\quad  \sum_{0\leq i_0<\cdots < i_k\leq N} \chi_{i_0\cdots i_k}\wedge h_{i_k}\cdots h_{i_0}\alpha
\end{align}
Next one assembles this operators on individual simplices of $X$ into the operator $K:\Omega^\bt(M)\ra \Omega^{\bt-1}(M)$ by setting $\left. (K\alpha)\right|_\sigma = K^{(\sigma)}(\alpha|_\sigma)$. The resulting operator satisfies the properties 
\begin{equation}
dK+Kd=\mr{id}-\iota\circ p,\quad K^2=0,\quad K\iota=0,\quad pK=0
\end{equation}
See \cite{Getzler} for the proof on $\Delta^N$ (the respective properties on $M$ follow from consistency of Dupont's operator with simplicial face maps).

\textbf{Effective theory on a triangulation.}
Thus, we have an induction data triple (cf. Definition \ref{def: ind data})
$$V=\Omega^\bt(M,\g)\stackrel{(\iota,p,K)}{\rightsquigarrow} V'=C^\bt(X,\g)$$
given by Whitney elementary forms, Poincar\'e integration map and Dupont's chain homotopy operator. This data allows us to define the effective theory, induced from the continuum $BF$ theory (regarded as the BV package associated to the dgLa $\Omega^\bt(V,\g)$ via the construction of Section \ref{sss: abstr BF theory}) on the space 
$$\FF_X=\FF'=C^\bt(X,\g)[1]\oplus C_{-\bt}(X,\g^*)[-2]$$
parameterized by the superfields
$$A_X=\sum_{\sigma\subset X} e_\sigma A^\sigma,\qquad B_X=\sum_{\sigma\subset X} B_\sigma e^\sigma $$
Here $A_\sigma$ is a $\g$-valued coordinate on (the cochain part of) $\FF_X$ of ghost number $1-|\sigma|$ and $B_\sigma$ is a $\g^*$-valued coordinate on (the chain part of) $\FF_X$ of ghost number $-2+|\sigma|$; we denote by $|\sigma|$ the dimension of the simplex $\sigma$.

The effective action $S_X$ is defined by the BV pushforward construction (\ref{l26_fiber_BV_int}), which is in this case a functional integral. Its stationary phase evaluation leads to the expansion (\ref{l26_S'}). Due to the way the gauge-fixing data $(\iota,p,K)$ is assembled local induction data for the simplices of $X$, the values of Feynman diagrams contributing to $S_X$ also split into local contributions of individual simplices. This leads to the splitting of $S_X$ into \textit{universal local building blocks} $\bar{S}_\sigma$ corresponding to the simplices $\sigma$ of $X$:
\begin{equation}\label{l27_S_X_locality}
S_X(A_X,B_X)=\sum_{\sigma\subset X} \bar{S}_\sigma(\left. A_X\right|_\sigma, B_\sigma)
\end{equation}
Here $\left. A_X\right|_\sigma=\sum_{\sigma'\subset \sigma} e_{\sigma'}A^{\sigma'}$. The local building blocks $\bar{S}_\sigma$ depend only on the dimension of the simplex $\sigma$ and do not depend on the combinatorics of $X$ beyond the simplex $\sigma$ itself.
To find the building block for an $N$-simplex, one computes the Feynman diagrams of Section \ref{sss: eff action for abstr BF theory} for a standard simplex $\Delta^N$, using the Whitney/Poincar\'e/Dupont induction data on $\Delta^N$, and subtracts the contributions of positive-codimension faces of $\Delta^N$, setting 
$$\bar{S}_{\Delta^N}:=S_{\Delta^N}-\sum_{\sigma\subset \Delta^N,\sigma\neq \Delta^N}\bar{S}_\sigma$$
Here we understand that we work by induction in the simplex dimension $N$.

\textbf{Local building blocks.}
\begin{itemize}
\item For a $0$-simplex $[0]$ (we understand $0$ as the vertex label), we have 
\begin{equation}\label{l27_Sbar_vertex}
\bar{S}_{[0]}=\frac12 \lan B_{[0]},[A^{[0]},A^{[0]}] \ran
\end{equation}
Here the field component $A^{[0]}$ takes values in $\g$ and $B_{[0]}$ is in $\g^*$, and the ghost numbers are: $\gh\,A^{[0]}=1$, $\gh\,B_{[0]}=-2$. The only contributing diagram is $\vcenter{\hbox{\input{l26_corolla.pdftex_t}}}$.
\item For a $1$-simplex $[01]$ 
$$\vcenter{\hbox{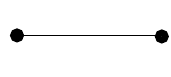}}$$
(with $0$ and $1$ the labels of the two vertices and $[01]$ the label of the top cell), we have
\begin{multline}
\bar{S}_{[01]}=\lan B_{[01]}, \left[A^{[01]},\frac{A^{[0]}+A^{[1]}}{2}\right]+\mathbb{F}([A^{[01]},\bt])\circ (A^{[1]}-A^{[0]}) \ran-\\
-i\hbar\; \tr_\g \mathbb{G}([A^{[01]},\bt])
\end{multline}
Here we introduced two functions 
\begin{align*}
\mathbb{F}(x)=\frac{x}{2}\coth\frac{x}{2}=\sum_{k=0}^\infty \frac{B_{2k}}{(2k)!}x^{2k}=1+\frac{x^2}{12}-\frac{x^4}{720}+\cdots,
\\ 
\mathbb{G}(x)=\log\left(\frac{2}{x}\sinh\frac{x}{2}\right)=\sum_{k=1}^\infty \frac{B_{2k}}{(2k)\cdot (2k)!}x^{2k}=\frac{x^2}{2\cdot 12}-\frac{x^4}{4\cdot 720}+\cdots
\end{align*}
where $B_0=1,B_1=-\frac12,B_2=\frac16,B_3=0,B_4=-\frac{1}{30},\ldots$ are Bernoulli numbers. Here the contributing diagrams are ``branches'' with $2k+1$ leaves and ``wheels'' with $2k$ leaves
\begin{equation}\label{l27_branch_and_wheel}
\vcenter{\hbox{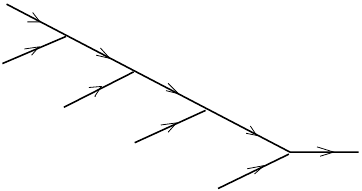}},\qquad \vcenter{\hbox{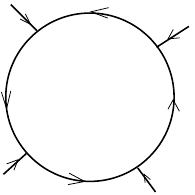}}
\end{equation}
plus the two simple diagrams (\ref{l26_simple_diagrams}). Explicit computation of these diagrams \cite{SimpBF,DiscrBF}, using Dupont's operator on $1$-simplex, leads to the result above with the Bernoulli numbers.
\item For a simplex  $\Delta^N$ of general dimension the building block has the following structure:
\begin{multline}\label{l27_Sbar_ansatz}
\bar{S}_{\Delta^N}=\sum_{k\geq 1}\sum_{\Gamma_0}\sum_{e_1,\ldots,e_k\subset \Delta^N} \frac{1}{|\mr{Aut}(\Gamma_0)|} C_{\Gamma_0;e_1, \ldots, e_k} \lan B_{\Delta^N},\mr{Jacobi}_{\Gamma_0}(A^{e_1},\ldots,A^{e_k})\ran -\\
-i\hbar \;\sum_{k\geq 2}\sum_{\Gamma_1}\sum_{e_1,\ldots,e_k\subset \Delta^N} \frac{1}{|\mr{Aut}(\Gamma_1)|} C_{\Gamma_1;e_1, \ldots, e_k}\; \mr{Jacobi}_{\Gamma_1}(A^{e_1},\ldots,A^{e_k})
\end{multline}
Here we are summing over binary rooted trees $\Gamma_0$ with $k\geq 1$ leaves colored by sub-simplices $e_1,\ldots,e_k$ of $\Delta^N$ and over trivalent 1-loop graphs (of the type  $\vcenter{\hbox{\input{l26_loop_generic.pdftex_t}}}$, i.e., a collection of binary trees with roots on the cycle), with $k\geq 2$ leaves also colored by sub-simplices $e_1,\ldots,e_k$ of $\Delta^N$. For a tree $\Gamma_0$, $\mr{Jacobi}_{\Gamma_0}:\g^{\otimes k}\ra \g$ stands for the nested commutator in $\g$ with nesting prescribed by the combinatorics of the tree $\Gamma_0$. Likewise, $\mr{Jacobi}_{\Gamma_1}:\g^{\otimes k}\ra \RR$ stands for the trace over $\g$ of a nested commutator prescribed by the graph $\Gamma_1$ (e.g., for a wheel with $k$ spikes, as in (\ref{l27_branch_and_wheel}), we have $\mr{Jacobi}_{\Gamma_1}(A^{e_1},\ldots,A^{e_k})=\tr_\g [A^{e_1},[A^{e_2},\cdots [A^{e_k},\bt]\cdots]]$). Coefficients $C_{\Gamma_0;e_1,\ldots,e_k},C_{\Gamma_1;e_1,\ldots,e_k}$ are certain real\footnote{
Structure constants for trees $\Gamma_0$ are rational by construction. 
Also, all structure constants can be made rational within the approach of \cite{CMRcell}, see Remark \ref{l27_rem_CMRcell},
 i.e. by applying a canonical transformation with a local generator to the action $S_X$. 
} structure constants depending on the dimension of the simplex $N$, on the graph $\Gamma$ and on the combinatorics of the $k$-tuple of sub-simplices $e_1,\ldots,e_k$ of $\Delta^N$. For example, from explicit computation of Feynman diagrams \cite{SimpBF,DiscrBF}, we have
$$
C\left(\vcenter{\hbox{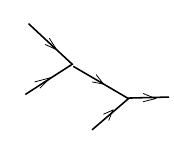}}\right) =\left\{\begin{array}{l}
\pm \frac{|e_1|!\,|e_2|!\,|e_3|!}{(|e_1|+|e_2|+1)\cdot (N+2)!} \\ 0
\end{array}\right.
,\qquad
C\left(\vcenter{\hbox{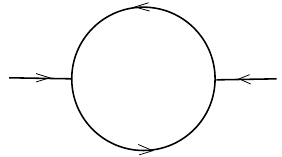}}\right)=\left\{\begin{array}{l}
\pm \frac{1}{(N+1)^2 (N+2)} \\ 0
\end{array}\right.
$$
where the sign and the vanishing/nonvanishing of the structure constant depends on the combinatorics of the tuple of sub-simplices in $\Delta^N$.
\end{itemize}

\begin{remark}
One can assemble the building blocks (\ref{l27_Sbar_ansatz}) using (\ref{l27_S_X_locality}) for $X$ any simplicial complex (not necessarily a triangulation of a manifold) and the result will be a solution of quantum master equation. This corresponds to replacing the de Rham algebra $\Omega^\bt(M,\g)$ in the construction of the theory from which we induce the action on $X$, with Sullivan's cdga of piecewise-polynomial forms on $X$, tensored with $\g$. We call the BV package $(\FF_X, S_X)$ the simplicial (or cellular) $BF$ theory.
\end{remark}

\begin{remark}\label{l27_rem_CMRcell}
Direct computation of Feynman diagrams for a simplex $\Delta^N$ with $N\geq 2$ is rather difficult, especially in the case of 1-loop diagrams (in particular, the latter are infinite-dimensional supertraces which need be regularized appropriately). One shortcut is to compute only the tree diagrams directly and derive the values of 1-loop diagrams by imposing the quantum master equation on $S_{\Delta^N}$. An altogether different approach was developed in \cite{CMRcell}: rather than obtaining the theory on  a triangulation $X$ as a BV pushforward from the continuum theory, one may ask whether a collection of structure constants $C$ exists such that the building blocks defined by the formula (\ref{l27_Sbar_ansatz}), when summed up over simplices of an arbitrary simplicial complex $X$ as in (\ref{l27_S_X_locality}), yield a solution $S_X$ of the quantum master equation. The theorem of \cite{CMRcell} (Theorem 8.6) ensures that the answer is positive\footnote{The construction is by induction in skeleta of $X$. For a zero-skeleton, we have a solution of quantum master equation given as a sum of terms (\ref{l27_Sbar_vertex}). Going by induction from $\mr{sk}_{k-1}X$ to $\mr{sk}_k X$, we extend the solution of QME on $(k-1)$-skeleton to the $k$-skeleton by obstruction theory: we first add the term bilinear in $A$ and $B$, corresponding to the cellular coboundary operator going from $(k-1)$-cells to the newly adjoined $k$-cells; then we successively correct for the error in the QME by adding higher and higher correction terms to the action.
} and, in the appropriate sense, unique (up to canonical transformations with generators satisfying 
an ansatz similar to (\ref{l27_S_X_locality},\ref{l27_Sbar_ansatz})).\footnote{To have uniqueness up to canonical transformations, one imposes two initial ``conditions'' on the building blocks: the building blocks for $0$-cells are fixed to be given by (\ref{l27_Sbar_vertex}); the quadratic part of $S_X\bmod \hbar$ must be $\lan B_X, dA_X \ran$.} Moreover, in this approach Dupont's chain homotopy and Whitney forms are not needed and one may allow $X$ to be an arbitrary regular cellular complex, with no restriction on cells being simplices. In this case the structure constants $C$ of the building block $\bar{S}_e$ for $e\subset X$ a cell can be chosen to depend only on the combinatorial type of the cell (i.e. on the combinatorial type of the cellular decomposition of the closure $\bar{e}$ induced from $X$).
Also, in this approach, for a fixed Riemannian manifold $M$ endowed with a sequence of cellular decompositions $X_n$ of maximal cell diameter going to zero as $n\ra\infty$, the cellular actions $S_{X_n}$ converge to the continuum action (\ref{l26_S_BF_AKSZ},\ref{l26_S_canonical}) in the limit $n\ra\infty$.
\end{remark}

\textbf{Effective theory on cohomology.}
Having constructed the cellular theory $(\FF_X,S_X)$ on a triangulation/cellular decomposition $X$ of $M$, we can proceed further and induce the effective theory corresponding to the cohomology $V'=H^\bt(M,\g)$ viewed as a subcomplex in $V=C^\bt(X,\g)$, using again the construction of Section \ref{sss: eff action for abstr BF theory}, Remark \ref{l26_rem_induction_for_BF_infty}. The result is a BV package $(\FF_{H^\bt},S_{H^\bt})$ (which we call the effective $BF$ theory ``on cohomology'' of $M$) with $\FF_{H^\bt}=H^\bt(M,\g)[1]\oplus H_{-\bt}(M,\g^*)[-2]$ and with $S_{H^\bt}$ satisfying the ansatz (\ref{l26_S_BF_infty}):
$$S_{H^\bt}=\sum_{n\geq 2} \frac{1}{n!}\lan b, l_n^{H^\bt}(a,\ldots,a)\ran-i\hbar\sum_{n\geq 2}\frac{1}{n!}q_n^{H^\bt}(a,\ldots,a)$$
where we denoted $a,b$ the superfields (\ref{l26_abstr_BF_superfields}) for the theory on the level of cohomology. 
Here the operations $l_n^{H^\bt}:\wedge^n H^\bt(M,\g)\ra H^\bt(M,\g)$ and $q_n^{H^\bt}:\wedge^n H^\bt(M,\g)\ra \RR$ are constructed using the Feynman diagram expansion (\ref{l26_S'}) out of the cellular action $S_X$ and endow the cohomology $H^\bt(M,\g)$ with the structure of a unimodular $L_\infty$ algebra (with zero differential $l_1=0$). In particular:
\begin{itemize}
\item Operations $l^{H^\bt}_n$ are the Lie-Massey operations on $\g$-valued cohomology of $M$.\footnote{
E.g., for $\alpha,\beta$ two closed forms on $M$ and $x,y\in\g$ we have $l_2^{H^\bt}(x\otimes[\alpha],y\otimes[\beta])=[x,y]\otimes [\alpha\wedge \beta]$ -- the usual cup product of the cohomology classes $[\alpha],[\beta]$, tensored with the Lie bracket for $\g$-coefficients. For $\alpha,\beta,\gamma$ a triple of closed forms and $x,y,z\in\g$, we have $l_3^{H^\bt}(x\otimes[\alpha],y\otimes [\beta],z\otimes [\gamma])=\pm[[x,y],z]\otimes p(-K(\alpha\wedge \beta)\wedge\gamma)+\cdots$ where $\cdots$ are two similar terms which come from simultaneous cyclic permutations of $x,y,z$ and $\alpha,\beta,\gamma$. Here $p,K$ are maps from an arbitrarily chosen induction triple $\Omega^\bt(M)\stackrel{(\iota,p,K)}{\rightsquigarrow} H^\bt(M)$ -- in particular one may infer such a triple from Hodge decomposition of forms on $M$ associated to a choice of Riemannian metric. In case when classes $[\alpha],[\beta],[\gamma]$ have pairwise vanishing cup products, we have $l_3^{H^\bt}(x\otimes[\alpha],y\otimes [\beta],z\otimes [\gamma])=\left[ [[x,y],x]\otimes \pm d^{-1}(\alpha\wedge\beta)\wedge\gamma+\cdots \right]$ and the expression in brackets on the r.h.s. is a closed form, and thus it is meaningful to take its cohomology class. The original reference on Lie-Massey brackets is \cite{Retakh}.
}
In the case of $M$ simply connected, they contain the complete information on the rational homotopy type of $M$. For $M$ non-simply connected, these operations define, via the deformation-theoretic model (as the homotopy Maurer-Cartan set)
$$\frac{\{a\in \epsilon H^1(M,\g)[[\epsilon]]\;\;|\;\; \sum_{k\geq 2}\frac{1}{k!}l_k^{H^\bt}(a,\ldots,a)=0\}}{a\sim a+\sum_{k\geq 1}\frac{1}{k!}l_{k+1}^{H^\bt}(a,\ldots,a,\beta),\;\; \beta\in 
H^0(M,\g)[[\epsilon]]}$$
for the formal neighborhood of zero connection on the moduli space of flat connections 
$$
\MM_{M,G}\cong \mr{Hom}(\pi_1(M),G)/G$$
\item Operations $q_n^{H^\bt}$ contain information on the singular behavior of Reidemeister-Ray-Singer torsion $\tau(M,\nabla)$ for a flat $G$-connection $\nabla$ in the formal neighborhood of zero connection in $
\MM_{M,G}$.
\end{itemize}

The unimodular $L_\infty$ structure on cohomology $H^\bt(M,\g)$ is a stronger invariant than just its classical $L_\infty$ part.
\begin{example}
For the circle $S^1$ and the Klein bottle $\mathsf{KB}$, the cohomology (with coefficients in $\RR$ or in $\g$) is isomorphic, with $H^0(-,\g)=\g$ and $H^1(-,\g)=\g$. Moreover (choosing the coefficients to be $\g$), they are isomorphic as graded Lie algebras and as $L_\infty$ algebras (operations $l_{\geq 3}^{H^\bt}$ vanish). However, they are distinguished by the quantum operations $q_n^{H^\bt}$ on cohomology. In terms of the effective actions on cohomology, we have:

\begin{tabular}{c||c}
circle & Klein Bottle \\ \hline
$S_{H^\bt}= \lan B_0, \frac12[A^0,A^0] \ran+\lan B_1,[A^1,A^0] \ran-$ &  $S_{H^\bt}= \lan B_0, \frac12[A^0,A^0] \ran+\lan B_1,[A^1,A^0] \ran-$\\
$-i\hbar\;\tr_\g \log\frac{\sinh\frac{\mr{ad}(A^1)}{2}}{\frac{\mr{ad}(A^1)}{2}}$ & $-i\hbar\;\tr_\g \log\left(\frac{\mr{ad}(A^1)}{2}\coth \frac{\mr{ad}(A^1)}{2} \right)^{-1}$
\end{tabular}
where indices $0$ and $1$ pertain to the generator of zeroth and first cohomology of the circle/Klein bottle.
\end{example}

If $X$ is a triangulation (or a more general cellular decomposition) of $M$ and $X'$ is a cellular subdivision of $X$ -- in this case we say that $X$ is an \emph{aggregation} of $X'$  -- then one can construct the BV pushforward of the cellular action $S_{X'}$ on $\FF_{X'}$ to $\FF_X$. The result is a canonical transformation of the cellular action $S_X$. Thus, cellular actions 
$S_X$ viewed modulo canonical transformations are compatible with BV pushforwards along cellular aggregations. And, in turn, from every $X$ one can induce the effective action on cohomology $S_{H^\bt}$ which, when viewed modulo canonical transformations, is independent of $X$ and independent on the induction data $C^\bt(X,\g)\stackrel{(\iota,p,K)}{\rightsquigarrow} H^\bt(M,\g)$ necessary to define the gauge-fixing for the BV pushforward to cohomology.

Schematically, for a fixed manifold $M$, we have the following picture of different realizations of $BF$ theory:
$$
\vcenter{\hbox{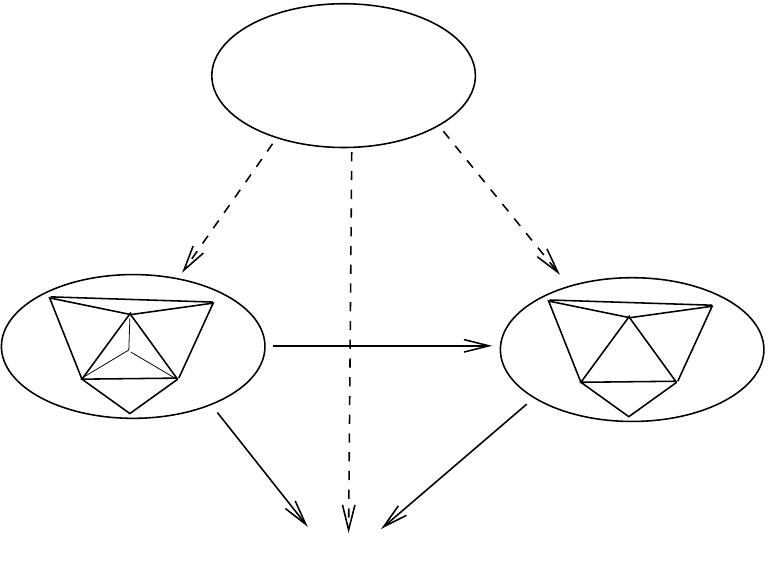}}
$$
Here we have three levels:
\begin{enumerate}
\item The ``upper'' level -- the continuum $BF$ theory on $M$, with infinite-dimensional space of fields, modelled on forms on $M$.
\item The ``middle'' level -- cellular $BF$ theories associated to cellular decompositions $X$ of $M$, with finite-dimensional fields modelled on cellular cochains. Cellular actions $S_X$ are local -- given as a sum over cells of local contributions. Considered modulo canonical transformations, cellular realizations $(\FF_X,S_X)$ are related by BV pushforwards along cellular aggregations.
\item The ``lower'' level -- the theory $(\FF_{H^\bt},S_{H^\bt})$ induced on cohomology of $M$. On this level we obtain the invariants of $M$, while the locality disappears completely.
\end{enumerate}
In the picture above the dashed arrows represent path integrals  (infinite-dimensional BV pushforwards) while the solid arrows represent finite-dimensional BV pushforwards. 

One motivation for cellular $BF$ theory is that it replaces the path integral computation of the invariants induced on cohomology (from the continuum level) by a coherent system of finite-dimensional integral formulae for same invariants (BV pushforwards $X\ra H^\bt$ for all possible $X$). Thus, ultimately, one may forget entirely about the continuum level in the picture above.

\subsubsection{Remarks}
\begin{remark}\label{l27_rem_xi}
One can  fine tune the normalization of the reference Berezinian (\ref{l26_abstr_BF_mu}) $\mu_X$ for the cellular $BF$ theory in such a way that the expression 
$$Z_X=\mu_X^{1/2}e^{\frac{i}{\hbar}S_X}\qquad \in \mr{Dens}^{1/2}(\FF_X)$$ reproduces itself under BV pushforward cellular aggregations, modulo $\Delta$-exact terms: $(\mc{P}_{X'\ra X})_*Z_{X'}=Z_X+\Delta_X(\cdots)$. This property is achieved if we set
$$\mu_X^{1/2}=
\prod_{\sigma\subset X} \xi_{|\sigma|}\cdot|\DD A^\sigma|^{\frac12} |\DD B_\sigma|^{\frac12}$$
where 
$$\xi_k:=(2\pi\hbar)^{-\frac14 +\frac12 k (-1)^{k-1}} (e^{-\frac{\pi i}{2}}\hbar)^{\frac14 +\frac12 k (-1)^{k-1}}$$
Thus, the reference half-density on cellular fields is is the product of the standard half-densities (built out of the reference volume element on $\g$) for individual cells, rescaled using local complex factors $\xi_k$. This nontrivial local rescaling of the standard cellular half-density can be thought of as a baby version of renormalization. Then the BV pushforward to cohomology yields 
\begin{equation}\label{l27_Z_H}
Z_{H^\bt}:=(\mc{P}_{X\ra H^\bt})_*Z_X= \xi_{H^\bt} \tau(M,\g) e^{\frac{i}{\hbar}S_{H^\bt}} 
\end{equation}
where $\xi_{H^\bt}=\prod_{k=0}^{\dim M}(\xi_k)^{\dim H^k(M,\g)}\;\;\in\CC$ and $\tau(M,\g)\in \mr{Dens}^{1/2}_\mr{const}(\FF_{H^\bt})\cong \mr{Det}\, H^\bt(M,\g)/\{\pm 1\}$ is the Reidemeister torsion of $M$ endowed with trivial connection. 
We call the expression (\ref{l27_Z_H}) the partition function.
\end{remark}

\begin{remark}
One can twist the differential forms on $M$, cochains on $X$ and the cohomology by a background $G$-local system $E$ (equivalently, a flat $G$-bundle). The entire story goes through.\footnote{In particular the local building blocks (\ref{l27_Sbar_ansatz}) do not depend on the local system and they are assembled into the cellular action $S_X$ again as in (\ref{l27_S_X_locality}), however now $\left.A_X\right|_{\sigma}$ is defined using the data of the local system, as $\sum_{\sigma'\subset \sigma} e_{\sigma'}\mr{Ad}_{E(\sigma\leftarrow\sigma')}A^{\sigma'}$. Here we assume that the local system is trivialized over the barycenters 
of the cells/simplices of $X$ and $E(\sigma\leftarrow\sigma')\in G$ is the parallel transport between the barycenters of $\sigma'$ and $\sigma$.
}
In particular in the partition function $Z_{H^\bt}$, $\tau$ becomes the Reidemeister torsion for a possibly non-trivial local system $E$ and $\xi_{H^\bt}$ is expressed in terms of the Betti numbers of $M$ twisted by $E$.
\end{remark}

\begin{remark} In the setting of arbitrary regular CW complexes (Remark \ref{l27_rem_CMRcell}), a fundamental property of the cellular theory is simple-homotopy invariance, which in particular implies invariance w.r.t. to cellular aggregations (which can be presented as special simple-homotopy equivalences). The fundamental observation \cite{CMRcell} is that if $X$ an \emph{elementary collapse} of $Y$, then the BV pushforward of $S_Y$ to $X$, along the collapse, yields a canonical transformation of $S_X$. 
$$
\vcenter{\hbox{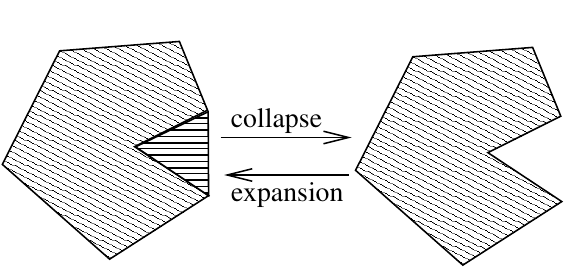}}
$$
This immediately implies that if $X$ and $X'$ are simple-homotopy equivalent CW complexes (i.e. can be connected by a sequence of collapses and their inverses -- expansions), then the actions induced from $X$ and $X'$ on cohomology agree up to a canonical transformation.
\end{remark}

\marginpar{\LARGE{Lecture 28, 12/07/2016.}}
\subsection{Perturbative Chern-Simons theory 
}
We give a brief discussion of the perturbative Chern-Simons theory after Witten \cite{Witten_CS} (one-loop partition function) 
and Axelrod-Singer \cite{AS1,AS2} (higher loop corrections).

Fix $M$ a closed connected oriented $3$-manifold and $G$ a compact simple, simply connected Lie group.
Recall the Chern-Simons action functional
$$S_{CS}(A)=\int_M\tr\;\frac12 A\wedge dA+\frac16 A\wedge [A,A]$$
on $G$-connections 
$A\in \mr{Conn}_{M,G}\simeq \Omega^1(M,\g)$.

\textbf{Idea (Witten, \cite{Witten_CS}):} Consider the path integral over the space of connections
\begin{equation}\label{l28_CS_PI}
Z(M,G;k)=\int_{\mr{Conn}_{M,G}}\DD A\;\; e^{\frac{ik}{2\pi}S_{CS}(A)}
\end{equation}
with $k\in\ZZ$ the ``level''.\footnote{One requests integrality of $k$ for the integrand in (\ref{l28_CS_PI}) to be invariant under gauge transformations of connections, see Section \ref{sss: CS gauge invariance}. For the perturbative treatment of the Chern-Simons path integral, integrality of $k$ is irrelevant.} One expects the path integral to yield an invariant of the $3$-manifold $M$ under orientation preserving diffeomorphisms, since the integrand is manifestly independent on geometric structures on $M$ besides the orientation.

\subsubsection{Perturbative contribution of an acyclic flat connection: one-loop part.}\label{sss: CS one-loop}
We want to evaluate the perturbative contribution of the (gauge orbit of) a flat connection $A_0$ (cf. Theorem \ref{thm: FP stat phase}; here instead of using Faddeev-Popov formalism for gauge-fixing, we employ AKSZ-BV formalism, which is equivalent to Faddeev-Popov in the case of Chern-Simons theory). Assume that the flat connection $A_0$ is \emph{acyclic}.\footnote{I.e. that 
the cohomology of the twisted de Rham operator $d_{A_0}$ acting on $\Omega^\bt(M,\g)$ vanishes in all degrees.
Or, equivalently, that the gauge class $[A_0]\in\MM_{M,G}$ is an isolated point in the moduli space of flat connections and corresponds to an irreducible local system. 
To treat the perturbation theory around a non-acyclic flat connection, one needs the machinery of BV pushforwards, see \cite{CM}.
} We have
\begin{equation}\label{l28_Z_pert_0}
Z_\mr{pert}(M,G;\hbar;A_0)=\int_{\Omega^1(M,\g)}^{\mr{gauge-fixed,\,perturbative\;at\;}a=0}\DD a\;e^{\frac{i}{\hbar}S_{CS}(A_0+a)}\\
\end{equation}
Here $\hbar$ is related to the level by $\hbar=\frac{2\pi}{k}$; $a$ is the fluctuation of the field $A$ around the fixed flat connection $A_0$. Chern-Simons action evaluated on a connection close to $A_0$ yields
$$S_{CS}(A_0+a)=S_{CS}(A_0)+\int_M \tr\; \frac12 a\wedge d_{A_0} a+\frac16 a\wedge [a,a]$$
-- the sum of a constant and a term which looks like the usual Chern-Simons action evaluated on the fluctuation, but with the de Rham operator twisted by $A_0$, denote it $S_{CS}(A_0;a)$. One can extend the classical action $S_{CS}(A_0;a)$ on $\Omega^1(M,\g)\ni a$ to a BV action, using the $A_0$-twisted version of the AKSZ construction, which yields the master action
$$S_\mr{CS-AKSZ}(A_0;\til{a})=\int_M \tr\; \frac12 \til{a}\wedge d_{A_0} \til{a}+\frac16 \til{a}\wedge [\til{a},\til{a}]$$
on nonhomogeneous forms $\Omega^\bt(M,\g)[1]\ni \til{a}$. Thus, continuing (\ref{l28_Z_pert_0}), we have the BV integral
\begin{equation}\label{l28_Z_per_1}
Z_\mr{pert}(M,G;\hbar;A_0)=e^{\frac{i}{\hbar}S_{CS}(A_0)}\int_{\LL\subset \Omega^\bt(M,\g)[1]}^\mr{perturbative}\DD\til{a}\; e^{\frac{i}{\hbar}S_\mr{CS-AKSZ}(A_0;\til{a})}
\end{equation}
where we choose the gauge-fixing Lagrangian $\LL$ to be given by the ``Lorentz gauge'' $d^*\til{a}=0$, with $d^*$ the Hodge dual of the de Rham operator, corresponding to some fixed metric $g$ on $M$. Finally, perturbative (stationary phase) evaluation of the BV integral on the r.h.s. yields (cf. Theorem \ref{thm: FP stat phase}):
\begin{equation}\label{l28_Z_pert_expansion}
Z_\mr{pert}(M,G;\hbar;A_0)=\underbrace{e^{\frac{i}{\hbar}S_{CS}(A_0)}}_{Z^{0\mr{-loop}}}\cdot \underbrace{\tau(M,A_0)^{\frac12}\cdot e^{\frac{\pi i}{4}\psi(A_0,g)}}_{Z^{1\mr{-loop}}}\cdot \underbrace{\exp\left( \sum_\Gamma \frac{(i\hbar)^{-\chi(\Gamma)}}{|\mr{Aut}(\Gamma)|}
\Phi(\Gamma)\right)}_{Z^{\geq 2\mr{-loop}}}
\end{equation}
Here 
$$\tau(M,A_0)=\prod_{p=1}^3 ({\det}_{\Omega^p(M,\g)}\Delta^{(p)}_{A_0})^{-\frac{p(-1)^p}{2}}\qquad \in \RR_{> 0}$$
is the Ray-Singer torsion; $\Delta^{(p)}=d_{A_0}d^*_{A_0}+d^*_{A_0}d_{A_0}$ is the Laplacian acting on $p$-forms, twisted by $A_0$. Another object appearing in the r.h.s. of (\ref{l28_Z_pert_expansion}) is $\psi(A_0,g)$ -- the Atiyah-Patodi-Singer eta invariant of the Dirac operator $L_-:= *d_{A_0}+d_{A_0}*$ acting on $\Omega^\mr{odd}(M,\g)=\Omega^1(M,\g)\oplus \Omega^3(M,\g)$.\footnote{Recall that for a self-adjoint elliptic operator $\OO:\HH\ra\HH$ on a compact manifold, its APS eta invariant is the zeta-regularized signature, $\eta=\lim_{s\ra 0}\sum_\lambda \mr{sign}(\lambda)\cdot |\lambda|^{-s}$ where the sum goes over the eigenvalues $\lambda$ of $\OO$.} 

\begin{remark}
It is instructive to write the r.h.s. of (\ref{l28_Z_per_1}) as a Faddeev-Popov integral: 
\begin{multline}\label{l28_CS_FP}
e^{\frac{i}{\hbar}S_{CS}(A_0)}\int\DD a\,\DD c\,\DD\bar{c}\,\DD\lambda\; e^{\frac{i}{\hbar}\int_M\tr\, (\frac12 a\wedge d_{A_0} a+\frac16 a\wedge [a,a])+\lambda\, d^*_{A_0}a+ \bar{c}\, d_{A_0}^*d_{A_0+a}c}\\
=e^{\frac{i}{\hbar}S_{CS}(A_0)}\int\DD a\,\DD c\,\DD\bar{c}\,\DD\lambda\; e^{\frac{i}{\hbar} 
\left(\frac12 \lan a+\lambda, L_- (a+\lambda) \ran+\lan *\bar{c}, \Delta_{A_0}^{(0)} c\ran + \int_M\tr \,(\frac16 a\wedge [a,a]+d_{A_0}^*\bar{c}\wedge[a,c])\right )}
\end{multline}
where we understand the fields $a,c,\bar{c},\lambda$ as $\g$-valued forms of degrees $1,0,3,3$, respectively, of ghost number $0,1,-1,0$, respectively; $\lan \alpha,\beta \ran:=\int_M\tr \,\alpha\wedge *\beta$ is the Hodge inner product of $\g$-valued forms.
Computing the Gaussian part of the result, we obtain (cf. Theorem \ref{thm: FP stat phase})
$$ 
e^{\frac{i}{\hbar}S_{CS}(A_0)}\cdot (\det L_-)^{-\frac 12} \cdot \det \Delta_{A_0}^{(0)} \cdot e^{\frac{\pi i}{4}\mr{sign}L_-}\cdot (\mbox{Feynman diagrams})
$$
Here the determinants and the signature are understood in the sense of zeta-regularization. Thus we precisely reproduced the formula (\ref{l28_Z_pert_expansion}). Note that the square root of Ray-Singer torsion can indeed be written as $\tau^{1/2}=(\det L_-)^{-1/2}\det \Delta^{(0)}_{A_0}$.
\end{remark}

\textbf{Cancelling metric dependence of 1-loop partition function (at the cost of introducing framing dependence).}
Ray-Singer torsion is independent of the choice of Riemannian metric $g$ that we used for the gauge-fixing. However, the eta invariant $\psi(A_0,g)$ does depend on the metric. This dependence can be unerstood from  Atiyah-Patodi-Singer index theorem which implies \cite{Witten_CS} that the combination
\begin{equation}\label{l28_psi_plus_Sgrav}
\frac{\pi i}{4}\psi(A_0,g)+\frac{i\dim G}{24}\cdot \frac{1}{2\pi}S_\mr{grav}(g,\phi)
\end{equation}
is a topological invariant (i.e. is independent on $g$). Here the ``gravitational Chern-Simons'' term $S_\mr{grav}$ is the Chern-Simons action evaluated on the Levi-\v{C}ivita connection associated to the metric $g$. To make sense of the latter, one needs to fix a trivialization (``framing'') $\phi:TM\ra M\times \RR^3$ of the tangent bundle of $M$, considered 
up to homotopy of trivializations. A known fact that all connected  orientable $3$-manifolds have trivializable tangent bundle, i.e. framings do exist. Moreover, the set of framings forms a $\ZZ$-torsor. In particular, there is a local operation on $M$ (changing the trivialization in a ball cut out of $M$) changing the framing by ``one unit''. Changing the framing $\phi$ by $n$ units changes $S_\mr{grav}$ by a multiple of $n$:
\begin{equation}\label{l28_Sgrav_change_under_shift_of_framing}
\frac{1}{2\pi}S_\mr{grav}(g,\phi+n)=\frac{1}{2\pi} S_\mr{grav}(g,\phi)+2\pi\cdot n
\end{equation}

Denote by $Z^{\leq 1\mr{-loop}}$ the right hand side of (\ref{l28_Z_pert_expansion}) without the last term -- we ignore  for the moment the contribution of Feynman graphs with $\geq 2$ loops. Performing a change of definition of partition function (a baby version of renormalization) 
\begin{equation}
Z^{\leq 1\mr{-loop}}\qquad \longrightarrow\qquad \til{Z}^{\leq 1\mr{-loop}}:=Z^{\leq 1\mr{-loop}}\cdot e^{\frac{i\dim G}{24}\cdot\frac{1}{2\pi}S_\mr{grav}(g,\phi)}
\end{equation}
we obtain an expression where dependence on the metric is cancelled (due to metric independence of (\ref{l28_psi_plus_Sgrav})), at the cost of introducing a dependence on framing $\phi$. The latter is under control via (\ref{l28_Sgrav_change_under_shift_of_framing}). In particular, shift of the framing $\phi\mapsto \phi+1$ induces the change 
$$\til{Z}^{\leq 1\mr{-loop}}\mapsto \til{Z}^{\leq 1\mr{-loop}}\cdot e^{\frac{2\pi i\dim G}{24}}$$

\begin{remark} Another consequence of Atiyah-Patodi-Singer theorem controls the dependence of $\psi(A_0,g)$ on the flat connection $A_0$ via
\begin{equation}
\frac{\pi i}{4}\psi(A_0,g)=\frac{\pi i}{4}\dim G\,\psi_0(g)+i\,c_2(G)\frac{1}{2\pi}S_{CS}(A_0)
\end{equation}
where $\psi_0(g)$ is the eta invariant of the operator $*d+d*$ on $\Omega^\mr{odd}(M)$ (without the twist by $A_0$) and $c_2(G)$ is the value of the quadratic Casimir of $G$ in adjoint representation (the dual Coxeter number of the Lie algebra $\g$). E.g., $c_2(SU(N))=N$. Therefore, one can write
$$\til{Z}^{\leq 1\mr{-loop}}=e^{\frac{i (k+c_2(G))}{2\pi}S_{CS}(A_0)}\tau(M,A_0)^{\frac12}e^{\frac{\pi i}{4}\dim G\,\psi_0(g)}\cdot e^{\frac{i}{24}\frac{\dim G}{2\pi}S_\mr{grav}(g,\phi)}$$
note that $A_0$-dependence of the eta invariant got absorbed into the shift of the level $k$ of Chern-Simons theory by $c_2(G)$.
\end{remark}

\subsubsection{Higher loop corrections, after Axelrod-Singer}
Now we proceed to the corrections to the perturbative Chern-Simons path integral in powers of $\hbar$ -- the last term in (\ref{l28_Z_pert_expansion}). Consider the operator 
\begin{equation}
\mathbb{K}=d^*_{A_0}\Delta_{A_0}^{-1}:\quad \Omega^\bt(M,\g)\ra \Omega^{\bt-1}(M,\g)
\end{equation}
-- the chain contraction for the twisted de Rham complex $\Omega^\bt(M,\g),d_{A_0}$ arising from Hodge-de Rham theory. It is an integral operator with certain integral kernel  
$$\eta\in \Omega^2(M\times M\backslash\diag,\g\otimes \g)$$
given by a $\g\otimes\g$-valued $2$-form on $M\times M$ which is singular at the diagonal and smooth away from the diagonal. The form $\eta$ is the propagator for Chern-Simons theory (or ``Green's function'' or ``parametrix''). The integral kernel property relating $\mathbb{K}$ and $\eta$ reads 
$$ \mathbb{K}\alpha=\int_{M_{(2)}}\tr\,\eta_{12}\wedge \alpha_{(2)}=(p_1)_*(\tr\,\eta\wedge p_2^* \alpha) $$
for $\alpha\in \Omega^\bt(M,\g)$ a test form and $p_1,p_2$ the projections from $M\times M$ to the first and second factor (subscripts $1,2$ in the middle formula above also relate to the fist and second factor in $M\times M$); $(p_1)_*$ is the fiber integral.

The crucial observation is that the singularity of $\eta$ on the diagonal of $M\times M$ is relatively mild: $\eta$ extends to a smooth form on the compactification $C_2(M)$ given by a differential geometric blow-up of the diagonal of $M\times M$ which replaces the diagonal with its tangent sphere bundle, $\diag\ra ST\diag$. Chain homotopy property $d_{A_0}\mathbb{K}+\mathbb{K}d_{A_0}=\mr{id}$ implies that 
\begin{enumerate}[(i)]
\item $d\eta=0$ with $d$ the de Rham operator on $C_2(M)$ twisted by $A_0$ on both copies of $M$;
\item fixing a point $y\in M$ and integrating the propagator over a $2$-sphere in $M$ surrounding $y$, we get $\int_{S^2\ni x}\eta(x,y)=1$. Put another way, the pushforward (fiber integral) of $\eta$ restricted to $\dd C_2(M)=ST\diag$ along the projection $ST\diag\ra \diag$ is $1$ (as the constant function on $\diag$).
\end{enumerate}

\textbf{Fulton-MacPherson-Axelrod-Singer compactified configuration spaces.} For $n\geq 2$, we have the \emph{open} configuration space of $n$ distinct ordered points on $M$, 
$$C_n^\mr{open}(M)=\{(x_1,\ldots,x_n)\in M^n\;\;|\;\; x_i\neq x_j\;\; \mr{for\; any\;}i\neq j\}\qquad = M^n\backslash \cup_S \diag_S$$
where $S=\{i_1,\ldots,i_k\}$ runs over subsets of $\{1,\ldots,n\}$ with $k\geq 2$ elements and $\diag_S=\{(x_1,\ldots,x_n)\;\;|\;\; x_{i_1}=\cdots=x_{i_k}\}\subset M^n$ is the diagonal in $M^n$ corresponding to $S$. Fulton-MacPherson-Axelrod-Singer compactification 
of the open configuration space $C_n^\mr{open}(M)$ consists in replacing all diagonals $\diag_S$ in $M^n$ with their unit normal bundles $N\diag_S/\RR_+$ (the differential geometric blow-ups of the diagonals; we also denote them $\mr{Bl}(\diag_S)$). We denote this compactification $C_n(M)$. In particular, in $C_n(M)$, the situation when $k\geq 2$ points collapse together at a single point gets endowed with the ``zoomed-in'' picture, containing the information on the relative positions of the collapsing points, modulo scaling.
$$
\vcenter{\hbox{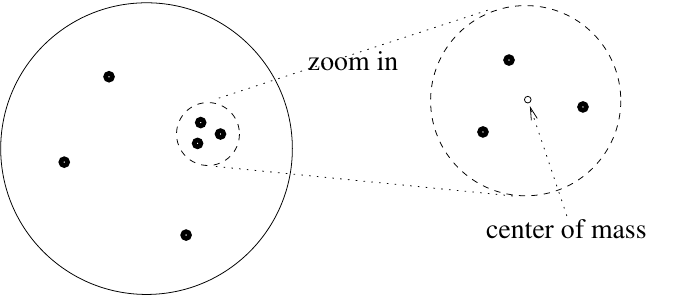}}
$$
In particular, to every subset $S\subset \{1,\ldots,n\}$ with $|S|\geq 2$ elements, there corresponds a boundary stratum of the compactified configuration space,
\begin{equation}\label{l28_bdry_stratum}
\dd_S C_n(M)=\mr{Bl}(\diag_S)\simeq C_{n-|S|+1}(M)\times \underline{C}_{|S|}(\RR^3)
\end{equation}
where the \emph{reduced} configuration space $\underline{C}_{k}(\RR^3)=C_k(\RR^3)/\RR_+\ltimes \RR^3$ appearing in the r.h.s. is the space of configurations of $k$ points modulo translations and rescalings (note that one can mod out translations by fixing the center of mass of the configuration to be at the origin of $\RR^3$).\footnote{\label{l28_footnote_pi_S}
Another way to phrase (\ref{l28_bdry_stratum}) is to note that there is a canonical projection $\pi_S:\dd_S C_{n}(M)\ra C_{n-|S|+1}(M)$ which identifies all the points from $S$ with the point of collapse $x_0\in M$; the fiber of $\pi_S$ is the reduced configuration space of $S$ points on the tangent space $T_{x_0}M$, which can be identified with $\underline{C}_{|S|}(\RR^3)$ if the trivialization of $TM$ is fixed. Thus $\pi_S$ is a trivial fiber bundle with fiber $\underline{C}_{|S|}(\RR^3)$, but its trivialization depends on the trivialization of $TM$.}
Boundary strata corresponding to collapses of pairs of points are called the ``principal'' boundary strata, while collapses of $|S|\geq 3$ points correspond to the ``hidden'' boundary strata. The space $C_n(M)$ is a smooth $3n$-dimensional manifold with corners; the corners are described by nested collapses of points on $M$.\footnote{It is instructive to check directly the dimension of the r.h.s. of (\ref{l28_bdry_stratum}) -- it is $3(n-|S|+1)+(3|S|-3-1)=3n-1$. Thus, it has the correct dimension to be a codimension $1$ stratum of $C_n(M)$.}

\textbf{Feynman diagrams.} In the last term of (\ref{l28_Z_pert_expansion}), we sum over connected $3$-valent graphs $\Gamma$ without ``short loops'' -- edges connecting a vertex to itself (whereas double/triple edges between distinct vertices are allowed). Here are the first contributing graphs (all admissible graphs with $2$ and $3$ loops):
$$
\vcenter{\hbox{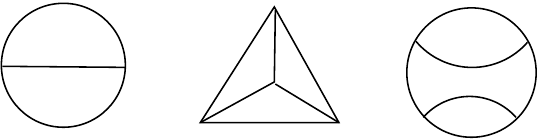}}
$$
The weight of a graph $\Gamma$ is given by the integral over the Fulton-MacPherson-Axelrod-Singer configuration space of $V$ points on $M$, where $V$ is the number of vertices:
\begin{equation}\label{l28_conf_space_int}
\Phi(\Gamma)=\int_{C_V(M)}\langle\bigwedge_{\mr{edges\,}(v_1,v_2)}\pi^*_{v_1v_2}\eta \;\;,\;\; \sigma_\Gamma\circ \bigotimes_{\mr{vertices}}f\rangle 
\end{equation}
Here $\pi_{v_1v_2}:C_V(M)\ra C_2(M)$ is the map which forgets all points except the two; $f\in (\g^*)^{\otimes 3}$ is the structure tensor of the Lie algebra $\g$ and $\sigma_\Gamma$ is the permutation of half-edges corresponding to the graph $\Gamma$; $\lan,\ran$ is the canonical pairing between $\g^{\otimes 2E}$ and $(\g^*)^{\otimes 3V}$ (note that in a $3$-valent graph the number of vertices and the number of edges are related by $2E=3V$).

\textbf{Finiteness.} The crucial point is that the values of Feynman diagrams are given by integrals (\ref{l28_conf_space_int}) of smooth forms over compact manifolds (with corners), hence they are finite numbers and there are no divergencies in the perturbative quantum theory. This property is ultimately due to the tame short-distance behavior of the propagator -- the Hodge-theoretic homotopy inverse of the de Rham operator.

\textbf{Dependence on the metric $g$.} The propagator $\eta$ and thus the Feynman weights $\Phi(\Gamma)$ depend on the chosen metric $g$ on $M$; it turns out that, when summed over all admissible graphs, as in (\ref{l28_Z_pert_expansion}), the dependence on the metric 
$g$ \emph{almost} cancels out.

Let $g_t$, for $t\in [0,1]$, be a smooth family of Riemannian metrics on $M$ and let $\eta_t$ be the respective family of propagators. An explicit calculation of the variation of the propagator with respect to a variation of metric yields a coboundary
\begin{equation}\label{l28_d_t_eta}
\frac{\dd}{\dd t}\eta_t=d\xi_t
\end{equation}
with $d$ the total de Rham operator on $\Omega^2(C_2(M),\g\otimes \g)$ twisted by $A_0$ and with $\xi_t\in \Omega^1(C_2(M),\g\otimes \g)$ a certain $t$-dependent $1$-form. We can assemble $\eta_t$ and $\xi_t$ into a composite object 
$$\til\eta_t=\eta_t+dt\cdot\xi_t \quad \in \Omega^2([0,1]\times C_2(M),\g\otimes\g)$$
which, by virtue of (\ref{l28_d_t_eta}), satisfies $d_\mr{tot}\til\eta_t=0$ with $d_\mr{tot}=dt\frac{\dd}{\dd t}+d$ the total de Rham operator, including the de Rham operator on the  interval parameterizing the family. Next, we can calculate the variation of the higher-loop part of the perturbative partition function (\ref{l28_Z_pert_expansion}) as follows:
\begin{multline}
dt\wedge \frac{\dd}{\dd t}\;Z^{\geq 2\mr{-loop}}_t=\sum_\Gamma\frac{(i\hbar)^{-\chi(\Gamma)}}{|\mr{Aut}(\Gamma)|}\quad dt\wedge\frac{\dd}{\dd t}\;\Phi_t(\Gamma)\\
=
\sum_\Gamma\frac{(i\hbar)^{-\chi(\Gamma)}}{|\mr{Aut}(\Gamma)|}\quad \int_{C_V(M)}d\;\langle \bigwedge_{\mr{edges}}\pi^*_{v_1v_2}\til\eta_t\;,\;\sigma_\Gamma\circ\bigotimes_\mr{vertices}f \rangle \\
\underset{\mr{Stokes'}}{=}
\sum_\Gamma\quad \sum_{S\subset \{1,\ldots,V\},\;|S|\geq 2}\frac{(i\hbar)^{-\chi(\Gamma)}}{|\mr{Aut}(\Gamma)|}\quad \int_{\dd_S C_V(M)}\langle \bigwedge_{\mr{edges}}\pi^*_{v_1v_2}\til\eta_t\;,\;\sigma_\Gamma\circ\bigotimes_\mr{vertices}f \rangle
\end{multline}
Here the subscript $t$ in $Z_t,\Phi_t(\Gamma)$ reminds that we use the propagator $\eta_t$; $\Gamma$ runs over all (possibly, disconnected) trivalent graphs. On the r.h.s. we have the following terms:
\begin{enumerate}
\item \label{l28_collapses_i} Collapses of pairs of vertices, $|S|=2$. If the two collapsing vertices are not connected by an edge in $\Gamma$, the contribution vanishes by degree reasons: the integrand is non-singular as the two points collide, therefore it is a pullback of a smooth form on the base of the fiber bundle $\pi_S:C_{V}(M)\ra C_{V-1}(M)$ (cf. footnote \ref{l28_footnote_pi_S}) and hence is not a full-degree form on the total space.
On the other hand, contributions of collapses of pairs of vertices connected by an edge (``collapses of edges'')
cancel out when summed over trivalent graphs $\Gamma$ due to 
the ``IHX relation'' -- vanishing of the sum of contributions of any three graphs for which the stars of the collapsing edge have the following form (and the rest of the graph is the same): 
$$
\vcenter{\hbox{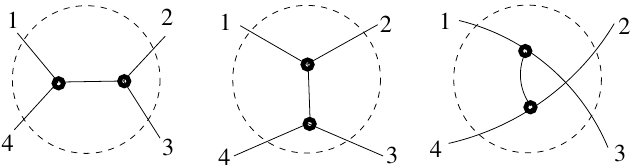}}
$$
This is in turn a consequence of 
Jacobi identity in $\g$.
\item \label{l28_collapses_ii}
Collapses of $\geq 3$ vertices, except the case when the collapsed subgraph of $\Gamma$ is an entire connected component of $\Gamma$. 
Their contribution vanishes by Kontsevich's vanishing lemmata \cite{Kontsevich_Feynman}. 
\item \label{l28_collapses_iii} Collapse of all vertices of a connected subgraph of $\Gamma$ at one point on $M$. The contribution of this boundary stratum can be computed and can be written in the form 
\begin{equation}
Z_t^{\geq 2\mr{-loop}}\cdot  dt\wedge \frac{\dd}{\dd t}\;\left(c_\Gamma(\g)\cdot S_\mr{grav}(g_t,\phi)\right)
\end{equation}
where $c_\Gamma(\g)$ is a universal numeric constant depending only on the combinatorics of the graph $\Gamma$ and on the Lie algebra $\g$.
\end{enumerate}
This calculation implies that there exists a power series $c(\hbar)\in \hbar \RR[[i\hbar]]$ with universal coefficients depending only on $\g$, such that the ``renormalized'' higher-loop part of the partition function
\begin{equation}
\til Z^{\geq 2\mr{-loop}}:=Z^{\geq 2\mr{-loop}}\cdot e^{i c(\hbar)S_\mr{grav}(g,\phi)}
\end{equation}
is independent of the metric $g$ (but depends on the framing $\phi$). Finally, putting this together with the discussion of Section \ref{sss: CS one-loop}, we obtain that the full renormalized answer of perturbative Chern-Simons theory (near an acyclic flat connection $A_0$) is:
\begin{multline}
\til Z_\mr{pert}(M,G,\phi;\hbar;A_0)=
\\=e^{\frac{i}{\hbar} S_{CS}(A_0)}\cdot \tau(M,A_0)^{\frac12}\cdot e^{\frac{\pi i}{4}\psi(A_0,g)}\cdot \exp\left(\sum_\Gamma\frac{(i\hbar)^{-\chi(\Gamma)}}{|\mr{Aut}(\Gamma)|}\Phi(\Gamma)\right)\cdot e^{i c'(\hbar) S_\mr{grav}(g,\phi)}
\end{multline}
with $c'(\hbar):=\frac{2\pi}{24}\dim G+c(\hbar)$. It is independent on the metric but is dependent on the framing of $M$. Dependence of the partition function of Chern-Simons theory on framing is a nontrivial effect of quantization.

\begin{remark} One can consider a finite-dimensional BV integral, based on a cyclic dgLa, modelling the Chern-Simons path integral, see \cite{CM} (the construction is a cyclic version of the abstract $BF$ theory associated to dgLa and its BV pushforward to a subcomplex, cf. Sections \ref{sss: abstr BF theory}, \ref{sss: eff action for abstr BF theory}). Independence of the result on the gauge-fixing is a consequence of BV-Stokes' theorem, but can also be proven directly on the level of Feynman diagrams. The resulting diagrammatic proof is an algebraic version of the argument above for the true Chern-Simons theory; however, we see only contractions of edges of Feynman graphs, as in (\ref{l28_collapses_i}) above, which again cancel out due to graph combinatorics and Jacobi identity. No analogs of situations (\ref{l28_collapses_ii},\ref{l28_collapses_iii}) show up in this finite dimensional/algebraic context. Therefore, one may say that our naive initial expectation that the Chern-Simons path integral should be independent of the Riemannian metric by BV-Stokes' theorem (which is only a theorem for finite-dimensional BV integrals) corresponds to the cancellation of contributions of principal boundary strata of configuration spaces to the variation of the path integral with the metric. Contributions of hidden boundary strata are an essentially infinite-dimensional field-theoretic phenomenon, which in the Chern-Simons case leads to the dependence of the result on the framing.
\end{remark}

\begin{remark} 
BV formalism is not necessary to construct the perturbative Chern-Simons path integral (\ref{l28_Z_pert_expansion}) -- one can obtain it purely using Faddeev-Popov method for gauge-fixing (unlike, e.g. $BF$ theory in dimension $\geq 4$ or the Poisson sigma model, where one is forced to use BV). However, BV (and, in particular, AKSZ) perspective is helpful. E.g. the simple form of Feynman weights (\ref{l28_conf_space_int}) corresponds to the AKSZ origin of the theory. If we would split the propagator $\eta$, viewed as a singular $2$-form on $M\times M$, into components $\eta^{2,0}+\eta^{1,1}+\eta^{0,2}$ according to the bi-degree of forms on the product, the Feynman weight (\ref{l28_conf_space_int}) of the graph $\Gamma$ splits into parts corresponding to decorating the half-edges of $\Gamma$ by degrees $\in \{0,1,2\}$ in such a way that the sum of the two degrees for every edge is $2$. In this way we get the perturbative expansion of Faddeev-Popov integral (\ref{l28_CS_FP}), with half-edges of degrees $0,1,2$ corresponding to the fields $c,a,d^*_{A_0}\bar{c}$, respectively. This was the perspective of the original papers \cite{AS1,AS2} and from that point of view the fact that Feynman diagrams coming from Faddeev-Popov construction assemble into the compact form (\ref{l28_conf_space_int}) looks like a miracle. BV formalism and, in particular, AKSZ construction of Chern-Simons theory explains this miracle.
\end{remark}

\begin{remark} 
One can construct the perturbative path integral for Chern-Simons theory around a non-acyclic connection using the technology of BV pushforwards, see \cite{CM}. The result, expressed as a sum over trivalent Feynman graphs with leaves decorated by cohomology classes of $M$, yields a volume element on the moduli space of flat connections on $M$; the total volume of the moduli space with respect to this volume element is the total Chern-Simons partition function. 
In this setting a Stokes' theorem argument on the compactified configuration space of points on $M$ is used to prove two statements: that a change of gauge-fixing induces a canonical transformation of the effective action on cohomology (which determines in turn the volume element on the moduli space), and that the effective action satisfies the quantum master equation. Generally, 
Stokes' theorem arguments on the configuration space, dealing directly with values of Feynman diagrams, are a refined field-theoretic version of the BV-Stokes' theorem, suited for perturbative AKSZ theories. In particular, the same technology works in the Poisson sigma model \cite{Kontsevich,CF}. Configuration space formalism 
can also be modified to allow source manifolds with boundary \cite{CMRpert} and pairs of manifold and a submanifold, corresponding to enrichment of an AKSZ theory with observables supported on submanifolds, cf. \cite{CR,AKSZobs}.
\end{remark}

\end{document}

%% file: theta.pdftex_t
\begin{picture}(0,0)%
\includegraphics{theta.pdf}%
\end{picture}%
%
%
\setlength{\unitlength}{1973sp}%
\begingroup\makeatletter\ifx\SetFigFont\undefined%
\gdef\SetFigFont#1#2#3#4#5{%
  \reset@font\fontsize{#1}{#2pt}%
  \fontfamily{#3}\fontseries{#4}\fontshape{#5}%
  \selectfont}%
\fi\endgroup%
\begin{picture}(1914,1908)(2351,-2713)
\end{picture}%

%% file: generic_graph.pdftex_t
\begin{picture}(0,0)%
\includegraphics{generic_graph.pdf}%
\end{picture}%
%
%
\setlength{\unitlength}{1973sp}%
\begingroup\makeatletter\ifx\SetFigFont\undefined%
\gdef\SetFigFont#1#2#3#4#5{%
  \reset@font\fontsize{#1}{#2pt}%
  \fontfamily{#3}\fontseries{#4}\fontshape{#5}%
  \selectfont}%
\fi\endgroup%
\begin{picture}(3136,2614)(1208,-2393)
\end{picture}%

%% file: l8_graph_i.pdftex_t
\begin{picture}(0,0)%
\includegraphics{l8_graph_i.pdf}%
\end{picture}%
%
%
\setlength{\unitlength}{987sp}%
\begingroup\makeatletter\ifx\SetFigFont\undefined%
\gdef\SetFigFont#1#2#3#4#5{%
  \reset@font\fontsize{#1}{#2pt}%
  \fontfamily{#3}\fontseries{#4}\fontshape{#5}%
  \selectfont}%
\fi\endgroup%
\begin{picture}(4282,4412)(695,-3965)
\end{picture}%

%% file: l8_theta.pdftex_t
\begin{picture}(0,0)%
\includegraphics{l8_theta.pdf}%
\end{picture}%
%
%
\setlength{\unitlength}{1973sp}%
\begingroup\makeatletter\ifx\SetFigFont\undefined%
\gdef\SetFigFont#1#2#3#4#5{%
  \reset@font\fontsize{#1}{#2pt}%
  \fontfamily{#3}\fontseries{#4}\fontshape{#5}%
  \selectfont}%
\fi\endgroup%
\begin{picture}(2052,1908)(2285,-2713)
\end{picture}%

%% file: l8_graph_iii.pdftex_t
\begin{picture}(0,0)%
\includegraphics{l8_graph_iii.pdf}%
\end{picture}%
%
%
\setlength{\unitlength}{1973sp}%
\begingroup\makeatletter\ifx\SetFigFont\undefined%
\gdef\SetFigFont#1#2#3#4#5{%
  \reset@font\fontsize{#1}{#2pt}%
  \fontfamily{#3}\fontseries{#4}\fontshape{#5}%
  \selectfont}%
\fi\endgroup%
\begin{picture}(1806,885)(1741,-1728)
\end{picture}%

%% file: l9_graph.pdftex_t
\begin{picture}(0,0)%
\includegraphics{l9_graph.pdf}%
\end{picture}%
%
%
\setlength{\unitlength}{1973sp}%
\begingroup\makeatletter\ifx\SetFigFont\undefined%
\gdef\SetFigFont#1#2#3#4#5{%
  \reset@font\fontsize{#1}{#2pt}%
  \fontfamily{#3}\fontseries{#4}\fontshape{#5}%
  \selectfont}%
\fi\endgroup%
\begin{picture}(4395,1920)(1328,-3901)
\end{picture}%

%% file: l9_graph1.pdftex_t
\begin{picture}(0,0)%
\includegraphics{l9_graph1.pdf}%
\end{picture}%
%
%
\setlength{\unitlength}{987sp}%
\begingroup\makeatletter\ifx\SetFigFont\undefined%
\gdef\SetFigFont#1#2#3#4#5{%
  \reset@font\fontsize{#1}{#2pt}%
  \fontfamily{#3}\fontseries{#4}\fontshape{#5}%
  \selectfont}%
\fi\endgroup%
\begin{picture}(1447,672)(2340,-2505)
\end{picture}%

%% file: l9_graph2.pdftex_t
\begin{picture}(0,0)%
\includegraphics{l9_graph2.pdf}%
\end{picture}%
%
%
\setlength{\unitlength}{789sp}%
\begingroup\makeatletter\ifx\SetFigFont\undefined%
\gdef\SetFigFont#1#2#3#4#5{%
  \reset@font\fontsize{#1}{#2pt}%
  \fontfamily{#3}\fontseries{#4}\fontshape{#5}%
  \selectfont}%
\fi\endgroup%
\begin{picture}(1467,1512)(2340,-3345)
\end{picture}%

%% file: l9_graph3.pdftex_t
\begin{picture}(0,0)%
\includegraphics{l9_graph3.pdf}%
\end{picture}%
%
%
\setlength{\unitlength}{987sp}%
\begingroup\makeatletter\ifx\SetFigFont\undefined%
\gdef\SetFigFont#1#2#3#4#5{%
  \reset@font\fontsize{#1}{#2pt}%
  \fontfamily{#3}\fontseries{#4}\fontshape{#5}%
  \selectfont}%
\fi\endgroup%
\begin{picture}(1461,1444)(1157,-2718)
\end{picture}%

%% file: l9_graph4.pdftex_t
\begin{picture}(0,0)%
\includegraphics{l9_graph4.pdf}%
\end{picture}%
%
%
\setlength{\unitlength}{987sp}%
\begingroup\makeatletter\ifx\SetFigFont\undefined%
\gdef\SetFigFont#1#2#3#4#5{%
  \reset@font\fontsize{#1}{#2pt}%
  \fontfamily{#3}\fontseries{#4}\fontshape{#5}%
  \selectfont}%
\fi\endgroup%
\begin{picture}(2746,606)(1010,-1595)
\end{picture}%

%% file: l12_theta.pdftex_t
\begin{picture}(0,0)%
\includegraphics{l12_theta.pdf}%
\end{picture}%
%
%
\setlength{\unitlength}{789sp}%
\begingroup\makeatletter\ifx\SetFigFont\undefined%
\gdef\SetFigFont#1#2#3#4#5{%
  \reset@font\fontsize{#1}{#2pt}%
  \fontfamily{#3}\fontseries{#4}\fontshape{#5}%
  \selectfont}%
\fi\endgroup%
\begin{picture}(1462,1448)(451,-1775)
\end{picture}%

%% file: l12_dumbbell.pdftex_t
\begin{picture}(0,0)%
\includegraphics{l12_dumbbell.pdf}%
\end{picture}%
%
%
\setlength{\unitlength}{789sp}%
\begingroup\makeatletter\ifx\SetFigFont\undefined%
\gdef\SetFigFont#1#2#3#4#5{%
  \reset@font\fontsize{#1}{#2pt}%
  \fontfamily{#3}\fontseries{#4}\fontshape{#5}%
  \selectfont}%
\fi\endgroup%
\begin{picture}(2157,740)(1166,-2129)
\end{picture}%

%% file: l13_wheel.pdftex_t
\begin{picture}(0,0)%
\includegraphics{l13_wheel.pdf}%
\end{picture}%
%
%
\setlength{\unitlength}{1579sp}%
\begingroup\makeatletter\ifx\SetFigFont\undefined%
\gdef\SetFigFont#1#2#3#4#5{%
  \reset@font\fontsize{#1}{#2pt}%
  \fontfamily{#3}\fontseries{#4}\fontshape{#5}%
  \selectfont}%
\fi\endgroup%
\begin{picture}(2501,2563)(991,-3203)
\end{picture}%

%% file: l13_HE_e.pdftex_t
\begin{picture}(0,0)%
\includegraphics{l13_HE_e.pdf}%
\end{picture}%
%
%
\setlength{\unitlength}{1381sp}%
\begingroup\makeatletter\ifx\SetFigFont\undefined%
\gdef\SetFigFont#1#2#3#4#5{%
  \reset@font\fontsize{#1}{#2pt}%
  \fontfamily{#3}\fontseries{#4}\fontshape{#5}%
  \selectfont}%
\fi\endgroup%
\begin{picture}(1060,435)(853,-1618)
\put(1701,-1411){\makebox(0,0)[lb]{\smash{{\SetFigFont{6}{7.2}{\rmdefault}{\mddefault}{\updefault}{\color[rgb]{0,0,0}$e$}%
}}}}
\end{picture}%

%% file: l13_HE_o.pdftex_t
\begin{picture}(0,0)%
\includegraphics{l13_HE_o.pdf}%
\end{picture}%
%
%
\setlength{\unitlength}{1381sp}%
\begingroup\makeatletter\ifx\SetFigFont\undefined%
\gdef\SetFigFont#1#2#3#4#5{%
  \reset@font\fontsize{#1}{#2pt}%
  \fontfamily{#3}\fontseries{#4}\fontshape{#5}%
  \selectfont}%
\fi\endgroup%
\begin{picture}(1060,435)(853,-1618)
\put(1701,-1411){\makebox(0,0)[lb]{\smash{{\SetFigFont{6}{7.2}{\rmdefault}{\mddefault}{\updefault}{\color[rgb]{0,0,0}$o$}%
}}}}
\end{picture}%

%% file: l13_E_e.pdftex_t
\begin{picture}(0,0)%
\includegraphics{l13_E_e.pdf}%
\end{picture}%
%
%
\setlength{\unitlength}{1381sp}%
\begingroup\makeatletter\ifx\SetFigFont\undefined%
\gdef\SetFigFont#1#2#3#4#5{%
  \reset@font\fontsize{#1}{#2pt}%
  \fontfamily{#3}\fontseries{#4}\fontshape{#5}%
  \selectfont}%
\fi\endgroup%
\begin{picture}(2084,430)(853,-1623)
\put(1251,-1421){\makebox(0,0)[lb]{\smash{{\SetFigFont{6}{7.2}{\rmdefault}{\mddefault}{\updefault}{\color[rgb]{0,0,0}$e$}%
}}}}
\put(2341,-1431){\makebox(0,0)[lb]{\smash{{\SetFigFont{6}{7.2}{\rmdefault}{\mddefault}{\updefault}{\color[rgb]{0,0,0}$e$}%
}}}}
\end{picture}%

%% file: l13_E_o.pdftex_t
\begin{picture}(0,0)%
\includegraphics{l13_E_o.pdf}%
\end{picture}%
%
%
\setlength{\unitlength}{1381sp}%
\begingroup\makeatletter\ifx\SetFigFont\undefined%
\gdef\SetFigFont#1#2#3#4#5{%
  \reset@font\fontsize{#1}{#2pt}%
  \fontfamily{#3}\fontseries{#4}\fontshape{#5}%
  \selectfont}%
\fi\endgroup%
\begin{picture}(2084,430)(853,-1623)
\put(1251,-1421){\makebox(0,0)[lb]{\smash{{\SetFigFont{6}{7.2}{\rmdefault}{\mddefault}{\updefault}{\color[rgb]{0,0,0}$o$}%
}}}}
\put(2341,-1431){\makebox(0,0)[lb]{\smash{{\SetFigFont{6}{7.2}{\rmdefault}{\mddefault}{\updefault}{\color[rgb]{0,0,0}$o$}%
}}}}
\end{picture}%

%% file: l13_HE_wavy.pdftex_t
\begin{picture}(0,0)%
\includegraphics{l13_HE_wavy.pdf}%
\end{picture}%
%
%
\setlength{\unitlength}{1973sp}%
\begingroup\makeatletter\ifx\SetFigFont\undefined%
\gdef\SetFigFont#1#2#3#4#5{%
  \reset@font\fontsize{#1}{#2pt}%
  \fontfamily{#3}\fontseries{#4}\fontshape{#5}%
  \selectfont}%
\fi\endgroup%
\begin{picture}(590,142)(811,-1004)
\end{picture}%

%% file: l13_HE_out.pdftex_t
\begin{picture}(0,0)%
\includegraphics{l13_HE_out.pdf}%
\end{picture}%
%
%
\setlength{\unitlength}{1973sp}%
\begingroup\makeatletter\ifx\SetFigFont\undefined%
\gdef\SetFigFont#1#2#3#4#5{%
  \reset@font\fontsize{#1}{#2pt}%
  \fontfamily{#3}\fontseries{#4}\fontshape{#5}%
  \selectfont}%
\fi\endgroup%
\begin{picture}(794,154)(879,-1141)
\end{picture}%

%% file: l13_HE_in.pdftex_t
\begin{picture}(0,0)%
\includegraphics{l13_HE_in.pdf}%
\end{picture}%
%
%
\setlength{\unitlength}{1973sp}%
\begingroup\makeatletter\ifx\SetFigFont\undefined%
\gdef\SetFigFont#1#2#3#4#5{%
  \reset@font\fontsize{#1}{#2pt}%
  \fontfamily{#3}\fontseries{#4}\fontshape{#5}%
  \selectfont}%
\fi\endgroup%
\begin{picture}(794,150)(879,-1141)
\end{picture}%

%% file: l13_E_wavy.pdftex_t
\begin{picture}(0,0)%
\includegraphics{l13_E_wavy.pdf}%
\end{picture}%
%
%
\setlength{\unitlength}{1973sp}%
\begingroup\makeatletter\ifx\SetFigFont\undefined%
\gdef\SetFigFont#1#2#3#4#5{%
  \reset@font\fontsize{#1}{#2pt}%
  \fontfamily{#3}\fontseries{#4}\fontshape{#5}%
  \selectfont}%
\fi\endgroup%
\begin{picture}(1074,213)(811,-1042)
\end{picture}%

%% file: l13_E_or.pdftex_t
\begin{picture}(0,0)%
\includegraphics{l13_E_or.pdf}%
\end{picture}%
%
%
\setlength{\unitlength}{1973sp}%
\begingroup\makeatletter\ifx\SetFigFont\undefined%
\gdef\SetFigFont#1#2#3#4#5{%
  \reset@font\fontsize{#1}{#2pt}%
  \fontfamily{#3}\fontseries{#4}\fontshape{#5}%
  \selectfont}%
\fi\endgroup%
\begin{picture}(1560,181)(879,-1153)
\end{picture}%

%% file: l13_QED_vertex.pdftex_t
\begin{picture}(0,0)%
\includegraphics{l13_QED_vertex.pdf}%
\end{picture}%
%
%
\setlength{\unitlength}{1973sp}%
\begingroup\makeatletter\ifx\SetFigFont\undefined%
\gdef\SetFigFont#1#2#3#4#5{%
  \reset@font\fontsize{#1}{#2pt}%
  \fontfamily{#3}\fontseries{#4}\fontshape{#5}%
  \selectfont}%
\fi\endgroup%
\begin{picture}(866,794)(1487,-1898)
\end{picture}%

%% file: l13_QED_graph.pdftex_t
\begin{picture}(0,0)%
\includegraphics{l13_QED_graph.pdf}%
\end{picture}%
%
%
\setlength{\unitlength}{1973sp}%
\begingroup\makeatletter\ifx\SetFigFont\undefined%
\gdef\SetFigFont#1#2#3#4#5{%
  \reset@font\fontsize{#1}{#2pt}%
  \fontfamily{#3}\fontseries{#4}\fontshape{#5}%
  \selectfont}%
\fi\endgroup%
\begin{picture}(3221,2200)(985,-3019)
\end{picture}%

%% file: l13_theta.pdftex_t
\begin{picture}(0,0)%
\includegraphics{l13_theta.pdf}%
\end{picture}%
%
%
\setlength{\unitlength}{1973sp}%
\begingroup\makeatletter\ifx\SetFigFont\undefined%
\gdef\SetFigFont#1#2#3#4#5{%
  \reset@font\fontsize{#1}{#2pt}%
  \fontfamily{#3}\fontseries{#4}\fontshape{#5}%
  \selectfont}%
\fi\endgroup%
\begin{picture}(1934,1920)(924,-3664)
\end{picture}%

%% file: l13_dumbbell.pdftex_t
\begin{picture}(0,0)%
\includegraphics{l13_dumbbell.pdf}%
\end{picture}%
%
%
\setlength{\unitlength}{1973sp}%
\begingroup\makeatletter\ifx\SetFigFont\undefined%
\gdef\SetFigFont#1#2#3#4#5{%
  \reset@font\fontsize{#1}{#2pt}%
  \fontfamily{#3}\fontseries{#4}\fontshape{#5}%
  \selectfont}%
\fi\endgroup%
\begin{picture}(2629,976)(694,-3192)
\end{picture}%

%% file: l14_Wilson.pdftex_t
\begin{picture}(0,0)%
\includegraphics{l14_Wilson.pdf}%
\end{picture}%
%
%
\setlength{\unitlength}{1579sp}%
\begingroup\makeatletter\ifx\SetFigFont\undefined%
\gdef\SetFigFont#1#2#3#4#5{%
  \reset@font\fontsize{#1}{#2pt}%
  \fontfamily{#3}\fontseries{#4}\fontshape{#5}%
  \selectfont}%
\fi\endgroup%
\begin{picture}(3810,3808)(222,-3663)
\put(3401,-1061){\makebox(0,0)[lb]{\smash{{\SetFigFont{8}{9.6}{\rmdefault}{\mddefault}{\updefault}{\color[rgb]{0,0,0}$\Lambda$}%
}}}}
\put(2841,-2131){\makebox(0,0)[lb]{\smash{{\SetFigFont{8}{9.6}{\rmdefault}{\mddefault}{\updefault}{\color[rgb]{0,0,0}$\Lambda'$}%
}}}}
\end{picture}%

%% file: l14_aggreg.pdftex_t
\begin{picture}(0,0)%
\includegraphics{l14_aggreg.pdf}%
\end{picture}%
%
%
\setlength{\unitlength}{1579sp}%
\begingroup\makeatletter\ifx\SetFigFont\undefined%
\gdef\SetFigFont#1#2#3#4#5{%
  \reset@font\fontsize{#1}{#2pt}%
  \fontfamily{#3}\fontseries{#4}\fontshape{#5}%
  \selectfont}%
\fi\endgroup%
\begin{picture}(9522,2453)(193,-2041)
\end{picture}%

%% file: l16_HE1.pdftex_t
\begin{picture}(0,0)%
\includegraphics{l16_HE1.pdf}%
\end{picture}%
%
%
\setlength{\unitlength}{1973sp}%
\begingroup\makeatletter\ifx\SetFigFont\undefined%
\gdef\SetFigFont#1#2#3#4#5{%
  \reset@font\fontsize{#1}{#2pt}%
  \fontfamily{#3}\fontseries{#4}\fontshape{#5}%
  \selectfont}%
\fi\endgroup%
\begin{picture}(1053,162)(860,-1132)
\end{picture}%

%% file: l16_HE2.pdftex_t
\begin{picture}(0,0)%
\includegraphics{l16_HE2.pdf}%
\end{picture}%
%
%
\setlength{\unitlength}{1973sp}%
\begingroup\makeatletter\ifx\SetFigFont\undefined%
\gdef\SetFigFont#1#2#3#4#5{%
  \reset@font\fontsize{#1}{#2pt}%
  \fontfamily{#3}\fontseries{#4}\fontshape{#5}%
  \selectfont}%
\fi\endgroup%
\begin{picture}(1053,162)(860,-1132)
\end{picture}%

%% file: l16_HE3.pdftex_t
\begin{picture}(0,0)%
\includegraphics{l16_HE3.pdf}%
\end{picture}%
%
%
\setlength{\unitlength}{1973sp}%
\begingroup\makeatletter\ifx\SetFigFont\undefined%
\gdef\SetFigFont#1#2#3#4#5{%
  \reset@font\fontsize{#1}{#2pt}%
  \fontfamily{#3}\fontseries{#4}\fontshape{#5}%
  \selectfont}%
\fi\endgroup%
\begin{picture}(1053,204)(860,-1143)
\end{picture}%

%% file: l16_HE4.pdftex_t
\begin{picture}(0,0)%
\includegraphics{l16_HE4.pdf}%
\end{picture}%
%
%
\setlength{\unitlength}{1973sp}%
\begingroup\makeatletter\ifx\SetFigFont\undefined%
\gdef\SetFigFont#1#2#3#4#5{%
  \reset@font\fontsize{#1}{#2pt}%
  \fontfamily{#3}\fontseries{#4}\fontshape{#5}%
  \selectfont}%
\fi\endgroup%
\begin{picture}(1053,204)(860,-1147)
\end{picture}%

%% file: l16_E1.pdftex_t
\begin{picture}(0,0)%
\includegraphics{l16_E1.pdf}%
\end{picture}%
%
%
\setlength{\unitlength}{1973sp}%
\begingroup\makeatletter\ifx\SetFigFont\undefined%
\gdef\SetFigFont#1#2#3#4#5{%
  \reset@font\fontsize{#1}{#2pt}%
  \fontfamily{#3}\fontseries{#4}\fontshape{#5}%
  \selectfont}%
\fi\endgroup%
\begin{picture}(1614,194)(1324,-1618)
\end{picture}%

%% file: l16_E2.pdftex_t
\begin{picture}(0,0)%
\includegraphics{l16_E2.pdf}%
\end{picture}%
%
%
\setlength{\unitlength}{1973sp}%
\begingroup\makeatletter\ifx\SetFigFont\undefined%
\gdef\SetFigFont#1#2#3#4#5{%
  \reset@font\fontsize{#1}{#2pt}%
  \fontfamily{#3}\fontseries{#4}\fontshape{#5}%
  \selectfont}%
\fi\endgroup%
\begin{picture}(1614,194)(1324,-1618)
\end{picture}%

%% file: l16_E3.pdftex_t
\begin{picture}(0,0)%
\includegraphics{l16_E3.pdf}%
\end{picture}%
%
%
\setlength{\unitlength}{1973sp}%
\begingroup\makeatletter\ifx\SetFigFont\undefined%
\gdef\SetFigFont#1#2#3#4#5{%
  \reset@font\fontsize{#1}{#2pt}%
  \fontfamily{#3}\fontseries{#4}\fontshape{#5}%
  \selectfont}%
\fi\endgroup%
\begin{picture}(1614,209)(1324,-1618)
\end{picture}%

%% file: l16_V1.pdftex_t
\begin{picture}(0,0)%
\includegraphics{l16_V1.pdf}%
\end{picture}%
%
%
\setlength{\unitlength}{1973sp}%
\begingroup\makeatletter\ifx\SetFigFont\undefined%
\gdef\SetFigFont#1#2#3#4#5{%
  \reset@font\fontsize{#1}{#2pt}%
  \fontfamily{#3}\fontseries{#4}\fontshape{#5}%
  \selectfont}%
\fi\endgroup%
\begin{picture}(824,607)(1519,-2316)
\end{picture}%

%% file: l16_V2.pdftex_t
\begin{picture}(0,0)%
\includegraphics{l16_V2.pdf}%
\end{picture}%
%
%
\setlength{\unitlength}{1973sp}%
\begingroup\makeatletter\ifx\SetFigFont\undefined%
\gdef\SetFigFont#1#2#3#4#5{%
  \reset@font\fontsize{#1}{#2pt}%
  \fontfamily{#3}\fontseries{#4}\fontshape{#5}%
  \selectfont}%
\fi\endgroup%
\begin{picture}(594,651)(1519,-2360)
\end{picture}%

%% file: l16_V3.pdftex_t
\begin{picture}(0,0)%
\includegraphics{l16_V3.pdf}%
\end{picture}%
%
%
\setlength{\unitlength}{1973sp}%
\begingroup\makeatletter\ifx\SetFigFont\undefined%
\gdef\SetFigFont#1#2#3#4#5{%
  \reset@font\fontsize{#1}{#2pt}%
  \fontfamily{#3}\fontseries{#4}\fontshape{#5}%
  \selectfont}%
\fi\endgroup%
\begin{picture}(824,1074)(1519,-2783)
\end{picture}%

%% file: l16_V4.pdftex_t
\begin{picture}(0,0)%
\includegraphics{l16_V4.pdf}%
\end{picture}%
%
%
\setlength{\unitlength}{1973sp}%
\begingroup\makeatletter\ifx\SetFigFont\undefined%
\gdef\SetFigFont#1#2#3#4#5{%
  \reset@font\fontsize{#1}{#2pt}%
  \fontfamily{#3}\fontseries{#4}\fontshape{#5}%
  \selectfont}%
\fi\endgroup%
\begin{picture}(1254,964)(1219,-2673)
\end{picture}%

%% file: l16_graph.pdftex_t
\begin{picture}(0,0)%
\includegraphics{l16_graph.pdf}%
\end{picture}%
%
%
\setlength{\unitlength}{1973sp}%
\begingroup\makeatletter\ifx\SetFigFont\undefined%
\gdef\SetFigFont#1#2#3#4#5{%
  \reset@font\fontsize{#1}{#2pt}%
  \fontfamily{#3}\fontseries{#4}\fontshape{#5}%
  \selectfont}%
\fi\endgroup%
\begin{picture}(3794,2981)(1039,-2688)
\put(3411,-2591){\makebox(0,0)[lb]{\smash{{\SetFigFont{8}{9.6}{\rmdefault}{\mddefault}{\updefault}{\color[rgb]{0,0,0}$\rho$-vertex}%
}}}}
\end{picture}%

%% file: l17_HE1.pdftex_t
\begin{picture}(0,0)%
\includegraphics{l17_HE1.pdf}%
\end{picture}%
%
%
\setlength{\unitlength}{1973sp}%
\begingroup\makeatletter\ifx\SetFigFont\undefined%
\gdef\SetFigFont#1#2#3#4#5{%
  \reset@font\fontsize{#1}{#2pt}%
  \fontfamily{#3}\fontseries{#4}\fontshape{#5}%
  \selectfont}%
\fi\endgroup%
\begin{picture}(997,457)(616,-1675)
\put(631,-1581){\makebox(0,0)[lb]{\smash{{\SetFigFont{8}{9.6}{\rmdefault}{\mddefault}{\updefault}{\color[rgb]{0,0,0}$x$}%
}}}}
\put(1131,-1571){\makebox(0,0)[lb]{\smash{{\SetFigFont{8}{9.6}{\rmdefault}{\mddefault}{\updefault}{\color[rgb]{0,0,0}$a,\mu$}%
}}}}
\end{picture}%

%% file: l17_HE2.pdftex_t
\begin{picture}(0,0)%
\includegraphics{l17_HE2.pdf}%
\end{picture}%
%
%
\setlength{\unitlength}{1973sp}%
\begingroup\makeatletter\ifx\SetFigFont\undefined%
\gdef\SetFigFont#1#2#3#4#5{%
  \reset@font\fontsize{#1}{#2pt}%
  \fontfamily{#3}\fontseries{#4}\fontshape{#5}%
  \selectfont}%
\fi\endgroup%
\begin{picture}(1057,496)(1346,-2185)
\put(1361,-2081){\makebox(0,0)[lb]{\smash{{\SetFigFont{8}{9.6}{\rmdefault}{\mddefault}{\updefault}{\color[rgb]{0,0,0}$x$}%
}}}}
\put(1961,-2091){\makebox(0,0)[lb]{\smash{{\SetFigFont{8}{9.6}{\rmdefault}{\mddefault}{\updefault}{\color[rgb]{0,0,0}$a$}%
}}}}
\end{picture}%

%% file: l17_HE3.pdftex_t
\begin{picture}(0,0)%
\includegraphics{l17_HE3.pdf}%
\end{picture}%
%
%
\setlength{\unitlength}{1973sp}%
\begingroup\makeatletter\ifx\SetFigFont\undefined%
\gdef\SetFigFont#1#2#3#4#5{%
  \reset@font\fontsize{#1}{#2pt}%
  \fontfamily{#3}\fontseries{#4}\fontshape{#5}%
  \selectfont}%
\fi\endgroup%
\begin{picture}(1057,498)(1346,-2185)
\put(1361,-2081){\makebox(0,0)[lb]{\smash{{\SetFigFont{8}{9.6}{\rmdefault}{\mddefault}{\updefault}{\color[rgb]{0,0,0}$x$}%
}}}}
\put(1961,-2091){\makebox(0,0)[lb]{\smash{{\SetFigFont{8}{9.6}{\rmdefault}{\mddefault}{\updefault}{\color[rgb]{0,0,0}$a$}%
}}}}
\end{picture}%

%% file: l17_V1_1.pdftex_t
\begin{picture}(0,0)%
\includegraphics{l17_V1_1.pdf}%
\end{picture}%
%
%
\setlength{\unitlength}{1973sp}%
\begingroup\makeatletter\ifx\SetFigFont\undefined%
\gdef\SetFigFont#1#2#3#4#5{%
  \reset@font\fontsize{#1}{#2pt}%
  \fontfamily{#3}\fontseries{#4}\fontshape{#5}%
  \selectfont}%
\fi\endgroup%
\begin{picture}(221,1003)(1830,-1862)
\end{picture}%

%% file: l17_V1_2.pdftex_t
\begin{picture}(0,0)%
\includegraphics{l17_V1_2.pdf}%
\end{picture}%
%
%
\setlength{\unitlength}{1973sp}%
\begingroup\makeatletter\ifx\SetFigFont\undefined%
\gdef\SetFigFont#1#2#3#4#5{%
  \reset@font\fontsize{#1}{#2pt}%
  \fontfamily{#3}\fontseries{#4}\fontshape{#5}%
  \selectfont}%
\fi\endgroup%
\begin{picture}(836,795)(1357,-2243)
\end{picture}%

%% file: l17_V3_1.pdftex_t
\begin{picture}(0,0)%
\includegraphics{l17_V3_1.pdf}%
\end{picture}%
%
%
\setlength{\unitlength}{1973sp}%
\begingroup\makeatletter\ifx\SetFigFont\undefined%
\gdef\SetFigFont#1#2#3#4#5{%
  \reset@font\fontsize{#1}{#2pt}%
  \fontfamily{#3}\fontseries{#4}\fontshape{#5}%
  \selectfont}%
\fi\endgroup%
\begin{picture}(1026,185)(1419,-2104)
\end{picture}%

%% file: l17_HE4.pdftex_t
\begin{picture}(0,0)%
\includegraphics{l17_HE4.pdf}%
\end{picture}%
%
%
\setlength{\unitlength}{1973sp}%
\begingroup\makeatletter\ifx\SetFigFont\undefined%
\gdef\SetFigFont#1#2#3#4#5{%
  \reset@font\fontsize{#1}{#2pt}%
  \fontfamily{#3}\fontseries{#4}\fontshape{#5}%
  \selectfont}%
\fi\endgroup%
\begin{picture}(1067,446)(1566,-2135)
\put(1581,-2041){\makebox(0,0)[lb]{\smash{{\SetFigFont{8}{9.6}{\rmdefault}{\mddefault}{\updefault}{\color[rgb]{0,0,0}$x$}%
}}}}
\put(2231,-2021){\makebox(0,0)[lb]{\smash{{\SetFigFont{8}{9.6}{\rmdefault}{\mddefault}{\updefault}{\color[rgb]{0,0,0}$i,\alpha$}%
}}}}
\end{picture}%

%% file: l17_HE5.pdftex_t
\begin{picture}(0,0)%
\includegraphics{l17_HE5.pdf}%
\end{picture}%
%
%
\setlength{\unitlength}{1973sp}%
\begingroup\makeatletter\ifx\SetFigFont\undefined%
\gdef\SetFigFont#1#2#3#4#5{%
  \reset@font\fontsize{#1}{#2pt}%
  \fontfamily{#3}\fontseries{#4}\fontshape{#5}%
  \selectfont}%
\fi\endgroup%
\begin{picture}(1067,435)(1566,-2135)
\put(1581,-2041){\makebox(0,0)[lb]{\smash{{\SetFigFont{8}{9.6}{\rmdefault}{\mddefault}{\updefault}{\color[rgb]{0,0,0}$x$}%
}}}}
\put(2231,-2021){\makebox(0,0)[lb]{\smash{{\SetFigFont{8}{9.6}{\rmdefault}{\mddefault}{\updefault}{\color[rgb]{0,0,0}$i,\alpha$}%
}}}}
\end{picture}%

%% file: l26_tree_generic.pdftex_t
\begin{picture}(0,0)%
\includegraphics{l26_tree_generic.pdf}%
\end{picture}%
%
%
\setlength{\unitlength}{1184sp}%
\begingroup\makeatletter\ifx\SetFigFont\undefined%
\gdef\SetFigFont#1#2#3#4#5{%
  \reset@font\fontsize{#1}{#2pt}%
  \fontfamily{#3}\fontseries{#4}\fontshape{#5}%
  \selectfont}%
\fi\endgroup%
\begin{picture}(1414,1176)(992,-1493)
\end{picture}%

%% file: l26_tree.pdftex_t
\begin{picture}(0,0)%
\includegraphics{l26_tree.pdf}%
\end{picture}%
%
%
\setlength{\unitlength}{1973sp}%
\begingroup\makeatletter\ifx\SetFigFont\undefined%
\gdef\SetFigFont#1#2#3#4#5{%
  \reset@font\fontsize{#1}{#2pt}%
  \fontfamily{#3}\fontseries{#4}\fontshape{#5}%
  \selectfont}%
\fi\endgroup%
\begin{picture}(2298,1558)(1640,-2101)
\end{picture}%

%% file: l26_loop_generic.pdftex_t
\begin{picture}(0,0)%
\includegraphics{l26_loop_generic.pdf}%
\end{picture}%
%
%
\setlength{\unitlength}{1184sp}%
\begingroup\makeatletter\ifx\SetFigFont\undefined%
\gdef\SetFigFont#1#2#3#4#5{%
  \reset@font\fontsize{#1}{#2pt}%
  \fontfamily{#3}\fontseries{#4}\fontshape{#5}%
  \selectfont}%
\fi\endgroup%
\begin{picture}(2459,1889)(1749,-2933)
\end{picture}%

%% file: l26_loop.pdftex_t
\begin{picture}(0,0)%
\includegraphics{l26_loop.pdf}%
\end{picture}%
%
%
\setlength{\unitlength}{1973sp}%
\begingroup\makeatletter\ifx\SetFigFont\undefined%
\gdef\SetFigFont#1#2#3#4#5{%
  \reset@font\fontsize{#1}{#2pt}%
  \fontfamily{#3}\fontseries{#4}\fontshape{#5}%
  \selectfont}%
\fi\endgroup%
\begin{picture}(3156,1606)(1052,-2803)
\end{picture}%

%% file: l26_corolla.pdftex_t
\begin{picture}(0,0)%
\includegraphics{l26_corolla.pdf}%
\end{picture}%
%
%
\setlength{\unitlength}{1973sp}%
\begingroup\makeatletter\ifx\SetFigFont\undefined%
\gdef\SetFigFont#1#2#3#4#5{%
  \reset@font\fontsize{#1}{#2pt}%
  \fontfamily{#3}\fontseries{#4}\fontshape{#5}%
  \selectfont}%
\fi\endgroup%
\begin{picture}(937,704)(1629,-1741)
\end{picture}%

%% file: l26_uL_inf_rel_1.pdftex_t
\begin{picture}(0,0)%
\includegraphics{l26_uL_inf_rel_1.pdf}%
\end{picture}%
%
%
\setlength{\unitlength}{1973sp}%
\begingroup\makeatletter\ifx\SetFigFont\undefined%
\gdef\SetFigFont#1#2#3#4#5{%
  \reset@font\fontsize{#1}{#2pt}%
  \fontfamily{#3}\fontseries{#4}\fontshape{#5}%
  \selectfont}%
\fi\endgroup%
\begin{picture}(2271,1784)(1562,-2813)
\put(2251,-2514){\makebox(0,0)[lb]{\smash{{\SetFigFont{8}{9.6}{\rmdefault}{\mddefault}{\updefault}{\color[rgb]{0,0,0}$l_s$}%
}}}}
\put(3211,-1606){\makebox(0,0)[lb]{\smash{{\SetFigFont{8}{9.6}{\rmdefault}{\mddefault}{\updefault}{\color[rgb]{0,0,0}$l_{r+1}$}%
}}}}
\end{picture}%

%% file: l26_uL_inf_rel_2.pdftex_t
\begin{picture}(0,0)%
\includegraphics{l26_uL_inf_rel_2.pdf}%
\end{picture}%
%
%
\setlength{\unitlength}{1973sp}%
\begingroup\makeatletter\ifx\SetFigFont\undefined%
\gdef\SetFigFont#1#2#3#4#5{%
  \reset@font\fontsize{#1}{#2pt}%
  \fontfamily{#3}\fontseries{#4}\fontshape{#5}%
  \selectfont}%
\fi\endgroup%
\begin{picture}(905,967)(2769,-2566)
\put(3579,-2491){\makebox(0,0)[lb]{\smash{{\SetFigFont{6}{7.2}{\rmdefault}{\mddefault}{\updefault}{\color[rgb]{0,0,0}$\mr{Str}$}%
}}}}
\put(3301,-1831){\makebox(0,0)[lb]{\smash{{\SetFigFont{8}{9.6}{\rmdefault}{\mddefault}{\updefault}{\color[rgb]{0,0,0}$l_{n+1}$}%
}}}}
\end{picture}%

%% file: l26_uL_inf_rel_3.pdftex_t
\begin{picture}(0,0)%
\includegraphics{l26_uL_inf_rel_3.pdf}%
\end{picture}%
%
%
\setlength{\unitlength}{1973sp}%
\begingroup\makeatletter\ifx\SetFigFont\undefined%
\gdef\SetFigFont#1#2#3#4#5{%
  \reset@font\fontsize{#1}{#2pt}%
  \fontfamily{#3}\fontseries{#4}\fontshape{#5}%
  \selectfont}%
\fi\endgroup%
\begin{picture}(1750,1784)(1562,-2813)
\put(2251,-2514){\makebox(0,0)[lb]{\smash{{\SetFigFont{8}{9.6}{\rmdefault}{\mddefault}{\updefault}{\color[rgb]{0,0,0}$l_s$}%
}}}}
\put(3211,-1614){\makebox(0,0)[lb]{\smash{{\SetFigFont{8}{9.6}{\rmdefault}{\mddefault}{\updefault}{\color[rgb]{0,0,0}$q_{r+1}$}%
}}}}
\end{picture}%

%% file: l27_branch.pdftex_t
\begin{picture}(0,0)%
\includegraphics{l27_branch.pdf}%
\end{picture}%
%
%
\setlength{\unitlength}{1973sp}%
\begingroup\makeatletter\ifx\SetFigFont\undefined%
\gdef\SetFigFont#1#2#3#4#5{%
  \reset@font\fontsize{#1}{#2pt}%
  \fontfamily{#3}\fontseries{#4}\fontshape{#5}%
  \selectfont}%
\fi\endgroup%
\begin{picture}(3464,1814)(1169,-2833)
\end{picture}%

%% file: l27_wheel.pdftex_t
\begin{picture}(0,0)%
\includegraphics{l27_wheel.pdf}%
\end{picture}%
%
%
\setlength{\unitlength}{1973sp}%
\begingroup\makeatletter\ifx\SetFigFont\undefined%
\gdef\SetFigFont#1#2#3#4#5{%
  \reset@font\fontsize{#1}{#2pt}%
  \fontfamily{#3}\fontseries{#4}\fontshape{#5}%
  \selectfont}%
\fi\endgroup%
\begin{picture}(1834,1844)(2209,-3183)
\end{picture}%

%% file: l27_loop_dec.pdftex_t
\begin{picture}(0,0)%
\includegraphics{l27_loop_dec.pdf}%
\end{picture}%
%
%
\setlength{\unitlength}{1973sp}%
\begingroup\makeatletter\ifx\SetFigFont\undefined%
\gdef\SetFigFont#1#2#3#4#5{%
  \reset@font\fontsize{#1}{#2pt}%
  \fontfamily{#3}\fontseries{#4}\fontshape{#5}%
  \selectfont}%
\fi\endgroup%
\begin{picture}(2677,1484)(1936,-3233)
\put(4441,-2721){\makebox(0,0)[lb]{\smash{{\SetFigFont{7}{8.4}{\rmdefault}{\mddefault}{\updefault}{\color[rgb]{0,0,0}$e_2$}%
}}}}
\put(1951,-2761){\makebox(0,0)[lb]{\smash{{\SetFigFont{7}{8.4}{\rmdefault}{\mddefault}{\updefault}{\color[rgb]{0,0,0}$e_1$}%
}}}}
\end{picture}%

%% file: l27_collapse.pdftex_t
\begin{picture}(0,0)%
\includegraphics{l27_collapse.pdf}%
\end{picture}%
%
%
\setlength{\unitlength}{1973sp}%
\begingroup\makeatletter\ifx\SetFigFont\undefined%
\gdef\SetFigFont#1#2#3#4#5{%
  \reset@font\fontsize{#1}{#2pt}%
  \fontfamily{#3}\fontseries{#4}\fontshape{#5}%
  \selectfont}%
\fi\endgroup%
\begin{picture}(5415,2559)(1948,-3438)
\put(6234,-1164){\makebox(0,0)[lb]{\smash{{\SetFigFont{8}{9.6}{\rmdefault}{\mddefault}{\updefault}{\color[rgb]{0,0,0}$X$}%
}}}}
\put(2656,-1111){\makebox(0,0)[lb]{\smash{{\SetFigFont{8}{9.6}{\rmdefault}{\mddefault}{\updefault}{\color[rgb]{0,0,0}$Y$}%
}}}}
\end{picture}%

%% file: l28_FMAS.pdftex_t
\begin{picture}(0,0)%
\includegraphics{l28_FMAS.pdf}%
\end{picture}%
%
%
\setlength{\unitlength}{1973sp}%
\begingroup\makeatletter\ifx\SetFigFont\undefined%
\gdef\SetFigFont#1#2#3#4#5{%
  \reset@font\fontsize{#1}{#2pt}%
  \fontfamily{#3}\fontseries{#4}\fontshape{#5}%
  \selectfont}%
\fi\endgroup%
\begin{picture}(6672,2815)(993,-4178)
\put(1031,-1591){\makebox(0,0)[lb]{\smash{{\SetFigFont{8}{9.6}{\rmdefault}{\mddefault}{\updefault}{\color[rgb]{0,0,0}$M$}%
}}}}
\end{picture}%

%% file: l28_CS_graphs.pdftex_t
\begin{picture}(0,0)%
\includegraphics{l28_CS_graphs.pdf}%
\end{picture}%
%
%
\setlength{\unitlength}{1973sp}%
\begingroup\makeatletter\ifx\SetFigFont\undefined%
\gdef\SetFigFont#1#2#3#4#5{%
  \reset@font\fontsize{#1}{#2pt}%
  \fontfamily{#3}\fontseries{#4}\fontshape{#5}%
  \selectfont}%
\fi\endgroup%
\begin{picture}(5164,1316)(926,-2947)
\end{picture}%

%% file: l28_IHX.pdftex_t
\begin{picture}(0,0)%
\includegraphics{l28_IHX.pdf}%
\end{picture}%
%
%
\setlength{\unitlength}{1973sp}%
\begingroup\makeatletter\ifx\SetFigFont\undefined%
\gdef\SetFigFont#1#2#3#4#5{%
  \reset@font\fontsize{#1}{#2pt}%
  \fontfamily{#3}\fontseries{#4}\fontshape{#5}%
  \selectfont}%
\fi\endgroup%
\begin{picture}(6114,1595)(2246,-3306)
\end{picture}%